%% file: Sujay_paper.tex
\documentclass[pra,aps,floatfix,amsmath,superscriptaddress,twocolumn,longbibliography,nofootinbib]{revtex4-2}
\usepackage{amssymb,enumerate}

\usepackage{amssymb,enumerate}
\usepackage{graphicx}
\usepackage{graphics}
\usepackage{amsmath}
\usepackage{amsthm,bbm}
\usepackage{color}
\usepackage{dsfont}
\usepackage{hyperref}

\usepackage[capitalise,nameinlink]{cleveref}

\usepackage{qcircuit}

\usepackage{graphicx}

\usepackage{xcolor}
\hypersetup{
    colorlinks,
    linkcolor={red!50!black},
    citecolor={blue!50!black},
    urlcolor={blue!80!black}
}

\usepackage{xfrac}

\usepackage{amsmath}
\usepackage{tikz-cd}
\usepackage{empheq}

\usepackage{enumitem}

\usepackage{empheq}
\newcommand*\widefbox[1]{\fbox{\hspace{2em}#1\hspace{2em}}}

\usepackage{cleveref}
\usepackage{nameref}

\usepackage{graphicx}

\usepackage{hyperref}
\usepackage[capitalise]{cleveref}
\usepackage{youngtab}

\usepackage{tikz}
\usepackage{mathtools}

\newcommand{\mD}{\mathcal{D}}
\newcommand{\mE}{\mathcal{E}}
\newcommand{\mF}{\mathcal{F}}

\newcommand{\mH}{\mathcal{H}}
\newcommand{\mI}{\mathcal{I}}
\newcommand{\mJ}{\mathcal{J}}

\newcommand{\mM}{\mathcal{M}}

\newcommand{\mT}{\mathcal{T}}
\newcommand{\mU}{\mathcal{U}}
\newcommand{\mV}{\mathcal{V}}

\newcommand{\Cbb}{\mathbb{C}}

\newcommand{\Ebb}{\mathbb{E}}

\newcommand{\Ibb}{\mathbb{I}}

\newcommand{\Nbb}{\mathbb{N}}

\newcommand{\Pbb}{\mathbb{P}}

\newcommand{\Rbb}{\mathbb{R}}

\newcommand{\Zbb}{\mathbb{Z}}

\newcommand{\Tr}{\operatorname{Tr}}

\newcommand{\<}{\langle}
\renewcommand{\>}{\rangle}

\newtheorem{thm}{Theorem}

\usepackage{lipsum}
\usepackage{lmodern}
\usepackage{tcolorbox}

\usepackage{amsfonts}
\usepackage{graphicx,graphics,epsfig,times,bm,bbm,amssymb,amsmath,amsfonts,mathrsfs}
\usepackage[normalem]{ulem}
\usepackage{setspace}
\usepackage{subfigure}
\usepackage{dsfont}
\usepackage{braket}
\usepackage{upgreek }
\usepackage{tikz}
\usepackage{natbib}
\usepackage{chngcntr}

\newtheorem{theorem}{Theorem}

\newtheorem{corollary}[theorem]{Corollary}

\newtheorem{definition}[theorem]{Definition}

\newtheorem{lemma}[theorem]{Lemma}

\makeatletter 
\renewcommand\onecolumngrid{%
  \do@columngrid{one}{\@ne}%
  \def\set@footnotewidth{\onecolumngrid}%
  \def\footnoterule{\kern-6pt\hrule width 1.5in\kern6pt}%
}
\makeatother

\newcommand{\bes} {\begin{subequations}}
\newcommand{\ees} {\end{subequations}}
\newcommand{\bea} {\begin{eqnarray}}
\newcommand{\eea} {\end{eqnarray}}
\newcommand{\be} {\begin{equation}}
\newcommand{\ee} {\end{equation}}

\def\>{\rangle}
\def\<{\langle}
\def\Tr{\textrm{Tr}}

\newcommand{\abs}[1]{\lvert #1 \rvert}

\newcommand{\ignore}[1]{}

\crefformat{step}{#2Step #1#3}

\begin{document}

\title{Optimal Distillation of Qubit Clocks}

\author{Sujay Kazi}
\affiliation{Department of Electrical and Computer Engineering, Duke University, Durham, NC 27708, USA}
\affiliation{Duke Quantum Center, Durham, NC 27708, USA}

\author{Iman Marvian}
\affiliation{Department of Electrical and Computer Engineering, Duke University, Durham, NC 27708, USA}
\affiliation{Duke Quantum Center, Durham, NC 27708, USA}
\affiliation{Department of Physics, Duke University, Durham, NC 27708, USA}

\begin{abstract}

We study coherence distillation under time-translation-invariant operations: given many copies of a quantum state containing coherence in the energy eigenbasis, the aim is to produce a purer coherent state while respecting the time-translation symmetry. This symmetry ensures that the output remains synchronized with the input and that the process 
can be realized by energy-conserving unitaries coupling the system to a reservoir initially in an energy eigenstate, thereby modeling thermal operations supplemented  by a work reservoir or battery. For qubit systems, we determine the optimal asymptotic fidelity and show that it is governed by the purity of coherence, a measure of asymmetry derived from the right logarithmic derivative (RLD) Fisher information. In particular, we find that the lowest achievable infidelity (one minus fidelity) scales as $1/N$ times the reciprocal of the purity of coherence of each input qubit, where $N$ is the number of copies, giving this quantity a clear operational meaning. We additionally study many other interesting aspects of the coherence distillation problem for qubits, including computing higher-order corrections to the lowest achievable infidelity up to $O(1/N^3)$, and expressing the optimal channel as a boundary value problem that can be solved numerically.

\end{abstract}

\maketitle

\section{Introduction}
\label{sec:intro}



The manipulation of energetic coherence plays a central role at the intersection of quantum thermodynamics and quantum metrology, where it finds its most direct operational manifestation in quantum clocks~\cite{helstrom1969quantum, holevo2011probabilistic, Braunstein1994,  GiovannettiMetrology,giovannetti2001quantum,giovannetti2011advances,  Bartlett2007,lostaglio2015quantum,lostaglio2015description, marvian2022operational,Marvian2020,streltsov2017colloquium,chiribella2013quantum,marvian2016quantify}. Any non-stationary state of a closed quantum system can serve as a clock in the sense that it keeps track of time, much like a pendulum that swings back and forth. Such a quantum state is said to be \emph{coherent} with respect to the eigenspaces of the Hamiltonian.

When several such clocks, each subject to noise and decoherence, are employed together, one may hope to recover a sharper notion of time by coherently processing their joint state. The problem of \emph{coherence distillation} asks to what extent one can convert many low-quality quantum clocks into a single, high-quality quantum clock \cite{Marvian2020}. In particular, consider an ensemble of quantum clocks, each synchronized to a common reference but affected by noise such as dephasing or depolarization. The goal is to transform this collection of noisy clocks into a single quantum state that remains synchronized with the reference while being as close as possible to a pure state. Because direct access to the reference clock is not possible, the protocol must be covariant under time translations: it should act independently of the clocks' specific phase in their evolution, producing an output at the corresponding point along its own time trajectory.

In this work, we present a complete analysis of the optimal distillation of quantum qubit clocks. Specifically, we determine the optimal asymptotic protocol for converting $N\gg 1$ noisy qubit clocks into a single, purer qubit clock. Our findings uncover the fundamental scaling of the achievable precision and resolve a conjecture posed in Ref. \cite{Marvian2020}. In particular, we show that the minimum achievable error, quantified by the infidelity (i.e., one minus the fidelity with the ideal pure state), is determined by the \emph{purity of coherence} of the input state~$\rho$, namely, the right logarithmic derivative (RLD) Fisher information metric \cite{Petz1996,Petz2011, hayashi2017quantum} for the family of time-evolved states $e^{-itH}\rho e^{+itH}$ with respect to the time parameter $t$. Hence, our work provides an operational interpretation of this lesser-known variant of the Fisher information and complements a recent result obtained for continuous-variable coherent states~\cite{Yadavalli2024}.

A particularly intriguing feature of the purity of coherence is that, for pure states, it is either zero or infinite. Our results offer a natural explanation of this behavior: pure states that possess energetic coherence cannot be reached from mixed states through time-translation–invariant operations. In contrast, the symmetric logarithmic derivative (SLD) Fisher information, commonly associated with estimation precision via the quantum Cram\'{e}r-Rao bound \cite{Braunstein1994, helstrom1969quantum, BarndorffNielsen2000, holevo2011probabilistic}, does not capture this operational limitation. (Interestingly, however, the SLD Fisher information also admits an operational interpretation as the \emph{coherence cost}, that is, the minimal rate of consumption of pure-state quantum clocks required to generate arbitrary mixed or pure clock states~\cite{marvian2022operational}.)

Finally, we analyze higher-order features of the qubit coherence-distillation problem, deriving corrections to the minimum infidelity up to order $O(1/N^3)$ and expressing the optimal channel as a boundary value problem amenable to numerical solution.

\section{Preliminaries}
\label{sec:preliminaries}

Coherence distillation can be understood as a problem in the \emph{resource theory of asymmetry}, which attempts to quantify the amount of asymmetry carried by a quantum state in operationally meaningful ways \cite{gour2008resource, marvian2012symmetry, Bartlett2007}. In the resource theory of asymmetry, the free states are symmetric states (those that commute with all elements of a specific representation of a symmetry group on the Hilbert space), and the free operations are covariant channels, which are defined as follows:

\begin{definition}[Group-covariant channel]
\label{def:group-covariant-channel}
Given a group $G$ and unitary representations $R_{\text{in}}$ and $R_{\text{out}}$ on Hilbert spaces $\mH_{\text{in}}$ and $\mH_{\text{out}}$, respectively, a quantum channel $\mE: \mD\left(\mH_{\text{in}}\right)\rightarrow\mD\left(\mH_{\text{out}}\right)$ is said to be $\left(G, R_{\text{in}}, R_{\text{out}}\right)$-covariant if, for all density operators $\rho\in\mD\left(\mH_{\text{in}}\right)$ and all group elements $g\in G$,
\begin{equation}
    \mE\left(R_{\text{in}}(g)\rho R_{\text{in}}(g)^\dagger\right) = R_{\text{out}}(g)\mE(\rho)R_{\text{out}}(g)^\dagger.
\end{equation}
When the representations $R_{\text{in}}$ and $R_{\text{out}}$ are understood from context, we may drop them and simply call such a channel \emph{$G$-covariant} for convenience.
\end{definition}

We now define the special example of group covariance that will be relevant for coherence distillation:

\begin{definition}[Time-translation invariant (TI) channel]
\label{def:TI-channel}
Given Hamiltonians $H_{\text{in}}$ and $H_{\text{out}}$ on the input and output Hilbert spaces, respectively, a \emph{time-translation invariant (TI) channel} is a quantum channel $\mE$ such that, for all times $t\in\mathbb{R}$,
\begin{equation}
    \mE\left(e^{-itH_{\text{in}}}\rho e^{+itH_{\text{in}}}\right) = e^{-itH_{\text{out}}}\mE(\rho)e^{+itH_{\text{out}}}.
\end{equation}
\end{definition}

The condition for time-translation invariance is a special case of the condition for group covariance. In particular, for a Hamiltonian $H$, the group $\{e^{-itH} \,|\, t\in\mathbb{R}\}$ is isomorphic to $U(1)$ if the Hamiltonian is periodic (all of its spectral gaps are rational multiples of each other) and $\mathbb{R}$ if the Hamiltonian is aperiodic (some spectral gap is an irrational multiple of another spectral gap). Hence a TI channel is also a $G$-covariant channel, where $G = U(1)$ or $G = \mathbb{R}$ depending on whether $H_{\text{in}}$ is periodic.

It is worth noting that the set of TI operations also naturally arises in the context of quantum thermodynamics. This connection follows from the covariant Stinespring dilation theorem \cite{marvian2012symmetry, keyl1999optimal}. According to this result, any TI operation $\mathcal{E}$ on a system $S$ with Hamiltonian $H_S$ can be realized by coupling it to an ancillary system (or \emph{battery}) with Hamiltonian $H_B$, as
\begin{equation}
 \mathcal{E}(\sigma) = \Tr_B \big[ U \big(\sigma_S \otimes |E\rangle\langle E|_B\big) U^\dagger \big],
\end{equation}
where $U$ is an energy-conserving unitary, i.e., a unitary that commutes with the total Hamiltonian $H_S \otimes \mathbb{I}_B + \mathbb{I}_S \otimes H_B$. Moreover, the initial state of the battery can be chosen as an eigenstate of its Hamiltonian $H_B$, denoted $|E\rangle_B$, which implies that it does not initially contain any coherence.  
Hence, our above distillation protocol can be interpreted as the distillation of coherence using thermal machines that are allowed to consume an arbitrary amount of work (see \cite{Marvian2020} for further discussion).

\subsection{Problem Setup: Coherence Distillation}
\label{subsec:problem-setup}

We now more carefully define the problem we are interested in. We wish to construct a single-shot coherence distillation protocol, meaning that our input state is a tensor product of identical states, and target state is a pure coherent state on a single system. (The adjective ``single-shot'' contrasts this task with, for instance, the task of generating coherent states at some linear rate.) More precisely:

\begin{definition}[Single-shot coherence distillation protocol]
\label{def:single-shot-protocol}
A single-shot coherence distillation protocol, or $N\rightarrow 1$ coherence distillation protocol, is a TI channel $\mE$ where 
 the following conditions are satisfied:
\begin{itemize}
    \item The input Hamiltonian is a sum of identical non-interacting Hamiltonians on the $N$ identical subsystems:
\begin{equation}
    H_{\text{in}} = \sum_{j=1}^{N}H^{(j)} = \sum_{j=1}^{N}\mathbb{I}^{\otimes j-1}\otimes H\otimes\mathbb{I}^{\otimes N-j}.
\end{equation}
    \item The output Hamiltonian is that same Hamiltonian on a single system identical to the $N$ input subsystems: $H_{\text{out}} = H$.
    \item The input state $\rho_{\text{in}} = \rho^{\otimes N}$ is a tensor product of $N$ identical states.
    \item The target state is a pure state $\rho_{\text{out}} = \ket{\psi_{\text{out}}}\bra{\psi_{\text{out}}}$.
\end{itemize}
\end{definition}

Now we restrict our attention to qubits. For convenience, label the input qubits $1$ through $N$, and label the output qubit ``out''. Without loss of generality, we assume that every qubit has an equal $Z$ Hamiltonian (for example, we can have every qubit in a magnetic field pointing in the $z$-direction). Hence the input Hamiltonian is $H_{\text{in}} = \sum_{i=1}^{N}Z_i$ and the output Hamiltonian is $H_{\text{out}} = Z_{\text{out}}$. This motivates the following definition:

\begin{definition}[Single-shot qubit coherence distillation protocol]
\label{def:single-shot-qubit-protocol}
A single-shot coherence distillation protocol (see Definition \ref{def:single-shot-protocol}) will additionally be called a \emph{qubit protocol} if the input and output states are qubits. Without loss of generality, we assume that the input Hamiltonian is $H_{\text{in}} = \sum_{i=1}^{N}Z_i$ and the output Hamiltonian is $H_{\text{out}} = Z_{\text{out}}$. Furthermore, we use $\Theta_{\text{in}}$ and $\Theta_{\text{out}}$ respectively to denote the polar angle (angle from the positive $z$-axis) of the input qubit state and the target qubit state on the Bloch sphere. This means that the time-evolved input state takes the form
\begin{align}
    \rho(t) &= e^{-itZ}\rho e^{+itZ} \\
    &= \frac{\mathbb{I} + \lambda_x\left(X\cos(2t) + Y\sin(2t)\right) + \lambda_zZ}{2},
\end{align}
where $0 < \lambda \le 1$ is called the \emph{purity parameter}, and where we have defined $\lambda_x = \lambda\sin\Theta_{\text{in}}$ and $\lambda_z = \lambda\cos\Theta_{\text{in}}$ for convenience. while the time-evolved target state takes the form
\begin{align}
    \ket{\psi_{\text{out}}(t)} &= e^{-itZ}\ket{\psi_{\text{out}}} \\
    &\propto \cos\frac{\Theta_{\text{out}}}{2}\ket{0} + e^{2it}\sin\frac{\Theta_{\text{out}}}{2}\ket{1}.
\end{align}
\end{definition}

The single-shot qubit coherence distillation problem is illustrated in Figure \ref{fig:coherence-distillation-problem}, which includes the four parameters that characterize any instance of this problem: $N$ (the number of input qubits), $\lambda$ (the purity parameter of the input qubits), $\Theta_{\text{in}}$ (the polar angle of the input qubits), and $\Theta_{\text{out}}$ (the polar angle of the desired pure coherent qubit state).

\begin{figure}
    \centering
    \includegraphics[scale=0.3]{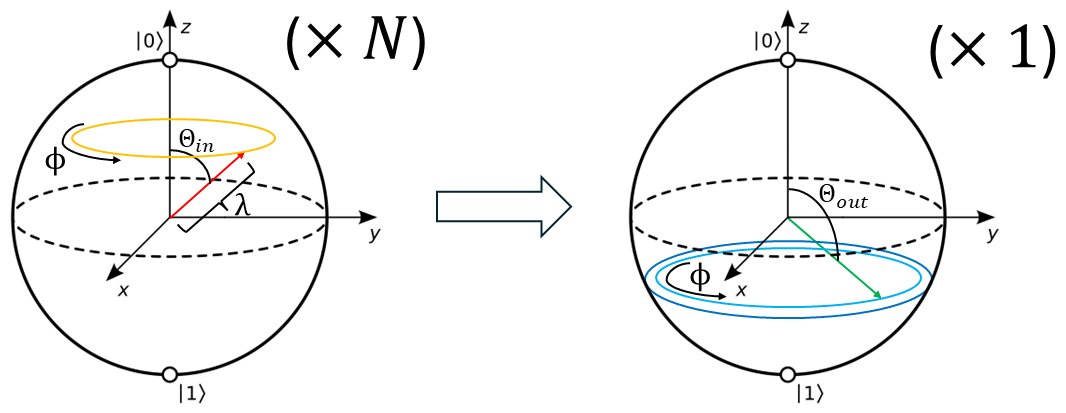}
    \caption{The single-shot qubit coherence distillation problem. Given $N$ copies of a noisy qubit coherent state (\textcolor{orange}{orange}), how can we create a state (\textcolor{cyan}{cyan}) that maximizes the fidelity with a pure coherent qubit state with the same time parameter $\phi$ (\textcolor{blue}{blue}), all without any knowledge of that time parameter?}
    \label{fig:coherence-distillation-problem}
\end{figure}

We now also define the special case where $\Theta_{\text{in}} = \Theta_{\text{out}} = \frac{\pi}{2}$, which will receive special attention throughout this paper due to its (relative) simplicity:

\begin{definition}[Single-shot equatorial qubit coherence distillation protocol]
\label{def:single-shot-qubit-protocol-equatorial}
A single-shot qubit coherence distillation protocol (see Definition \ref{def:single-shot-qubit-protocol}) will additionally be called \emph{equatorial} if the input polar angle $\Theta_{\text{in}}$ and the target polar angle $\Theta_{\text{out}}$ both equal $\frac{\pi}{2}$. This is equivalent to saying that the input qubit states and the target qubit state lie on the equatorial plane ($xy$-plane) of the Bloch sphere, hence the name.
\end{definition}

In fact, this equatorial special case was posed by Marvian \cite{Marvian2020} as the simplest interesting example of a coherence distillation problem (see Figure \ref{fig:coherence-distillation-three-graphs} for more details).

\subsection{Orders of Optimality}
\label{subsec:orders-of-optimality}

Next, we need to articulate how we will evaluate the performance of a single-shot coherence distillation protocol. We wish to maximize the fidelity of our output state with the desired pure coherent state. This motivates the following definition:

\begin{definition}[Infidelity of a single-shot distillation protocol]
\label{def:infidelity-of-protocol}
Suppose a single-shot distillation protocol $\mE$ has input state $\rho_{\text{in}} = \rho^{\otimes N}$, output state $\tilde{\rho} = \mE\left(\rho_{\text{in}}\right)$, and target state $\rho_{\text{out}} = \ket{\psi_{\text{out}}}\bra{\psi_{\text{out}}}$. Then the \emph{infidelity} of the protocol $\mI(\mE)$ is the infidelity (one minus fidelity) between the output state and the target state:
\begin{equation}
    \mI(\mE) = 1 - \text{Fid}\left(\tilde{\rho}, \rho_{\text{out}}\right) = 1 - \bra{\psi_{\text{out}}}\tilde{\rho}\ket{\psi_{\text{out}}}.
\end{equation}
\end{definition}

Of course, we should call a protocol optimal if it minimizes the infidelity. Furthermore, since we are most interested in the large-$N$ regime, it also makes sense to define a weaker, but arguably more useful, notion of optimality, where we only require the infidelity to be minimized up to some polynomial error:

\begin{definition}[Absolute optimality, $p^{\text{th}}$-order optimality]
\label{def:optimality-order}
A distillation protocol $\mE^{\text{opt}}_N$ is said to be \emph{absolutely optimal} if it achieves the minimum possible infidelity for every value of $N$. A distillation protocol $\mE_N$ is said to be \emph{$p^{\text{th}}$-order optimal} if, for an absolutely optimal protocol $\mE^{\text{opt}}_N$,
\begin{equation}
    \mI(\mE_N) - \mI\left(\mE^{\text{opt}}_N\right) = o(N^{-p}).
\end{equation}
\end{definition}

One can consider these definitions to be the distillation analogues of efficiency and $p^{\text{th}}$-order efficiency in classical parameter estimation. We are primarily interested in $1^{\text{st}}$-order optimality, which is already highly nontrivial, and which will provide interesting connections to quantum information geometry.

Furthermore, most of the protocols we will study in this work have their infidelity expressible as a power series in $\frac{1}{N}$, so it makes sense to define the following:

\begin{definition}[Infidelity coefficients, infidelity factor]
\label{def:infidelity-coeffs-infidelity-factor}
Suppose a distillation protocol $\mE$ has infidelity given by the asymptotic series
\begin{equation}
    \mI(\mE) = \sum_{p=1}^{\infty}\frac{a_p}{N^p} = \frac{a_1}{N} + \frac{a_2}{N^2} + \frac{a_3}{N^3} + \cdots.
\end{equation}
We call such a series an \emph{infidelity series}. Then the coefficients $a_1, a_2, a_3, \cdots$ are called the \emph{infidelity coefficients} of $\mE$, and we write $\delta_p(\mE) = a_p$ for each $p\in\mathbb{N}$. In particular, the leading coefficient $\delta_1(\mE)$ is called the \emph{infidelity factor} of $\mE$.
\end{definition}

One can, of course, construct distillation protocols whose infidelity does not take the form of a power series in $N^{-1}$. However, as we show in Appendices \ref{appendix:asymptotic-expansion}, \ref{appendix:1st-order-optimality}, \ref{appendix:2nd-order-optimality}, and \ref{appendix:equatorial-3rd-order-optimality}, the set of protocols we will naturally consider in the large-$N$ regime will have an infidelity series.

Note that, if the optimal protocol $\mE_{\text{opt}}$ and some other protocol $\mE$ both have an infidelity series, then $\mE$ is $p^{\text{th}}$-order optimal if and only if $\delta_q(\mE) = \delta_q(\mE_{\text{opt}})$ for each $1\le q\le p$. In particular, $\mE$ is $1^{\text{st}}$-order optimal if its infidelity factor $\delta_1(\mE)$ achieves the lowest possible value.

At first glance, it seems that one should also include an initial term $a_0$ in the above asymptotic series. The reason we exclude this term is that any ``reasonable'' distillation protocol should achieve $O\left(N^{-1}\right)$ infidelity. For example, there is a measure-and-prepare protocol that achieves $\Theta\left(N^{-1}\right)$ infidelity by performing suitable single-copy measurements on the input states, computing the maximum likelihood estimator (MLE) $\hat{\phi}$ for the time parameter $\phi$, and then preparing the target pure state according to $\hat{\phi}$ \cite{Marvian2020}. In fact, any single-shot distillation problem in the resource theory of asymmetry has a measure-and-prepare protocol that achieves ``reasonable'' performance in this sense. Therefore, there is no need for us to pay attention to single-shot distillation protocols that have $\omega\left(N^{-1}\right)$ infidelity, as they are obviously worse than various well-studied measure-and-prepare protocols.

\section{Purity of Coherence}
\label{sec:PH}

It turns out that coherence distillation cannot be done perfectly, and in fact, any coherence distillation protocol $\{\mE_N\}$ must have $\Omega(N^{-1})$ infidelity. The reason can be understood by studying the following quantity:

\begin{definition}[Purity of coherence \cite{Marvian2020}]
\label{def:purity-of-coherence}
The \emph{purity of coherence} of a quantum state $\rho$ with respect to a Hamiltonian $H$ is defined to be
\begin{equation}
    P_H(\rho) = \Tr\left[\rho^2H\rho^{-1}H\right] - \Tr\left[\rho H^2\right]
\end{equation}
if $\sup(H\rho H)\in\sup(\rho)$ and $\infty$ otherwise.
\end{definition}

It may be unclear at first glance why one should care about this quantity, or how one would come up with it at all. However, one way to motivate purity of coherence is as a distance metric on the space of states $\{e^{-itH}\rho e^{+itH} \,|\, t\in\Rbb\}$ known as right logarithmic derivative (RLD) Fisher information, which we discuss further in Section \ref{sec:RLD}. Purity of coherence is also proportional to the the second derivative of certain divergence measures between $\rho$ and $e^{-itH}\rho e^{+itH}$ at $t=0$. An example of such a divergence measure is the Petz-R\'{e}nyi $\alpha$-divergence with $\alpha=2$, given by $D_2(\rho \,||\, \sigma) = \ln\text{Tr}\left[\rho^2\sigma^{-1}\right]$. Therefore, purity of coherence is one way to quantify how rapidly a state time-evolved under the Schr\"{o}dinger equation becomes distinguishable from the original state. The interested reader is encouraged to peruse Supplementary Note 2 of \cite{Marvian2020} for a more detailed exposition on the properties of this quantity.

The most crucial fact about purity of coherence is that it is a measure of asymmetry with respect to time translation. In particular, it is zero for all incoherent states (those that commute with the Hamiltonian), positive for all coherent states (those that do not commute with the Hamiltonian), and non-increasing under TI channels:

\begin{theorem}[Monotonicity of purity of coherence under TI channels \cite{Marvian2020}]
\label{thm:PH-mononicity}
Suppose a TI channel $\mE_{\text{TI}}$ maps a system with Hamiltonian $H_{\text{in}}$ to a system with Hamiltonian $H_{\text{out}}$. Then for all possible input states $\rho$,
\begin{equation}
    P_{H_{\text{out}}}\left(\mE_{\text{TI}}(\rho)\right) \le P_{H_{\text{in}}}(\rho).
\end{equation}
\end{theorem}

The monotonicity of purity of coherence under TI channels can be seen as an immediate consequence of the monotonicity of various divergence measures (including the Petz-R\'{e}nyi $\alpha$-divergence with $\alpha=2$) under general quantum channels.

Other than monotonicity under TI channels, the most important property of purity of coherence is that it is infinite for pure coherent states, and it can be shown to be very large for all states with sufficiently high fidelity relative to a pure coherent state \cite{Marvian2020}. The result is a lower bound proved by \cite{Marvian2020} on the infidelity factor of any coherence distillation protocol, which we restate here for convenience:

\begin{theorem}[General lower bound on infidelity factor \cite{Marvian2020}]
\label{thm:infidelity-factor-lower-bound-general}
Any single-shot distillation protocol $\mE$ with input state $\rho_{\text{in}} = \rho^{\otimes N}$ and target state $\rho_{\text{out}} = \ket{\psi_{\text{out}}}\bra{\psi_{\text{out}}}$ must have infidelity factor
\begin{equation}
    \delta_1(\mE) \ge \frac{V_{H_{\text{out}}}\left(\ket{\psi_{\text{out}}}\bra{\psi_{\text{out}}}\right)}{P_H(\rho)}.
\end{equation}
\end{theorem}

Applying this theorem to the problem of single-shot qubit coherence distillation, as defined in Definition \ref{def:single-shot-qubit-protocol}, yields the following lower bound on the infidelity factor:

\begin{theorem}[Lower bound on infidelity factor for qubit coherence distillation]
\label{thm:infidelity-factor-lower-bound-qubit}
Any $N\rightarrow 1$ distillation protocol $\mE$ that maps input qubit states with polar angle $\Theta_{\text{in}}$ and purity parameter $\lambda$ to output qubit states with polar angle $\Theta_{\text{out}}$ must have infidelity factor
\begin{equation}
    \delta_1(\mE) \ge \frac{1-\lambda^2}{4\lambda^2}\frac{\sin^2\Theta_{\text{out}}}{\sin^2\Theta_{\text{in}}}.
\end{equation}
\end{theorem}

The proofs for Theorem \ref{thm:infidelity-factor-lower-bound-general} and Theorem \ref{thm:infidelity-factor-lower-bound-qubit} are provided in Appendix \ref{appendix:infidelity-lower-bound}.

\section{Main Results}
\label{sec:main-results}

The primary achievement of our work is that we find a first-order optimal single-shot qubit coherence distillation protocol and demonstrate that it has infidelity
\begin{equation}
    \mI(\mE_N) = \frac{1-\lambda^2}{4\lambda^2}\frac{\sin^2\Theta_{\text{out}}}{\sin^2\Theta_{\text{in}}}\frac{1}{N} + O\left(N^{-2}\right),
\end{equation}
meaning that its infidelity factor saturates the lower bound set by the monotonicity of purity of coherence. This result establishes an operational interpretation of the purity of coherence, and hence RLD Fisher information.

Furthermore, beyond first-order optimality, we additionally study the single-shot qubit coherence distillation problem in much more extensive detail.

First, we show how to compute absolutely optimal coherence distillation protocols for arbitrary parameter values. In fact, in the equatorial special case with odd $N$, we can actually find the absolutely optimal protocol analytically. As a result, we have the ability to study exactly optimal distillation protocols in much more detail than what the leading-order analyses would allow.

Second, we extend our analysis to solve the coherence distillation problem to $2^{\text{nd}}$-order optimality as well, and we even solve the equatorial special case to $3^{\text{rd}}$-order optimality. These analyses make this, to our knowledge, the first problem in the resource theory of asymmetry for which higher-order calculations have been explicitly carried out. We additionally interpret the $2^{\text{nd}}$-order optimal fidelity in the equatorial case as resulting from some small but nonzero wastage of purity of coherence.

Third, in the equatorial special case, we demonstrate how the optimal coherence distillation protocol behaves in the low-noise limit by treating it perturbatively around $\lambda=1$. In particular, we compute the second-order perturbation series in the parameter $c_0 = \frac{1-\lambda}{2}$, which denotes the infidelity of the input qubits with the target qubit state.

\section{Optimal Distillation Protocol}
\label{sec:optimal-protocol}

We now explain how we derived the optimal distillation protocol and showed that it asymptotically saturates the bound set by the monotonicity of purity of coherence. The first step, as described in Subsection \ref{subsec:preprocessing-schur-transform}, is to exploit a procedure known as Schur sampling \cite{Cirac1999, harrow2005applications}, which is reversible for permutation-invariant states such as $\rho^{\otimes N}$, to restrict the set of channels we need to consider to a collection with only $\Theta(N)$ real parameters. The second step, as described in Subsection \ref{subsec:boundary-value-problem}, is to write out the infidelity as a function of these parameters and use standard calculus to derive a boundary value problem that the parameters satisfy for the optimal protocol. This boundary value problem can then be solved numerically to derive an exactly optimal protocol for any single-shot qubit coherence distillation protocol, and it can be used to gain insight into how the optimal protocol behaves in the asymptotic regime. The third step, as described in Subsection \ref{subsec:concentration-power-series}, is to write out the parameters using power series in $N^{-1}$ and then use asymptotic properties of the distilled multi-qubit states obtained from Schur sampling \cite{Cirac1999} to do an order expansion and solve for those parameters sequentially.

\subsection{Preprocessing via the Schur Transform}
\label{subsec:preprocessing-schur-transform}

A powerful primitive that considerably reduces the space of protocols we need to consider is the Schur transform \cite{harrow2005applications, Cirac1999}. This transformation allows us to exploit the permutation symmetry of $\rho^{\otimes N}$. Recall that by Schur–Weyl duality, the $N$-qubit Hilbert space decomposes under the joint action of the unitary and permutation groups as
\begin{align}
    \left(\Cbb^2\right)^{\otimes N} &\cong \bigoplus_{j}\Cbb^{2j+1}\otimes\Cbb^{d(J,j)},
\end{align}
where $j = J, J-1, J-2, \cdots $ with $J=\tfrac{N}{2}$, $\Cbb^{2j+1}$ denotes the irreducible representations (irreps) of $SU(2)$, and $\Cbb^{d(J,j)}$ denotes the irreps of $S_N$. The irreps of $S_N$ contain no information about the state $\rho$, and hence that information can be discarded; only the irreps of $SU(2)$ contain useful information about $\rho$. This means that one can perform a total angular momentum measurement on the state $\rho^{\otimes N}$ that yields some value $j$ (a process known as \emph{Schur sampling}), and then apply a suitably chosen unitary to put the last $N-2j$ qubits in singlet states, which carry no useful information and thus can be freely discarded \cite{Cirac1999}.

The remaining state is a state on $N_C = 2j$ qubits that lives solely in the symmetric subspace, meaning that it lives in a far smaller Hilbert space of dimension $N_C+1$. In addition, this process is reversible \cite{Cirac1999}, meaning that it never hurts to start with this protocol before doing anything else. Furthermore, the multi-qubit state emerging from this procedure (which we call a \emph{Schur-sampled state} for convenience) retains no ``memory'' of the original value of $N$; as a result, we can just treat our protocol as something we do on $N_C$ qubits after Schur sampling.

We note that it has recently been shown that the Schur transform on $N$ copies of a qubit state $\rho$ can be efficiently implemented via random SWAP tests on the copies, followed by discarding the detected singlets \cite{brahmachari2025optimal}. In particular, after approximately $N\log_2N$ SWAP tests, the procedure yields, with high probability, a remaining subsystem of $N_C = 2j$ qubits in a totally symmetric state, with the success probability approaching one in the large-$N$ limit \cite{brahmachari2025optimal}.

Once we perform Schur sampling to reduce to a state on the $N_C$-qubit symmetric subspace and incorporate a few other ``obvious'' optimizations, we can narrow down our distillation protocol dramatically. In particular, in Appendix \ref{appendix:deriving-kraus-rep}, we prove that there must be an optimal distillation protocol with Kraus operators of the form
\begin{equation}
    K_w = \cos\theta_{w-1}\ket{0}\bra{(w-1)^{(s)}} + \sin\theta_w\ket{1}\bra{w^{(s)}}
\end{equation}
for $1\le w\le N_C$, where $\theta_0 = 0$, $\theta_{N_C} = \frac{\pi}{2}$, and $\theta_1,\cdots,\theta_{N_C-1}\in\left[0,\frac{\pi}{2}\right]$. (The state $\ket{w^{(s)}}$, commonly called a Dicke state, refers to the equal superposition of all computational basis states on $N_C$ qubits with Hamming weight $w$.) Hence we have only $N_C-1$ nontrivial real parameters that we need to optimize.

For the equatorial case $\Theta_{\text{in}} = \Theta_{\text{out}} = \frac{\pi}{2}$, we can additionally impose \emph{negation (bit-flip) symmetry}. In particular, if we negate all of the input qubits, then we should also negate the output qubit:
\begin{equation}
    \mE\left(X^{\otimes N}\rho_{\text{in}}X^{\otimes N}\right) = X\mE\left(\rho_{\text{in}}\right)X.
\end{equation}
For the Kraus operators shown above, bit-flip symmetry imposes an additional condition $\theta_w + \theta_{N_C-w} = \frac{\pi}{2}$. Although this only cuts the number of nontrivial real parameters by about half (from $N_C-1$ to $\lceil N_C/2\rceil-1$), it does have some unique consequences that substantially simplify the problem, which is why we give the equatorial case special attention throughout the paper.

\subsection{The First-Order Optimal Protocol}
\label{subsec:first-order-optimal-example}

In Appendix \ref{appendix:lower-bound-saturation}, we show that there exists a coherence distillation protocol that saturates the bound set by purity of coherence. In particular, we use the convenient form for the Kraus operators derived in Appendix \ref{appendix:deriving-kraus-rep}, which depend on parameters $\theta_w$ for $1\le\theta_w\le N_C-1$. We can consider the protocol whose $\theta_w$ values satisfy the following:
\begin{equation}
    \sin^2\theta_w = \frac{1-\cos\Theta_{\text{out}}}{2} + \lambda\frac{\sin^2\Theta_{\text{out}}}{\sin^2\Theta_{\text{in}}}\left(\frac{w}{N_C} - \frac{1-\cos\Theta_{\text{in}}}{2}\right).
\end{equation}
In Appendix \ref{appendix:lower-bound-saturation}, we show that this protocol achieves the lowest possible infidelity factor, namely
\begin{equation}
    \delta_1(\mE) = \frac{1-\lambda^2}{4\lambda^2}\frac{\sin^2\Theta_{\text{out}}}{\sin^2\Theta_{\text{in}}}.
\end{equation}
This confirms that the monotonicity of purity of coherence under TI channels sets the tightest possible bound on the leading-order infidelity of single-shot coherence distillation for qubits, and thereby imbues purity of coherence with operational significance.

Of course, this choice of $\sin^2\theta_w$ values may seem quite unmotivated. However, in the following subsections, we explain the systematic optimization method that allowed us to find this protocol.

\subsection{Boundary Value Problem}
\label{subsec:boundary-value-problem}

Once we have the Kraus representation shown at the end of Subsection \ref{subsec:preprocessing-schur-transform}, we need to optimize the values of $\theta_w$ for $1\le w\le N_C$. As we show in Appendix \ref{appendix:boundary-value-problem}, one can use ordinary multivariable calculus to find optimality conditions, but the solution only writes each $\theta_w$ value in terms of its neighbors $\theta_{w-1}$ and $\theta_{w+1}$. These three-angle relations, combined with the boundary conditions $\theta_0 = 0$ and $\theta_{N_C} = \frac{\pi}{2}$, produce a \emph{boundary value problem}.

Fortunately, this boundary value problem can be solved numerically to provide a wealth of insight into specific instances of this problem. There are also many other interesting studies one can carry out. For example, in the equatorial case with $N_C$ odd, the boundary value problem can actually be solved exactly, as we show in Appendix \ref{appendix:boundary-value-problem}\ref{subsec:solve-recurrence-odd-N-equatorial}. Furthermore, the boundary value problem can be solved perturbatively in the low-noise regime $\lambda\approx 1$, as we show in Appendix \ref{appendix:perturbative-protocols}.

Unfortunately, it is hard to use this boundary value problem to prove properties of an optimal or approximately optimal protocol in the asymptotic regime of large $N$. In order to perform the type of asymptotic analysis that revealed to us the $1^{\text{st}}$-order protocol stated in Subsection \ref{subsec:first-order-optimal-example}, we need to use a different method.

\subsection{Concentration Argument and Power Series Expansion}
\label{subsec:concentration-power-series}

To estimate the optimal distillation protocol in the asymptotic regime, we define $x=\frac{w}{N_C}$ for convenience and write $\sin^2\theta_w$ as a function $f(x)$ for some function $f:[0,1]\rightarrow[0,1]$. We then expand $f(x)$ as a power series about $x = x_0 \coloneqq \frac{1-\cos\Theta_{\text{in}}}{2}$, because the $P_{w-1,w}$ and $P_{w,w}$ values that enter the fidelity formula are concentrated heavily around $x = x_0$. In particular, the $k^{\text{th}}$ derivative $f^{(k)}(x_0)$ cannot affect $\delta_p(\mE)$ unless $p\ge\frac{k}{2}$; and furthermore, we find that the conditions for optimality do not involve the $k^{\text{th}}$ derivative until we reach $p=k$. This explains why taking $\sin^2\theta_w$ as a suitably chosen linear function of $w$ is sufficient to achieve $1^{\text{st}}$-order optimality. This also allows us to optimize the derivatives sequentially to find distillation protocols that are optimal to any order we wish. (For $2^{\text{nd}}$-order optimality and beyond, we have to consider the fact that the derivatives themselves could have correction terms that decay with $N_C$. Nonetheless, the essential optimization strategy remains the same.)

In Appendix \ref{appendix:asymptotic-expansion}, we explain the more general strategy that we use to sequentially carry out the asymptotic analysis at each order of optimality. In Appendix \ref{appendix:1st-order-optimality}, we show the calculations for the zeroth derivative $f(x_0)$ and first derivative $f'(x_0)$, which are all we need to compute a $1^{\text{st}}$-order optimal distillation protocol to demonstrate the operational significance of purity of coherence and RLD Fisher information. We find that the optimal choices are
\begin{equation}
    f(x_0) = \frac{1-\cos\Theta_{\text{out}}}{2}, \quad f'(x_0) = \lambda\frac{\sin^2\Theta_{\text{out}}}{\sin^2\Theta_{\text{in}}},
\end{equation}
which is exactly how we obtain the choice of $\sin^2\theta_w$ that we presented in Subsection \ref{subsec:first-order-optimal-example}.

In Appendix \ref{appendix:2nd-order-optimality}, we find the conditions for $2^{\text{nd}}$-order optimality in the general case. Finally, in Appendix \ref{appendix:equatorial-3rd-order-optimality}, we extend the analysis of the equatorial case to $3^{\text{rd}}$-order optimality.

\section{Connection to RLD Fisher Information}
\label{sec:RLD}

Purity of coherence is closely linked to a distance metric on quantum states known as the right logarithmic derivative (RLD) Fisher information. More precisely, the RLD Fisher information squared distance between $\rho$ and its time-evolved version $e^{-itH}\rho e^{+itH}$ equals $P_H(\rho)t^2 + O(t^4)$ for small values of $t$ \cite{Marvian2020}. As a result, the operational interpretation of purity of coherence can be understood from the properties of RLD Fisher information, which we explain in this section.

In classical information geometry, the Fisher information metric defines a Riemannian metric on the space of probability distributions. By Chentsov's theorem \cite{Chentsov1982}, Fisher information is the unique (up to global scaling) Riemannian metric on probability distributions that is non-increasing under stochastic maps. This imbues Fisher information with great significance in the field of classical parameter estimation, since parameter estimation is all about distinguishing a continuum of different probability distributions using a number of samples. Most notably, the Cram\'{e}r-Rao bound states that the covariance matrix of an unbiased estimator is greater than or equal to the inverse of the Fisher information matrix, a fact that arises from the monotonicity of the Fisher information metric (though it is usually not presented this way).

Generalizing Chentsov's theorem to quantum states would mean finding all Riemannian metrics on density operators that are non-increasing under quantum channels. However, instead of just one metric, one finds a whole family of monotone Riemannian metrics. These metrics were fully classified by Morozova, Chentsov, and Petz \cite{Morozova1991,Petz1996}, whose primary result we restate here for convenience:

\begin{theorem}[Morozova-Chentsov-Petz theorem \cite{Morozova1991,Petz1996}]
\label{thm:morozova-chentsov-petz}
A Riemannian metric on density operators is monotone (that is, non-increasing under quantum channels) if and only if, in the neighborhood of a diagonalized density matrix $\rho = \sum_{i}p_i\ket{i}\bra{i}$, it takes the form
\begin{equation}
    ds^2 = \sum_{i}\frac{d\rho_{ii}^2}{p_i} + \sum_{i\neq j}\frac{\abs{d\rho_{ij}}^2}{m_f(p_i,p_j)},
\end{equation}
where $m_f(p_i,p_j) = p_jf(p_i/p_j)$ and $f:[0,\infty)\rightarrow[0,\infty)$ is a Morozova-Chentsov (MC) function, meaning that it satisfies the following three properties:
\begin{itemize}
    \item normalized: $f(1) = 1$
    \item self-inverse: $f(t) = tf(1/t)$
    \item operator monotone: for any two positive semidefinite matrices $A$ and $B$, $A\ge B$ implies $f(A)\ge f(B)$.
\end{itemize}
These monotone Riemannian metrics are commonly referred to as \emph{quantum Fisher information (QFI) metrics}.
\end{theorem}

Notice that the first summation, which involves infinitesimal changes in the diagonal entries of $\rho$, looks exactly like the classical Fisher information metric on the space of probability distributions on a finite sample space. This is an immediate consequence of Chentsov's theorem for classical probability distributions \cite{Chentsov1982}. However, the second summation, which involves infinitesimal changes in the off-diagonal entries of $\rho$, has considerable freedom coming from the choice of MC function $f$.

Intuitively, $m_f(p_i,p_j)$ is some way of taking an average of $p_i$ and $p_j$, which is why $f$ must satisfy the normalization and self-inverse conditions. The operator monotone condition restricts MC functions to lie between a minimum MC function $f_{\text{HM}}(t) = \frac{2t}{1+t}$ and a maximum MC function $f_{\text{AM}}(t) = \frac{1+t}{2}$, where ``HM'' and ``AM'' stand for ``harmonic mean'' and ``arithmetic mean'', respectively. Since the MC function is in the denominator, the maximum MC function $f_{\text{AM}}$ yields the smallest QFI metric, known as \emph{symmetric logarithmic derivative (SLD) Fisher information}, while the minimum MC function $f_{\text{HM}}$ yields the largest QFI metric, known as \emph{right logarithmic derivative (RLD) Fisher information}. Two other notable QFI metrics are the Wigner-Yanase metric and the Kubo-Mori metric, with respective MC functions $f_{\text{WY}}(t) = \left(\frac{1+\sqrt{t}}{2}\right)^2$ and $f_{\text{KM}}(t) = \frac{t-1}{\ln t}$.

We can classify QFI metrics into two broad categories based on their behavior for rank-deficient states. If $f(0)\neq 0$, such as for the SLD and Wigner-Yanase metrics, then the QFI metric can be extended to rank-deficient states. In contrast, if $f(0)=0$, such as for the Kubo-Mori and RLD metrics, then the QFI metric becomes infinite for rank-deficient states. More specifically, a path on the manifold of rank-deficient density operators can have infinite length with respect to the latter type of QFI metric. This is a marked contrast from the classical Fisher information, since a path on the manifold of probability distributions on a finite sample space will always have finite length with respect to the classical Fisher information metric. Furthermore, the fact that some path lengths on density operators become infinite with respect to QFI metrics with $f(0)=0$ allows one to prove the impossibility of coherence distillation at a positive linear rate $r > 0$, since any of those QFI metrics on the family of input states $\rho(\theta)^{\otimes N}$ would scale as $N$, while that QFI metric on any family of states with vanishing infidelity relative to the family of $\ket{\psi(\theta)}\bra{\psi(\theta)}^{\otimes \lfloor rN\rfloor}$ would scale strictly faster than $N$.

The distinguished significance of the RLD metric comes from its status as the largest of the QFI metrics. Intuitively, to derive the strongest possible resource-theoretic bound on a quantum state conversion task, you should pick a resource that is small for the states you \textit{have} and large for the states you \textit{want}. Since coherence distillation seeks to produce pure states (or as close as possible to pure states), The RLD metric is the ideal resource to study, since among all QFI metrics, it diverges the fastest as one approaches a pure state. Any of the metrics that diverges for a pure state can be used to derive the conclusion that coherence distillation at a positive linear rate is impossible, but the RLD metric is preferred because it gives us the lowest upper bound on the sub-linear rate at which one can perform distillation with vanishing error, and conversely it gives us the highest lower bound on the infidelity that must remain in the output qubit for single-shot distillation.

\section{Discussion}
\label{sec:discussion}

Despite recent advances in the resource theory of asymmetry, e.g., in \cite{Marvian2020, tajima2022universal, Yamaguchi2024, yamaguchi2023beyond, marvian2022operational, tajima2024gibbs}, many important questions remain open. In this work, we solved an open problem posed in \cite{Marvian2020} about how well qubit clocks can be distilled. This work observed that, in the equatorial special case, the lower bound on the infidelity factor enforced by purity of coherence can be saturated at both the low-noise and high-noise limits using more well-established protocols (see Figure \ref{fig:coherence-distillation-three-graphs}). It was further conjectured that this lower bound can be saturated at all noise levels, though no proof was attempted. By establishing first-order optimality, we demonstrate that the purity of coherence bound can indeed be saturated across all noise regimes.

\begin{figure}
    \includegraphics[scale=0.39]{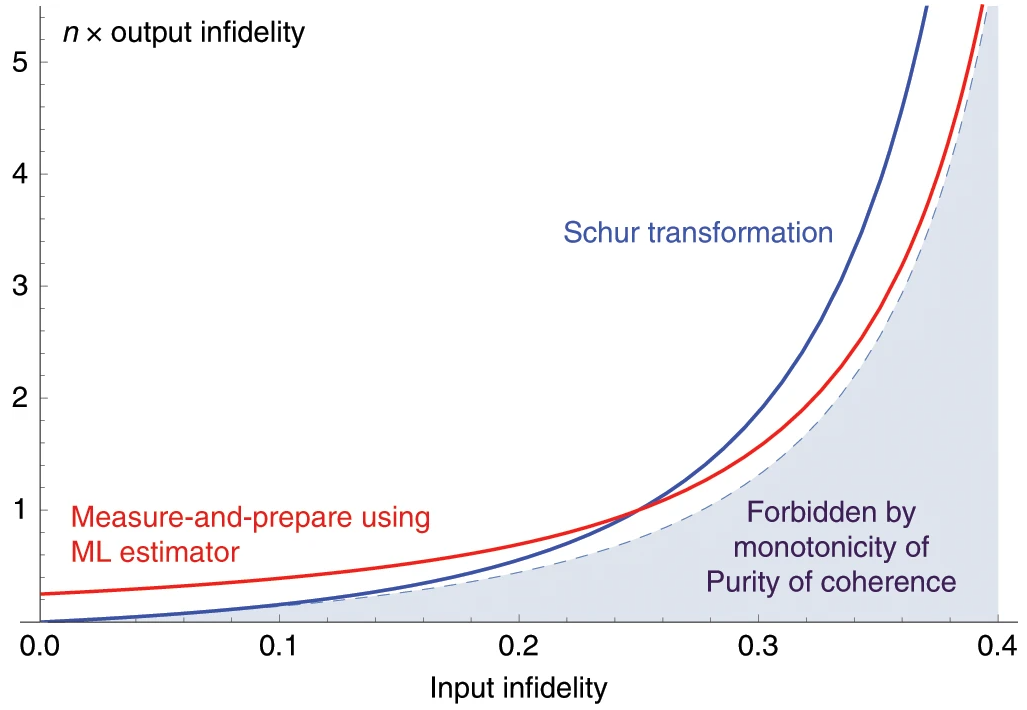}
    \caption{A comparison of the first-order infidelity of three different single-shot qubit coherence distillation protocols. The horizontal axis is the input infidelity $\frac{1-\lambda}{2}$, and the vertical axis is the infidelity factor $\delta_1(\mE)$, meaning that the output infidelity is $\mI(\mE_N) = \frac{\delta_1(\mE)}{N} + o(N^{-1})$. The blue curve $\mathcolor{blue}{\delta_1(\mE) = \frac{1-\lambda}{2\lambda^2}}$ corresponds to the qubit distillation protocol with full $SU(2)$ covariance devised by Cirac et al., which is based on the Schur transform \cite{Cirac1999}. The red curve $\mathcolor{red}{\delta_1(\mE) = \frac{1}{4\lambda^2}}$ corresponds to a TI measure-and-prepare protocol devised in \cite{Marvian2020}. Finally, the violet curve $\mathcolor{violet}{\delta_1(\mE) = \frac{1-\lambda^2}{4\lambda^2}}$ corresponds to our first-order optimal protocol. Anything below the purple curve is forbidden by the monotonicity of purity of coherence \cite{Marvian2020}. (This figure is taken from \cite{Marvian2020}.)}
    \label{fig:coherence-distillation-three-graphs}
\end{figure}

The primary importance of our work is that it provides an operational significance to purity of coherence (RLD Fisher information) as the quantity whose monotonicity sets the tightest possible leading-order bound on the performance of single-shot qubit coherence distillation. This same fact was also recently established for the problem of distilling thermal coherent states of a quantum harmonic oscillator \cite{Yadavalli2024}. We are thus increasingly confident that the coherence distillation problem will be governed at the leading order by purity of coherence for any input family of noisy coherent states and output family of pure coherent states (although this is still open). A natural direction for future study would be to find a first-order optimal coherence distillation protocol for a general input coherent state and output coherent state. Of particular interest would be whether RLD Fisher information continues to enforce the tightest possible bound on the infidelity factor. One could even study distillation problems for symmetry groups other than $U(1)$, and similarly ask whether there are natural information-theoretic considerations (QFI metrics or otherwise) that yield the minimal infidelity factor.

Although the Morozova-Chentsov-Petz classification of Riemannian CPTP-monotone metrics on density operators has been known for thirty years, relatively little attention has been paid to any of these metrics other than SLD Fisher information, which is by far the most famous, largely because it sets the tightest asymptotic bound on the mean squared error of quantum state estimation on a one-parameter family \cite{Braunstein1994}, analogous to the Cram\'{e}r-Rao bound for classical parameter estimation. This work showcases the other extreme of the monotone metrics family as an important quantity in its own right. In particular, the fact that RLD Fisher information blows up for rank-deficient states is crucial for its importance in quantum state distillation problems such as this one.

Beyond first-order optimality, the ability to numerically solve for an absolutely optimal protocol for arbitrary values means that coherence distillation can be studied for non-asymptotic regimes as well. This distinguishes it from many other similar problems, for which very little is known in non-asymptotic regimes. One potential use case could be the experimental realization of single-shot coherence distillation protocols on a small number of qubits.

The second-order analysis in the general case and the third-order analysis in the equatorial special case are also significant in their own right. Many other problems in the resource theory of asymmetry and quantum state estimation theory have been addressed at the leading order using methods that preclude the possibility of higher-order analysis (see Appendix \ref{appendix:2nd-order-optimality}\ref{subsec:2nd-order-optimality-difficult}).  Therefore, although our method of finding a first-order optimal protocol is bespoke and may be difficult to apply to other problems, the fact that it can be naturally extended to higher orders at all is a virtue.

Interestingly, as part of our own work, we also perform some amount of second-order and third-order analysis for Schur sampling applied to a tensor product of $N$ identical qubits, which we show in Appendix \ref{appendix:angular-momentum-moments}. This work may be of independent interest for other problems where Schur sampling is used as a primitive.

In fact, our work is the first to explicitly compute the $2^{\text{nd}}$-order and $3^{\text{rd}}$-order infidelity for a distillation problem in the resource theory of asymmetry. The only other distillation problem to be solved beyond $1^{\text{st}}$-order optimality is that of single-shot qudit distillation with full $SU(d)$ symmetry, solved for qubits by Cirac et al. \cite{Cirac1999} and extended to qudits by Li et al. \cite{Li2024}. However, neither of these works explicitly carries out the infidelity analysis to higher orders.

Therefore, another future endeavor would be to study higher orders of optimality for distillation protocols. In the classical setting, information geometry has been used to study higher orders of efficiency for parameter estimation beyond the Cram\'{e}r-Rao bound (which governs $1^{\text{st}}$-order efficiency, with one notable result being that maximum likelihood estimation supplemented by bias correction is $2^{\text{nd}}$-order efficient but \textit{not} $3^{\text{rd}}$-order efficient \cite{Amari1985,Ghosh1982}. It would be interesting to see whether there are analogues of these facts and others for higher-order optimal distillation protocols.


\section*{Acknowledgments}
\label{sec:acknowledgments}

 We acknowledge support from  NSF Phy-2046195, NSF FET-2106448, and 
NSF QLCI grant OMA-2120757. SK is funded by the National Defense Science and Engineering Graduate (NDSEG) Fellowship. SK and IM would like to thank Shiv Akshar Yadavalli, Shrigyan Brahmachari, Yash Chitgopekar, Govind Lal Sidhardh, Plato Deliyannis, Austin Hulse, Nikolaos Koukoulekidis, and David Jakab for many useful discussions.

\bibliography{references}

\onecolumngrid

\newpage

\maketitle
\vspace{-5in}
\begin{center}
\Large{Supplementary Material}
\end{center}

This paper has a lot of supplementary material, so for the reader's benefit, we organize them as follows:
\begin{itemize}
    \item If you are only interested in how coherence distillation provides an operational interpretation for purity of coherence (RLD Fisher information), you can focus on Appendices \ref{appendix:infidelity-lower-bound}, \ref{appendix:deriving-kraus-rep}, and \ref{appendix:lower-bound-saturation}.
    \item To see how the problem of optimal single-shot qubit coherence distillation can be understood as a boundary value problem that can be solved numerically, you can refer to Appendix \ref{appendix:boundary-value-problem}.
    \item For an explanation of how one can systematically find protocols at any order of optimality in the large-$N$ asymptotic regime, you can refer to Appendix \ref{appendix:asymptotic-expansion}.
    \item To see how we came up with a $1^{\text{st}}$-order optimal protocol that saturates the purity-of-coherence bound, you can refer to Appendix \ref{appendix:1st-order-optimality}.
    \item For $2^{\text{nd}}$-order and $3^{\text{rd}}$-order analyses in the large-$N$ asymptotic regime, you can refer to Appendices \ref{appendix:2nd-order-optimality}, \ref{appendix:perfect-conversion}, and \ref{appendix:equatorial-3rd-order-optimality}.
    \item For proofs of the asymptotic formulas used throughout this paper, which concern the distribution of the angular momentum measurement outcome when performing Schur sampling, as well as the matrix entries of the resulting states on the symmetric subspace, you can refer to Appendices \ref{appendix:angular-momentum-moments} and \ref{appendix:understanding-P-vals}.
    \item For a deeper study of other facets of the single-shot coherence distillation problem that are fascinating in their own right, you can refer to Appendices \ref{appendix:PH-dissipation}, \ref{appendix:entanglement-breaking}, and \ref{appendix:perturbative-protocols}. In all of these appendices, we restrict our attention to the equatorial case.
\end{itemize}


\newcommand\appitemtwo[2]{
\newcommand\appitem[1]{\hyperref[{#1}]
{\textbf{\cref{#1}}} \textbf{\nameref*{#1}}
\dotfill \pageref{#1}\vspace{5pt}}
\newcommand\subappitem[1]{
\makeatletter
\newcommand{\appsec}[2]{%
  \section{#1}%
  \def\@currentlabelname{#1}%
  \def\@currentlabel{\thesection}%
  \label{#2}%
  \addcontentsline{toc}{section}{#1}%
}
\newcommand{\appsubsec}[2]{%
  \subsection{#1}%
  \def\@currentlabelname{#1}%
  \def\@currentlabel{\thesubsection}%
  \label{#2}%
}
\newcommand{\appsubsubsec}[2]{%
  \refstepcounter{subsubsection}%
  \subsubsection{#1}%
  \addtocounter{subsubsection}{-1}%
  \def\@currentlabelname{#1}%
  \def\@currentlabel{\thesubsubsection}%
  \label{#2}%
}
\makeatother


\section*{Supplementary Material: Table of Contents}
\begin{itemize}[label={}]
\item \appitem{appendix:infidelity-lower-bound}

\item \appitem{appendix:deriving-kraus-rep}
\subitem \subappitem{subsec:cirac-distillation-primitive}
\subitem \subappitem{subsec:covariance-conditions-parameter-counting}
\subitem \subappitem{subsec:distillation-protocol-form-proof-choi-matrices}
\subitem \subappitem{subsec:distillation-protocol-form-proof-kraus-operators}
\subitem \subappitem{subsec:distillation-protocol-implementation-stinespring-dilation}

\item \appitem{appendix:lower-bound-saturation}

\item \appitem{appendix:boundary-value-problem}
\subitem \subappitem{subsec:brute-force-optimization}
\subitem \subappitem{subsec:solve-recurrence-pure-input-matching-target}
\subitem \subappitem{subsec:solve-recurrence-odd-N-equatorial}

\item \appitem{appendix:asymptotic-expansion}

\item \appitem{appendix:1st-order-optimality}
\subitem \subappitem{subsec:0th-order-optimality-derivation}
\subitem \subappitem{subsec:1st-order-optimality-derivation}

\item \appitem{appendix:2nd-order-optimality}
\subitem \subappitem{subsec:2nd-order-optimality-difficult}
\subitem \subappitem{subsec:2nd-order-optimality-derivation}
\subitem \subappitem{subsec:2nd-order-optimality-special-case-equatorial}
\subitem \subappitem{subsec:2nd-order-optimality-special-case-lam1}

\item \appitem{appendix:perfect-conversion}

\item \appitem{appendix:equatorial-3rd-order-optimality}

\item \appitem{appendix:angular-momentum-moments}
\subitem \subappitem{subsec:J2-integer-moments}
\subitem \subappitem{subsec:J2-moments-to-NC-moments}
\subitem \subappitem{subsec:positive-moments-to-negative-moments}
\subitem \subappitem{subsec:NC-moments-ese}
\subitem \subappitem{subsec:NC-negative-moments-full-computation}

\item \appitem{appendix:understanding-P-vals}
\subitem \subappitem{subsec:P-vals-miscellaneous-facts}
\subitem \subappitem{subsec:P-vals-moments-defining}
\subitem \subappitem{subsec:P-vals-moments-offset0}
\subitem \subappitem{subsec:P-vals-moments-offset1-equatorial}
\subitem \subappitem{subsec:P-vals-moments-offset1-general}

\item \appitem{appendix:PH-dissipation}

\item \appitem{appendix:entanglement-breaking}
\subitem \subappitem{subsec:characterizing-eb-protocols}
\subitem \subappitem{subsec:explicit-separability}

\item \appitem{appendix:perturbative-protocols}
\subitem \subappitem{subsec:lam1-order1}
\subitem \subappitem{subsec:lam1-order2}
\end{itemize}

\appendix

\color{black}
\onecolumngrid

\include{app_infidelity_lower_bound}

\newpage

\include{app_deriving_Kraus_rep}

\newpage

\include{app_lower_bound_saturation}

\newpage

\include{app_boundary_value_problem}

\newpage

\include{app_asymptotic_expansion}

\newpage

\include{app_first_order}

\newpage

\include{app_second_order}

\newpage

\include{app_perfect_conversion}

\newpage

\include{app_third_order_equatorial}





\newpage

\include{app_ang_mom_moments}

\newpage

\include{app_P_vals}

\newpage

\include{app_PH_dissipation}

\newpage

\include{app_entanglement_breaking}



\newpage

\include{app_perturbative_protocols}



\end{document}

%% file: app_infidelity_lower_bound.tex
\appsec{Lower Bound on Infidelity Factor of Single-Shot Coherence Distillation}
{appendix:infidelity-lower-bound}

In this appendix, we use the monotonicity of purity of coherence under TI channels (Theorem \ref{thm:PH-mononicity}) to prove a lower bound on the infidelity factor of a general single-shot coherence distillation protocol (Theorem \ref{thm:infidelity-factor-lower-bound-general}). We then apply this result to prove a lower bound on the infidelity factor of a single-shot qubit distillation protocol (Theorem \ref{thm:infidelity-factor-lower-bound-qubit}).

\vspace{0.5\baselineskip}

To lower-bound the infidelity, we must show that any mixed state that is sufficiently close to a pure coherent state must have large purity of coherence. Fortunately, this work has already been done for us by Supplementary Note 3, Lemma 2, Corollary 1 of \cite{Marvian2020}, which we restate here for convenience:

\begin{lemma}[Purity of coherence of mixed state near pure state; Supplementary Note 3, Lemma 2, Corollary 1 of \cite{Marvian2020}]
\label{lem:PH-mixed-near-pure}
If a state $\sigma$ has infidelity $\delta$ with pure state $\ket{\phi}$, i.e., $\bra{\phi}\sigma\ket{\phi} = 1-\delta$, then
\begin{equation}
    P_H(\sigma) \ge V_H(\ket{\phi}\bra{\phi})\left[\frac{(1-\delta)^2}{\delta} - 1\right].
\end{equation}
\end{lemma}

This will be the principal fact we use to lower-bound the infidelity factor.

\vspace{0.5\baselineskip}

\begin{proof}[Proof of Theorem \ref{thm:infidelity-factor-lower-bound-general}]

Consider a single-shot distillation protocol $\mE_N$ with infidelity
\begin{equation}
    \mI(\mE) = \frac{\delta_1(\mE)}{N} + o\left(N^{-1}\right).
\end{equation}
Using Lemma \ref{lem:PH-mixed-near-pure}, the purity of coherence of the output state is
\begin{align}
    P_{H_{\text{out}}}\left(\mE\left(\rho_{\text{in}}\right)\right) &\ge V_H\left(\ket{\psi_{\text{out}}}\bra{\psi_{\text{out}}}\right)\left[\frac{\left(1-\mI(\mE)\right)^2}{\mI(\mE)}-1\right] \\
    &= \left[\frac{V_H\left(\ket{\psi_{\text{out}}}\bra{\psi_{\text{out}}}\right)}{\delta_1(\mE)} + o(1)\right]N.
\end{align}
The monotonicity of purity of coherence states that
\begin{equation}
    P_{H_{\text{out}}}\left(\mE\left(\rho_{\text{in}}\right)\right) \le P_{H_{\text{in}}}\left(\rho_{\text{in}}\right) = P_H(\rho)N,
\end{equation}
which immediately implies that
\begin{equation}
    \delta_1(\mE) \ge \frac{V_{H_{\text{out}}}\left(\ket{\psi_{\text{out}}}\bra{\psi_{\text{out}}}\right)}{P_H(\rho)}.
\end{equation}

\end{proof}

\begin{proof}[Proof of Theorem \ref{thm:infidelity-factor-lower-bound-qubit}]
The purity of coherence of a qubit state can be computed conveniently using results from Supplementary Note 2 of \cite{Marvian2020}. (In particular, the subsection titled ``Purity of coherence for Qubits'' provides several nice formulas.) The purity of coherence of each input qubit is
\begin{equation}
    P_H(\rho) = \frac{4\lambda^2}{1-\lambda^2}\sin^2\Theta_{\text{in}}.
\end{equation}
The energy variance of the target qubit state is
\begin{equation}
    V_H\left(\ket{\psi_{\text{out}}}\bra{\psi_{\text{out}}}\right) = \sin^2\Theta_{\text{out}}.
\end{equation}
The resulting infidelity factor lower bound is thus
\begin{equation}
    \delta_1(\mE) \ge \frac{1-\lambda^2}{4\lambda^2}\frac{\sin^2\Theta_{\text{out}}}{\sin^2\Theta_{\text{in}}}.
\end{equation}
\end{proof}

If we restrict our attention to equatorial distillation, we can make a more precise statement. In particular, as shown in Supplementary Note 10 of \cite{Marvian2020}, one can exploit symmetry with respect to the $X$ operator to show that
\begin{equation}
    \mI(\mE_N) \ge \frac{1}{2}\left[1 - \left(1 + \frac{1-\lambda^2}{\lambda^2}\frac{1}{N}\right)^{-1/2}\right] = \frac{1-\lambda^2}{4\lambda^2}\frac{1}{N} - \frac{3\left(1-\lambda^2\right)^2}{16\lambda^2}\frac{1}{N^2} + \frac{5\left(1-\lambda^2\right)^3}{32\lambda^6}\frac{1}{N^3} + O\left(N^{-4}\right).
\end{equation}
Just as in the general case, the first term can be saturated, in the sense that there exists a protocol $\mE_N$ satisfying $\delta_1(\mE) = \frac{1-\lambda^2}{4\lambda^2}$.

\vspace{0.5\baselineskip}

In the equatorial case, since it is convenient to write a lower bound on infidelity that includes higher-order terms, one may also naturally wonder whether there exists a protocol that can match the higher-order terms in the above inequality. For example, does there exist a protocol $\mE_N$ satisfying $\delta_1(\mE) = \frac{1-\lambda^2}{4\lambda^2}$ and also $\delta_2(\mE) = -\frac{3\left(1-\lambda^2\right)^2}{16\lambda^4}$, thereby saturating the purity-of-coherence bound at the second order? However, the answer to this question turns out to be no, as we show in Appendix \ref{appendix:2nd-order-optimality}\ref{subsec:2nd-order-optimality-special-case-equatorial}. As a result, it is unclear whether the higher-order terms resulting from monotonicity of purity of coherence have any operational significance.

%% file: app_deriving_Kraus_rep.tex
\appsec{Deriving the Kraus Representation of the Optimal Distillation Protocol}
{appendix:deriving-kraus-rep}

This appendix is dedicated to proving the following theorem:

\begin{theorem}
\label{thm:optimal-protocol-kraus-rep}
There is an optimal $N\rightarrow 1$ coherence distillation channel $\mE$ that begins with Schur sampling and subsequently applies a channel with the following Kraus operators on the $N_C$-qubit symmetric subspace:
\begin{equation}
    K_w = \cos\theta_{w-1}\ket{0}\bra{(w-1)^{(s)}} + \sin\theta_w\ket{1}\bra{w^{(s)}} \quad (1\le w\le N_C),
\end{equation}
for some collection of angles $\theta_0, \theta_1, \theta_2, \cdots, \theta_{N_C}\in\left[0,\frac{\pi}{2}\right]$, where $\theta_0 = 0$ and $\theta_{N_C} = \frac{\pi}{2}$. Here, $\ket{w^{(s)}}$ denotes the fully symmetric $N_C$-qubit state with Hamming weight $w$ in the computational basis (also commonly called a \textbf{Dicke state}), that is:
\begin{equation}
    \ket{w^{(s)}} = \binom{N_C}{w}^{-1/2}\sum_{w(b)=w}\ket{b},
\end{equation}
where the sum is taken over all $N_C$-bit strings $b$ with Hamming weight $w$.
\end{theorem}

We will provide two different proofs of this theorem, which are based on two different common descriptions of a quantum channel. The first proof is based on Choi matrices (see Appendix \ref{appendix:deriving-kraus-rep}\ref{subsec:distillation-protocol-form-proof-choi-matrices}), whereas the second proof is based on Kraus operators (see Appendix \ref{appendix:deriving-kraus-rep}\ref{subsec:distillation-protocol-form-proof-kraus-operators}). The second proof (using Kraus operators) is probably the most appealing to well-versed quantum information theorists. However, the first proof (using Choi matrices) introduces a more general expression for covariant protocols which will be valuable in Appendix \ref{appendix:entanglement-breaking}, where we discuss entanglement-breaking distillation channels. One popular tool for determining whether a channel is entanglement-breaking is the positive partial transpose (PPT) condition, which is easy to assess given a Choi matrix.

\vspace{0.5\baselineskip}

Furthermore, following those two proofs, we will demonstrate that a distillation protocol of the form stated in Theorem \ref{thm:optimal-protocol-kraus-rep} has a nice description in terms of yet another common description of a quantum channel, namely the Stinespring dilation (see Appendix \ref{appendix:deriving-kraus-rep}\ref{subsec:distillation-protocol-implementation-stinespring-dilation}). This will shed some light on how one might actually implement one of these coherence distillation protocols in practice, including an operational meaning for the $\theta_w$ values in the theorem statement.

\vspace{0.5\baselineskip}

Prior to these discussions, we first provide more detail about how Schur sampling and additional covariance conditions reduce the space of protocols we need to consider. In Appendix \ref{appendix:deriving-kraus-rep}\ref{subsec:cirac-distillation-primitive}, we say a bit more about Schur sampling. In addition, in Appendix \ref{appendix:deriving-kraus-rep}\ref{subsec:covariance-conditions-parameter-counting}, we explain how imposing additional covariance conditions reduces the number of parameters but does not capture the full benefits of Schur sampling.

\appsubsec{Using Schur Sampling as a Primitive}
{subsec:cirac-distillation-primitive}

An important primitive in all three proofs, and throughout this paper, is the process of Schur sampling \cite{Cirac1999}. Given $\rho^{\otimes N}$ for a qubit state $\rho = \lambda\ket{\psi}\bra{\psi} + (1-\lambda)\frac{\Ibb}{2}$, first perform a total angular momentum measurement, which yields some value $j$ and a multiplicity index $\alpha$. Afterward, apply a suitable unitary $U_{j,\alpha}$ to transform the state into a state where the last $N-2j$ qubits are all in singlet states. These states carry no information about $\rho$, and thus they can be freely discarded. The result is a state on the symmetric space of $N_C=2j$ qubits, where $0\le N_C\le N$ and $N-N_C$ is even (but for large $N$, $N_C\sim\lambda N$ with high probability). We discuss the distribution of $N_C$ (as a function of $N$ and $\lambda$) in much more detail in Appendix \ref{appendix:angular-momentum-moments}.

\vspace{0.5\baselineskip}

Cirac et al. showed that Schur sampling can be understood as a distillation procedure because the reduced state of any one of these remaining qubits has very high fidelity with $\ket{\psi}\bra{\psi}$ \cite{Cirac1999}. In fact, if $\Theta_{\text{in}} = \Theta_{\text{out}}$, just performing Schur sampling and then discarding all but one qubit already achieves decent performance, with infidelity factor $\delta_1(\mE) = \frac{1-\lambda}{2\lambda^2}$ \cite{Cirac1999}. However, this procedure is not optimal for the problem we care about, first because it satisfies full $SU(2)$ covariance, which is more than what we need, and second because  $\Theta_{\text{in}}\neq\Theta_{\text{out}}$ in general.

\vspace{0.5\baselineskip}

For this discussion, the crucial points about the Schur sampling are as follows. First, it is actually reversible, in the sense that one can undo it by reintroducing the appropriate number of singlet states and performing a random permutation on the qubits \cite{Cirac1999}. Hence, it never hurts to begin our distillation procedure by performing Schur sampling. Second, the Schur-sampled state lives in the $N_C$-qubit symmetric subspace, which has dimension $N_C+1$, far smaller than the $2^N$ dimension of the original Hilbert space.

\vspace{0.5\baselineskip}

The real kicker is that any output state of Schur sampling applied to $\rho^{\otimes N}$, which lives in the $N_C$-qubit symmetric subspace, depends only on $N_C$ and $\lambda$; it maintains no ``memory'' of the original value of $N$. This means that, when we optimize the post-Schur-sampling protocol for a specific value of $N_C$, we actually do it in one fell swoop for any value of $N$ that could produce that value of $N_C$ from Schur sampling (in other words, any $N\ge N_C$ such that $N-N_C$ is even). This allows us to optimize the $\theta_w$ values just based on $N_C$, which is what ensures that we have $N_C-1 = \Theta(N)$ real parameters. If the ``memoryless'' property were not true, then perhaps we would need to optimize the $\theta_w$ values based on both the original number of qubits $N$ and the Schur-sampled number of qubits $N_C$, which would increase the number of parameters to $\Theta(N^2)$.

\vspace{0.5\baselineskip}

There are some additional practical benefits of always starting our procedure with Schur sampling. In particular, the $N_C$-qubit symmetric subspace has dimension $N_C+1$, which is far smaller than $2^N$ for the original Hilbert space.

Also, this way of thinking automatically incorporates the extra assumption of permutation invariance but goes even further. Permutation invariance is an addition which we discuss \ref{appendix:deriving-kraus-rep}\ref{subsec:covariance-conditions-parameter-counting}. As a result, it greatly benefits us to think of our distillation protocols in a ``post-Schur-sampling'' sense, and we will essentially always do so from this point forward.

\appsubsec{Covariance Conditions and Parameter Counting}
{subsec:covariance-conditions-parameter-counting}

To further demonstrate the advantage of starting with Schur sampling, we count the number of parameters in the single-shot qubit coherence distillation protocol under various constraints.

\vspace{0.5\baselineskip}

An arbitrary quantum channel from $N$ qubits to $1$ qubit, has $3\times 4^N$ real parameters, which we can see by considering the Choi matrix. There are $4^{N+1}$ real parameters for a Hermitian matrix of dimension $2^{N+1}$ (so that the channel is Hermiticity-preserving), but then each $2\times 2$ block has to have trace $1$ for the channel to be trace-preserving.

\vspace{0.5\baselineskip}

The time-translation invariance (TI) condition reduces this parameter count somewhat. In particular, as we will show in Appendix \ref{appendix:deriving-kraus-rep}\ref{subsec:distillation-protocol-form-proof-choi-matrices}, the TI condition enforces a block-diagonal condition on the Choi matrix, where the different blocks have dimension $\binom{N+1}{E}$ for each integer $0\le E\le N+1$. Hence the number of free parameters drops to approximately $\binom{2N+2}{N+1}$ (slightly less due to the trace-preserving condition). This is only a reduction by a factor of about $\sqrt{\pi N}$, and in particular, the number of parameters is still exponentially large.

\vspace{0.5\baselineskip}

This is where imposing additional covariance conditions helps dramatically. Since the input state $\rho^{\otimes N}$ is permutation-invariant, we are free to impose \emph{permutation symmetry} on our distillation protocol as well. In particular, if we permute the input qubits, then we should get the same output qubit state. That is, for any permutation $\sigma\in S_N$,
\begin{equation}
    \mE\left(P_{\sigma}\rho_{\text{in}}P_{\sigma^{-1}}\right) = \mE\left(\rho_{\text{in}}\right),
\end{equation}
where $P_{\sigma}$ denotes the qubit permutation operator corresponding to $\sigma$. This can be understood as $S_N$-covariance, where $S_N$ takes the qubit-permuting representation on the input Hilbert space and the trivial representation on the output Hilbert space. This symmetry reduces the number of free parameters from exponentially many to only polynomially many, so it is a very substantial simplification without which this problem would be very unwieldy. In particular, if $b,b'$ are $N$-bit strings and $c,c'$ are single bits, then the $\ket{b}\bra{b'}\otimes\ket{c}\bra{c'}$ entry of the Choi matrix must be the same even if one shuffles the bits of $b$ and the bits of $b'$ according to the same permutation $\sigma\in S_N$. Even without time-translation invariance, this drops the number of free parameters all the way down to $\Theta(N^3)$. When combined with time-translation invariance, the number of free parameters falls to $\Theta(N^2)$.

\vspace{0.5\baselineskip}

However, notice that this number of parameters is still greater than the $\Theta(N)$ count we obtain by taking advantage of Schur sampling. The reason for this is that Schur sampling automatically incorporates permutation symmetry, but it \textit{also} incorporates the memoryless nature of the Schur-sampled state, in the sense that it depends only on $N_C$ and $\lambda$, and not on the original number of qubits $N$. Hence the advantages of Schur sampling cannot be fully appreciated from permutation symmetry alone.

\vspace{0.5\baselineskip}

For the equatorial case $\Theta_{\text{in}} = \Theta_{\text{out}} = \frac{\pi}{2}$, we can additionally impose \emph{negation (bit-flip) symmetry}. In particular, if we negate all of the input qubits, then we should also negate the output qubit:
\begin{equation}
    \mE\left(X^{\otimes N}\rho_{\text{in}}X^{\otimes N}\right) = X\mE\left(\rho_{\text{in}}\right)X.
\end{equation}
This can be understood as $\mathbb{Z}_2$-covariance, where the single non-identity element of $\mathbb{Z}_2$ is represented by $X^{\otimes N}$ on the input Hilbert space and $X$ on the output Hilbert space.

\vspace{0.5\baselineskip}

It is also worth commenting on the impact of \emph{negation (bit-flip) symmetry}. In Subsection \ref{subsec:preprocessing-schur-transform}, we noted that, in the equatorial case $\Theta_{\text{in}} = \Theta_{\text{out}} = \frac{\pi}{2}$, we can additionally impose negation symmetry. We also pointed out that, for the Kraus operators given by Theorem \ref{thm:optimal-protocol-kraus-rep}, the number of parameters is only cut by about half, which follows from the fact that the covariance group $\Zbb_2$ has size $2$.

\vspace{0.5\baselineskip}

The advantage of negation symmetry is thus not really about cutting down the parameter count; rather, negation symmetry has other unique consequences that substantially simplify the problem. First, as we show in Appendix \ref{appendix:boundary-value-problem}\ref{subsec:brute-force-optimization}, the boundary value problem that we need to solve incorporates the matrix entries in the computational basis of the $N_C$-qubit state that results from Schur sampling. Ordinarily, optimizing the coherence distillation protocol requires studying the entries both on the diagonal (which we call $P_{w,w}$) and one away from the diagonal (which we call $P_{w-1,w}$). However, negation symmetry simplifies the problem in such a way that the $P_{w,w}$ entries no longer matter, and we only need to care about the $P_{w-1,w}$ entries. Second, as we show in Appendix \ref{appendix:boundary-value-problem}\ref{subsec:solve-recurrence-odd-N-equatorial}, negation symmetry allows the optimal equatorial distillation protocol to be computed efficiently and exactly for any odd value of $N$. This was greatly helpful for some of the numerical computations that we carried out for large values of $N$ to notice useful patterns. Third, as we show in Appendix \ref{appendix:2nd-order-optimality}\ref{subsec:2nd-order-optimality-special-case-equatorial}, a $1^{\text{st}}$-order optimal distillation protocol that satisfies bit-flip symmetry is automatically $2^{\text{nd}}$-order optimal as well.

\appsubsec{Proof Using Choi Matrices}
{subsec:distillation-protocol-form-proof-choi-matrices}

We begin with the proof using Choi matrices. We suspect that this proof will be most appealing to those who do not yet have much comfort with TI channels and would wish to see in the most detailed way possible what conditions time-translation invariance imposes on a channel when written in the energy eigenbasis:

\begin{proof}[Proof using Choi matrices]
The time-translation invariance (TI) condition states that, for any $N$-qubit state $\rho$ and any angle $\theta\in\mathbb{R}$,
\begin{equation}
    \mE\left(\left(e^{i\theta Z}\right)^{\otimes N}\rho\left(e^{-i\theta Z}\right)^{\otimes N}\right) = e^{i\theta Z}\mE(\rho)e^{-i\theta Z}.
\end{equation}
For the sake of more generally demonstrating how time-translation invariance restricts a quantum channel, we actually start by not treating the protocol in a ``post-Schur-sampling'' sense, and rather just starting with the original $N$-qubit tensor product state. We will return to the ``post-Schur-sampling'' viewpoint shortly.

\vspace{0.5\baselineskip}

We can consider the action of $\mE$ on computational basis states. In particular, if $b,b'$ are two $N$-bit strings, then
\begin{equation}
    \left(e^{i\theta Z}\right)^{\otimes N}\ket{b}\bra{b'}\left(e^{-i\theta Z}\right)^{\otimes N} = e^{i(N-2w(b))}e^{-i(N-2w(b')}\ket{b}\bra{b'} = e^{2i(w(b')-w(b))}\ket{b}\bra{b'}.
\end{equation}
This immediately implies that
\begin{equation}
    e^{i\theta Z}\mE\left(\ket{b}\bra{b'}\right)e^{-i\theta Z} = e^{2i(w(b')-w(b))}\mE\left(\ket{b}\bra{b'}\right),
\end{equation}
which in turn restricts which entries in the $2\times 2$ matrix $\mE\left(\ket{b}\bra{b'}\right)$ are even allowed to be nonzero. In particular, we obtain
\begin{equation}
\label{eq:TI-comp-basis}
    \mE\left(\ket{b}\bra{b'}\right) = \begin{cases}
        0 & \abs{w(b') - w(b)} \ge 2 \\
        z(b,b')\ket{0}\bra{1} & w(b') - w(b) = +1 \\
        z(b,b')\ket{1}\bra{0} & w(b') - w(b) = -1 \\
        r(b,b')\ket{0}\bra{0} + (1-r(b,b'))\ket{1}\bra{1} & w(b') - w(b) = 0,
    \end{cases}
\end{equation}
where in this case, $z(b,b')$ refers to some complex number and $r(b,b')$ refers to some real number. Note that, by choosing the coefficients in the $w(b') - w(b) = 0$ case to add up to $1$, we have already enforced the trace-preserving condition on the channel.

\vspace{0.5\baselineskip}

This is a special case of a more general observation. In particular, if a channel $\mE_{A\rightarrow B}$ is time-translation invariant with input Hamiltonian $H_A$ and output Hamiltonian $H_B$, then its Choi matrix
\begin{equation}
    J\left(\mE_{A\rightarrow B}\right) = \sum_{i,j=1}^{d_A}\ket{i}\bra{j}_A\otimes\mE(\ket{i}\bra{j})_B
\end{equation}
commutes with $H_{\text{tot}} = H_A\otimes\mathbb{I}_B + \mathbb{I}_A\otimes\left(-H_B^\intercal\right)$. This is discussed in much more detail in Supplementary Note 1 of Ref. \cite{Marvian2020}. In our special case above, $H_A = \frac{\hbar\pi}{\tau}\sum_{i=1}^{N}Z_i$ and $H_B = \frac{\hbar\pi}{\tau}Z_{\text{out}}$ (where ``out'' denotes the output qubit and $\tau$ is the period of the quantum clock), so the Choi matrix $J(\mathcal{E})$ commutes with $Z_1 + Z_2 + \cdots + Z_N - Z_{\text{out}}$. This fact immediately imposes a block-diagonal structure on $J(\mE)$ that yields the same restrictions on $\mE(\ket{b}\bra{b'})$ that are shown above in Equation \ref{eq:TI-comp-basis}. 

\vspace{0.5\baselineskip}

Also, we note that in the language of \cite{marvian2014modes}, Eq.~\ref{eq:TI-comp-basis} determines which modes of asymmetry of the input are relevant. 
In particular, time-translation symmetry of the channel implies that, for a single-qubit output, only modes of asymmetry corresponding to frequencies $w(b') - w(b) = 0, \pm 1$ are relevant. 
Ref.~\cite{marvian2014modes} shows how similar arguments can be extended to channels that are covariant with respect to an arbitrary compact group, using the notion of Fourier transform over general groups.

\vspace{0.5\baselineskip}

We now return to thinking of our protocol in a post-Schur-sampling sense, meaning that we will have the channel $\mE$ act on the $N_C$-qubit symmetric subspace. The same observation as above applies, except now with the permutation-invariant Hamming weight states $\{\ket{w^{(s)}} \,|\, 0\le w\le N_C\}$ taking the place of the computational basis states $\{\ket{b} \,|\, b\in\{0,1\}^N\}$. In particular, for exactly the same reason as above,
\begin{equation}
    \mE\left(\ket{w^{(s)}}\bra{(w')^{(s)}}\right) = \begin{cases}
        0 & \abs{w'-w} \ge 2 \\
        A_w & w'-w = +1 \\
        \overline{A}_w & w'-w = -1 \\
        \cos^2\theta_w\ket{0}\bra{0} + \sin^2\theta_w\ket{1}\bra{1} & w'-w = 0,
    \end{cases}
\end{equation}
where we have already incorporated the fact that $\mE$ is trace-preserving and Hermiticity-preserving. (We have also ensured that $\mE(\ket{w^{(s)}}\bra{w^{(s)}})$ is positive semidefinite, which at least ensures that $\mE$ does not fail positivity in an ``obvious'' way.) Therefore, the Choi matrix adopts the following block-diagonal form:

\begin{equation}
    J(\mE) = \begin{array}{|cc|cc|cc|cc|cc|cc|}
        \hline
        \cos^2\theta_0 & & & A_1 & & & & & & & & \\
        & \sin^2\theta_0 & & & & & & & & & & \\
        \hline
        & & \cos^2\theta_1 & & & A_2 & & & & & & \\
        \overline{A}_1 & & & \sin^2\theta_1 & & & & & & & & \\
        \hline
        & & & & \cos^2\theta_2 & & & \ddots & & & & \\
        & & \overline{A}_2 & & & \sin^2\theta_2 & & & & & & \\
        \hline
        & & & & & & \ddots & & & A_{N_C-1} & & \\
        & & & & \ddots & & & \ddots & & & & \\
        \hline
        & & & & & & & & \cos^2\theta_{N_C-1} & & & A_{N_C} \\
        & & & & & & \overline{A}_{N_C-1} & & & \sin^2\theta_{N_C-1} & & \\
        \hline
        & & & & & & & & & & \cos^2\theta_{N_C} & \\
        & & & & & & & & \overline{A}_{N_C} & & & \sin^2\theta_{N_C} \\
        \hline
    \end{array}.
\end{equation}

In the above Choi matrix, we have already enforced the constraints that $\mE$ is trace-preserving and Hermiticity-preserving (as well as not ``obviously'' failing positivity, as we explained above). The $\theta_w$ values actually do have some significance in the context of a covariant Stinespring dilation, as we will show in Appendix \ref{appendix:deriving-kraus-rep}\ref{subsec:distillation-protocol-implementation-stinespring-dilation}, so the choice to label the diagonal entries $\cos^2\theta_w$ and $\sin^2\theta_w$ is deeper than merely ensuring that those entries add up to $1$. Also, since we have only used $\cos^2\theta_w$ and $\sin^2\theta_w$ values in the above Choi matrices, we can restrict the $\theta_w$ values to $\left[0,\frac{\pi}{2}\right]$ for convenience.

\vspace{0.5\baselineskip}

The only remaining condition on $\mE$ we need to check is complete positivity, which is equivalent to $J(\mE)$ being positive semidefinite (PSD). This matrix is block-diagonal with $2$ blocks of size $1\times 1$ and $N_C$ blocks of size $2\times 2$. As a result, it is easy to determine that $J(\mE)$ is PSD if and only if
\begin{equation}
    \abs{A_w} \le \cos\theta_{w-1}\sin\theta_w \quad (1\le w\le N_C).
\end{equation}

Now remember that our aim is to map $N$ copies of a qubit state to a single copy of another qubit state with the \emph{same} azimuthal angle on the Bloch sphere. So for example, if our input qubits have azimuthal angle $0$ (that is, they lie in the $xz$-plane of the Bloch sphere with positive $x$-coordinate), then our target state also has azimuthal angle $0$. The target state thus takes the form
\begin{equation}
    \rho_{\text{target}} = \ket{\psi_{\text{target}}}\bra{\psi_{\text{target}}} = \frac{1}{2}\begin{bmatrix}
        1+C_{\text{out}} & S_{\text{out}} \\
        S_{\text{out}} & 1-C_{\text{out}}
    \end{bmatrix},
\end{equation}
where we have defined $C_{\text{out}} \coloneqq \cos\Theta_{\text{out}}$ and $S_{\text{out}} \coloneqq \sin\Theta_{\text{out}}$ for convenience. Therefore, the fidelity of the output qubit state with the target state is
\begin{equation}
    \mF = \sum_{w=0}^{N_C}\left(\frac{1+C_{\text{out}}}{2}\cos^2\theta_w + \frac{1-C_{\text{out}}}{2}\sin^2\theta_w\right)P_{w,w} + S_{\text{out}}\sum_{w=1}^{N_C}\text{Re}[A_w]P_{w-1,w},
\end{equation}
where the $P_{w,w'}$ values are the matrix entries of the Schur-sampled state $\rho_{N_C/2}$ when written in the basis of Dicke states:
\begin{equation}
    P_{w,w'} \coloneqq \bra{w^{(s)}}\rho_{N_C/2}\ket{(w')^{(s)}}.
\end{equation}
These matrix elements are studied in detail in Appendix \ref{appendix:understanding-P-vals}.

\vspace{0.5\baselineskip}

To maximize the fidelity, we must therefore maximize each $\text{Re}[A_w]$, which means choosing
\begin{equation}
    A_w = \cos\theta_{w-1}\sin\theta_w \quad (1\le w\le N_C).
\end{equation}
Finally, we need to show that $\theta_0 = 0$ and $\theta_{N_C} = \frac{\pi}{2}$. Notice that the dependence of $\mF$ on $\theta_0$ can be written as follows:
\begin{align}
    \mF &= \left(\frac{1+C_{\text{out}}}{2}\cos^2\theta_0 + \frac{1-C_{\text{out}}}{2}\sin^2\theta_0\right)P_{00} + S_{\text{out}}\cos\theta_0\sin\theta_1P_{01} + \text{(other terms)} \\
    &= \left(\frac{1-C_{\text{out}}}{2} + C_{\text{out}}\cos^2\theta_0\right)P_{00} + S_{\text{out}}\cos\theta_0\sin\theta_1P_{01} + \text{(other terms)}.
\end{align}
When $\mF$ is written in this way, it is much more obvious that maximizing $\cos\theta_0$ is the optimal choice, irrespective of the other $\theta_w$ values. A similar logic will show that maximizing $\sin\theta_{N_C}$ is the optimal choice, irrespective of the other $\theta_w$ values. Hence, $\theta_0 = 0$ and $\theta_{N_C} = \frac{\pi}{2}$. The above Choi matrix with the above expression for the $A_w$ values and the above values for $\theta_0$ and $\theta_{N_C}$ perfectly matches the Kraus representation in the statement of Theorem \ref{thm:optimal-protocol-kraus-rep}, so we are done.
\end{proof}

\appsubsec{Proof Using Kraus Operators}
{subsec:distillation-protocol-form-proof-kraus-operators}

Let us now proceed to a proof that directly uses Kraus operators. We suspect that this proof may be more appealing to those well-versed in the theory of covariant channels and what the covariance conditions says about their Kraus operators:

\begin{proof}[Proof using Kraus operators]
It was shown by Gour \cite{gour2008resource} and again by Marvian \cite{Marvian2020} that a TI channel has a Kraus representation $\{K_{(E,\alpha)}\}$ where, in addition to the usual normalization condition
\begin{equation}
    \sum_{(E,\alpha)}K_{(E,\alpha)}^\dagger K_{(E,\alpha)} = \Ibb,
\end{equation}
each individual Kraus operator satisfies
\begin{equation}
    e^{-iH_{\text{out}}t}K_{(E,\alpha)}e^{+iH_{\text{in}}t} = e^{-iEt}K_{(E,\alpha)}.
\end{equation}
Intuitively, $E$ refers to the amount of energy that is being added to the system, and $\alpha$ is a multiplicity index, since there can be more than one Kraus operator with the same value of $E$.

\vspace{0.5\baselineskip}

Once again, we assume that we start with Schur sampling to bring us to an $N_C$-qubit state in the symmetric subspace \cite{Cirac1999}. Therefore, the energy eigenstates states on the input space are $\ket{w^{(s)}}$ for $0\le w\le N_C$, and the energy eigenstates on the output space are $\ket{0}$ and $\ket{1}$. Notice that
\begin{equation}
    e^{-iH_{\text{out}}t}\ket{b}\bra{w^{(s)}}e^{+iH_{\text{in}}t} = e^{-itZ}\ket{b}\bra{w^{(s)}}e^{-it\sum_{j=1}^{N_C}Z_j} = e^{it(w-b)}\ket{b}\bra{w^{(s)}}.
\end{equation}
Therefore, the only terms of the form $\ket{b}\bra{w^{(s)}}$ that can be put in the same Kraus operator are those with the same value of $w-b$. In other words, our Kraus operators must take the form
\begin{equation}
    K_{(w,\alpha)} = c(w-1,\alpha)\ket{0}\bra{(w-1)^{(s)}} + d(w,\alpha)\ket{1}\bra{w^{(s)}} \quad (0\le w\le N_C+1)
\end{equation}
for some complex numbers $c(w,\alpha)$ and $d(w,\alpha)$, and where we set $\ket{(-1)^{(s)}} = \ket{(N_C+1)^{(s)}} = 0$ by convention. Now notice that
\begin{align}
    \sum_{(w,\alpha)}K_{(w,\alpha)}^\dagger K_{(w,\alpha)} &= \sum_{(w,\alpha)}\Big\{\abs{c(w-1,\alpha)}^2\ket{(w-1)^{(s)}}\bra{(w-1)^{(s)}} \\
    & \quad\quad + \abs{d(w,\alpha)}^2\ket{w^{(s)}}\bra{w^{(s)}} \\
    & \quad\quad + c(w-1,\alpha)^*d(w,\alpha)\ket{(w-1)^{(s)}}\bra{w^{(s)}} \\
    & \quad\quad + c(w-1,\alpha)d(w,\alpha)^*\ket{w^{(s)}}\bra{(w-1)^{(s)}}\Big\}.
\end{align}
Therefore, the normalization condition implies that
\begin{equation}
    \sum_{(w,\alpha)}\left[\abs{c(w,\alpha)}^2 + \abs{d(w,\alpha)}^2\right] = 1 \quad (0\le w\le N_C).
\end{equation}
To narrow down $c(w,\alpha)$ and $d(w,\alpha)$ further, we now need to invoke optimality. In particular, we can directly compute that the fidelity of the output qubit state with the target state is
\begin{equation}
    \mF = \frac{1+C_{\text{out}}}{2}\sum_{(w,\alpha)}P_{w,w}\abs{c(w,\alpha)}^2 + \frac{1-C_{\text{out}}}{2}\sum_{(w,\alpha)}P_{w,w}\abs{d(w,\alpha)}^2 + S_{\text{out}}\sum_{(w,\alpha)}P_{w-1,w}\text{Re}\left[c(w-1,\alpha)^*d(w,\alpha)\right].
\end{equation}
By the Cauchy-Schwarz inequality, for any fixed $w$,
\begin{equation}
    \sum_{\alpha}\text{Re}\left[c(w-1,\alpha)^*d(w,\alpha)\right] \le \abs{\sum_{\alpha}c(w-1,\alpha)^*d(w,\alpha)} \le \sqrt{\left[\sum_{\alpha}\abs{c(w-1,\alpha)}^2\right]\left[\sum_{\alpha}\abs{d(w,\alpha)}^2\right]}.
\end{equation}
The equality case of the Cauchy-Schwarz inequality (the second inequality above) is achieved if and only if the ratio $c(w-1,\alpha):d(w,\alpha)$ is the same for all $\alpha$. However, this means that, for that value of $w$, the Kraus operators $K_{(w,\alpha)}$ are all the same up to constant factor. If a quantum channel has multiple Kraus operators $K_\alpha$ that are the same up to constant factors $b_\alpha$, in the sense that $K_\alpha = b_\alpha K_0$ for some fixed operator $K_0$, then they can all be replaced by a single Kraus operator $K_{\text{tot}} = b_{\text{tot}}K_0$, where $\abs{b_{\text{tot}}}^2 = \sum_{\alpha}\abs{b_\alpha}^2$. In other words, we are best off going multiplicity-free and having just one Kraus operator for each Hamming weight:
\begin{equation}
    K_w = c(w-1)\ket{0}\bra{(w-1)^{(s)}} + d(w)\ket{1}\bra{w^{(s)}} \quad (0\le w\le N_C+1).
\end{equation}
(Note that the values $c(-1)$ and $d(N_C+1)$ do not matter, because we set $\ket{(-1)^{(s)}} = \ket{(N_C+1)^{(s)}} = 0$ by convention.) Now, to saturate the first inequality above, we need $c(w-1)^*d(w)$ to be nonnegative real, meaning that, for each $0\le w\le N_C+1$, $c(w-1)$ and $d(w)$ have the same argument. Since we can freely multiply Kraus operators by global phase factors, we can set $c(w-1)$ and $d(w)$ to be nonnegative real. Then the normalization condition states that $c(w)^2 + d(w)^2 = 1$ for each $0\le w\le N_C+1$. Hence we can write $c(w) = \cos\theta_w$ and $d(w) = \sin\theta_w$ for some angles $0\le\theta_w\le\frac{\pi}{2}$. Finally, when we plug all of these values into the fidelity formula, we find exactly the same fidelity $\mF$ as we did in Appendix \ref{appendix:deriving-kraus-rep}\ref{subsec:distillation-protocol-form-proof-choi-matrices} (of course we must, since it is the same channel). Therefore, by the same logic, we conclude that $\theta_0 = 0$ and $\theta_{N_C} = \frac{\pi}{2}$, which finishes the proof.
\end{proof}

\appsubsec{Implementation of Coherence Distillation via Covariant Stinespring Dilation}
{subsec:distillation-protocol-implementation-stinespring-dilation}

Finally, we show that a quantum channel of the form shown in Theorem \ref{thm:optimal-protocol-kraus-rep} has an especially nice Stinespring dilation. We suspect that this may be of interest to anyone wishing to actually implement coherence distillation in practice.

\vspace{0.5\baselineskip}

The covariant Stinespring dilation theorem \cite{marvian2012symmetry, keyl1999optimal} states that any TI channel has a Stinespring dilation where the ancillary system is initialized in an incoherent state and the unitary operators applied to the joint system are all energy-conserving unitaries with respect to the joint Hamiltonian $H_{\text{in}}\otimes\Ibb_{\text{anc}} + \Ibb_{\text{in}}\otimes H_{\text{anc}}$.

\vspace{0.5\baselineskip}

A general quantum channel acting on $N$ qubits may require up to $N$ ancillary qubits to implement via a Stinespring dilation. However, the form of Kraus operators shown in Theorem \ref{thm:optimal-protocol-kraus-rep} will imply that we actually need only $1$ ancilla qubit. The crucial point is that, as we showed in Appendix \ref{appendix:deriving-kraus-rep}\ref{subsec:distillation-protocol-form-proof-kraus-operators}, although a general TI channel would require multiple Kraus operators $K_{(w,\alpha)}$ for each Hamming weight $w$, an optimal coherence distillation protocol will only require one Kraus operator $K_w$ for each Hamming weight $w$.

\vspace{0.5\baselineskip}

Without loss of generality, we can assume that the ancilla qubit has Hamiltonian $H_{\text{anc}} = Z$. We can initialize it in any incoherent state, namely $p\ket{0}\bra{0} + (1-p)\ket{1}\bra{1}$ for some $0\le p\le 1$. For convenience, we will initialize it in the state $\ket{0}\bra{0}$. In this case, the total energy (equivalently, Hamming weight) eigenspaces are as follows:
\begin{itemize}
    \item $\text{span}\{\ket{0^{(s)}}\ket{0}\}$ (total Hamming weight $0$)
    \item $\text{span}\{\ket{w^{(s)}}\ket{1},\ket{w^{(s)}}\ket{0}\}$ for $1\le w\le N_C$ (total Hamming weight $w$);
    \item $\text{span}\{\ket{N_C^{(s)}}\ket{1}\}$ (total Hamming weight $N_C+1$).
\end{itemize}
The total energy-conserving unitary can be understood as the direct sum of an energy-conserving unitary on each total energy eigenspace:
\begin{equation}
    U_{\text{sys,anc}} = \bigoplus_{w=0}^{N_C+1}U^{(w)},
\end{equation}
where $U^{(0)} = [u^{(0)}_{00}]$ and $U^{(N_C+1)} = [u^{(N_C+1)}_{11}]$ are $1\times 1$ unitaries, and $U^{(w)} = [u^{(w)}_{ij}]$ for $1\le w\le N_C$ are $2\times 2$ unitaries.

\vspace{0.5\baselineskip}

As a reminder, we use $P_{w,w'}$ to denote the entries of the Schur-sampled state $\rho_{N_C/2}$ in the computational basis. Therefore, we can write the joint state before the unitary as follows:
\begin{equation}
    \sum_{w,w'=0}^{N_C}P_{w,w'}\ket{w^{(s)}}\bra{(w')^{(s)}}\otimes\ket{0}\bra{0}.
\end{equation}
Then, after applying the joint unitary $U_{\text{sys,anc}}$, the resulting joint state is
\begin{equation}
    \sum_{w,w'=0}^{N_C}P_{w,w'}\left(u^{(w)}_{00}\ket{w^{(s)}}\ket{0} + u^{(w)}_{10}\ket{(w-1)^{(s)}}\ket{1}\right)\left(\overline{u}^{(w')}_{00}\bra{(w')^{(s)}}\bra{0} + \overline{u}^{(w')}_{10}\bra{(w'-1)^{(s)}}\bra{1}\right),
\end{equation}
where we set $\ket{(-1)^{(s)}}=0$ by convention.

\vspace{0.5\baselineskip}

When we trace out the original system and return the ancilla qubit, the state we get is the sum of all the $2\times 2$ blocks on the diagonal of the above matrix. Therefore, the final output state is
\begin{equation}
    \sum_{w=0}^{N_C}P_{w,w}\left[\abs{u^{(w)}_{00}}^2\ket{0}\bra{0} + \abs{u^{(w)}_{10}}^2\ket{1}\bra{1}\right] + \sum_{w=1}^{N_C}P_{w-1,w}\left[u^{(w-1)}_{00}\overline{u}^{(w)}_{10}\ket{0}\bra{1} + \overline{u}^{(w-1)}_{00}u^{(w)}_{10}\ket{1}\bra{0}\right],
\end{equation}
where we set the ``out of bounds'' value $u^{(0)}_{10}$ to equal zero by convention. Therefore, the fidelity with the target state is
\begin{equation}
    \mF = \sum_{w=0}^{N_C}P_{w,w}\left[\frac{1+C_{\text{out}}}{2}\abs{u^{(w)}_{00}}^2 + \frac{1-C_{\text{out}}}{2}\abs{u^{(w)}_{10}}^2\ket{1}\bra{1}\right] + S_{\text{out}}\sum_{w=1}^{N_C}P_{w-1,w}\text{Re}\left[u^{(w-1)}_{00}\overline{u}^{(w)}_{10}\right].
\end{equation}
Once again, there is no reason not to choose $u^{(w)}_{00}$ and $u^{(w)}_{10}$ to be nonnegative real numbers. Also, in order for $U^{(w)}$ to be unitary, these two numbers must have their squared magnitudes add up to $1$. Hence, we can say $u^{(w)}_{00} = \cos\theta_w$ and $u^{(w)}_{10} = \sin\theta_w$ for some $0\le\theta_w\le\frac{\pi}{2}$. Once we set these values, we can conclude that the channel has Kraus operators
\begin{equation}
    K_w = \cos\theta_{w-1}\ket{0}\bra{(w-1)^{(s)}} + \sin\theta_w\ket{1}\bra{w^{(s)}} \quad (0\le w\le N_C).
\end{equation}
This is exactly the form we want. Finally, when we plug all of these values into the fidelity formula, we find exactly the same formula as we did in Appendix \ref{appendix:deriving-kraus-rep}\ref{subsec:distillation-protocol-form-proof-choi-matrices} (of course we must, since it is the same channel). Therefore, by the same logic, we conclude that $\theta_0 = 0$ and $\theta_{N_C} = \frac{\pi}{2}$.

\vspace{0.5\baselineskip}

In the two proofs of Theorem \ref{thm:optimal-protocol-kraus-rep} shown in Appendices \ref{appendix:deriving-kraus-rep}\ref{subsec:distillation-protocol-form-proof-choi-matrices} and \ref{appendix:deriving-kraus-rep}\ref{subsec:distillation-protocol-form-proof-kraus-operators}, the $\theta_w$ values were basically just abstractions so that we could define quantities $\cos\theta_w$ and $\sin\theta_w$ whose squares add up to $1$. However, the covariant Stinespring dilation above shows us a meaningful interpretation for the $\theta_w$ values. Because we initialized the ancilla in the state $\ket{0}\bra{0}$, we never had to worry about the entries $u^{(w)}_{01}$ and $u^{(w)}_{11}$ in the energy-conserving unitaries $U^{(w)}$. At this point, we can choose them to be anything that ensures that $U^{(w)}$ is unitary, meaning that $u^{(w)}_{01} = -e^{i\phi_w}\sin\theta_w$ and $u^{(w)}_{11} = e^{i\phi_w}\cos\theta_w$ for some angles $\phi_w$. For convenience, we can choose $\phi_w$ across the board, meaning that
\begin{align}
    U^{(0)} &= \begin{bmatrix} \cos\theta_0 \end{bmatrix} = \begin{bmatrix} 1 \end{bmatrix} \\
    U^{(w)} &= \begin{bmatrix}
        \cos\theta_w & -\sin\theta_w \\
        \sin\theta_w & \cos\theta_w
    \end{bmatrix} \quad (1\le w\le N_C).
\end{align}
(We never solved for $U^{(N_C+1)}$, but it does not matter, because our choice of ancilla initialization means that the joint state before the unitary has no presence in Hamming weight $N_C+1$.) In other words, each $U^{(w)}$ is just a $2\times 2$ rotation matrix on the $2$-dimensional energy eigenspace with total Hamming weight $w$, and $\theta_w$ is nothing but the angle of that rotation.

\vspace{0.5\baselineskip}

Let us show this more viscerally by writing out the matrices explicitly for the very simple example $N_C=3$. Before applying the energy-conserving unitaries, the joint state looks as follows:
\begin{equation}
    \begin{array}{|cc|cc|cc|cc|}
        \hline
        P_{00} & & P_{01} & & P_{02} & & P_{03} & \\
        & & & & & & & \\
        \hline
        P_{10} & & P_{11} & & P_{12} & & P_{13} & \\
        & & & & & & & \\
        \hline
        P_{20} & & P_{21} & & P_{22} & & P_{23} & \\
        & & & & & & & \\
        \hline
        P_{30} & & P_{31} & & P_{32} & & P_{33} & \\
        & & & & & & & \\
        \hline
    \end{array}.
\end{equation}
And after applying the energy-conserving unitaries, the joint state looks as follows:
\begin{equation}
    \begin{array}{|cc|cc|cc|cc|}
        \hline
        P_{00}\cos^2\theta_0 & P_{01}\cos\theta_0\sin\theta_1 & P_{01}\cos\theta_0\cos\theta_1 & P_{02}\cos\theta_0\sin\theta_2 & P_{02}\cos\theta_0\cos\theta_2 & P_{03}\cos\theta_0\sin\theta_3 & P_{03}\cos\theta_0\cos\theta_3 & \\
        P_{10}\sin\theta_1\cos\theta_0 & P_{11}\sin^2\theta_1 & P_{12}\sin\theta_1\sin\theta_2 & P_{12}\sin\theta_1\cos\theta_2 & P_{13}\sin\theta_1\sin\theta_3 & P_{13}\sin\theta_1\cos\theta_3 & \\
        \hline
        P_{10}\cos\theta_1\cos\theta_0 & P_{11}\cos\theta_1\sin\theta_1 & P_{11}\cos^2\cos\theta_1 & P_{12}\cos\theta_1\sin\theta_2 & P_{12}\cos\theta_1\cos\theta_2 & P_{13}\cos\theta_1\sin\theta_3 & P_{13}\cos\theta_1\cos\theta_3 & \\
        P_{20}\sin\theta_2\cos\theta_0 & P_{21}\sin\theta_2\sin\theta_1 & P_{21}\sin\theta_2\cos\theta_1 & P_{22}\sin^2\theta_2 & P_{22}\sin\theta_2\cos\theta_2 & P_{23}\sin\theta_2\sin\theta_3 & P_{23}\sin\theta_2\cos\theta_3 & \\
        \hline
        P_{20}\cos\theta_2\cos\theta_0 & P_{21}\cos\theta_2\sin\theta_1 & P_{21}\cos\theta_2\cos\theta_1 & P_{22}\cos\theta_2\sin\theta_2 & P_{22}\cos^2\theta_2 & P_{23}\cos\theta_2\sin\theta_3 & P_{23}\cos\theta_2\cos\theta_3 & \\
        P_{30}\sin\theta_3\cos\theta_0 & P_{31}\sin\theta_3\sin\theta_1 & P_{31}\sin\theta_3\cos\theta_1 & P_{32}\sin\theta_3\sin\theta_2 & P_{32}\sin\theta_3\cos\theta_2 & P_{33}\sin^2\theta_3 & P_{33}\sin\theta_3\cos\theta_3 & \\
        \hline
        P_{30}\cos\theta_3\cos\theta_0 & P_{31}\cos\theta_3\sin\theta_1 & P_{31}\cos\theta_3\cos\theta_1 & P_{32}\cos\theta_3\sin\theta_2 & P_{32}\cos\theta_3\cos\theta_2 & P_{33}\cos\theta_3\sin\theta_3 & P_{33}\cos^2\theta_3 & \\
        & & & & & & & \\
        \hline
    \end{array}.
\end{equation}
When we trace out the original $N_C$-qubit system and return the ancilla qubit, we can see that the diagonal blocks add up to exactly what we stated in the proof. Furthermore, we can more readily see from that final state that the Kraus operators have the form shown in Theorem \ref{thm:optimal-protocol-kraus-rep}.

%% file: app_lower_bound_saturation.tex
\appsec{Asymptotic Saturation of the Purity-of-Coherence Bound}
{appendix:lower-bound-saturation}

In Appendix \ref{appendix:infidelity-lower-bound} (see Theorem \ref{thm:infidelity-factor-lower-bound-qubit}), we showed that any coherence distillation protocol $\mE$ must have infidelity factor
\begin{equation}
    \delta_1(\mE) \ge \frac{1-\lambda^2}{4\lambda^2}\frac{\sin^2\Theta_{\text{out}}}{\sin^2\Theta_{\text{in}}}.
\end{equation}
In this appendix, we will show an explicit example of a protocol that saturates this bound, thereby confirming that purity of coherence (RLD Fisher information) sets the tightest possible leading-order bound on the achievable infidelity of single-shot qubit coherence distillation.

\vspace{0.5\baselineskip}

In Appendix \ref{appendix:deriving-kraus-rep}, we showed that an optimal coherence distillation protocol can be characterized with a simple set of Kraus operators. In particular, we begin by performing Schur sampling to obtain a state on the symmetric subspace of $N_C$ qubits. We then apply a channel with Kraus operators
\begin{equation}
    K_w = \cos\theta_{w-1}\ket{0}\bra{(w-1)^{(s)}} + \sin\theta_{w-1}\ket{1}\bra{w^{(s)}} \quad (1\le w\le N_C),
\end{equation}
where $\theta_0 = 0$ and $\theta_{N_C} = \frac{\pi}{2}$. Finding the optimal coherence distillation protocol can now be reduced to the question of optimizing the $\theta_w$ values for $1\le w\le N_C-1$.

\vspace{0.5\baselineskip}

For convenience, define the following values:
\begin{equation}
    C_{\text{out}} \coloneqq \cos\Theta_{\text{out}}, \quad\quad S_{\text{out}} \coloneqq \sin\Theta_{\text{out}}, \quad\quad C_{\text{in}} \coloneqq \cos\Theta_{\text{in}}, \quad\quad S_{\text{in}} \coloneqq \sin\Theta_{\text{in}}.
\end{equation}
In this appendix, we will show that the $\theta_w$ values satisfying
\begin{equation}
    \boxed{\sin^2\theta_w = \frac{1-C_{\text{out}}}{2} + \lambda\frac{S_{\text{out}}^2}{S_{\text{in}}^2}\left(\frac{w}{N_C} - \frac{1-C_{\text{in}}}{2}\right)}
\end{equation}
will produce a protocol that achieves the smallest possible infidelity factor $\delta_1(\mE) = \frac{1-\lambda^2}{4\lambda^2}\frac{S_{\text{out}}^2}{S_{\text{in}}^2}$.

\vspace{0.5\baselineskip}

This choice of protocol may seem very unmotivated, but in fact it was found through a systematic optimization process. We encourage the curious reader to check out Appendices \ref{appendix:asymptotic-expansion} and \ref{appendix:1st-order-optimality} to learn how we came up with this $1^{\text{st}}$-order optimal protocol.

\vspace{0.5\baselineskip}

We first define the following two variables, which are just linear functions of $w$:
\begin{align}
    z_0 &\coloneqq \frac{w}{N_C} - \frac{1-C_{\text{in}}}{2} \\
    z_1 &\coloneqq \frac{w-1/2}{N_C} - \frac{1-C_{\text{in}}}{2} = z_0 - \frac{1}{2N_C}.
\end{align}
We now use these to define another collection of values:
\begin{align}
    \mM_p^{(0)} &\coloneqq \sum_{w=0}^{N_C}z_0^pP_{w,w} \\
    \mM_p^{(1)} &\coloneqq \sum_{w=0}^{N_C}z_1^pP_{w-1,w}.
\end{align}
The motivation for defining these quantities is that the $P_{w,w}$ values define a probability distribution that is centered at $w \approx \frac{1-C_{\text{in}}}{2}N_C$. Therefore, the $\mM_p^{(0)}$ values can be intuitively interpreted as the centered moments of this distribution. A similar fact is true for the $P_{w-1,w}$ values, with the exception that this ``distribution'' is not quite normalized, since the $P_{w-1,w}$ values add up to slightly less than $1$. It is also worth mentioning that we define $z_1$ to be slightly different from $z_0$ (in particular, replacing $w$ with $w-\frac{1}{2}$) to account for the fact that $P_{w-1,w}$ has indices $w-1$ and $w$, so $w-\frac{1}{2}$ is a suitable way to ``balance'' these indices.

\vspace{0.5\baselineskip}

To prove the $1^{\text{st}}$-order optimality of the above protocol, we need the following result:

\begin{theorem}[Moments of the $P_{w,w}$ and $P_{w-1,w}$ Values, up to $1^{\text{st}}$-order precision]
\label{thm:P-vals-centered-moments-1st-order}
Up to $1^{\text{st}}$-order precision in $N_C^{-1}$, the moments of the $P_{w,w}$ values look as follows:
\begin{align}
    \mM_0^{(0)} &= 1 \\
    \mM_2^{(0)} &= \frac{S_{\text{in}}^2}{4\lambda}\frac{1}{N_C} + O\left(N_C^{-2}\right) \\
    \mM_p^{(0)} &= O\left(N_C^{-2}\right) \quad (p\neq 0,2).
\end{align}
Furthermore, up to $1^{\text{st}}$-order precision in $N_C^{-1}$, the moments of the $P_{w-1,w}$ values look as follows:
\begin{align}
    \mM_0^{(1)} &= 1 - \frac{1}{2\lambda S_{\text{in}}^2}\frac{1}{N_C} + O\left(N_C^{-2}\right) \\
    \mM_2^{(1)} &= \frac{S_{\text{in}}^2}{4\lambda}\frac{1}{N_C} + O\left(N_C^{-2}\right) \\
    \mM_p^{(1)} &= O\left(N_C^{-2}\right) \quad (p\neq 0,2).
\end{align}
\end{theorem}

We prove this result for the $\mM_p^{(1)}$ values in Appendix \ref{appendix:understanding-P-vals}\ref{subsec:P-vals-moments-offset0}, and we prove this result for the $\mM_p^{(1)}$ values in Appendix \ref{appendix:understanding-P-vals}\ref{subsec:P-vals-moments-offset1-general}.

\vspace{0.5\baselineskip}

These results can be understood intuitively as a consequence of the fact that the ``distributions'' of the $P_{w,w}$ and $P_{w-1,w}$ values are approximately Gaussian, both with mean $\sim\frac{1-C_{\text{in}}}{2}N_C$ and variance $\sim\frac{S_{\text{in}}^2}{4\lambda}N_C$, but with slightly different amplitudes. However, we do not actually need the Gaussian character of these ``distributions'' to prove asymptotic optimality. We only need to know that the errors in the above moment values are all $O\left(N_C^{-2}\right)$.

\vspace{0.5\baselineskip}

Let us begin by writing out the fidelity achieved by a general protocol with the form shown in Theorem \ref{thm:optimal-protocol-kraus-rep}. As we show in the different proofs of this theorem in Appendix \ref{appendix:deriving-kraus-rep}, the fidelity takes the form
\begin{equation}
    \mF = \frac{1+C_{\text{out}}}{2} - C_{\text{out}}\sum_{w=0}^{N_C}P_{w,w}\sin^2\theta_w + S_{\text{out}}\sum_{w=1}^{N_C}P_{w-1,w}\cos\theta_{w-1}\sin\theta_w.
\end{equation}
We now use our $\sin^2\theta_w$ ansatz and Theorem \ref{thm:P-vals-centered-moments-1st-order} to evaluate the first summation:
\begin{align}
    \sum_{w=0}^{N_C}P_{w,w}\sin^2\theta_w &= \sum_{w=0}^{N_C}P_{w,w}\left(\frac{1-C_{\text{out}}}{2} + \lambda\frac{S_{\text{out}}^2}{S_{\text{in}}^2}z_0\right) \\
    &= \frac{1-C_{\text{out}}}{2}\mM_0^{(0)} + \lambda\frac{S_{\text{out}}^2}{S_{\text{in}}^2}\mM_1^{(0)} \\
    &= \frac{1-C_{\text{out}}}{2}\left[1\right] + \lambda\frac{S_{\text{out}}^2}{S_{\text{in}}^2}\left[O\left(N_C^{-2}\right)\right] \\
    &= \frac{1-C_{\text{out}}}{2} + O\left(N_C^{-2}\right).
\end{align}
Now we need to evaluate the second summation. This will require a bit more effort, because we need to use the protocol ansatz to compute the quantity $\cos\theta_{w-1}\sin\theta_w$. First, we compute $\cos^2\theta_{w-1}$:
\begin{align}
    \cos^2\theta_{w-1} &= 1 - \sin^2\theta_{w-1} \\
    &= 1 - \left[\frac{1-C_{\text{out}}}{2} + \lambda\frac{S_{\text{out}}^2}{S_{\text{in}}^2}\left(\frac{w-1}{N_C} - \frac{1-C_{\text{in}}}{2}\right)\right] \\
    &= \frac{1+C_{\text{out}}}{2} - \lambda\frac{S_{\text{out}}^2}{S_{\text{in}}^2}\left(z_1 - \frac{1}{2N_C}\right).
\end{align}
Next, we compute $\cos^2\theta_{w-1}\sin^2\theta_w$:
\begin{align}
    \cos^2\theta_{w-1}\sin^2\theta_w &= \left[\frac{1+C_{\text{out}}}{2} - \lambda\frac{S_{\text{out}}^2}{S_{\text{in}}^2}\left(z_1 - \frac{1}{2N_C}\right)\right]\left[\frac{1-C_{\text{out}}}{2} + \lambda\frac{S_{\text{out}}^2}{S_{\text{in}}^2}\left(z_1 + \frac{1}{2N_C}\right)\right] \\
    &= \frac{S_{\text{out}}^2}{4} + \frac{\lambda C_{\text{out}}S_{\text{out}}^2}{S_{\text{in}}^2}z + \frac{\lambda S_{\text{out}}^2}{2S_{\text{in}}^2}\frac{1}{N_C} - \frac{\lambda^2S_{\text{out}}^4}{S_{\text{in}}^4}z^2 + \cdots \\
    &= S_{\text{out}}^2\left[\frac{1}{4} + \frac{\lambda C_{\text{out}}}{S_{\text{in}}^2}z + \frac{\lambda}{2S_{\text{in}}^2}\frac{1}{N_C} - \frac{\lambda^2S_{\text{out}}^2}{S_{\text{in}}^4}z^2 + \cdots\right]
\end{align}
The above expression has terms with $z_1/N_C$ and $1/N_C^2$, but we can ignore these terms. More generally, any term that equals $z^pN_C^{-q}$ (where $z$ can refer to $z_0$ or $z_1$) times a constant factor (i.e., a factor depending only on $\lambda$, $\Theta_{\text{in}}$, $\Theta_{\text{out}}$) can be ignored as long as $p+2q\ge 3$, because then it will not affect the $1^{\text{st}}$-order fidelity. Intuitively, the typical values of $z$ are $\Theta\left(N_C^{-1/2}\right)$, so a term with $z^pN_C^{-q}$ only contributes at the order of $N_C^{-(p+2q)/2}$. (In fact, even if $p+2q=3$, the moments in Theorem \ref{thm:P-vals-centered-moments-1st-order} actually imply that such a term will only contribute at the order of $N_C^{-2}$, as opposed to $N_C^{-3/2}$.)

\vspace{0.5\baselineskip}

Finally, we take the square root of the above quantity to obtain $\cos\theta_{w-1}\sin\theta_w$. We use the following Taylor series:
\begin{equation}
    \sqrt{\frac{1}{4} + \varepsilon} = \frac{1}{2} + \varepsilon - \varepsilon^2 + O\left(\varepsilon^3\right)
\end{equation}
Applying this Taylor series, and again truncating terms that will not contribute to the $1^{\text{st}}$-order infidelity, we obtain
\begin{align}
    \cos\theta_{w-1}\sin\theta_w &= S_{\text{out}}\sqrt{\frac{1}{4} + \frac{\lambda C_{\text{out}}}{S_{\text{in}}^2}z + \frac{\lambda}{2S_{\text{in}}^2}\frac{1}{N_C} - \frac{\lambda^2S_{\text{out}}^2}{S_{\text{in}}^4}z^2 + \cdots} \\
    &= S_{\text{out}}\Bigg\{\frac{1}{2} + \left[\frac{\lambda C_{\text{out}}}{S_{\text{in}}^2}z + \frac{\lambda}{2S_{\text{in}}^2}\frac{1}{N_C} - \frac{\lambda^2S_{\text{out}}^2}{S_{\text{in}}^4}z^2 + \cdots\right] \\
    & \quad - \left[\frac{\lambda C_{\text{out}}}{S_{\text{in}}^2}z + \frac{\lambda}{2S_{\text{in}}^2}\frac{1}{N_C} - \frac{\lambda^2S_{\text{out}}^2}{S_{\text{in}}^4}z^2 + \cdots\right]^2 + \cdots\Bigg\} \\
    &= S_{\text{out}}\Bigg\{\frac{1}{2} + \left[\frac{\lambda C_{\text{out}}}{S_{\text{in}}^2}z + \frac{\lambda}{2S_{\text{in}}^2}\frac{1}{N_C} - \frac{\lambda^2S_{\text{out}}^2}{S_{\text{in}}^4}z^2 + \cdots\right] \\
    & \quad - \left[\frac{\lambda^2C_{\text{out}}^2}{S_{\text{in}}^4}z^2 + \cdots\right] + \cdots\Bigg\} \\
    &= S_{\text{out}}\Bigg\{\frac{1}{2} + \frac{\lambda C_{\text{out}}}{S_{\text{in}}^2}z + \frac{\lambda}{2S_{\text{in}}^2}\frac{1}{N_C} - \frac{\lambda^2}{S_{\text{in}}^4}z^2 + \cdots\Bigg\}.
\end{align}
We can now finally compute the second summation in the fidelity formula:
\begin{align}
    \sum_{w=0}^{N_C}P_{w,w}\cos\theta_{w-1}\sin\theta_w &= S_{\text{out}}\sum_{w=0}^{N_C}P_{w-1,w}\Bigg\{\frac{1}{2} + \frac{\lambda C_{\text{out}}}{S_{\text{in}}^2}z + \frac{\lambda}{2S_{\text{in}}^2}\frac{1}{N_C} - \frac{\lambda^2}{S_{\text{in}}^4}z^2 + \cdots\Bigg\} \\
    &= S_{\text{out}}\Bigg\{\left(\frac{1}{2} + \frac{\lambda}{2S_{\text{in}}^2}\frac{1}{N_C}\right)\mM_0^{(1)} + \left(\frac{\lambda C_{\text{out}}}{S_{\text{in}}^2}\right)\mM_1^{(1)} - \frac{\lambda^2}{S_{\text{in}}^4}\mM_2^{(0)} + O\left(N_C^{-2}\right)\Bigg\} \\
    &= S_{\text{out}}\Bigg\{\left(\frac{1}{2} + \frac{\lambda}{2S_{\text{in}}^2}\frac{1}{N_C}\right)\left[1 - \frac{1}{2\lambda S_{\text{in}}^2}\frac{1}{N_C} + O\left(N_C^{-2}\right)\right] \\
    & \quad + \left(\frac{\lambda C_{\text{out}}}{S_{\text{in}}^2}\right)\left[O\left(N_C^{-2}\right)\right] - \frac{\lambda^2}{S_{\text{in}}^4}\left[\frac{S_{\text{in}}^2}{4\lambda}\frac{1}{N_C} + O\left(N_C^{-2}\right)\right] + O\left(N_C^{-2}\right)\Bigg\} \\
    &= S_{\text{out}}\Bigg\{\frac{1}{2} + \left[\frac{\lambda}{2S_{\text{in}}^2} - \frac{1}{4\lambda S_{\text{in}}^2} - \frac{\lambda}{4S_{\text{in}}^2}\right]\frac{1}{N_C} + O\left(N_C^{-2}\right)\Bigg\} \\
    &= S_{\text{out}}\Bigg\{\frac{1}{2} - \frac{1-\lambda^2}{4\lambda S_{\text{in}}^2}\frac{1}{N_C} + O\left(N_C^{-2}\right)\Bigg\}.
\end{align}
We now put everything together and compute the fidelity of the protocol in question:
\begin{align}
    \mF &= \frac{1+C_{\text{out}}}{2} - C_{\text{out}}\sum_{w=0}^{N_C}P_{w,w}\sin^2\theta_w + S_{\text{out}}\sum_{w=1}^{N_C}P_{w-1,w}\cos\theta_{w-1}\sin\theta_w \\
    &= \frac{1+C_{\text{out}}}{2} - C_{\text{out}}\Bigg\{\frac{1-C_{\text{out}}}{2} + O\left(N_C^{-2}\right)\Bigg\} + S_{\text{out}}^2\Bigg\{\frac{1}{2} - \frac{1-\lambda^2}{4\lambda S_{\text{in}}^2}\frac{1}{N_C} + O\left(N_C^{-2}\right)\Bigg\} \\
    &= 1 - \frac{1-\lambda^2}{4\lambda}\frac{S_{\text{out}}^2}{S_{\text{in}}^2}\frac{1}{N_C} + O\left(N_C^{-2}\right).
\end{align}
Therefore, the infidelity of this protocol, which is just one minus fidelity, takes the form
\begin{equation}
    \mI(\mE_N) = \frac{1-\lambda^2}{4\lambda}\frac{S_{\text{out}}^2}{S_{\text{in}}^2}\frac{1}{N_C} + O\left(N_C^{-2}\right).
\end{equation}

\vspace{0.5\baselineskip}

This infidelity is written in terms of $N_C$, the number of qubits following the initial Schur sampling step; but we would like it to be written in terms of $N$, the original number of qubits. To this end, we need to know something about the distribution of $N_C$, or in other words, the distribution of the total angular momentum $j=N_C/2$ when a total angular momentum measurement is applied to $\rho^{\otimes N}$ for a qubit state $\rho$. In Appendix \ref{appendix:angular-momentum-moments}, we study the distribution of $N_C$ in great detail. However, for our purposes here, it is enough to use the following: 
\begin{align}
    \Ebb\left[N_C^{-1}\right] &= \frac{1}{\lambda}\frac{1}{N} + O\left(N^{-2}\right) \\
    \Ebb\left[N_C^{-p}\right] &= O\left(N^{-p}\right) \quad (p\ge 2).
\end{align}
In fact, at this level of precision, these results follow from the fact that $N_C$ is also asymptotically Gaussian with mean $\sim\lambda N$ and variance $\sim(1-\lambda^2)N$, a fact which has been known for many years \cite{Keyl2001}.

\vspace{0.5\baselineskip}

Applying these results to the above formula for $\mI(\mE)$ immediately tells us that the average infidelity is
\begin{equation}
    \mI(\mE_N) = \frac{1-\lambda^2}{4\lambda^2}\frac{S_{\text{out}}^2}{S_{\text{in}}^2}\frac{1}{N} + O\left(N^{-2}\right),
\end{equation}
which means that the infidelity factor is
\begin{equation}
    \delta_1(\mE) = \frac{1-\lambda^2}{4\lambda^2}\frac{S_{\text{out}}^2}{S_{\text{in}}^2}.
\end{equation}
We conclude that our protocol ansatz saturates the lower bound proved in Theorem \ref{thm:infidelity-factor-lower-bound-qubit}. This confirms that purity of coherence has an operational interpretation as a time-asymmetry resource whose monotonicity under TI channels (derived from the monotonicity of RLD Fisher information for general quantum channels) sets the tightest possible leading-order bound on the infidelity of single-shot qubit coherence distillation.

%% file: app_boundary_value_problem.tex
\appsec{Finding the Optimal Distillation Protocol by Solving a Boundary Value Problem}
{appendix:boundary-value-problem}

In this appendix, we use the form of the Kraus operators proved in Appendix \ref{appendix:deriving-kraus-rep} to compute the fidelity and then set up the problem of maximizing the fidelity. The result ends up being a boundary value problem, as we show in Appendix \ref{appendix:boundary-value-problem}\ref{subsec:brute-force-optimization}. This boundary value problem can be solved numerically in general, and furthermore, it can also be solved exactly in some special cases, as we show in Appendices \ref{appendix:boundary-value-problem}\ref{subsec:solve-recurrence-pure-input-matching-target} and \ref{appendix:boundary-value-problem}\ref{subsec:solve-recurrence-odd-N-equatorial}.

\appsubsec{Brute-Force Optimization}
{subsec:brute-force-optimization}

We begin by showing that straightforward multivariable calculus can be used to optimize an arbitrary instance of the single-shot qubit coherence distillation problem, with the result taking the form of a boundary value problem.

\vspace{0.5\baselineskip}

We must first write out the form of the Schur-sampled state, and to do so nicely, we must define some additional notation. The computational basis states for a specific direction $\hat{n}$ are defined so that $\ket{0_{\hat{n}}}$ is a $-1$ eigenstate of $\hat{n}\cdot\vec{\sigma}$, while $\ket{1_{\hat{n}}}$ is a $+1$ eigenstate of $\hat{n}\cdot\vec{\sigma}$. (This convention may seem backward, but it is done to keep in line with the convention used by Cirac et al. \cite{Cirac1999}.) In particular, let $\hat{n} = \langle\sin\Theta,0,\cos\Theta\rangle$ for some $0\le\Theta\le\pi$ be some unit vector with azimuthal angle $\phi=0$. Then the computational basis states for direction $\hat{n}$ are
\begin{align}
    \ket{0_{\hat{n}}} &= \sin\frac{\Theta}{2}\ket{0} - \cos\frac{\Theta}{2}\ket{1} \\
    \ket{1_{\hat{n}}} &= \cos\frac{\Theta}{2}\ket{0} + \sin\frac{\Theta}{2}\ket{1}.
\end{align}
Once these single-qubit states are defined for an arbitrary direction $\hat{n}$, the Dicke states for direction $\hat{n}$, which we denote $\ket{k_{\hat{n}}^{(s)}}$, are defined completely analogously to the usual Dicke states. In particular, $\ket{k_{\hat{n}}^{(s)}}$ is the uniform superposition of all $\binom{N_C}{k}$ distinct permutations of $k$ qubits in the state $\ket{1_{\hat{n}}}$ and $N_C-k$ qubits in the state $\ket{0_{\hat{n}}}$.

\vspace{0.5\baselineskip}

Using this notation, we can write out the state we have following Schur sampling \cite{Cirac1999}, which is denoted $\rho_j$ for $j = N_C/2$:
\begin{equation}
    \rho_{N_C/2} = \frac{c_1-c_0}{c_1^{N_C+1} - c_0^{N_C+1}}\sum_{k=0}^{N_C}c_1^kc_0^{N_C-k}\ket{k_{\hat{n}}^{(s)}}\bra{k_{\hat{n}}^{(s)}}.
\end{equation}
Next, we write out the Kraus representation of our protocol, as shown in Theorem \ref{thm:optimal-protocol-kraus-rep}:
\begin{equation}
    K_w = \cos\theta_{w-1}\ket{0}\bra{(w-1)^{(s)}} + \sin\theta_w\ket{1}\bra{w^{(s)}} \quad (1\le w\le N_C).
\end{equation}
Therefore, the output of the protocol is as follows:
\begin{align}
    \mE\left(\rho_{N_C/2}\right) &= \sum_{w=1}^{N_C}K_w\left(\rho_{N_C/2}\right)K_w^\dagger \\
    &= \left[\sum_{w=1}^{N_C}P_{w-1,w-1}\cos^2\theta_{w-1}\right]\ket{0}\bra{0} \\
    & \quad + \left[\sum_{w=1}^{N_C}P_{w,w}\sin^2\theta_w\right]\ket{1}\bra{1} \\
    & \quad + \left[\sum_{w=1}^{N_C}P_{w-1,w}\cos\theta_{w-1}\sin\theta_w\right]\left(\ket{0}\bra{1} + \ket{1}\bra{0}\right).
\end{align}
Our target state is the following:
\begin{equation}
    \rho_{\text{target}} = \ket{\psi_{\text{target}}}\bra{\psi_{\text{target}}} = \frac{1}{2}\begin{bmatrix}
        1+C_{\text{out}} & S_{\text{out}} \\
        S_{\text{out}} & 1-C_{\text{out}}
    \end{bmatrix}.
\end{equation}
Therefore, we can compute the fidelity of $\mE\left(\rho_{N_C/2}\right)$ with $\rho_{\text{target}} = \ket{\psi_{\text{target}}}\bra{\psi_{\text{target}}}$:
\begin{align}
    \mF &= \frac{1+C_{\text{out}}}{2}\sum_{w=0}^{N_C}P_{w,w}\cos^2\theta_w + \frac{1-C_{\text{out}}}{2}\sum_{w=0}^{N_C}P_{w,w}\sin^2\theta_w + S_{\text{out}}\sum_{w=1}^{N_C}P_{w-1,w}\cos\theta_{w-1}\sin\theta_w \\
    &= \frac{1+C_{\text{out}}}{2} - C_{\text{out}}\sum_{w=0}^{N_C}P_{w,w}\sin^2\theta_w + S_{\text{out}}\sum_{w=1}^{N_C}P_{w-1,w}\cos\theta_{w-1}\sin\theta_w.
\end{align}

\vspace{0.5\baselineskip}

Now we can attempt to maximize this fidelity in the most straightforward way, which is to use multivariable calculus and just take partial derivatives:
\begin{equation}
    \frac{\partial\mF}{\partial\theta_w} = -2C_{\text{out}}P_{w,w}\cos\theta_w\sin\theta_w + S_{\text{out}}\left(\cos\theta_{w-1}\cos\theta_wP_{w-1,w} - \sin\theta_w\sin\theta_{w+1}P_{w,w+1}\right).
\end{equation}
Setting each partial derivative to zero yields the following condition:
\begin{equation}
    \boxed{\cot\Theta_{\text{out}} = \frac{\cos\theta_{w-1}\cos\theta_wP_{w-1,w} - \sin\theta_w\sin\theta_{w+1}P_{w,w+1}}{2\cos\theta_w\sin\theta_wP_{w,w}}}.
\end{equation}
And we're done! We have determined the optimal distillation protocol!

\vspace{0.5\baselineskip}

$\ldots$ Well, kind of. What we have here is a \textbf{boundary value problem}, given by a $3$-angle recurrence relation that relates every set of $3$ consecutive angles, as well as boundary conditions given by the first and last angle ($\theta_0 = 0$ and $\theta_N = \frac{\pi}{2}$, respectively). In this case, this does indeed uniquely determine all of the intermediate angles, so in that sense, we are done: the optimal choice of angles is precisely those angles that solve this boundary value problem. However, actually finding these angles is far from a trivial process.

\vspace{0.5\baselineskip}

The rest of this appendix will be dedicated to making certain observations about this boundary value problem in special cases. Furthermore, in Appendix \ref{appendix:perturbative-protocols}, we use the boundary value problem in an especially interesting way, namely to compute the first-order and second-order perturbations to the optimal protocol in the equatorial case when $\lambda$ is perturbed slightly away from $1$. However, the boundary value problem is tough to study analytically, so for the order analysis of the type needed to prove a connection with information-theoretic constraints such as purity of coherence and RLD Fisher information, the boundary value problem does not do us much good. As a result, to find approximately optimal protocols in the asymptotic regime, we will have to resort to a completely different method.

\vspace{0.5\baselineskip}

Let us first make an observation about the case where the target state is equatorial, i.e., $\Theta_{\text{out}} = \frac{\pi}{2}$, meaning $C_{\text{out}} = 0$ and $S_{\text{out}} = 1$. In this case, the fidelity formula simplifies as follows:
\begin{equation}
    \mF = \frac{1}{2} + \sum_{w=1}^{N_C}P_{w-1,w}\cos\theta_{w-1}\sin\theta_w.
\end{equation}
Notice that the dependence on the $P_{w,w}$ values has completely dropped. The three-angle relation also simplifies considerably:
\begin{equation}
    \boxed{\tan\theta_w = \frac{\cos\theta_{w-1}}{\sin\theta_{w+1}}\frac{P_{w-1,w}}{P_{w,w+1}}}.
\end{equation}
We can rewrite this as
\begin{equation}
    \tan\theta_w = \frac{\cos\theta_{w-1}}{\sin\theta_{w+1}}R_w,
\end{equation}
where for convenience we have defined
\begin{equation}
    R_w \coloneqq \frac{P_{w-1,w}}{P_{w,w+1}}.
\end{equation}
In other words, in the special case where $\Theta_{\text{out}} = \frac{\pi}{2}$ (meaning that the target state is equatorial), the optimal protocol does not depend on the $P_{w,w}$ values; it only depends on the $P_{w-1,w}$ values. Furthermore, it does not depend on their values intrinsically, only on their ratios given by $R_w$. This case is especially nice to solve numerically because it is very easy to determine the optimal value of one angle $\theta_w$ given its neighbors $\theta_{w-1}$ and $\theta_{w+1}$.

\appsubsec{Exact Solution for Pure Input States and Matching Target State}
{subsec:solve-recurrence-pure-input-matching-target}

One special case where we can solve the fidelity optimization problem exactly is the case where $\lambda = 1$ and $\Theta_{\text{in}} = \Theta_{\text{out}} \eqqcolon \Theta$. Of course, in such a situation, each input state is already equal to the target state, and the Schur-sampled state will also just be $N_C = N$ copies of that same state \cite{Cirac1999}. Hence, we should find that just discarding all the qubits except one should achieve perfect fidelity.

\vspace{0.5\baselineskip}

Let us confirm that this is actually what happens. First, refer to the formula for $P_{w,w'}$ in Appendix \ref{appendix:understanding-P-vals}\ref{subsec:P-vals-miscellaneous-facts}, where $P_{w,w'}$ is written in terms of other values called $Q_{wk}$ and $Q_{w'k}$. When $k=N_C$, that expression simplifies as follows:
\begin{equation}
    P_{w,w'} = Q_{w,N_C}Q_{w',N_C}.
\end{equation}
Furthermore, when $k=N_C$, the expressions for $Q_{wk}$ also simplify as follows:
\begin{equation}
    Q_{w,N_C} = \sqrt{\binom{N_C}{w}}\left(\sin\frac{\Theta}{2}\right)^w\left(\cos\frac{\Theta}{2}\right)^{N_C-w}.
\end{equation}
Therefore, the fidelity can be written as follows:
\begin{align}
    \mF &= \frac{1+C_{\text{out}}}{2}\sum_{w=0}^{N_C}P_{w,w}\cos^2\theta_w + \frac{1-C_{\text{out}}}{2}\sum_{w=0}^{N_C}P_{w,w}\sin^2\theta_w + S_{\text{out}}\sum_{w=1}^{N_C}P_{w-1,w}\cos\theta_{w-1}\sin\theta_w \\
    &= \cos^2\frac{\Theta}{2}\sum_{w=0}^{N_C}Q_{w,N_C}^2\cos^2\theta_w + \sin^2\frac{\Theta}{2}\sum_{w=0}^{N_C}Q_{w,N_C}^2\sin^2\theta_w \\
    & \quad + 2\cos\frac{\Theta}{2}\sin\frac{\Theta}{2}\sum_{w=1}^{N_C}Q_{w-1,N_C}Q_{w,N_C}\cos\theta_{w-1}\sin\theta_w.
\end{align}
Now, we can use the iconic \textbf{arithmetic mean-geometric mean (AM-GM) inequality}, which states that the arithmetic mean of a collection of nonnegative real numbers is always at least as large as its geometric mean, with equality if and only if all the numbers are equal. For just two numbers $x$ and $y$, the proof is very simple:
\begin{equation}
    (\sqrt{x} - \sqrt{y})^2 \ge 0 \implies x - 2\sqrt{xy} + y \ge 0 \implies \frac{x+y}{2} \ge \sqrt{xy}.
\end{equation}
The AM-GM inequality implies that
\begin{equation}
    2\cos\frac{\Theta}{2}\sin\frac{\Theta}{2}Q_{w-1,N_C}Q_{w,N_C}\cos\theta_{w-1}\sin\theta_w \le \sin^2\frac{\Theta}{2}Q_{w-1,N_C}^2\cos^2\theta_{w-1} + \cos^2\frac{\Theta}{2}Q_{w,N_C}^2\sin^2\theta_w.
\end{equation}
Therefore, we can upper bound the fidelity as follows:
\begin{align}
    \mF &\le \cos^2\frac{\Theta}{2}\sum_{w=0}^{N_C}Q_{w,N_C}^2\cos^2\theta_w + \sin^2\frac{\Theta}{2}\sum_{w=0}^{N_C}Q_{w,N_C}^2\sin^2\theta_w \\
    & \quad + \sum_{w=1}^{N_C}\left(\sin^2\frac{\Theta}{2}Q_{w-1,N_C}^2\cos^2\theta_{w-1} + \cos^2\frac{\Theta}{2}Q_{w,N_C}^2\sin^2\theta_w\right) \\
    &= \left(\cos^2\frac{\Theta}{2} + \sin^2\frac{\Theta}{2}\right)\sum_{w=0}^{N_C}Q_{w,N_C}^2\cos^2\theta_w \\
    & \quad + \left(\sin^2\frac{\Theta}{2} + \cos^2\frac{\Theta}{2}\right)\sum_{w=0}^{N_C}Q_{w,N_C}^2\sin^2\theta_w \\
    &= \sum_{w=0}^{N_C}Q_{w,N_C}^2\cos^2\theta_w + \sum_{w=0}^{N_C}Q_{w,N_C}^2\sin^2\theta_w \\
    &= \sum_{w=0}^{N_C}Q_{w,N_C}^2 \\
    &= \sum_{w=0}^{N_C}P_{w,w} \\
    &= 1
\end{align}
Of course, it is not a surprise that the fidelity is at most $1$. The more interesting point is that there is an equality condition. In particular, the AM-GM inequality is saturated if and only if
\begin{align}
    \sin^2\frac{\Theta}{2}Q_{w-1,N_C}^2\cos^2\theta_{w-1} &= \cos^2\frac{\Theta}{2}Q_{w,N_C}^2\sin^2\theta_w \\
    \iff \binom{N_C}{w-1}\cos^2\theta_{w-1} &= \binom{N_C}{w}\sin^2\theta_w \\
    \iff \sin^2\theta_w &= \frac{w}{N_C-w+1}\cos^2\theta_w.
\end{align}
If we define $S_w\coloneqq\sin^2\theta_w$ for convenience, then we get the recurrence relation
\begin{equation}
    S_w = \frac{w}{N-w+1}\left(1 - S_{w-1}\right),
\end{equation}
which along with $S_0 = \sin^2\theta_0 = 0$ and an easy induction proof, yields a remarkably simple formula:
\begin{equation}
    \sin^2\theta_w = \frac{w}{N_C} \iff \theta_w = \arcsin\sqrt{\frac{w}{N_C}} \quad (0\le w\le N_C).
\end{equation}
In other words, perfect fidelity is achieved if and only if we make the above choice of $\theta_w$ values. But it turns out that this is precisely the discarding protocol, in which we throw out all the qubits except one! To see this, one can directly compute the partial trace of a computational basis state $\ket{k_{\hat{n}}^{(s)}}\bra{(k')_{\hat{n}}^{(s)}}$ with respect to all the qubits except one and confirm that, for all $0\le k,k'\le N_C$, it matches the $(k,k')$ block of the Choi matrix corresponding to the channel with the above choice of $\theta_w$ values.

\vspace{0.5\baselineskip}

If $\lambda=1$ but $\Theta_{\text{in}}\neq\Theta_{\text{out}}$, then it turns out that it is not possible to achieve perfect fidelity. However, this is not a matter of great concern, because it turns out that one can achieve perfect fidelity with an exponentially small chance of failure. The idea is that there is a time translation-invariant procedure that, with some positive probability of success, will convert a pure qubit state with polar angle $\Theta_{\text{in}}$ and azimuthal angle $\phi$ into the pure qubit state with polar angle $\Theta_{\text{in}}$ and the same azimuthal angle $\phi$. Since you only need to succeed once in $N$ tries, the chance of failure is exponentially small (in other words, the fidelity is $1$ minus an exponentially small amount). We explicitly show how to do this probabilistic conversion in Appendix \ref{appendix:perfect-conversion}.

\appsubsec{Exact Solution for Equatorial Case with Odd Number of Qubits}
{subsec:solve-recurrence-odd-N-equatorial}

Suppose we are in the equatorial special case $\Theta_{\text{in}} = \Theta_{\text{out}} = \frac{\pi}{2}$, and also suppose that $N_C$ is odd (which also means that the original $N$ must have been odd). Then we can actually solve the recurrence relation exactly by exploiting bit-flip symmetry, which enforces opposite angles to be complementary, i.e., $\theta_w + \theta_{N-w} = \frac{\pi}{2}$. Let $N_C = 2n+1$ be odd. Then the recurrence relation with $w=n$ tells us that
\begin{equation}
    \tan\theta_n = \frac{\cos\theta_{n-1}}{\sin\theta_{n+1}}R_n.
\end{equation}
Since $\theta_n + \theta_{n+1} = \frac{\pi}{2}$, we have $\sin\theta_{n+1} = \cos\theta_n$, from which it follows that
\begin{equation}
    \sin\theta_n = \cos\theta_{n-1}R_n.
\end{equation}
We can now plug this relation into the recurrence relation with $w=n-1$ as follows:
\begin{align}
    \tan\theta_{n-1} = \frac{\cos\theta_{n-2}}{\sin\theta_n}R_{n-1} = \frac{\cos\theta_{n-2}}{\cos\theta_{n-1}}\frac{R_{n-1}}{R_n} \\
    \implies \sin\theta_{n-1} = \cos\theta_{n-2}\frac{R_{n-1}}{R_n}.
\end{align}
Next, we plug this relation into the recurrence relation with $w=n-2$ in exactly the same way:
\begin{align}
    \tan\theta_{n-2} = \frac{\cos\theta_{n-3}}{\sin\theta_{n-1}}R_{n-2} = \frac{\cos\theta_{n-3}}{\cos\theta_{n-2}}\frac{R_{n-2}R_n}{R_{n-1}} \\
    \implies \sin\theta_{n-2} = \cos\theta_{n-3}\frac{R_{n-2}R_n}{R_{n-1}}.
\end{align}
We can now repeat this trick again and again. In the end, we obtain
\begin{equation}
    \sin\theta_1 = \cos\theta_0\frac{R_1R_3R_5\cdots}{R_2R_4R_6\cdots},
\end{equation}
where the numerator has $R_w$ for all odd values of $w$ from $1$ to $n$, and the denominator has $R_w$ for all even values of $w$ from $1$ to $n$. Finally, we use the fact that $\cos\theta_0 = 1$ to obtain an exact expression for $\sin\theta_1$:
\begin{equation}
    \sin\theta_1 = \prod_{w=1}^{n}R_w^{(-1)^{w-1}} = \begin{cases}
        \frac{R_1R_3R_5\cdots R_{n-1}}{R_2R_4R_6\cdots R_n} & n\text{ even }(N_C \equiv 1\pmod{4}) \\
        \frac{R_1R_3R_5\cdots R_n}{R_2R_4R_6\cdots R_{n-1}} & n\text{ odd }(N_C \equiv 3\pmod{4}),
    \end{cases}
\end{equation}
where as a reminder, we have set $N_C = 2n+1$ to be odd.

\vspace{0.5\baselineskip}

We can now solve the recurrence relation very quickly, since we can use $\theta_0$ and $\theta_1$ to obtain $\theta_2$, and then use $\theta_1$ and $\theta_2$ to obtain $\theta_3$, and so on. Unfortunately, the expression for $\sin\theta_1$ is fairly complicated, so attempting to move from this point to get closed-form expressions for $\theta_2$, $\theta_3$, and so on appears to be tedious and not particularly worthwhile. However, this exact solution can at least be used in a computer program to efficiently solve for all the angles, which saves at least some time. As a result, the equatorial special case with an odd number of qubits is especially nice when trying to use numerics to better understand trends in the optimal protocol.

\vspace{0.5\baselineskip}

So why does this trick fail when $N_C$ is even? At the middle, the bit-flip symmetry enforces $\theta_{N_C/2} = \frac{\pi}{4}$, rather than relating two angles. As a result, we are left with a boundary value problem on the angles $\theta_w$ for $0\le w\le \frac{N_C}{2}$, where the first and last angles are known to be $\theta_0 = 0$ and $\theta_{N_C/2} = \frac{\pi}{4}$. Furthermore, none of the $3$-angle recurrence relations tell us anything new to take advantage of the bit-flip symmetry. For example, plugging $\theta_{N_C/2} = \frac{\pi}{4}$ into the $3$-angle recurrence relation with $w = N_C/2$ yields a statement that is satisfied for any choice of $\theta_{w-1}$, as long as $\theta_{w-1}$ and $\theta_{w+1}$ are complementary. As a result, we are stuck with the same type of boundary value problem we started with, just on a domain of half the length.

\vspace{0.5\baselineskip}

One cool consequence of the exact solution for odd $N_C$ can be seen by rewriting the $3$-angle recurrence relation in terms of $S_w = \sin^2\theta_w$. One way of doing so is the following:
\begin{equation}
    S_{w+1} = \left(1 - S_{w-1}\right)\left(\frac{1}{S_w} - 1\right)R_w^2.
\end{equation}
Based on the formula for the $P_{w,w'}$ values shown in Corollary \ref{cor:P-val-single-summation-equatorial}, $R_w^2$ is always a rational function of $\lambda$ with integer coefficients. Therefore, if any two of $S_{w-1}$, $S_w$, and $S_{w+1}$ are rational functions of $\lambda$ with integer coefficients, then the same must also be true of the third one. Therefore, as long as this holds for $S_0$ and $S_1$, induction implies that this holds for all $S_w$. But we know that $S_0 = 0$, and we now have a formula for $S_1$ (obtained by squaring the above result for $\sin\theta_1$) that is just a quotient of two products of $R_w^2$ values, meaning that $S_1$ is also a rational function of $\lambda$ with integer coefficients. We conclude that, if $N_C$ is odd, then $S_w = \sin^2\theta_w$ is a rational function of $\lambda$ with integer coefficients for all integers $0\le w\le N_C$. (We suspect that this fact is true for even $N_C$ as well, but we have not attempted to prove or disprove this.)

%% file: app_asymptotic_expansion.tex
\appsec{Series Expansion of the Distillation Protocol in the Asymptotic Regime}
{appendix:asymptotic-expansion}

Although we can use the boundary value problem described in Appendix \ref{appendix:boundary-value-problem} to numerically derive the exactly optimal protocol, it is hard to analytically derive the asymptotic behavior of the protocol from that boundary value problem. As a result, we turn to a different method for the sake of the asymptotic analysis, which is what provides us the connection to purity of coherence and RLD Fisher information.

\vspace{0.5\baselineskip}

As we show in Appendix \ref{appendix:deriving-kraus-rep}, an optimal protocol can be fully described by the values of $\sin^2\theta_w$ for $0\le w\le N_C$. Furthermore, we know that the $P_{w,w'}$ values are highly concentrated in a relatively narrow range of $w$ and $w'$ values. As a result, the fidelity formula shown in Appendix \ref{appendix:boundary-value-problem} demonstrates that only a relatively small range of $\sin^2\theta_w$ values contribute meaningfully to the fidelity. This motivates the use of an asymptotic expansion. In particular, we can set $z = \frac{w}{N}-\mu$, where the $P_{w,w'}$ values concentrate around $w,w'\approx\mu N_C$. Then, we can write $\sin^2\theta_w$ as a power series in $z$:
\begin{equation}
    \sin^2\theta_w = a_0 + a_1z + a_2z^2 + a_3z^3 + a_4z^4 + \cdots
\end{equation}
However, the coefficients of the power series may themselves be power series in $N_C^{-1}$. As a result, we modify our asymptotic expansion to look as follows:
\begin{align}
    \sin^2\theta_w = \,\, & \left(b_{00} + \frac{b_{01}}{N_C} + \frac{b_{02}}{N_C^2} + \cdots\right) + \left(b_{10} + \frac{b_{11}}{N_C} + \cdots\right)z \\
    & + \left(b_{20} + \frac{b_{21}}{N_C} + \cdots\right)z^2 + \left(b_{30} + \cdots\right)z^3 + \left(b_{40} + \cdots\right)z^4 + \cdots
\end{align}
In Appendix \ref{appendix:understanding-P-vals}, we compute the moments of the $P_{w-\alpha,w}$ values for $\alpha=0,1$. The decay of those moments shows that, at the leading order, the $P_{w,w'}$ values tend to a Gaussian with mean $\approx\mu N_C$ and variance $\approx\sigma^2N_C$ for some values of $\mu$ and $\sigma$ that do not depend on $N_C$. This means that the typical values of $z$ are on the order of $N_C^{-1/2}$. More specifically, we will have
\begin{align}
    z^p &= O\left(N_C^{-p/2}\right) \\
    \mathbb{E}\left[z^p\right] &= O\left(N_C^{-\lceil p/2\rceil}\right).
\end{align}

\vspace{0.5\baselineskip}

Using this information, we can break the coefficients into tiers, based on what power of $N_C$ they carry. For example, the term $\frac{b_{31}}{N_C}z^3$ is on the order of $N_C^{-5/2}$, so it can only contribute at the $3^{\text{rd}}$ order. In general, the term $\frac{b_{pq}}{N_C^q}z^p$ is on the order of $N_C^{-(p+2q)/2}$, meaning that the breakdown looks as follows:

\begin{itemize}
    \item $0^{\text{th}}$-order infidelity: only $b_{00}$ contributes
    \item $1^{\text{st}}$-order infidelity: $b_{01}$, $b_{10}$, $b_{20}$ also contribute
    \item $2^{\text{nd}}$-order infidelity: $b_{02}$, $b_{11}$, $b_{21}$, $b_{30}$, $b_{40}$ also contribute
    \item $3^{\text{rd}}$-order infidelity: $b_{03}$, $b_{12}$, $b_{22}$, $b_{31}$, $b_{41}$, $b_{50}$, $b_{60}$ also contribute
    \item (and so on...)
\end{itemize}

This may look a bit scary, since even for our primary objective, which is $1^{\text{st}}$-order optimality, we would need to think about four variables, namely $b_{00}$, $b_{01}$, $b_{10}$, $b_{20}$. However, it turns out that our analysis will be a lot easier than this. In the vicinity of a local optimum, the deviation of a function is quadratic in the deviations of its inputs, as opposed to linear. Therefore, although the term $\frac{b_{pq}}{N_C^q}z^p$ is on the order of $N_C^{-(p+2q)/2}$, when we plug in the optimum values for all previous variables, we will find that $b_{pq}$ will only contribute at the order of $N_C^{-(p+2q)}$. Therefore, when we solve for successive orders of optimality, the workflow will look as follows:

\begin{itemize}
    \item $0^{\text{th}}$-order optimality: solve for $b_{00}$
    \item $1^{\text{st}}$-order optimality: solve for $b_{10}$
    \item $2^{\text{nd}}$-order optimality: solve for $b_{01}$, $b_{20}$
    \item $3^{\text{rd}}$-order optimality: solve for $b_{11}$, $b_{30}$
    \item $4^{\text{th}}$-order optimality: solve for $b_{02}$, $b_{21}$, $b_{40}$
    \item $5^{\text{th}}$-order optimality: solve for $b_{12}$, $b_{31}$, $b_{50}$
    \item (and so on...)
\end{itemize}

To clarify this point, let us give an example. The variable $b_{31}$, since it is saddled with a factor of $N_C^{-5/2}$, will first show up in the most general formula for the $3^{\text{rd}}$-order infidelity. However, if we have already optimized the $2^{\text{nd}}$-order infidelity, then deviations in $b_{31}$ will only cause the infidelity to change quadratically in $N_C^{-5/2}$, meaning that it has deviations on the order of $N_C^{-5}$. As a result, when we plug in the optimum values for $b_{00}$, $b_{01}$, $b_{10}$, $b_{20}$ (which are all the quantities that we need to solve for to determine $2^{\text{nd}}$-order optimality), the terms containing $b_{31}$ will ``magically'' cancel out. When we proceed to the expression for the $4^{\text{th}}$-order infidelity, then once again, there will be terms depending on $b_{31}$. However, when we plug in the optimum values for the four variables previously mentioned, and additionally the optimum values for $b_{11}$ and $b_{30}$ (which we need to solve for to determine $3^{\text{rd}}$-order optimality), the terms containing $b_{31}$ will again ``magically'' cancel out. Only when we proceed to the expression of the $5^{\text{th}}$-order infidelity will we find that, even when plugging in the optimum values for all previous variables, the terms containing $b_{31}$ will not cancel out. As a result, we will finally be forced to solve for $b_{31}$ in order to achieve $5^{\text{th}}$-order optimality.

\vspace{0.5\baselineskip}

For the most general version of our single-shot qubit coherence distillation problem, where we allow $\lambda$, $\Theta_{\text{in}}$, and $\Theta_{\text{out}}$ to be anything, we tackle $1^{\text{st}}$-order optimality in Appendix \ref{appendix:1st-order-optimality} and $2^{\text{nd}}$-order optimality in Appendix \ref{appendix:2nd-order-optimality}. In particular, $1^{\text{st}}$-order optimality will demonstrate the connection with purity of coherence and RLD Fisher information.

\vspace{0.5\baselineskip}

However, for the equatorial special case, where we restrict to $\Theta_{\text{in}} = \Theta_{\text{out}} = \frac{\pi}{2}$, we can additionally impose bit-flip symmetry on our protocol to simplify the asymptotic analysis considerably. A protocol that additionally satisfies bit-flip symmetry will also satisfy $\theta_w + \theta_{N_C-w} = \frac{\pi}{2}$, meaning that $\mu = \frac{1}{2}$, $a_0 = \frac{1}{2}$, and $a_p = 0$ for all even values of $p\ge 1$. Furthermore, the behavior of the $P_{w,w'}$ coefficients also somewhat simplifies. As a result, we additionally solve for $3^{\text{rd}}$-order optimality in this special case, as we show in Appendix \ref{appendix:equatorial-3rd-order-optimality}.

\vspace{0.5\baselineskip}

As those later appendices will reveal, these computations become incredibly tedious (and also less illuminating) very fast, which is why we do not attempt to proceed further. However, at least in theory, there is no barrier to determining optimality to arbitrarily high orders. In particular, in order to solve for $p^{\text{th}}$-order optimality, it suffices to compute the following values:
\begin{itemize}
    \item the $q^{\text{th}}$ centered moment of the $P_{w-1,w}$ and $P_{w,w}$ values, up to $\Theta(N^{-p})$ precision, and for each $1\le q\le p$ (we compute these for the first few orders in Appendix \ref{appendix:understanding-P-vals});
    \item $\Ebb\left[N_C^{-q}\right]$, up to $\Theta(N^{-p})$ precision, for each $1\le q\le p$ (we compute these for the first few orders in Appendix \ref{appendix:angular-momentum-moments}).
\end{itemize}
Even the latter quantities are not needed to find the optimal protocol; they are just needed to convert the infidelity expressions from a power series in $N_C^{-1}$ to a power series in $N^{-1}$.

\vspace{0.5\baselineskip}

The fact that this method \textit{can} be carried to arbitrarily high orders is itself a fact of independent interest. Many problems in the resource theory of asymmetry are tackled at the leading order using methods that prevent higher-order analysis. The most notable example of such a method is quantum local asymptotic normality (QLAN), which observes that, if a state $\rho$ is within $O(N^{-(\varepsilon+1/4)})$ distance of some fixed diagonalized state $\rho_0$, then the state $\rho^{\otimes N}$ can be understood as a tensor product of a classical Gaussian random variable (for the diagonal entries of $\rho$) and several thermal coherent states (one for each independent off-diagonal entry of $\rho$) \cite{Guta2006,Kahn2009}. Although QLAN is a very powerful and general tool, it requires one to localize the state up to $O(N^{-(\varepsilon+1/4)})$ error, which is often done by measuring a sublinear number of copies to estimate the state. Such a method is fine for $1^{\text{st}}$-order optimality, where a sublinear number of copies can be wasted without issue, but it cannot be used for $2^{\text{nd}}$-order optimality, as we explain in more detail at the beginning of Appendix \ref{appendix:2nd-order-optimality}. Therefore, we believe that our method, although much more bespoke, is superior for the problem at hand.

%% file: app_first_order.tex
\appsec{First-Order Optimality for General Qubit Distillation}
{appendix:1st-order-optimality}

In this appendix, we derive a $1^{\text{st}}$-order optimal protocol and show that it saturates the bound imposed by the monotonicity of purity of coherence. In particular, this appendix will show how we came up with the ansatz that we used in Appendix \ref{appendix:lower-bound-saturation}.

\vspace{0.5\baselineskip}

We begin by defining the following for convenience:
\begin{equation}
    C_{\text{out}} \coloneqq \cos\frac{\Theta_{\text{out}}}{2}, \quad\quad S_{\text{out}} \coloneqq \sin\frac{\Theta_{\text{out}}}{2}, \quad\quad C_{\text{in}} \coloneqq \cos\frac{\Theta_{\text{in}}}{2}, \quad\quad S_{\text{in}} \coloneqq \sin\frac{\Theta_{\text{in}}}{2}.
\end{equation}
(WARNING: This ``half-angle'' convention happens to be more convenient for this one appendix, but it does \textit{not} match what is used in the rest of this paper. Every other appendix in this paper uses the ``full-angle'' convention $C_{\text{out}} \coloneqq \cos\Theta_{\text{out}}$, and similarly for the others.)

\vspace{0.5\baselineskip}

Using this notation, the formula for the fidelity that we derived in Appendix \ref{appendix:boundary-value-problem} can be written as follows:
\begin{equation}
    \mF = \sum_{w=0}^{N_C}\left(C_{\text{out}}^2\cos^2\theta_w + S_{\text{out}}^2\sin^2\theta_w\right)P_{w,w} + 2C_{\text{out}}S_{\text{out}}\sum_{w=1}^{N_C}\cos\theta_{w-1}\sin\theta_wP_{w-1,w}.
\end{equation}

\vspace{0.5\baselineskip}

Now we use the crucial fact about how the $P_{w,w'}$ values behave in the asymptotic regime. For a fixed value of the ``offset'' $\alpha$, the $P_{w-\alpha,w}$ values roughly follow an unnormalized Gaussian. The three parameters that depend on $\alpha$ are the amplitude $A_\alpha$, the mean $\mu_\alpha$, and the variance $\sigma_\alpha^2$. For our problem, we care about $\alpha=0$ and $\alpha=1$. In particular, for $\alpha=0$, the parameters are as follows:
\begin{align}
    A_0 &= 1 \\
    \mu_0 &= \frac{1-\cos\Theta_{\text{in}}}{2} + \cos\Theta_{\text{in}}\frac{1-\lambda}{2\lambda}\frac{1}{N_C} + O\left(N_C^{-2}\right) \\
    \sigma_0^2 &= \frac{\sin^2\Theta_{\text{in}}}{4\lambda}\frac{1}{N_C} + O\left(N_C^{-2}\right),
\end{align}
while for $\alpha=1$, the parameters are as follows:
\begin{align}
    A_1 &= 1 - \frac{1}{2\lambda\sin^2\Theta_{\text{in}}}\frac{1}{N_C} + O\left(N_C^{-2}\right) \\
    \mu_1 &= \frac{1-\cos\Theta_{\text{in}}}{2} + \cos\Theta_{\text{in}}\frac{1-\lambda}{2\lambda}\frac{1}{N_C} + O\left(N_C^{-2}\right) \\
    \sigma_1^2 &= \frac{\sin^2\Theta_{\text{in}}}{4\lambda}\frac{1}{N_C} + O\left(N_C^{-2}\right).
\end{align}
This Gaussian behavior can be understood as a consequence of the moment formulas that we prove in Appendix \ref{appendix:understanding-P-vals}.

\vspace{0.5\baselineskip}

As a brief aside, we could present this entire discussion just in terms of those moment formulas, with no special attention paid to this Gaussian behavior. In this case, the relevant moments are $\mM_p^{(\alpha)}$ for $p=0,1,2$ and $\alpha=0,1$. In fact, once we proceed to higher orders of optimality (see Appendices \ref{appendix:2nd-order-optimality} and \ref{appendix:equatorial-3rd-order-optimality}), we have no choice but to do this, since the Gaussian approximation for the $P_{w-\alpha,w}$ values will no longer be precise enough. However, because we want to highlight the Gaussianity of these matrix entries, we present this appendix in this alternative way, characterizing each distribution by its amplitude, mean, and variance instead of by its zeroth, first, and second moments.

\vspace{0.5\baselineskip}

Since the values of $\mu_0$ and $\mu_1$ agree to $\Theta(N_C^{-1})$ precision, we will define for convenience
\begin{align}
    \mu &\coloneqq \frac{1-\cos\Theta_{\text{out}}}{2} + \cos\Theta_{\text{in}}\frac{1-\lambda}{2\lambda}\frac{1}{N_C} \\
    x &\coloneqq \frac{w}{N_C} \\
    z &\coloneqq x - \mu.
\end{align}
In particular, since
\begin{equation}
    \mu = \mu_0 + O\left(N_C^{-2}\right) = \mu_1 + O\left(N_C^{-2}\right),
\end{equation}
we can treat both $P_{w-1,w}$ and $P_{w,w}$ as being centered at $\mu$. In fact, even the $\Theta(N_C^{-1})$ term in $\mu$ is irrelevant for this appendix, and it will only concern us when we study $2^{\text{nd}}$-order optimality. However, as we will show in Appendix \ref{appendix:2nd-order-optimality}, the $2^{\text{nd}}$-order analysis will have to take into account the difference between $\mu_0$ and $\mu_1$ at the $\Theta(N_C^{-2})$ level.

\vspace{0.5\baselineskip}

As mentioned in Appendix \ref{appendix:asymptotic-expansion}, we will write $\sin^2\theta_w$ as a power series in $z$ as follows:
\begin{equation}
    \sin^2\theta_{N_Cx} = a_0 + a_1z + a_2z^2 + a_3z^3 + \cdots
\end{equation}
However, each coefficient $a_p$ can also be a power series in $N_C^{-1}$:
\begin{equation}
    a_p = \sum_{q=0}^{\infty}\frac{b_{pq}}{N_C^q} = b_{p0} + \frac{b_{p1}}{N_C} + \frac{b_{p2}}{N_C^2} + \cdots.
\end{equation}
For example, if we write all the terms that contribute up to $\Theta(N_C^{-2})$ precision, we obtain
\begin{equation}
    \sin^2\theta_{N_Cx} = \left(b_{00} + \frac{b_{01}}{N_C} + \frac{b_{02}}{N_C^2}\right) + \left(b_{10} + \frac{b_{11}}{N_C}\right)z + b_{20}z^2 + b_{30}z^3 + O\left(N_C^{-5/2}\right).
\end{equation}

\vspace{0.5\baselineskip}

Finally, for convenience, we will also define the following notation:
\begin{equation}
    \theta_{\pm} \coloneqq \theta_{N_Cx\pm\frac{1}{2}} \quad \theta_\circ \coloneqq \theta_{N_Cx}.
\end{equation}
This notation is needed because, even for $1^{\text{st}}$-order optimality, we will need to use the fact that the inputs to cosine and sine in the second summation in the fidelity formula are slightly different angles $\theta_{w-1}$ and $\theta_{w+1}$, as opposed to being the same angle.

\appsubsec{Deriving Zeroth-Order Optimality}
{subsec:0th-order-optimality-derivation}

Even before we tackle $1^{\text{st}}$-order optimality, we must tackle $0^{\text{th}}$-order optimality. In other words, we need to ensure that $\mF$ tends to $1$, which is the bare minimum for a reasonable distillation protocol. However, this will turn out to be quite simple, so we do not put $0^{\text{th}}$-order optimality in its own appendix.

\vspace{0.5\baselineskip}

At the $0^{\text{th}}$-order, only $b_{00}$ matters, so everything is quite simple:
\begin{align}
    \sin^2\theta_\circ &= b_{00} + O\left(N_C^{-1/2}\right) \\
    \cos^2\theta_\circ &= 1 - b_{00} + O\left(N_C^{-1/2}\right) \\
    \cos\theta_-\sin\theta_+ &= \sqrt{b_{00}\left(1 - b_{00}\right)} + O\left(N_C^{-1/2}\right)
\end{align}
We can plug these into the formula for $\mF$ and use the Cauchy-Schwarz inequality as follows:
\begin{align}
    \mF &= C_{\text{out}}^2\left(1 - b_{00}\right) + S_{\text{out}}^2b_{00} + 2C_{\text{out}}S_{\text{out}}\sqrt{b_{00}\left(1 - b_{00}\right)} + O\left(N_C^{-1}\right) \\
    &= \left(C_{\text{out}}\sqrt{1 - b_{00}} + S_{\text{out}}\sqrt{b_{00}}\right)^2 + O\left(N_C^{-1}\right) \\
    &\le \left(C_{\text{out}}^2 + S_{\text{out}}^2\right)\left(\left(1 - b_{00}\right) + b_{00}\right) + O\left(N_C^{-1}\right) \\
    &= 1 + O\left(N_C^{-1}\right).
\end{align}
The fidelity equals $1 + O(N_C^{-1})$ if and only if the above application of the Cauchy-Schwarz inequality is saturated, which occurs if and only if
\begin{equation}
    \boxed{b_{00} = S_{\text{out}}^2 = \frac{1-\cos\Theta_{\text{out}}}{2}}.
\end{equation}
Thus we have completed the $0^{\text{th}}$-order analysis.

\appsubsec{Deriving First-Order Optimality}
{subsec:1st-order-optimality-derivation}

Now the $1^{\text{st}}$-order analysis truly begins. We can plug in our answer for $b_{00}$ above, and $b_{01}$, $b_{10}$, $b_{20}$ also contribute:
\begin{equation}
    \sin^2\theta_{N_Cx} = \left(S_{\text{out}}^2 + \frac{b_{01}}{N_C}\right) + b_{10}z + b_{20}z^2 + O\left(N_C^{-3/2}\right).
\end{equation}
When we subtract this from $1$, we obtain
\begin{equation}
    \cos^2\theta_{N_Cx} = \left(C_{\text{out}}^2 - \frac{b_{01}}{N_C}\right) - b_{10}z - b_{20}z^2 + O\left(N_C^{-3/2}\right).
\end{equation}
If we replace $x$ with $x\pm\frac{1}{2N_C}$ in the above formulas and multiply them, we obtain, after some straightforward but tedious algebra,
\begin{align}
    \cos^2\theta_-\sin^2\theta_+ = \,\, & \left[\left(S_{\text{out}}^2 + \frac{b_{01}}{N_C}\right) + b_{10}\left(z + \frac{1}{2N_C}\right) + b_{20}\left(z + \frac{1}{2N_C}\right)^2 + O\left(N_C^{-3/2}\right)\right] \\
    & \left[\left(C_{\text{out}}^2 - \frac{b_{01}}{N_C}\right) - b_{10}\left(z - \frac{1}{2N_C}\right) - b_{20}\left(z - \frac{1}{2N_C}\right)^2 + O\left(N_C^{-3/2}\right)\right] \\
    = \,\, & C_{\text{out}}^2S_{\text{out}}^2 + b_{01}\left(C_{\text{out}}^2 - S_{\text{out}}^2\right)\frac{1}{N_C} + b_{10}\left(C_{\text{out}}^2 - S_{\text{out}}^2\right)z \\
    & + \frac{b_{10}}{2}\left(C_{\text{out}}^2 + S_{\text{out}}^2\right)\frac{1}{N_C} - b_{10}^2z^2 + b_{20}\left(C_{\text{out}}^2 - S_{\text{out}}^2\right)z^2 + O\left(N_C^{-3/2}\right) \\
    = \,\, & C_{\text{out}}^2S_{\text{out}}^2 + \left[b_{01}\left(C_{\text{out}}^2 - S_{\text{out}}^2\right) + \frac{b_{10}}{2}\right]\frac{1}{N_C} + b_{10}\left(C_{\text{out}}^2 - S_{\text{out}}^2\right)z \\
    & + \left[b_{20}\left(C_{\text{out}}^2 - S_{\text{out}}^2\right) - b_{10}^2\right]z^2 + O\left(N_C^{-3/2}\right) \\
    = \,\, & \frac{1}{4}\sin^2\Theta_{\text{out}} + \left(b_{01}\cos\Theta_{\text{out}} + \frac{b_{10}}{2}\right)\frac{1}{N_C} + b_{10}\cos\Theta_{\text{out}}z \\
    & + \left(b_{20}\cos\Theta_{\text{out}} - b_{10}^2\right)z^2 + O\left(N_C^{-3/2}\right).
\end{align}
Now we need to take the square root of the above quantity. The dominant term is the leading one, which does not have any decay when $N_C$ increases. As a result, we use the following Taylor series approximation:
\begin{equation}
    \sqrt{\frac{1}{4}\sin^2\Theta_{\text{out}} + \epsilon} = \frac{1}{2}\sin\Theta_{\text{out}} + \frac{\epsilon}{\sin\Theta_{\text{out}}} - \frac{\epsilon^2}{\sin^3\Theta_{\text{out}}} + O\left(\epsilon^3\right)
\end{equation}
When we apply this Taylor series approximation, we obtain
\begin{align}
    \cos\theta_-\sin\theta_+ = \,\, & \frac{1}{2}\sin\Theta_{\text{out}} + b_{10}\cot\Theta_{\text{out}}z \\
    & + \left(b_{01}\cot\Theta_{\text{out}} + \frac{b_{10}}{2}\csc\Theta_{\text{out}}\right)\frac{1}{N} \\
    & + \left(b_{20}\cot\Theta_{\text{out}} - b_{10}^2\csc^3\Theta_{\text{out}}\right)z^2 + O\left(N_C^{-3/2}\right).
\end{align}

\vspace{0.5\baselineskip}

Now that we have formulas for $\cos^2\theta_\circ$, $\sin^2\theta_\circ$, and $\cos\theta_-\sin\theta_+$, we can plug them into the formula for $\mF$:
\begin{align}
    \mF = \,\, & \sum_{w=0}^{N_C}\left(C_{\text{out}}^2\cos^2\theta_w + S_{\text{out}}^2\sin^2\theta_w\right)P_{w,w} + 2C_{\text{out}}S_{\text{out}}\sum_{w=1}^{N_C}\cos\theta_{w-1}\sin\theta_wP_{w-1,w} \\
    = \,\, & C_{\text{out}}^2\sum_{w=0}^{N_C}P_{w,w} - \left(C_{\text{out}}^2 - S_{\text{out}}^2\right)\sum_{w=0}^{N_C}\sin^2\theta_wP_{w,w} + 2C_{\text{out}}S_{\text{out}}\sum_{w=1}^{N_C}\cos\theta_{w-1}\sin\theta_wP_{w-1,w} \\
    = \,\, & C_{\text{out}}^2A_0 - \left(C_{\text{out}}^2 - S_{\text{out}}^2\right)\sum_{w=0}^{N_C}P_{w,w}\left[\left(S_{\text{out}}^2 + \frac{b_{01}}{N_C}\right) + b_{10}z + b_{20}z^2 + O\left(N_C^{-3/2}\right)\right] \\
    & + 2C_{\text{out}}S_{\text{out}}\sum_{w=1}^{N_C}P_{w-1,w}\Bigg\{\frac{1}{2}\sin\Theta_{\text{out}} + b_{10}\cot\Theta_{\text{out}}z \\
    & + \left(b_{01}\cot\Theta_{\text{out}} + \frac{b_{10}}{2}\csc\Theta_{\text{out}}\right)\frac{1}{N_C} \\
    & + \left(b_{20}\cot\Theta_{\text{out}} - b_{10}^2\csc^3\Theta_{\text{out}}\right)z^2 + O\left(N_C^{-3/2}\right)\Bigg\} \\
    = \,\, & \frac{1 + \cos\Theta_{\text{out}}}{2}A_0 - \cos\Theta_{\text{out}}\left[\frac{1 - \cos\Theta_{\text{out}}}{2} + \frac{b_{01}}{N_C} + b_{20}\sigma_0^2\right]A_0 \\
    & + \sin\Theta_{\text{out}}\Bigg\{\frac{\sin\Theta_{\text{out}}}{2} + \frac{b_{01}\cot\Theta_{\text{out}}}{N_C} + \frac{b_{10}\csc\Theta_{\text{out}}}{2N_C} \\
    & + \left(b_{20}\cot\Theta_{\text{out}} - b_{10}^2\csc^3\Theta_{\text{out}}\right)\sigma_1^2\Bigg\}A_1 + O\left(N_C^{-2}\right) \\
    = \,\, & A_0\Bigg\{\frac{1 + \cos^2\Theta_{\text{out}}}{2} - b_{01}\cos\Theta_{\text{out}}\frac{1}{N_C} - b_{20}\cos\Theta_{\text{out}}\sigma_0^2\Bigg\} \\
    & + A_1\Bigg\{\frac{\sin^2\Theta_{\text{out}}}{2} + b_{01}\cos\Theta_{\text{out}}\frac{1}{N_C} + \frac{b_{10}}{2}\frac{1}{N_C} \\
    & + \left(b_{20}\cos\Theta_{\text{out}} - b_{10}^2\csc^2\Theta_{\text{out}}\right)\sigma_1^2\Bigg\} + O\left(N_C^{-2}\right)
\end{align}
Notice that, along the way, we are able to simplify certain summations by writing them in terms of the parameters $A_0$, $\sigma_0$, $A_1$, $\sigma_1$. It is worth mentioning that Gaussianity is not actually required for these substitutions to work, only that the $P_{w-1,w}$ and $P_{w,w}$ values decay fast enough for any higher centered moments to be $O(N_C^{-2})$.

\vspace{0.5\baselineskip}

Now we can plug in the values for $A_0$, $\sigma_0$, $A_1$, $\sigma_1$. We thus obtain the following:
\begin{align}
    \mF = \,\, & \frac{1 + \cos^2\Theta_{\text{out}}}{2} - \textcolor{red}{b_{01}\cos\Theta_{\text{out}}\frac{1}{N_C}} - \textcolor{blue}{b_{20}\frac{\cos\Theta_{\text{out}}\sin^2\Theta_{\text{out}}}{4\lambda}\frac{1}{N_C}} \\
    & + \frac{\sin^2\Theta_{\text{out}}}{2} - \frac{\sin^2\Theta_{\text{out}}}{4\lambda\sin^2\Theta_{\text{in}}}\frac{1}{N_C} + \textcolor{red}{b_{01}\cos\Theta_{\text{out}}\frac{1}{N_C}} + \frac{b_{10}}{2}\frac{1}{N_C} \\
    & + \textcolor{blue}{b_{20}\frac{\cos\Theta_{\text{out}}\sin^2\Theta_{\text{out}}}{4\lambda}\frac{1}{N_C}} - b_{10}^2\frac{\sin^2\Theta_{\text{in}}}{4\lambda\sin^2\Theta_{\text{out}}}\frac{1}{N_C} + O\left(N_C^{-2}\right) \\
    = \,\, & 1 - \left(\frac{\sin^2\Theta_{\text{out}}}{4\lambda\sin^2\Theta_{\text{in}}} - \frac{b_{10}}{2} + \frac{\sin^2\Theta_{\text{in}}}{4\lambda\sin^2\Theta_{\text{out}}}b_{10}^2\right)\frac{1}{N_C} + O\left(N_C^{-2}\right)
\end{align}
This is where the ``magical'' cancellation described in Appendix \ref{appendix:asymptotic-expansion} comes into play. Notice that the terms involving $b_{01}$ (marked above in \textcolor{red}{red}) and the terms involving $b_{20}$ (marked above in \textcolor{blue}{blue}) perfectly cancel out. Furthermore, this cancellation does not happen for arbitrary protocols; it happens specifically because we have already plugged in $b_{00} = S_{\text{out}}^2$, which is the necessary and sufficient condition for $0^{\text{th}}$-order optimality. As we explained in Appendix \ref{appendix:asymptotic-expansion}, this occurs because, in the vicinity of a local optimum, deviations in the input only affect the function quadratically. Therefore, deviations in $b_{10}$ only contribute at the $\Theta(N_C^{-1})$ order, while deviations in $b_{01}$ and $b_{20}$ only contribute at the $\Theta(N_C^{-2})$ order.

\vspace{0.5\baselineskip}

As a result of the cancellation, our leading-order infidelity (one minus fidelity) can be written as follows:
\begin{equation}
    \mI(\mE) = \left(\frac{\sin^2\Theta_{\text{out}}}{4\lambda\sin^2\Theta_{\text{in}}} - \frac{b_{10}}{2} + \frac{\sin^2\Theta_{\text{in}}}{4\lambda\sin^2\Theta_{\text{out}}}b_{10}^2\right)\frac{1}{N_C} + O\left(N_C^{-2}\right)
\end{equation}
The coefficient of $N_C^{-1}$ is a quadratic in $b_{10}$, and it achieves its minimum at the following value:
\begin{equation}
    \boxed{b_{10} = \frac{\sin^2\Theta_{\text{out}}}{\sin^2\Theta_{\text{in}}}\lambda}.
\end{equation}
Just for clarity, this means that any protocol that has $\sin^2\theta_{N_Cx} = f(x)$, with $f(x)$ differentiable and satisfying
\begin{align}
    f\left(\frac{1-\cos\Theta_{\text{out}}}{2}\right) &= \frac{1-\cos\Theta_{\text{out}}}{2} \\
    f'\left(\frac{1-\cos\Theta_{\text{out}}}{2}\right) &= \frac{\sin^2\Theta_{\text{out}}}{\sin^2\Theta_{\text{in}}}\lambda,
\end{align}
will be $1^{\text{st}}$-order optimal. (Of course, for the actual protocol, this function only gets evaluated when $x$ is an integer multiple of $N_C^{-1}$. However, it is convenient to think of $x$ as a continuous parameter.) This also demonstrates that the $1^{\text{st}}$-order optimal protocol that we used in Appendix \ref{appendix:lower-bound-saturation} emerges organically from this systematic optimization procedure.

\vspace{0.5\baselineskip}

When we plug this into the formula for $\mI(\mE)$, we obtain
\begin{equation}
    \mI(\mE_{\text{opt}}) = \left(\frac{1-\lambda^2}{4\lambda}\frac{\sin^2\Theta_{\text{out}}}{\sin^2\Theta_{\text{in}}}\right)\frac{1}{N_C} + O\left(N_C^{-2}\right).
\end{equation}
This is nice, but we would prefer to have this written as a power series in $N^{-1}$, rather than $N_C^{-1}$. To do this, we need to understand the outcome distribution of the total angular momentum measurement required for Schur sampling \cite{Cirac1999}. In this case, we only need to know $\Ebb[N_C^{-1}]$, and as we show in Appendix \ref{appendix:angular-momentum-moments}, that has a very simple expression:
\begin{equation}
    \mathbb{E}\left[\frac{1}{N_C}\right] = \frac{1}{\lambda}\frac{1}{N} + O\left(N^{-2}\right).
\end{equation}
As a result, the infidelity series for a $1^{\text{st}}$-order optimal distillation protocol starts as follows:
\begin{equation}
    \mI(\mE_{\text{opt}}) = \left(\frac{1-\lambda^2}{4\lambda^2}\frac{\sin^2\Theta_{\text{out}}}{\sin^2\Theta_{\text{in}}}\right)\frac{1}{N} + O\left(N^{-2}\right)
\end{equation}
Hence the first coefficient of the infidelity series, which we call the infidelity factor by Definition \ref{def:infidelity-coeffs-infidelity-factor}, is
\begin{equation}
    \boxed{\delta_1(\mE_{\text{opt}}) = \frac{1-\lambda^2}{4\lambda^2}\frac{\sin^2\Theta_{\text{out}}}{\sin^2\Theta_{\text{in}}}}.
\end{equation}
This exactly saturates the lower bound proved in Appendix \ref{appendix:infidelity-lower-bound} using the monotonicity of purity of coherence! Hence the monotonicity of purity of coherence sets the tightest possible lower bound on the infidelity factor.

%% file: app_second_order.tex
\appsec{Second-Order Optimality for General Qubit Distillation}
{appendix:2nd-order-optimality}

In this appendix, we push beyond the foundation set by first-order optimality, which we derived in Appendix \ref{appendix:1st-order-optimality}, and tackle the question of second-order optimality. Although this appendix will no longer reveal a connection to any more general information-theoretic idea (as far as we currently know), we believe it has value as a more detailed treatment of the single-shot qubit coherence distillation problem in its own right, and as a demonstration of the power of the more bespoke method that we use for the asymptotic analysis, as described in Appendix \ref{appendix:asymptotic-expansion}.

\appsubsec{Why Is Second-Order Optimality Difficult?}
{subsec:2nd-order-optimality-difficult}

Second-order optimality is much more punishing than first-order optimality. In first-order optimality, you can use a sublinear number of copies of your input state however you want; even throwing them away is perfectly fine. This fact is essential for many leading-order analyses in many problems in quantum parameter estimation theory and the resource theory of asymmetry. For example, the first-order optimal measurement scheme for single-parameter quantum state estimation, as outlined by Barndorff-Nielsen and Gill \cite{BarndorffNielsen2000}, relies on consuming a sublinear number of copies to find a crude approximation of the parameter, after which one can measure the rest of the qubits using a measurement designed to be maximally informative for parameter values near the crude estimate. As another example, the proof for the maximum i.i.d. conversion rate between pure states in the resource theory of asymmetry for an arbitrary compact Lie group \cite{Yamaguchi2024} uses a result known as quantum local asymptotic normality \cite{Guta2006,Kahn2009}, which is a very powerful and general tool, but which also requires one to first localize the state to a small neighborhood by measuring a sublinear number of copies.

\vspace{0.5\baselineskip}

However, such tricks are not permissible for higher orders of optimality. In fact, wasting even a single copy of your input state will ruin your ability to achieve second-order optimality! To see this, note that an optimal distillation protocol $\mE_N$ will have infidelity series
\begin{equation}
    \mI(\mE_N) = \frac{\delta_1(\mE)}{N} + \frac{\delta_2(\mE)}{N^2} + O\left(N^{-3}\right).
\end{equation}
Now observe the geometric series
\begin{equation}
    \frac{1}{N-1} = \sum_{q=1}^{\infty}\frac{1}{N^q} = \frac{1}{N} + \frac{1}{N^2} + \frac{1}{N^3} + \cdots
\end{equation}
More generally, when you raise both sides of the above equation to an integer power $p$, you obtain:
\begin{equation}
    \frac{1}{(N-1)^p} = \sum_{q=p}^{\infty}\frac{\binom{q-1}{p-1}}{N^q} = \frac{1}{N^p} + \frac{p}{N^{p+1}} + \frac{p(p+1)/2}{N^{p+2}} + \cdots
\end{equation}
Now suppose that you waste a single copy and run $\mE_{N-1}$ on the remaining input copies. The resulting infidelity is
\begin{align}
    \mI(\mE_{N-1}) &= \frac{\delta_1(\mE)}{N-1} + \frac{\delta_2(\mE)}{(N-1)^2} + O\left((N-1)^{-3}\right) \\
    &= \frac{\delta_1(\mE)}{N} + \frac{\delta_1(\mE) + \delta_2(\mE)}{N^2} + O\left(N^{-3}\right).
\end{align}
Notice that the coefficient of the $N^{-2}$ term has increased from $\delta_2(\mE)$ to $\delta_1(\mE) + \delta_2(\mE)$. Hence any protocol with $\delta_1(\mE) > 0$ (which is always the case for the distillation problems we are considering) will be suboptimal at the second order if we waste an input copy.

\vspace{0.5\baselineskip}

Of course, the above argument may seem a little artificial, because there is no particular reason to waste an input copy. However, we strongly suspect that the same conclusion will hold for any protocol that uses any number of input copies in a ``sufficiently suboptimal'' way, and this is potentially very common. In particular, we believe that all existing protocols for quantum state estimation or conversion that rely on first developing a crude estimate of the state will fall prey to this issue.

\appsubsec{Deriving Second-Order Optimality}
{subsec:2nd-order-optimality-derivation}

Let us now proceed to the actual derivation. We return to the usual notational convention for the sines and cosines of $\Theta_{\text{in}}$ and $\Theta_{\text{out}}$, meaning that we define the following quantities for convenience:
\begin{equation}
    C_{\text{in}} \equiv \cos\Theta_{\text{in}}, \quad S_{\text{in}} \equiv \sin\Theta_{\text{in}}, \quad C_{\text{out}} \equiv \cos\Theta_{\text{out}}, \quad S_{\text{out}} \equiv \sin\Theta_{\text{out}}.
\end{equation}
Furthermore, as explained in Appendix \ref{appendix:asymptotic-expansion}, we need to use the centered moments of the $P_{w-\alpha,w}$ values for offset values $\alpha=0$ and $\alpha=1$. 
In particular, these moments are defined to be centered on the value $\mu$, which is (up to a factor of $N_C$) the first moment of the $P_{w,w}$ values:
\begin{equation}
    \mu \equiv \sum_{w=0}^{N_C}\frac{w}{N_C}P_{w,w} = \frac{1 - C_{\text{in}}}{2} + \frac{C_{\text{in}}(1-\lambda)}{2\lambda}\frac{1}{N_C} + \text{e.s.e.}
\end{equation}
This formula for $\mu$ is proven in Appendix \ref{appendix:understanding-P-vals}\ref{subsec:P-vals-moments-offset0}. Once $\mu$ is defined, as we show in Appendix \ref{appendix:understanding-P-vals}\ref{subsec:P-vals-moments-defining}, the $p^{\text{th}}$ moment centered at $\mu$ of the $P_{w,w'}$ values with offset $\alpha$ can then be defined as
\begin{equation}
    \mM_{p}^{(\alpha)} \equiv \sum_{w=\alpha}^{N_C}\left(\frac{w - \frac{\alpha}{2}}{N_C} - \mu\right)^pP_{w-\alpha,w}.
\end{equation}
It may seem like a concern that we are defining all of these centered moments based on the mean for the $P_{w,w}$ values (i.e., offset zero). After all, the $P_{w-\alpha,w}$ values for other $\alpha$ values may have different means. However, the mean for $\alpha=1$ agrees with $\mu$ up to $O(N_C^{-2})$ error, so we never have to worry about this discrepancy in the first-order analysis. However, in the second-order analysis, we do have to worry about the $O(N_C^{-2})$ difference between the mean for the $P_{w,w}$ values and the mean for the $P_{w-1,w}$ values. (This will become clearer later in this appendix.)

\vspace{0.5\baselineskip}

In general, as one can see from the results in Appendix \ref{appendix:understanding-P-vals}, the distribution of the $P_{w,w}$ values (i.e., offset zero) is slightly ``nicer'' than those for the other offset values, because the $p^{\text{th}}$ moment, after the $N_C^{-p}$ term, has only exponentially small error (which we abbreviate as ``e.s.e.'' for convenience). Since $\mu$ acts as the reference point for computing the centered moments and the point relative to which the protocol is written as a power series in $z\coloneqq \frac{w}{N_C} - \mu$, it thus makes sense to define $\mu$ based on the $P_{w,w}$ values to make these calculations as convenient as possible.

\vspace{0.5\baselineskip}

In Appendices \ref{appendix:understanding-P-vals}\ref{subsec:P-vals-moments-offset0} and \ref{appendix:understanding-P-vals}\ref{subsec:P-vals-moments-offset1-general}, we show how to compute these moments. Since we are doing single-shot distillation to an output qubit clock of the same frequency as the input qubit clocks, we only care about $\alpha=0$ and $\alpha=1$. Furthermore, we only need to compute these values to $N_C^{-2}$ precision. Since the typical values of $\frac{w}{N_C}$ are only $O(N_C^{-1/2})$ away from $\mu$, we only need to care about the moments up to $p=4$.

\vspace{0.5\baselineskip}

As we show in Lemma \ref{lem:centered-moments-offset-0}, the first few centered moments for offset $\alpha=0$ are as follows:
\begin{align}
    \mM_0^{(0)} &= 1 \\
    \mM_1^{(0)} &= 0 \\
    \mM_2^{(0)} &= \frac{S_{\text{in}}^2}{4\lambda}\frac{1}{N_C} + \frac{(1-\lambda)\left(\left(C_{\text{in}}^2 - S_{\text{in}}^2\right) + C_{\text{in}}^2\lambda\right)}{4\lambda^2}\frac{1}{N_C^2} + \text{e.s.e.} \\
    \mM_3^{(0)} &= \frac{C_{\text{in}}S_{\text{in}}^2(3 - \lambda^2)}{8\lambda^2}\frac{1}{N_C^2} + O\left(N_C^{-3}\right) \\
    \mM_4^{(0)} &= \frac{3S_{\text{in}}^4}{16\lambda^2}\frac{1}{N_C^2} + O\left(N_C^{-3}\right).
\end{align}
Because the sum of the $P_{w,w}$ values equals the trace of the Schur-sampled state, it must equal exactly $1$. Furthermore, since $\mu$ was defined using the $P_{w,w}$ values, the centered first moment must be exactly $0$. The previously mentioned fact about the zero-offset distribution explains the presence of the ``e.s.e.'' in $\mM_2{(0)}$ after the $N_C^{-2}$ term. Similarly, if we were to extend the above formula for $\mM_3^{(0)}$ to the $N_C^{-3}$ term, then the remaining error would also be exponentially small (and similar for all higher moments too).

\vspace{0.5\baselineskip}

As we show in Lemma \ref{lem:centered-moments-offset-0}, the first few centered moments for offset $\alpha=1$ are as follows:
\begin{align}
    \mM_0^{(1)} &= 1 - \frac{1}{2S_{\text{in}}^2\lambda}\frac{1}{N_C} + \left(-\frac{C_{\text{in}}^2}{2S_{\text{in}}^4} + \frac{1}{2S_{\text{in}}^2\lambda} - \frac{3}{8S_{\text{in}}^4\lambda^2}\right)\frac{1}{N_C^2} + O\left(N_C^{-3}\right) \\
    \mM_1^{(1)} &= \frac{C_{\text{in}}(1 + \lambda^2)}{4S_{\text{in}}^2\lambda^2}\frac{1}{N_C^2} + O\left(N_C^{-3}\right) \\
    \mM_2^{(1)} &= \frac{S_{\text{in}}^2}{4\lambda}\frac{1}{N_C} + \left[-\frac{1}{8} + \frac{(1-\lambda)\left(-1 + \left(C_{\text{in}}^2 - S_{\text{in}}^2\right)(2+\lambda)\right)}{8\lambda^2}\right]\frac{1}{N_C^2} + O\left(N_C^{-3}\right) \\
    \mM_3^{(1)} &= \frac{C_{\text{in}}S_{\text{in}}^2(3 - \lambda^2)}{8\lambda^2}\frac{1}{N_C^2} + O\left(N_C^{-3}\right) \\
    \mM_4^{(1)} &= \frac{3S_{\text{in}}^4}{16\lambda^2}\frac{1}{N_C^2} + O\left(N_C^{-3}\right).
\end{align}
Let us briefly compare and contrast the $\alpha=0$ and $\alpha=1$ distributions. The two distributions both tend to normal distributions with the same leading-order parameters, so the leading term (that is, the $\Theta\left(N_C^{-\lceil p/2\rceil}\right)$ term) of $\mM_p^{(\alpha)}$ will always be the same regardless of $\alpha$. However, beyond the leading terms of each moment, there can be appreciable differences. For example, notice that $\mM_0^{(1)}$ is slightly less than $1$. In addition, notice that $\mM_1^{(1)}$ has an $N_C^{-2}$ term, reflecting the slight difference in the means of the $\alpha=0$ and $\alpha=1$ distributions.

\vspace{0.5\baselineskip}

We now invoke the asymptotic expansion for the distillation protocol that we described in Appendix \ref{appendix:asymptotic-expansion}:
\begin{align}
    \sin^2\theta_w = \,\, & \left(b_{00} + \frac{b_{01}}{N_C} + \frac{b_{02}}{N_C^2} + \cdots\right) + \left(b_{10} + \frac{b_{11}}{N_C} + \cdots\right)z \\
    & + \left(b_{20} + \frac{b_{21}}{N_C} + \cdots\right)z^2 + \left(b_{30} + \cdots\right)z^3 + \left(b_{40} + \cdots\right)z^4 + \cdots
\end{align}
We now compute $\sin^2\theta_+$ by plugging $w=N_Cx+\frac{1}{2}$ into the asymptotic expansion:
\begin{align}
    \sin^2\theta_+ = \,\, & \left[b_{00} + \frac{b_{01}}{N_C} + \frac{b_{02}}{N_C^2} + \cdots\right] + \left[b_{10} + \frac{b_{11}}{N_C} + \cdots\right]\left(z + \frac{1}{2N_C}\right) \\
    & + \left[b_{20} + \frac{b_{21}}{N_C} + \cdots\right]\left(z + \frac{1}{2N_C}\right)^2 + \left[b_{30} + \cdots\right]\left(z + \frac{1}{2N_C}\right)^3 \\
    & + \left[b_{40} + \cdots\right]\left(z + \frac{1}{2N_C}\right)^4 + \cdots \\
    = \,\, & b_{00} + b_{10}z + \left(b_{01} + \frac{b_{10}}{2}\right)\frac{1}{N_C} + b_{20}z^2 + \left(b_{11} + b_{20}\right)\frac{z}{N_C} + b_{30}z^3 \\
    & + \left(b_{02} + \frac{b_{11}}{2} + \frac{b_{20}}{4}\right)\frac{1}{N_C^2} + \left(b_{21} + \frac{3}{2}b_{30}\right)\frac{z^2}{N_C} + b_{40}z^4 + \cdots
\end{align}
We similarly compute $\cos^2\theta_-$ by plugging $w=N_Cx-\frac{1}{2}$ into the asymptotic expansion:
\begin{align}
    \cos^2\theta_- = \,\, & \left[\left(1 - b_{00}\right) - \frac{b_{01}}{N_C} - \frac{b_{02}}{N_C^2} - \cdots\right] + \left[-b_{10} - \frac{b_{11}}{N_C} - \cdots\right]\left(z - \frac{1}{2N_C}\right) \\
    & + \left[-b_{20} - \frac{b_{21}}{N_C} - \cdots\right]\left(z - \frac{1}{2N_C}\right)^2 + \left[-b_{30} - \cdots\right]\left(z - \frac{1}{2N_C}\right)^3 \\
    & + \left[-b_{40} - \cdots\right]\left(z - \frac{1}{2N_C}\right)^4 + \cdots \\
    = \,\, & \left(1 - b_{00}\right) - b_{10}z + \left(-b_{01} + \frac{b_{10}}{2}\right)\frac{1}{N_C} - b_{20}z^2 + \left(-b_{11} + b_{20}\right)\frac{z}{N_C} - b_{30}z^3 \\
    & + \left(-b_{02} + \frac{b_{11}}{2} - \frac{b_{20}}{4}\right)\frac{1}{N_C^2} + \left(-b_{21} + \frac{3}{2}b_{30}\right)\frac{z^2}{N_C} - b_{40}z^4 + \cdots
\end{align}
Multiplying the two previous quantities yields
\begin{align}
    \cos^2\theta_-\sin^2\theta_+ = \,\, & b_{00}\left(1 - b_{00}\right) + b_{10}\left(1 - 2b_{00}\right)z \\
    & + \left[b_{01}\left(1 - 2b_{00}\right) + \frac{b_{10}}{2}\right]\frac{1}{N_C} + \left[b_{20}\left(1 - 2b_{00}\right) - b_{10}^2\right]z^2 \\
    & + \left[b_{11}\left(1 - 2b_{00}\right) + b_{20} - 2b_{10}b_{01}\right]\frac{z}{N_C} + \left[b_{30}\left(1 - 2b_{00}\right) - 2b_{10}b_{20}\right]z^3 \\
    & + \left[b_{02}\left(1 - 2b_{00}\right) + \frac{b_{11}}{2} + \frac{b_{20}}{4}\left(1 - 2b_{00}\right) - b_{01}^2 + \frac{b_{10}^2}{4}\right]\frac{1}{N_C^2} \\
    & + \left[b_{21}\left(1 - 2b_{00}\right) + \frac{3}{2}b_{30} - 2b_{01}b_{20} - 2b_{10}b_{11}\right]\frac{z^2}{N_C} \\
    & + \left[b_{40}\left(1 - 2b_{00}\right) - 2b_{10}b_{30} - b_{20}^2\right]z^4 + \cdots
\end{align}
Because the leading quantity $b_{00}\left(1 - b_{00}\right)$ will show up a lot, we give it a special name for convenience:
\begin{equation}
    B_{00} \equiv b_{00}\left(1 - b_{00}\right).
\end{equation}
Now we need use the following Taylor series:
\begin{equation}
    \sqrt{1 + \epsilon} = 1 + \frac{\epsilon}{2} - \frac{\epsilon^2}{8} + \frac{\epsilon^3}{16} - \frac{5}{128}\epsilon^4 + O\left(\epsilon^5\right)
\end{equation}
When we use the above Taylor series to take the square root of $\cos^2\theta_-\sin^2\theta_+$, we obtain (take a deep breath now...)
\begin{align}
    \cos\theta_-\sin\theta_+ = \,\, & \sqrt{B_{00}}\Bigg\{1 + \frac{1}{2B_{00}}\left[b_{10}\left(1 - 2b_{00}\right)\right]z + \frac{1}{2B_{00}}\left[b_{01}\left(1 - 2b_{00}\right) + \frac{b_{10}}{2}\right]\frac{1}{N_C} \\
    & + \frac{1}{2B_{00}}\left[b_{20}\left(1 - 2b_{00}\right) - b_{10}^2\right]z^2 + \frac{1}{2B_{00}}\left[b_{11}\left(1 - 2b_{00}\right) + b_{20} - 2b_{10}b_{01}\right]\frac{z}{N_C} \\
    & + \frac{1}{2B_{00}}\left[b_{30}\left(1 - 2b_{00}\right) - 2b_{10}b_{20}\right]z^3 \\
    & + \frac{1}{2B_{00}}\left[b_{02}\left(1 - 2b_{00}\right) + \frac{b_{11}}{2} + \frac{b_{20}}{4}\left(1 - 2b_{00}\right) - b_{01}^2 + \frac{b_{10}^2}{4}\right]\frac{1}{N_C^2} \\
    & + \frac{1}{2B_{00}}\left[b_{21}\left(1 - 2b_{00}\right) + \frac{3}{2}b_{30} - 2b_{01}b_{20} - 2b_{10}b_{11}\right]\frac{z^2}{N_C} \\
    & + \frac{1}{2B_{00}}\left[b_{40}\left(1 - 2b_{00}\right) - 2b_{10}b_{30} - b_{20}^2\right]z^4 - \frac{1}{8B_{00}^2}\left[b_{10}\left(1 - 2b_{00}\right)\right]^2z^2 \\
    & - \frac{1}{4B_{00}^2}\left[b_{10}\left(1 - 2b_{00}\right)\right]\left[b_{01}\left(1 - 2b_{00}\right) + \frac{b_{10}}{2}\right]\frac{z}{N_C} \\
    & - \frac{1}{4B_{00}^2}\left[b_{10}\left(1 - 2b_{00}\right)\right]\left[b_{20}\left(1 - 2b_{00}\right) - b_{10}^2\right]z^3 \\
    & - \frac{1}{4B_{00}^2}\left[b_{10}\left(1 - 2b_{00}\right)\right]\left[b_{11}\left(1 - 2b_{00}\right) + b_{20} - 2b_{10}b_{01}\right]\frac{z^2}{N_C} \\
    & - \frac{1}{4B_{00}^2}\left[b_{10}\left(1 - 2b_{00}\right)\right]\left[b_{30}\left(1 - 2b_{00}\right) - 2b_{10}b_{20}\right]z^4 \\
    & - \frac{1}{8B_{00}^2}\left[b_{01}\left(1 - 2b_{00}\right) + \frac{b_{10}}{2}\right]^2\frac{1}{N_C^2} \\
    & - \frac{1}{4B_{00}^2}\left[b_{01}\left(1 - 2b_{00}\right) + \frac{b_{10}}{2}\right]\left[b_{20}\left(1 - 2b_{00}\right) - b_{10}^2\right]\frac{z^2}{N_C} \\
    & - \frac{1}{8B_{00}^2}\left[b_{20}\left(1 - 2b_{00}\right) - b_{10}^2\right]^2z^4 + \frac{1}{16B_{00}^3}\left[b_{10}\left(1 - 2b_{00}\right)\right]^3z^3 \\
    & + \frac{3}{16B_{00}^3}\left[b_{10}\left(1 - 2b_{00}\right)\right]^2\left[b_{01}\left(1 - 2b_{00}\right) + \frac{b_{10}}{2}\right]\frac{z^2}{N_C} \\
    & + \frac{3}{16B_{00}^3}\left[b_{10}\left(1 - 2b_{00}\right)\right]^2\left[b_{20}\left(1 - 2b_{00}\right) - b_{10}^2\right]z^4 \\
    & - \frac{5}{128B_{00}^4}\left[b_{10}\left(1 - 2b_{00}\right)\right]^4z^4\Bigg\} + \cdots
\end{align}
For the sake of readability, we have organized the terms in the above expression in the following way. Within the curly braces, we start with the term $1$ in the Taylor series. Next, we put all the terms coming from the $+\epsilon/2$ term in the Taylor series. After that, we put all the terms coming from the $-\epsilon^2/8$ term in the Taylor series. Then, we put all the terms coming from the $+\epsilon^3/16$ term in the Taylor series. Finally, we put all the terms coming from the $-5\epsilon^4/128$ term in the Taylor series.

\vspace{0.5\baselineskip}

It is now time to plug in the values for $b_{00}$ and $b_{10}$, which we derived in Appendix \ref{appendix:1st-order-optimality} from $0^{\text{th}}$-order optimality and $1^{\text{st}}$-order optimality, respectively:
\begin{equation}
    \boxed{b_{00} = \frac{1 - C_{\text{out}}}{2}} \quad \boxed{b_{10} = \lambda\frac{S_{\text{out}}^2}{S_{\text{in}}^2}}.
\end{equation}
Now we recall the fidelity formula that we derived in Appendix \ref{appendix:boundary-value-problem}:
\begin{align}
    \mF &= \sum_{w=0}^{N_C}\left(\frac{1 + C_{\text{out}}}{2}\cos^2\theta_w + \frac{1 - C_{\text{out}}}{2}\sin^2\theta_w\right)P_{w,w} + \sum_{w=1}^{N_C}\cos\theta_{w-1}\sin\theta_wP_{w-1,w} \\
    &= \frac{1 + C_{\text{out}}}{2}\sum_{w=0}^{N_C}P_{w,w} - C_{\text{out}}\sum_{w=0}^{N_C}\sin^2\theta_wP_{w,w} + S_{\text{out}}\sum_{w\in\Zbb+\frac{1}{2}, \,\, 0\le w\le N_C}\cos\theta_-\sin\theta_+P_{-,+}.
\end{align}
We now plug everything we have derived so far (the asymptotic expansion, the values of $b_{00}$ and $b_{10}$, and the $\mM_p^{(\alpha)}$ values) into the fidelity formula. After a lot of tedious algebra, we find that the $\Theta(N_C^{-2})$ contribution to the fidelity is a quadratic polynomial in the variables $b_{01}$ and $b_{20}$. (This is the point at which the terms that depend on  $b_{11}$, $b_{30}$, $b_{02}$, $b_{21}$, $b_{40}$ at $\Theta(N_C^{-2})$ precision ``magically'' cancel out.) In particular, we can write
\begin{equation}
    \left[\Theta\left(N_C^{-2}\right)\text{ fidelity term}\right]  = Ab_{01}^2 + 2Bb_{01}b_{20} + Cb_{20}^2 + Db_{01} + Eb_{20} + F
\end{equation}

\begin{equation}
    A = -\frac{1}{S_{\text{out}}^2}
\end{equation}

\begin{equation}
    B = -\frac{1}{4\lambda}\frac{S_{\text{in}}^2}{S_{\text{out}}^2}
\end{equation}

\begin{equation}
    C = -\frac{3}{16\lambda^2}\frac{S_{\text{in}}^4}{S_{\text{out}}^2}
\end{equation}

\begin{equation}
    D = -\frac{1-\lambda^2}{2\lambda}\frac{C_{\text{out}}}{S_{\text{in}}^2}
\end{equation}

\begin{equation}
    E = \frac{-1 + 5\lambda^2}{8\lambda^2}C_{\text{out}} - \frac{3 - \lambda^2}{4\lambda}C_{\text{in}}
\end{equation}

\begin{align}
    F = \,\,& \frac{S_{\text{out}}^2}{16S_{\text{in}}^4\lambda^2}\Bigg\{-(1-\lambda)^2(1+\lambda)(3-\lambda) - 8\lambda^4C_{\text{out}}^2 \\
    & + 4\lambda\left(-1 - 3\lambda + \lambda^2 + \lambda^3\right)C_{\text{in}}^2 + 4\lambda\left(1 + 4\lambda^2 - \lambda^4\right)C_{\text{out}}C_{\text{in}}\Bigg\}.
\end{align}
One good sanity check is to look at the equatorial special case. In that case, the coefficients of the quadratic polynomial simplify dramatically:
\begin{align}
    & \Theta_{\text{in}} = \Theta_{\text{out}} = \frac{\pi}{2} \quad \left(C_{\text{in}} = C_{\text{out}} = 0\right) \\
    \implies & A = -1, \quad B = -\frac{1}{4\lambda}, \quad C = -\frac{3}{16\lambda^2}, \quad D = 0, \quad E = 0, \quad F = -\frac{(1-\lambda)^2(1+\lambda)(3-\lambda)}{16\lambda^2}.
\end{align}
Fortunately, a quadratic polynomial is extremely easy to optimize. If a quadratic polynomial has the form
\begin{align}
    f(x,y) &= Ax^2 + 2Bxy + Cy^2 + Dx + Ey + F
\end{align}
with $AC\neq B^2$, then it has a unique critical point $(x_0,y_0)$ that satisfies the following:
\begin{align}
    (x_0, y_0) &= \left(\frac{BE - CD}{2(AC - B^2)}, \frac{BD - AE}{2(AC - B^2)}\right) \\
    f(x_0, y_0) &= \frac{2BDE - CD^2 - AE^2}{4(AC - B^2)} + F = \frac{D}{2}x_0 + \frac{E}{2}y_0 + F.
\end{align}
In particular, if $A,C<0$ and $AC-B^2>0$, then the Hessian is negative definite, so the critical point is a local maximum, and thus also the global maximum. In this case, it is easy to verify that $A,B,C$ indeed satisfy these conditions. As a result, by using the formulas for $x_0$ and $y_0$ above, we can find the optimum values for the parameters $b_{01}$ and $b_{20}$:
\begin{equation}
    \boxed{b_{01} = \frac{1}{4\lambda}\frac{S_{\text{out}}^2}{S_{\text{in}}^2}\left[\lambda(3 - \lambda^2)C_{\text{in}} - (1 + \lambda^2)C_{\text{out}}\right]}
\end{equation}

\begin{equation}
    \boxed{b_{20} = \lambda\frac{S_{\text{out}}^4}{S_{\text{in}}^2}\left[2\lambda C_{\text{out}} - (3 - \lambda^2)C_{\text{in}}\right]}.
\end{equation}
Once again, a good sanity check is to examine the equatorial special case:
\begin{equation}
    \Theta_{\text{in}} = \Theta_{\text{out}} = \frac{\pi}{2} \quad \left(C_{\text{in}} = C_{\text{out}} = 0\right) \implies b_{01} = b_{20} = 0.
\end{equation}
As we mentioned in Appendix \ref{appendix:asymptotic-expansion}, bit-flip symmetry implies that $b_{pq}=0$ for all even $p$ (except $p=q=0$, in which case $b_{00} = \frac{1}{2}$). So the fact that we get $b_{01} = b_{20} = 0$ in this case is reassuring.

\vspace{0.5\baselineskip}

We now compute the $2^{\text{nd}}$-order contribution to the fidelity, which equals $\frac{D}{2}b_{01} + \frac{E}{2}b_{20} + F$. We first compute the following two auxiliary quantities:
\begin{equation}
    \frac{D}{2}b_{01} = -\frac{1-\lambda^2}{16\lambda^2}\frac{C_{\text{out}}S_{\text{out}}^2}{S_{\text{in}}^4}\left[\lambda(3 - \lambda^2)C_{\text{in}} - (1 + \lambda^2)C_{\text{out}}\right]
\end{equation}

\begin{equation}
    \frac{E}{2}b_{20} = \frac{1}{16\lambda}\frac{S_{\text{out}}^2}{S_{\text{in}}^4}\left[(-1 + 5\lambda^2)C_{\text{out}} - 2\lambda(3 - \lambda^2)C_{\text{in}}\right]\left[2\lambda C_{\text{out}} - (3 - \lambda^2)C_{\text{in}}\right]
\end{equation}
Adding the above two quantities and $F$ yields
\begin{align}
    \left[\Theta\left(N_C^{-2}\right)\text{ fidelity term}\right] = \frac{1-\lambda}{16\lambda^2}\frac{S_{\text{out}}^2}{S_{\text{in}}^4}\Big\{ & C_{\text{out}}^2\left(1 + \lambda - \lambda^2 - \lambda^3 + 8\lambda^4\right) \\
    + \,\, & C_{\text{in}}^2\left[2\lambda^2\left(7 + \lambda - 3\lambda^2 - \lambda^3\right)\right] \\
    + \,\, & C_{\text{out}}C_{\text{in}}\left[2\lambda^2\left(7 + \lambda - 3\lambda^2 - \lambda^3\right)\right] \\
    + \,\, & \left[-(1-\lambda)(1+\lambda)(3-\lambda)\right]\Big\}.
\end{align}
Finally, for the sake of completeness, we re-introduce the $0^{\text{th}}$-order and $1^{\text{st}}$-order contributions to write out the full fidelity to $\Theta(N_C^{-2})$ precision:
\begin{align}
    \mF = 1 - \frac{1-\lambda^2}{4\lambda}\frac{S_{\text{out}}^2}{S_{\text{in}}^2}\frac{1}{N_C} + \frac{1-\lambda}{16\lambda^2}\frac{S_{\text{out}}^2}{S_{\text{in}}^4}\Big\{ & C_{\text{out}}^2\left(1 + \lambda - \lambda^2 - \lambda^3 + 8\lambda^4\right) \\
    + \,\, & C_{\text{in}}^2\left[2\lambda^2\left(7 + \lambda - 3\lambda^2 - \lambda^3\right)\right] \\
    + \,\, & C_{\text{out}}C_{\text{in}}\left[2\lambda^2\left(7 + \lambda - 3\lambda^2 - \lambda^3\right)\right] \\
    + \,\, & \left[-(1-\lambda)(1+\lambda)(3-\lambda)\right]\Big\}\frac{1}{N_C^2} + O\left(N_C^{-3}\right).
\end{align}

\vspace{0.5\baselineskip}

Our last task is to convert this power series in $N_C^{-1}$ to a power series in $N^{-1}$. To do this, we need to compute the expectation values of both $N_C^{-1}$ and $N_C^{-2}$ to $\Theta(N^{-2})$ precision. We explain how to do this in Appendix \ref{appendix:angular-momentum-moments}. The relevant expectation values are as follows:
\begin{align}
    \mathbb{E}\left[N_C^{-1}\right] &= \frac{1}{\lambda}\frac{1}{N} + \frac{1-\lambda}{\lambda^2}\frac{1}{N^2} + O\left(N^{-3}\right) \\
    \mathbb{E}\left[N_C^{-2}\right] &= \frac{1}{\lambda^2}\frac{1}{N^2} + O\left(N^{-3}\right).
\end{align}
When we plug in these formulas, we can obtain the $2^{\text{nd}}$-order contribution to the fidelity:
\begin{align}
    \left[\Theta\left(N^{-2}\right)\text{ fidelity term}\right] = \frac{1-\lambda}{16\lambda^4}\frac{S_{\text{out}}^2}{S_{\text{in}}^4}\Big\{ & C_{\text{out}}^2\left(1 + \lambda - \lambda^2 - \lambda^3 + 8\lambda^4\right) \\
    + \,\, & C_{\text{in}}^2\left[2\lambda\left(2 + 7\lambda - \lambda^2 - 3\lambda^3 - \lambda^4\right)\right] \\
    + \,\, & C_{\text{out}}C_{\text{in}}\left[4\lambda^2\left(1 - 5\lambda - \lambda^2 + \lambda^3\right)\right] \\
    + \,\, & \left[-3(1-\lambda)(1+\lambda)^2\right]\Big\}.
\end{align}
Finally, for the sake of completeness, we re-introduce the $0^{\text{th}}$-order and $1^{\text{st}}$-order contributions to write out the full fidelity to $\Theta(N^{-2})$ precision:
\begin{align}
    \mF = 1 - \frac{1-\lambda^2}{4\lambda^2}\frac{S_{\text{out}}^2}{S_{\text{in}}^2}\frac{1}{N} + \frac{1-\lambda}{16\lambda^4}\frac{S_{\text{out}}^2}{S_{\text{in}}^4}\Big\{ & C_{\text{out}}^2\left(1 + \lambda - \lambda^2 - \lambda^3 + 8\lambda^4\right) \\
    + \,\, & C_{\text{in}}^2\left[2\lambda\left(2 + 7\lambda - \lambda^2 - 3\lambda^3 - \lambda^4\right)\right] \\
    + \,\, & C_{\text{out}}C_{\text{in}}\left[4\lambda^2\left(1 - 5\lambda - \lambda^2 + \lambda^3\right)\right] \\
    + \,\, & \left[-3(1-\lambda)(1+\lambda)^2\right]\Big\}\frac{1}{N^2} + O\left(N^{-3}\right).
\end{align}
This completes the $2^{\text{nd}}$-order analysis of single-shot qubit coherence distillation. It is unclear whether the above fidelity expression has any deeper meaning. Regardless, we find that the fact that this problem \textit{can} be solved to higher orders at all interesting in its own right.

\appsubsec{Equatorial Special Case}
{subsec:2nd-order-optimality-special-case-equatorial}

One special case we always like to check is the equatorial case. We already incorporated some of this discussion into the general derivation in Appendix \ref{appendix:2nd-order-optimality}\ref{subsec:2nd-order-optimality-derivation}, but let us restate the key results. First, the protocol is actually extremely simple to describe, as both of the variables we solved for collapse to zero:
\begin{equation}
    \boxed{b_{01} = b_{20} = 0}.
\end{equation}
This can be understood as a direct consequence of bit-flip symmetry, which actually enforces $b_{pq}=0$ for all even $p$ (except $p=q=0$, in which case $b_{00} = \frac{1}{2}$). One interesting way to state this is as follows: \textbf{for the equatorial special case, a $1^{\text{st}}$-order optimal protocol that satisfies bit-flip symmetry will automatically be $2^{\text{nd}}$-order optimal as well.}

\vspace{0.5\baselineskip}

Another interesting thing to check is whether the lower bound on infidelity imposed by purity of coherence is saturated. In particular, as shown in Supplementary Note 10 of \cite{Marvian2020}, one can exploit bit-flip symmetry to lower bound the $2^{\text{nd}}$-order infidelity as well. In particular,
\begin{equation}
    \mI(\mE_N) \ge \frac{1}{2}\left[1 - \left(1 + \frac{1-\lambda^2}{\lambda^2}\frac{1}{N}\right)^{-1/2}\right] = \frac{1-\lambda^2}{4\lambda^2}\frac{1}{N} - \frac{3\left(1-\lambda^2\right)^2}{16\lambda^2}\frac{1}{N^2} + \frac{5\left(1-\lambda^2\right)^3}{32\lambda^6}\frac{1}{N^3} + O\left(N^{-4}\right).
\end{equation}
So is this bound saturated? To find out, we take the optimal fidelity, both as a power series in $N_C^{-1}$ and as a power series in $N^{-1}$, and we simply plug in $\Theta_{\text{in}} = \Theta_{\text{out}} = \frac{\pi}{2}$:
\begin{align}
    \Theta_{\text{in}} &= \Theta_{\text{out}} = \frac{\pi}{2} \quad \left(C_{\text{in}} = C_{\text{out}} = 0\right) \\
    \implies \mF &= 1 - \frac{1-\lambda^2}{4\lambda}\frac{1}{N_C} - \frac{(1-\lambda)^2(1+\lambda)(3-\lambda)}{16\lambda^2}\frac{1}{N_C^2} + O\left(N_C^{-3}\right) \\
    &= 1 - \frac{1-\lambda^2}{4\lambda^2}\frac{1}{N} - \frac{3(1-\lambda^2)^2}{16\lambda^2}\frac{1}{N^2} + O\left(N^{-3}\right) \\
    \implies \mI(\mE_N) &= \frac{1-\lambda^2}{4\lambda^2}\frac{1}{N} + \frac{3(1-\lambda^2)^2}{16\lambda^2}\frac{1}{N^2} + O\left(N^{-3}\right).
\end{align}
As we can see, the bound on the $N^{-2}$ term coming from purity of coherence is NOT saturated. (The two terms look the same, but one has a minus sign, while the other has a plus sign!) In fact, in Appendix \ref{appendix:PH-dissipation}, we show that this $2^{\text{nd}}$-order optimal protocol indeed wastes some purity of coherence, not at the $\Theta(N)$ order (which is the leading order), but at the $\Theta(1)$ order. The fact that the bound coming from purity of coherence monotonicity is not saturated at the $2^{\text{nd}}$ order is directly tied to this wastage.

\vspace{0.5\baselineskip}

Perhaps this is unsurprising; after all, even in classical parameter estimation, the Cram\'{e}r-Rao bound (which is by far the most famous application of Fisher information) can only be saturated at the leading order in general. (However, it can be saturated at higher orders, or even exactly, in some special cases.) It remains an open question to study whether there is some other natural information-geometric consideration that motivates the $\Theta\left(N^{-2}\right)$ term in the minimum infidelity. If there were such a consideration, it would be a distillation analogue of a result in classical parameter estimation which computes the $\Theta\left(N^{-2}\right)$ term in the variance of a bias-corrected, first-order efficient estimator and thus derives the condition for second-order efficiency \cite{Amari1985}.

\appsubsec{Special Case for Pure Input States}
{subsec:2nd-order-optimality-special-case-lam1}

In addition to the equatorial special case, there is one other especially interesting case to consider more closely. This is the $\lambda=1$ case, meaning that the input qubits are already pure. (Of course, the term ``distillation'' may no longer be appropriate in this case, since the input qubits are already pure.)

\vspace{0.5\baselineskip}

The first thing we can do is to plug $\lambda=1$ into the formulas for $b_{01}$ and $b_{20}$:
\begin{equation}
    \boxed{b_{01} = \frac{S_{\text{out}}^2}{2S_{\text{in}}^2}(C_{\text{in}} - C_{\text{out}})}
\end{equation}
\begin{equation}
    \boxed{b_{20} = \frac{2S_{\text{out}}^4}{2S_{\text{in}}^2}(C_{\text{out}} - C_{\text{in}})}.
\end{equation}
In general, these are probably not that illuminating. However, in the even more restricted case $\Theta_{\text{in}} = \Theta_{\text{out}}$, these values both collapse to zero! This actually makes perfect sense, because the discarding protocol $\sin^2\theta_w = \frac{w}{N_C}$ achieves perfect fidelity in this case. Therefore, in such a setting, plugging $z=0$ into the asymptotic expansion for the optimal protocol should produce
\begin{equation}
    z = 0 \implies \sin^2\theta = \mu = \frac{1 - C_{\text{in}}}{2} + \frac{C_{\text{in}}(1-\lambda)}{2\lambda}\frac{1}{N_C} + \text{e.s.e.} = \frac{1-C_{\text{in}}}{2} + \text{e.s.e.}
\end{equation}
However, plugging $z=0$ into the general asymptotic expansion yields
\begin{equation}
    z = 0 \implies \sin^2\theta_w = b_{00} + \frac{b_{01}}{N_C} + \frac{b_{02}}{N_C^2} + \cdots
\end{equation}
Furthermore, the protocol should always have slope exactly $1$. Therefore, we should always have $b_{00} = \frac{1-C_{\text{in}}}{2} = \frac{1-C_{\text{out}}}{2}$, $b_{10} = 1$, and all other values equal to $0$, and sure enough, the values we have derived so far match that perfectly.

\vspace{0.5\baselineskip}

The second thing we can do is to compute the fidelity. When we plug $\lambda=1$ into the fidelity formulas derived in Appendix \ref{appendix:2nd-order-optimality}\ref{subsec:2nd-order-optimality-derivation}, then regardless of $\Theta_{\text{in}}$ and $\Theta_{\text{out}}$, both the $N_C^{-1}$ and $N_C^{-2}$ terms vanish (equivalently, both the $N^{-1}$ and $N^{-2}$ terms vanish), leaving us with
\begin{equation}
    \lambda = 1 \implies \mF = 1 + O\left(N_C^{-3}\right) = 1 + O\left(N^{-3}\right).
\end{equation}
Given that we already showed the connection between the $1^{\text{st}}$-order infidelity and purity of coherence, the fact that the $1^{\text{st}}$-order infidelity vanishes is perhaps unsurprising. However, here we see that even the $2^{\text{nd}}$-order infidelity vanishes! If $\Theta_{\text{in}}=\Theta_{\text{out}}$, then this is obvious, because then the target state is the same as one input copy, but as long as $\Theta_{\text{in}}\neq\Theta_{\text{out}}$, this is not so obvious. In fact, if $\lambda=1$, then both the input and target states have infinite purity of coherence, so the purity of coherence cannot say anything about this case. However, this fact does have a satisfying intuitive explanation: in the $\lambda=1$ case, there exist a protocol that achieves zero $p^{\text{th}}$-order infidelity for \textit{every} $p\in\Nbb$, since the infidelity will actually be exponentially small. We present this protocol in Appendix \ref{appendix:perfect-conversion}.

%% file: app_perfect_conversion.tex
\appsec{Perfect Conversion with Post-Selection Between Pure Coherent States at Different Latitudes}
{appendix:perfect-conversion}

One interesting observation that comes out of the second-order analysis in Appendix \ref{appendix:2nd-order-optimality} is that, when the input qubit state is pure ($\lambda=1$), the second-order infidelity is also zero, regardless of $\Theta_{\text{in}}$ and $\Theta_{\text{out}}$. This fact is obvious for $\Theta_{\text{in}} = \Theta_{\text{out}}$, since the discarding protocol achieves perfect fidelity, but it is, in fact, non-obvious for $\Theta_{\text{in}} \neq \Theta_{\text{out}}$. For instance, Appendix \ref{appendix:boundary-value-problem}\ref{subsec:solve-recurrence-pure-input-matching-target} argues that, for $\Theta_{\text{in}} = \Theta_{\text{out}}$, the discarding protocol is the unique protocol on the symmetric subspace that achieves perfect fidelity, but this argument falls apart for $\Theta_{\text{in}} \neq \Theta_{\text{out}}$. In particular, one can show that it is no longer possible to saturate all of the inequalities that upper bound the fidelity at one, meaning that one cannot achieve perfect conversion from $N$ pure coherent qubit states to even $1$ pure coherent qubit state at a different latitude.

\vspace{0.5\baselineskip}

Furthermore, considerations involving purity of coherence no longer assist us here, since both the input and target pure coherent states have infinite purity of coherence. In fact, to draw an analogy to classical statistics, there are classical parameter estimation problems where the family of probability distributions has infinite Fisher information, yet the minimum-variance unbiased estimator (MVUE) still has $\Theta(N^{-2})$ variance. (A particularly simple one is a continuum analogue of the German tank problem, where the underlying distribution is the uniform distribution on $[0,\theta]$ for an unknown $\theta > 0$.) As a result, one might be led to believe that, even with infinite RLD Fisher information, one might be forced to have $\Theta(N^{-2})$ infidelity of distillation.

\vspace{0.5\baselineskip}

As a result, it is natural to ask: is there a simple and satisfying reason for why the second-order infidelity should vanish when the input states are pure? In fact, the answer is yes! In this appendix, we will demonstrate a procedure that, with positive probability under post-selection, achieves perfect one-to-one conversion between different latitudes of pure qubit states. This can be understood as providing a positive linear rate of perfect conversion, or alternatively an exponentially small chance of failure in the case of single-shot conversion.

\vspace{0.5\baselineskip}

For convenience, we will first assume that $\Theta_{\text{out}} \le \Theta_{\text{in}}$, i.e., the output coherent states are northward of the input coherent states. (As we will describe later, an analogous strategy will work in the opposite case, where the output coherent states are southward of the input coherent states.) In this case, the conversion strategy looks as follows:
\begin{itemize}
    \item Begin with the input state $\alpha\ket{0} + \beta\ket{1}$, where $\alpha = \cos\frac{\Theta_{\text{in}}}{2}$ and $\beta = e^{i\phi}\sin\frac{\Theta_{\text{in}}}{2}$, where $\phi\in[0,2\pi)$ is arbitrary and unknown to us.
    \item Introduce an ancilla in the $\ket{0}$ state to obtain the joint state $\alpha\ket{00} + \beta\ket{10}$.
    \item Apply the energy-conserving unitary $U(\eta) = \exp\left[-i\eta(\ket{01}\bra{10} - \ket{10}\bra{01})\right]$, which is a rotation on the $2$-dimensional subspace with total Hamming weight $1$. (We will compute the value of $\eta$ later.) The resulting state is $\alpha\ket{00} + \beta c\ket{10} + \beta s\ket{01}$, where $c\equiv\cos\eta$ and $s\equiv\sin\eta$ for convenience.
    \item Measure the ancilla qubit in the computational basis.
    \begin{itemize}
        \item If the measurement yields $\ket{0}$, then the conversion succeeds. The post-measurement state of the original qubit is $\frac{\alpha\ket{0} + \beta c\ket{1}}{\abs{\alpha}^2 + \abs{\beta}^2c^2}$. This occurs with probability $\Pbb = \abs{\alpha}^2 + \abs{\beta}^2c^2 = 1 - \abs{\beta}^2s^2$.
        \item If the measurement yields $\ket{1}$, then the conversion fails. The post-measurement state of the original qubit is $\ket{0}$. This occurs with probability $1 - \Pbb = \abs{\beta}^2s^2$.
    \end{itemize}
\end{itemize}

\vspace{0.5\baselineskip}

We now compute the value of the angle $\eta$ needed to carry out the conversion. Rather than computing $\eta$ directly, we compute $c^2$ and $s^2$, with the understanding that $0\le\eta\le\frac{\pi}{2}$:
\begin{align}
    & \quad \frac{\beta}{\alpha} = e^{i\phi}\tan\frac{\Theta_{\text{in}}}{2}, \quad \frac{\beta c}{\alpha} = e^{i\phi}\tan\frac{\Theta_{\text{out}}}{2} \\
    \implies & \quad c^2 = \frac{\tan^2\frac{\Theta_{\text{out}}}{2}}{\tan^2\frac{\Theta_{\text{in}}}{2}} = \frac{\left(1 - C_{\text{out}}\right)\left(1 + C_{\text{in}}\right)}{\left(1 + C_{\text{out}}\right)\left(1 - C_{\text{in}}\right)} \\
    \implies & \quad s^2 = \frac{2\left(C_{\text{out}} - C_{\text{in}}\right)}{\left(1 + C_{\text{out}}\right)\left(1 - C_{\text{in}}\right)}.
\end{align}
It is a good sanity check to confirm that $\Theta_{\text{out}} \le \Theta_{\text{in}}$ implies that $c^2$ and $s^2$ are both nonnegative, and to see that $c^2 = 1$ and $s^2 = 0$ precisely when $\Theta_{\text{in}} = \Theta_{\text{out}}$, which renders the ancilla useless, since in that case we can just return the input qubit with no alteration.

\vspace{0.5\baselineskip}

We can then also compute the probability of success, which turns out to have a very nice form:
\begin{align}
    \Pbb &= 1 - \abs{\beta}^2s^2 \\
    &= 1 - \sin^2\frac{\Theta_{\text{in}}}{2}\frac{2\left(C_{\text{out}} - C_{\text{in}}\right)}{\left(1 + C_{\text{out}}\right)\left(1 - C_{\text{in}}\right)} \\
    &= \frac{1 + C_{\text{in}}}{1 + C_{\text{out}}}.
\end{align}
Notice that this formula makes sense in the case we are considering, because $\Theta_{\text{out}} \le \Theta_{\text{in}}$ implies that $C_{\text{out}} \ge C_{\text{in}}$, so the success probability $\Pbb$ is always in the interval $[0,1]$. It is interesting to observe that, if we measure from the south pole, $\Pbb$ is simply the ratio of the input state elevation to the output state elevation. In other words, if we define the Hamiltonian of each qubit to be $H = \ket{0}\bra{0} = \frac{I+Z}{2}$, $\Pbb$ is the ratio of the input state energy to the output state energy.

\vspace{0.5\baselineskip}

A simple modification of the above strategy works in the opposite case, namely $\Theta_{\text{out}} \ge \Theta_{\text{in}}$. First, we initialize the ancilla qubit to $\ket{1}$ instead of $\ket{0}$. Next, the unitary will now be a rotation in the opposite direction, i.e., $U(\eta) = \exp\left[+i\eta(\ket{01}\bra{10} - \ket{10}\bra{01})\right]$. Finally, the successful measurement outcome on the ancilla is $\ket{1}$ instead of $\ket{0}$ (in both cases, the successful measurement outcome matches the initial state of the ancilla). Then the trigonometric functions of the rotation angle $\eta$ and the success probability assume the new formulas
\begin{align}
    c^2 = \frac{\tan^2\frac{\Theta_{\text{in}}}{2}}{\tan^2\frac{\Theta_{\text{out}}}{2}} &= \frac{\left(1 + C_{\text{out}}\right)\left(1 - C_{\text{in}}\right)}{\left(1 - C_{\text{out}}\right)\left(1 + C_{\text{in}}\right)} \\
    s^2 = 1 - c^2 &= \frac{2\left(C_{\text{in}} - C_{\text{out}}\right)}{\left(1 - C_{\text{out}}\right)\left(1 + C_{\text{in}}\right)} \\
    \Pbb &= \frac{1 - C_{\text{in}}}{1 - C_{\text{out}}}.
\end{align}
Once again, we can confirm that these values are all sensible in the setting where $\Theta_{\text{out}} \ge \Theta_{\text{in}}$, including in the special case where $\Theta_{\text{out}} = \Theta_{\text{in}}$ and the ancilla becomes useless.

\vspace{0.5\baselineskip}

So what does this all mean? Notice that each execution of this conversion strategy consumes only one input coherent state and outputs the desired pure coherent state with a positive constant probability. We have $N$ attempts at this strategy, and we only need to succeed once, which occurs with probability $1 - (1 - \Pbb)^N$. The average fidelity is at least this amount, since if we succeed, we achieve perfect fidelity, and if we fail, we achieve some incoherent state with nonnegative fidelity that we do not really care about. Therefore, the infidelity of this protocol is $\mI(\mE) \le (1-\Pbb)^N$, which is exponentially vanishing! Therefore, even though we cannot saturate all the inequalities in Appendix \ref{appendix:boundary-value-problem}\ref{subsec:solve-recurrence-pure-input-matching-target} to achieve perfect fidelity, we can achieve something exponentially close to perfect fidelity. Furthermore, we have found a natural explanation for why the second-order infidelity vanishes for $\lambda=1$, regardless of the values of $\Theta_{\text{in}}$ and $\Theta_{\text{out}}$, as we discovered in Appendix \ref{appendix:2nd-order-optimality}. In fact, the exponentially vanishing infidelity of this protocol tells us that we would continue to find zero $p^{\text{th}}$-order infidelity at $\lambda=1$ for all $p\in\Nbb$ and all possible values of $\Theta_{\text{in}}$ and $\Theta_{\text{out}}$.

\vspace{0.5\baselineskip}

This conversion strategy also has implications for another popular problem in the resource theory of asymmetry, namely that of linear conversion rate, where we try to produce $\approx rN$ copies of the desired state for the largest possible value of $r$, which is then called the linear conversion rate. For linear-rate conversion between families of pure states with \textbf{asymptotically vanishing error}, the maximum linear conversion rate was recently found by Yamaguchi et al. to be closely related to the quantum geometric tensor, which has interesting connections to both SLD and RLD Fisher information \cite{Yamaguchi2024}. However, what we have achieved here is actually \textbf{``perfect'' linear-rate conversion} between families of pure states, or in other words, linear-rate conversion between families of pure states with \textbf{zero error}. Optimizing the linear rate of perfect conversion between families of pure states in the resource theory of asymmetry may very well be an interesting question in its own right. For example, is there some asymmetry resource that is monotone for pure state conversions but not for general CPTP maps? In fact, it is not even immediately apparent whether the rate we achieved above is the best possible for converting between pure coherent states at different latitudes.

%% file: app_third_order_equatorial.tex
\appsec{Third-Order Optimality for Equatorial Distillation}
{appendix:equatorial-3rd-order-optimality}

To conclude our asymptotic analysis, we extend the equatorial special case to $3^{\text{rd}}$-order optimality. Solving for $2^{\text{nd}}$-order optimality in the general case, as we did in Appendix \ref{appendix:2nd-order-optimality}, is already quite painful. However, in Appendix \ref{appendix:2nd-order-optimality}\ref{subsec:2nd-order-optimality-special-case-equatorial}, we observe that bit-flip symmetry actually makes $2^{\text{nd}}$-order optimality somewhat trivial. As a result, we charge forward to $3^{\text{rd}}$-order optimality in the equatorial case.

\vspace{0.5\baselineskip}

To derive the conditions for $3^{\text{rd}}$-order optimality, we need to compute the relevant ``moments'' of the $P_{w-1,w}$ values to $\Theta\left(N_C^{-3}\right)$ precision. Refer to Lemma \ref{lem:equatorial-offset1-moments0246} in Appendix \ref{appendix:understanding-P-vals}\ref{subsec:P-vals-moments-offset1-equatorial} for a derivation of these formulas:
\begin{align}
    \mM_0^{(1)} &= 1 + \left(-\frac{1}{2\lambda}\right)\frac{1}{N_C} + \left(\frac{-3 + 4\lambda}{8\lambda^2}\right)\frac{1}{N_C^2} + \left(\frac{-9 + 12\lambda - 10\lambda^2}{16\lambda^3}\right)\frac{1}{N_C^3} + O\left(N_C^{-4}\right) \\
    \mM_2^{(1)} &= \left(\frac{1}{4\lambda}\right)\frac{1}{N_C} + \left(\frac{-3 + 2\lambda}{8\lambda^2}\right)\frac{1}{N_C^2} + \left(\frac{-3 + 4\lambda^2}{32\lambda^3}\right)\frac{1}{N_C^3} + O\left(N_C^{-4}\right) \\
    \mM_4^{(1)} &= \left(\frac{3}{16\lambda^2}\right)\frac{1}{N_C^2} + \left(\frac{-21 + 12\lambda + 2\lambda^2}{32\lambda^3}\right)\frac{1}{N_C^3} + O\left(N_C^{-4}\right) \\
    \mM_6^{(1)} &= \left(\frac{15}{64\lambda^3}\right)\frac{1}{N_C^3} + O\left(N_C^{-4}\right).
\end{align}
As we show in Appendix \ref{appendix:boundary-value-problem}\ref{subsec:brute-force-optimization}, we do not need to care about the $P_{w,w}$ values anymore. Furthermore, the odd centered moments are automatically zero. Hence the above values are truly all we need.

\vspace{0.5\baselineskip}

We begin by writing out our protocol, which we characterize using $\sin^2\theta_w$. As we explain in Appendix \ref{appendix:asymptotic-expansion}, bit-flip symmetry allows us to forcibly set $b_{00} = \frac{1}{2}$ and $b_{pq} = 0$ for all $(p,q)\neq(0,0)$ with $p$ even. Therefore, our protocol looks as follows, where we only show terms that can contribute at the $\Theta(N_C^{-3})$ order:
\begin{equation}
    \sin^2\theta_w = \frac{1}{2} + \left(\lambda + \frac{b_{11}}{N_C} + \frac{b_{12}}{N_C^2}\right)z + \left(b_{30} + \frac{b_{31}}{N_C}\right)z^3 + b_{50}z^5 + O\left(N_C^{-7/2}\right).
\end{equation}
Just as before, we use for convenience the notation
\begin{equation}
    \theta_{\pm} \coloneqq \theta_{N_Cx\pm\frac{1}{2}}.
\end{equation}
We need to compute $\cos\theta_-\sin\theta_+$, since this is the quantity that affects the fidelity of the protocol. We begin by computing $\sin^2\theta_+$:
\begin{align}
    \sin^2\theta_+ &= \frac{1}{2} + \left(\lambda + \frac{b_{11}}{N_C} + \frac{b_{12}}{N_C^2}\right)\left(z + \frac{1}{2N_C}\right) \\
    & \quad + \left(b_{30} + \frac{b_{31}}{N_C}\right)\left(z + \frac{1}{2N_C}\right)^3 + b_{50}\left(z + \frac{1}{2N_C}\right)^5 + O\left(N_C^{-7/2}\right) \\
    &= \frac{1}{2} + \lambda z + \frac{\lambda}{2}\frac{1}{N_C} + b_{11}\frac{z}{N_C} + b_{30}z^3 + \frac{b_{11}}{2}\frac{1}{N_C^2} + \frac{3}{2}b_{30}\frac{z^2}{N_C} \\
    & \quad + \left[b_{12} + \frac{3}{4}b_{30}\right]\frac{z}{N_C^2} + b_{31}\frac{z^3}{N_C} + b_{50}z^5 + \left[\frac{b_{12}}{2} + \frac{b_{30}}{8}\right]\frac{1}{N_C^3} \\
    & \quad + \frac{3}{2}b_{31}\frac{z^2}{N_C^2} + \frac{5}{2}b_{50}\frac{z^4}{N_C} + O\left(N_C^{-7/2}\right).
\end{align}
We now compute $\cos^2\theta_-$:
\begin{align}
    \cos^2\theta_- &= \frac{1}{2} - \left(\lambda + \frac{b_{11}}{N_C} + \frac{b_{12}}{N_C^2}\right)\left(z - \frac{1}{2N_C}\right) \\
    & \quad - \left(b_{30} + \frac{b_{31}}{N_C}\right)\left(z - \frac{1}{2N_C}\right)^3 - b_{50}\left(z - \frac{1}{2N_C}\right)^5 + O\left(N_C^{-7/2}\right) \\
    &= \frac{1}{2} - \lambda z + \frac{\lambda}{2}\frac{1}{N_C} - b_{11}\frac{z}{N_C} - b_{30}z^3 + \frac{b_{11}}{2}\frac{1}{N_C^2} + \frac{3}{2}b_{30}\frac{z^2}{N_C} \\
    & \quad + \left[-b_{12} - \frac{3}{4}b_{30}\right]\frac{z}{N_C^2} - b_{31}\frac{z^3}{N_C} - b_{50}z^5 + \left[\frac{b_{12}}{2} + \frac{b_{30}}{8}\right]\frac{1}{N_C^3} \\
    & \quad + \frac{3}{2}b_{31}\frac{z^2}{N_C^2} + \frac{5}{2}b_{50}\frac{z^4}{N_C} + O\left(N_C^{-7/2}\right).
\end{align}
We now multiply the two previous quantities together:
\begin{align}
    \cos^2\theta_-\sin^2\theta_+ &= \frac{1}{4} + \frac{\lambda}{2}\frac{1}{N_C} - \lambda^2z^2 + \left[\frac{\lambda^2}{4} + \frac{b_{11}}{2}\right]\frac{1}{N_C^2} \\
    & \quad + \left[-2\lambda b_{11} + \frac{3}{2}b_{30}\right]\frac{z^2}{N_C} - 2\lambda b_{30}z^4 \\
    & \quad + \left[\frac{\lambda b_{11}}{2} + \frac{b_{12}}{2} + \frac{b_{30}}{8}\right]\frac{1}{N_C^3} \\
    & \quad + \left[-b_{11}^2 \mathcolor{red}{+ \frac{3}{2}\lambda b_{30}} - 2\lambda b_{12} \mathcolor{red}{- \frac{3}{2}\lambda b_{30}} + \frac{3}{2}b_{31}\right]\frac{z^2}{N_C^2} \\
    & \quad + \left[\frac{5}{2}b_{50} - 2\lambda b_{31} - 2b_{11}b_{30}\right]\frac{z^4}{N_C} \\
    & \quad + \left[-b_{30}^2 - 2\lambda b_{50}\right]z^6 + O\left(N_C^{-4}\right).
\end{align}
For readability, we highlight two of the terms in \textcolor{red}{red} to show that they cancel out. (It is unclear whether there is any ``meaningful'' reason for this cancellation.) The more interesting observation at this point is that the terms with an odd power of $z$ all cancel out. This is a unique feature of the equatorial special case, and it is a direct consequence of the condition $\theta_w + \theta_{N_C-w} = \frac{\pi}{2}$.

\vspace{0.5\baselineskip}

We now wish to take the square root of the above quantity. To do so, we use the Taylor series
\begin{equation}
    \sqrt{\frac{1}{4} + \varepsilon} = \frac{1}{2} + \varepsilon - \varepsilon^2 + 2\varepsilon^3 + O(\varepsilon^4).
\end{equation}
When we use this Taylor series to compute $\cos\theta_-\sin\theta_+$, we obtain (time for another deep breath...)
\begin{align}
    \cos\theta_-\sin\theta_+ &= \frac{1}{2} + \Bigg\{\frac{\lambda}{2}\frac{1}{N_C} - \lambda^2z^2 + \left[\textcolor{red}{\frac{\lambda^2}{4}} + \frac{b_{11}}{2}\right]\frac{1}{N_C^2} \\
    & \quad + \left[-2\lambda b_{11} + \frac{3}{2}b_{30}\right]\frac{z^2}{N_C} - 2\lambda b_{30}z^4 \\
    & \quad + \left[\textcolor{orange}{\frac{\lambda b_{11}}{2}} + \frac{b_{12}}{2} + \frac{b_{30}}{8}\right]\frac{1}{N_C^3} + \left[-b_{11}^2 - 2\lambda b_{12} + \frac{3}{2}b_{31}\right]\frac{z^2}{N_C^2} \\
    & \quad + \left[\frac{5}{2}b_{50} - 2\lambda b_{31} - 2b_{11}b_{30}\right]\frac{z^4}{N_C} + \left[-b_{30}^2 - 2\lambda b_{50}\right]z^6\Bigg\} \\
    & \quad - \Bigg\{\textcolor{red}{\frac{\lambda^2}{4}\frac{1}{N_C^2}} - \lambda^3\frac{z^2}{N_C} + \lambda^4z^4 + \left[\textcolor{green}{\frac{\lambda^3}{4}} \textcolor{orange}{+ \frac{\lambda b_{11}}{2}}\right]\frac{1}{N_C^3} \\
    & \quad + \left[\textcolor{teal}{-\frac{\lambda^4}{2}} \textcolor{blue}{- \lambda^2b_{11}} \textcolor{blue}{- 2\lambda^2b_{11}} + \frac{3}{2}\lambda b_{30}\right]\frac{z^2}{N_C^2} \\
    & \quad + \left[\textcolor{violet}{-2\lambda^2b_{30}} + 4\lambda^3b_{11} \textcolor{violet}{- 3\lambda^2b_{30}}\right]\frac{z^4}{N_C} + 4\lambda^3b_{30}z^6\Bigg\} \\
    & \quad + 2\Bigg\{\textcolor{green}{\frac{\lambda^3}{8}\frac{1}{N_C^3}} \textcolor{teal}{- \frac{3}{4}\lambda^4\frac{z^2}{N_C^2}} + \frac{3}{2}\lambda^5\frac{z^4}{N_C} - \lambda^6z^6\Bigg\} + O\left(N_C^{-4}\right).
\end{align}
Once again for readability, we highlight some of the terms in \textcolor{red}{red}, \textcolor{orange}{orange}, and \textcolor{green}{green} to show that they cancel out. We also highlight some of the terms in \textcolor{teal}{teal}, \textcolor{blue}{blue}, and \textcolor{violet}{violet} to show that they are like terms and hence can be combined (but do not cancel out).

\vspace{0.5\baselineskip}

We are now ready to evaluate, and then maximize, the fidelity. When working in the equatorial special case, because the output state is also equatorial, it is common to work with the output purity parameter $\tilde{\lambda}$, rather than the output fidelity $\mF$. Of course, these two quantities are related in a very straightforward way:
\begin{equation}
    \mF = \frac{1+\tilde{\lambda}}{2}.
\end{equation}
As we showed in Appendix \ref{appendix:boundary-value-problem}, the formula for the output fidelity simplifies to no longer depend on the $P_{w,w}$ values in the equatorial case. In particular, the output purity parameter $\tilde{\lambda}$ can be written as follows:
\begin{equation}
    \tilde{\lambda} = 2\mF - 1 = 2\sum_{w=1}^{N_C}P_{w-1,w}\cos\theta_{w-1}\sin\theta_w.
\end{equation}
The above expression for $\cos\theta_-\sin\theta_+$ tells us that $\tilde{\lambda}$ will take the form
\begin{equation}
    \tilde{\lambda} = (\cdots)\mM_0^{(1)} + (\cdots)\mM_2^{(1)} + (\cdots)\mM_4^{(1)} + (\cdots)\mM_6^{(1)} + O\left(N_C^{-4}\right),
\end{equation}
where the coefficients of $\mM_p^{(1)}$ for $p=0,2,4,6$ are as follows:
\begin{align}
    \left[\mM_0^{(1)}\text{ coeff}\right] &= 1 + \lambda\frac{1}{N_C} + b_{11}\frac{1}{N_C^2} + \left[b_{12} + \frac{b_{30}}{4}\right]\frac{1}{N_C^3} \\
    \left[\mM_2^{(1)}\text{ coeff}\right] &= -2\lambda^2 + \left[-4\lambda b_{11} + 3b_{30} + 2\lambda^3\right]\frac{1}{N_C} \\
    & \quad + \left[-2b_{11} - 4\lambda b_{12} + 3b_{31} - 2\lambda^4 + 6\lambda^2b_{11} - 3\lambda b_{30}\right]\frac{1}{N_C^2} \\
    \left[\mM_4^{(1)}\text{ coeff}\right] &= \left[-4\lambda b_{30} - 2\lambda^4\right] + \Big[5b_{50} - 4\lambda b_{31} \\
    & \quad - 4b_{11}b_{30} + 10\lambda^2b_{30} - 8\lambda^3b_{11} + 6\lambda^5\Big]\frac{1}{N_C} \\
    \left[\mM_6^{(1)}\text{ coeff}\right] &= -2b_{30}^2 - 4\lambda b_{50} - 8\lambda^3b_{30} - 4\lambda^6.
\end{align}

\vspace{0.5\baselineskip}

Now comes the part where we plug in the formulas for $\mM_p^{(1)}$, which we stated at the start of this appendix, and which we show how to derive in Appendix \ref{appendix:understanding-P-vals}\ref{subsec:P-vals-moments-offset1-equatorial}. This is the part where the effect of quadratic variation in the vicinity of a local optimum will kick in, and the terms in $\tilde{\lambda}$ that depend on $b_{12}$, $b_{31}$, and $b_{50}$ will ``magically'' cancel out, leaving us with $b_{11}$ and $b_{30}$ as the variables we need to optimize. When we put everything together, we obtain
\begin{equation}
    \tilde{\lambda} = 1 + \left(-\frac{1-\lambda^2}{2\lambda}\right)\frac{1}{N_C} + \left[-\frac{(1-\lambda)^2(1+\lambda)(3-\lambda)}{8\lambda^2}\right]\frac{1}{N_C^2} + (\cdots)\frac{1}{N_C^3} + O\left(N_C^{-4}\right),
\end{equation}
where we write the coefficient of $N_C^{-3}$ separately for readability:
\begin{align}
    \left[N_C^{-3}\text{ coeff}\right] &= Ab_{11}^2 + 2Bb_{11}b_{30} + Cb_{30}^2 + Db_{11} + Eb_{30} + F \\
    A &= -\frac{1}{2\lambda} \\
    B &= -\frac{3}{8\lambda^2} \\
    C &= -\frac{15}{32\lambda^3} \\
    D &= \frac{1-\lambda}{\lambda} \\
    E &= \frac{3(1-\lambda)(2+\lambda)}{4\lambda^2} \\
    F &= \frac{(1-\lambda)^2(-9 - 6\lambda - 16\lambda^2 - 18\lambda^3 - 7\lambda^4)}{16\lambda^3}.
\end{align}

\vspace{0.5\baselineskip}

In general, a two-variable quadratic polynomial of the form
\begin{equation}
    f(x,y) = Ax^2 + 2Bxy + Cy^2 + Dx + Ey + F
\end{equation}
with $A^2 - BC \neq 0$ has the unique critical point $(x_0, y_0)$ and evaluation at that critical point as follows:
\begin{align}
    \left(x_0, y_0\right) &= \left(\frac{BE - CD}{2(AC - B^2)}, \frac{BD - AE}{2(AC - B^2)}\right) \\
    f\left(x_0, y_0\right) &= \frac{2BDE - CD^2 - AE^2}{4(AC - B^2)} + F.
\end{align}
Furthermore, since $A,C < 0$ and $AC - B^2 > 0$, the Hessian is negative definite, so this critical point is indeed a local maximum (and thus also the global maximum). Plugging in the values shown above yields the following solution for our protocol:
\begin{align}
    b_{11} = \frac{BE - CD}{2(AC - B^2)} &= -\frac{(1-\lambda)(1+3\lambda)}{2} \\
    b_{30} = \frac{BD - AE}{2(AC - B^2)} &= 2\lambda(1-\lambda^2).
\end{align}
The resulting coefficient of $N_C^{-3}$ in $\tilde{\lambda}$ is
\begin{equation}
    \left[N_C^{-3}\text{ coeff}\right] = -\frac{(1-\lambda)^3(1+\lambda)(9 + 6\lambda + 5\lambda^2)}{16\lambda^3}.
\end{equation}
Putting this all together, we obtain $\tilde{\lambda}$ and the output infidelity for a $3^{\text{rd}}$-order equatorial optimal protocol:
\begin{align}
    \tilde{\lambda} = & \,\, 1 + \left(-\frac{1-\lambda^2}{2\lambda}\right)\frac{1}{N_C} + \left[-\frac{(1-\lambda)^2(1+\lambda)(3-\lambda)}{8\lambda^2}\right]\frac{1}{N_C^2} \\
    & \,\, + \left[-\frac{(1-\lambda)^3(1+\lambda)(9 + 6\lambda + 5\lambda^2)}{16\lambda^3}\right]\frac{1}{N_C^3} + O\left(N_C^{-4}\right) \\
    \mI(\mE) = & \,\, \frac{1 - \tilde{\lambda}}{2} \\
    = & \,\, \left(\frac{1-\lambda^2}{4\lambda}\right)\frac{1}{N_C} + \left[\frac{(1-\lambda)^2(1+\lambda)(3-\lambda)}{16\lambda^2}\right]\frac{1}{N_C^2} \\
    & + \left[\frac{(1-\lambda)^3(1+\lambda)(9 + 6\lambda + 5\lambda^2)}{32\lambda^3}\right]\frac{1}{N_C^3} + O\left(N_C^{-4}\right).
\end{align}

\vspace{0.5\baselineskip}

Finally, as always, we use the negative moments of the angular momentum outcome distribution to convert the power series in $N_C^{-1}$ to a power series in $N^{-1}$. We now need to expand the expectation values of $N_C^{-p}$ for $p=1,2,3$ to the $\Theta\left(N_C^{-3}\right)$ order. As we show in Lemma \ref{lem:NC-negative-moments} in Appendix \ref{appendix:angular-momentum-moments}, the relevant moments look as follows:
\begin{align}
    \mathbb{E}\left[\frac{1}{N_C}\right] &= \left(\frac{1}{\lambda}\right)\frac{1}{N} + \left(\frac{1 - \lambda}{\lambda^2}\right)\frac{1}{N^2} + \left[\frac{(1 - \lambda)(1 + 2\lambda - \lambda^2)}{\lambda^4}\right]\frac{1}{N^3} + O\left(N^{-4}\right) \\
    \mathbb{E}\left[\frac{1}{N_C^2}\right] &= \left(\frac{1}{\lambda^2}\right)\frac{1}{N^2} + \left[\frac{(1 - \lambda)(1 + 3\lambda)}{\lambda^4}\right]\frac{1}{N^3} + O\left(N^{-4}\right) \\
    \mathbb{E}\left[\frac{1}{N_C^3}\right] &= \left(\frac{1}{\lambda^3}\right)\frac{1}{N^3} + O\left(N^{-4}\right).
\end{align}
Plugging these values into the infidelity as a power series in $N_C^{-3}$ finally gives us the infidelity series of a $3^{\text{rd}}$-order optimal protocol (of course, only the $N^{-3}$ term is new):
\begin{equation}
    \mI(\mE) = \frac{1 - \lambda^2}{4\lambda^2}\frac{1}{N} + \frac{3(1 - \lambda^2)^2}{16\lambda^4}\frac{1}{N^2} + \frac{(1 - \lambda^2)^2(15 - 7\lambda^2)}{32\lambda^6}\frac{1}{N^3} + O\left(N^{-4}\right).
\end{equation}
Just as for $2^{\text{nd}}$-order optimality, it is unclear whether there is any deeper meaning to this value. One observation is that, while the $N^{-2}$ term has a coefficient that is just a constant times a power of the purity of coherence $P_H(\rho) = \frac{4\lambda^2}{1-\lambda^2}$, the $N^{-3}$ term has a coefficient that does not satisfy this.

%% file: app_ang_mom_moments.tex
\appsec{Distribution of Angular Momentum Measurement Outcomes in Schur Sampling}
{appendix:angular-momentum-moments}

Throughout this paper, we have made frequent use of the following all-important quantity:
\begin{equation}
    N_C \equiv \text{number of qubits remaining after Schur sampling.}
\end{equation}
This quantity is important to define because, as we discuss in Appendix \ref{appendix:deriving-kraus-rep}, it always makes sense to start with Schur sampling \cite{Cirac1999}, and then follow it up with another procedure that is defined by the Kraus operators on the $N_C$-qubit symmetric subspace (since the Schur-sampled state $\rho_j$ with $j = N_C/2$ has support only on the $N_C$-qubit symmetric subspace).

\vspace{0.5\baselineskip}

When studied in a ``post-Schur-sampling'' sense, a distillation protocol will produce contributions to the infidelity in terms of $N_C$. (In particular, for the asymptotic protocols we study, the infidelity will be a power series in $N_C^{-1}$.) However, we would ideally like to write our infidelity in terms of the original number of qubits $N$. As a result, we need to compute expectations of the form $\mathbb{E}\left[N_C^{-p}\right]$ for positive integers $p$. To optimize the leading-order infidelity and demonstrate the operational meaning of purity of coherence, we only need $p=1$, but to find more fine-grained contributions to the infidelity, we need higher values of $p$ as well.

\vspace{0.5\baselineskip}

Cirac et al. state the mean of $N_C$ with one extra correction term \cite{Cirac1999}, although they do not prove it:
\begin{equation}
    \mathbb{E}\left[N_C\right] = N\left[\lambda + \frac{1-\lambda}{\lambda}\frac{1}{N} + O\left(\frac{1}{N^2}\right)\right].
\end{equation}
As a brief note, the ``yield'' is generally defined as the \textit{fraction} of the original number of qubits that are present in the final state, so $N_C$ is precisely $N$ times the yield. In particular, the expected yield is the quantity in brackets above.

\vspace{0.5\baselineskip}

In this appendix, we will prove this formula and many others using a variety of methods. In the past, it has already been shown that, for large $N$, the angular momentum distribution outcome is approximately a Gaussian with mean $\sim\lambda N$ and variance $\sim(1-\lambda^2)N$ \cite{Keyl2001}. In fact, this is already good enough for our first-order results in Appendix \ref{appendix:1st-order-optimality}. However, we will dive much deeper than that. As just one example, we will prove the above formula for $\Ebb[N_C]$ including the $\Theta(N^{-1})$ correction term, but we will also show that the correction afterward is much smaller than $O(N^{-2})$.

\vspace{0.5\baselineskip}

To the best of our knowledge, this is the most detailed study into the moments of this distribution to date. As result, we believe that this appendix may be of independent interest to people who use Schur sampling as a primitive for various problems.

\vspace{0.5\baselineskip}

This appendix is broken down as follows:
\begin{itemize}
    \item In Appendix \ref{appendix:angular-momentum-moments}\ref{subsec:J2-integer-moments}, we show how to compute the positive integer moments of the total angular momentum operator $J^2$, which has eigenvalues $j(j+1) = \frac{N_C(N_C+2)}{4}$.
    \item In Appendix \ref{appendix:angular-momentum-moments}\ref{subsec:J2-moments-to-NC-moments}, we show how to use the positive moments of $N_C(N_C+2)$ to compute the positive moments of $N_C$. We suspect that this method of computing the moments of $N_C$ (by going through the moments of $J^2$) will be most pleasing to those well-versed in the representation theory of $SU(2)$. This method also has the advantage of allowing one to compute the moments of $N_C$ systematically without guessing the answer in advance.
    \item In Appendix \ref{appendix:angular-momentum-moments}\ref{subsec:positive-moments-to-negative-moments}, we show how to use the positive moments of $N_C$ to compute the negative moments of $N_C$, i.e., expressions of the form $\Ebb[N_C^{-p}]$ for $p\in\Nbb$. These negative moments are the ones used throughout this paper to convert power series in $N_C^{-1}$ to power series in $N^{-1}$.
    \item In Appendix \ref{appendix:angular-momentum-moments}\ref{subsec:NC-moments-ese}, we show an original method to compute the positive moments of $N_C$ up to exponentially small error, which we present as Lemma \ref{lem:NC-moments-exp-small-error}. The fact that we can achieve exponentially small error in these approximations is not apparent from the method we use in Appendix \ref{appendix:angular-momentum-moments}\ref{subsec:J2-moments-to-NC-moments}, and as far as we know, it is not mentioned in any prior literature that invokes the Schur transform. The only disadvantage of this method is that we had to guess the answers based on numerics before we could prove them.
    \item In Appendix \ref{appendix:angular-momentum-moments}\ref{subsec:NC-negative-moments-full-computation}, we combine the improved results for the positive moments of $N_C$ from Appendix \ref{appendix:angular-momentum-moments}\ref{subsec:NC-moments-ese} with the method shown in Appendix \ref{appendix:angular-momentum-moments}\ref{subsec:positive-moments-to-negative-moments} to finally prove the negative moments of $N_C$, which we present as Lemma \ref{lem:NC-negative-moments}.
\end{itemize}

\appsubsec{Computing Positive Integer Moments of $J^2$ Total Angular Momentum Operator}
{subsec:J2-integer-moments}

When computing the moments of $N_C$, a useful trick is to take advantage of the total angular momentum operator
\begin{equation}
    J^2 = J_x^2 + J_y^2 + J_z^2,
\end{equation}
where the term $J_x^2$ can be expanded as
\begin{equation}
    J_x = \frac{1}{2}\sum_{i=1}^{N}X_i \implies 4J_x^2 = N\mathbb{I} + 2\sum_{1\le i<j\le N}X_iX_j,
\end{equation}
and similarly for $J_y^2$ and $J_z^2$. The $J^2$ operator is known to have eigenvalues $j(j+1)$, and $N_C = 2j$, so
\begin{equation}
    \mathbb{E}\left[\left(N_C(N_C+2)\right)^p\right] = \text{Tr}\left[\rho^{\otimes N}\left(4J^2\right)^p\right].
\end{equation}
Therefore, we can more easily compute the expected value of specific polynomials of $N_C$, in particular those that can be written as polynomials of $N_C(N_C+2)$. The first $p$ equations tell us about moments of $N_C$ up to $2p$, but we would need $2p$ equations to solve for all of them, so these equations only give us ``half'' of what we want. For example, we wish to have $\mathbb{E}\left[N_C\right]$ and $\mathbb{E}\left[N_C^2\right]$ separately, but the above equation with $p=1$ only gives us a specific combination of these two quantities. However, it is certainly a step in the right direction, and if we ever compute one of these expressions using some other method, we will get the other for free. Similarly, when we move up to $p=2$, we obtain an expression for $\mathbb{E}\left[N_C^4 + 4N_C^3 + 4N_C^2\right]$. Therefore, if we compute any two of the three quantities $\mathbb{E}\left[N_C^2\right]$, $\mathbb{E}\left[N_C^3\right]$, and $\mathbb{E}\left[N_C^4\right]$, we can quickly obtain the third.

\vspace{0.5\baselineskip}

Without loss of generality, we may assume that $\rho$ has its Bloch vector pointing in the positive $x$-direction, which implies that $
\rho^{\otimes N}$ takes the following form as a linear combination of Pauli strings:
\begin{equation}
    \rho = \frac{\mathbb{I} + \lambda X}{2} \implies \rho^{\otimes N} = \frac{1}{2^N}\sum_{S\subseteq[N]}\lambda^{\abs{S}}X^{\otimes S}.
\end{equation}
From here, we can easily compute $\text{Tr}\left[\rho^{\otimes N}\left(4J_x^2\right)\right]$, $\text{Tr}\left[\rho^{\otimes N}\left(4J_y^2\right)\right]$, and $\text{Tr}\left[\rho^{\otimes N}\left(4J_z^2\right)\right]$. Recall that the trace of an $N$-qubit Pauli string is zero unless it is the identity, in which case it is $2^N$. Therefore, to compute $\text{Tr}\left[\rho^{\otimes N}\left(4J_x^2\right)\right]$, for example, we simply find all instances where a Pauli string in $\rho^{\otimes N}$ matches a Pauli string in $4J_x^2$, multiply their coefficients, add all these numbers, and multiply by $2^N$. When we do this for $4J_x^2$, $4J_y^2$, and $4J_z^2$, we obtain the following:
\begin{align}
    \text{Tr}\left[\rho^{\otimes N}\left(4J_x^2\right)\right] &= \frac{1}{2^N}\cdot \left[N + 2\binom{N}{2}\lambda^2\right]\cdot 2^N = N + N(N-1)\lambda^2 \\
    \text{Tr}\left[\rho^{\otimes N}\left(4J_y^2\right)\right] &= \frac{1}{2^N}\cdot N\cdot 2^N = N \\
    \text{Tr}\left[\rho^{\otimes N}\left(4J_z^2\right)\right] &= \frac{1}{2^N}\cdot N\cdot 2^N = N.
\end{align}
We now add these three expressions to obtain
\begin{align}
    & \text{Tr}\left[\rho^{\otimes N}\left(4J^2\right)\right] = N(N-1)\lambda^2 + 3N \\
    \iff & \boxed{\mathbb{E}\left[N_C^2 + 2N_C\right] = \lambda^2N^2 + \left(3 - \lambda^2\right)N},
\end{align}
where in the second line, we have rewritten the expression in standard polynomial form with respect to $N$ (instead of with respect to $\lambda$). We have thus obtained $\mathbb{E}\left[N_C(N_C+2)\right]$, as desired.

\vspace{0.5\baselineskip}

We now proceed to compute $\mathbb{E}\left[N_C^2(N_C+2)^2\right]$. We begin by expanding $\left(4J^2\right)^2$ as follows:
\begin{align}
    \left(4J^2\right)^2 &= 16\sum_{\text{sym}}J_x^4 + 16\sum_{\text{sym}}J_x^2J_y^2 \\
    &= 16\left(J_x^4 + J_y^4 + J_z^4 + J_x^2J_y^2 + J_y^2J_x^2 + J_x^2J_z^2 + J_z^2J_x^2 + J_y^2J_z^2 + J_z^2J_y^2\right).
\end{align}
We now expand each of these quantities in turn. For example,
\begin{align}
    16J_x^4 &= \left[N + \binom{4}{2}\binom{N}{2}\right]\mathbb{I} + 24\sum_{i<j<k<l}X_iX_jX_kX_l + \left[(N-2)\binom{4}{2}\cdot 2 + 2\cdot 4\right]\sum_{i<j}X_iX_j \\
    &= (3N^2 - 2N)\mathbb{I} + 24\sum_{i<j<k<l}X_iX_jX_kX_l + (12N - 16)\sum_{i<j}X_iX_j,
\end{align}
and similarly for $16J_y^4$ and $16J_z^4$. Furthermore,
\begin{align}
    16J_x^2J_y^2 = N^2\mathbb{I} + 2N\sum_{i<j}X_iX_j + 2N\sum_{i<j}Y_iY_j - 4\sum_{i<j}Z_iZ_j + 4i\sum_{i,j,k\text{ distinct}}X_iY_jZ_k,
\end{align}
and similarly for the other five products of two distinct operators chosen from $4J_x^2$, $4J_y^2$, and $4J_z^2$. As a result, we can derive the following expectation values:
\begin{align}
    \text{Tr}\left[\rho^{\otimes N}\left(16J_x^4\right)\right] &= (3N^2 - 2N)1 + 24\binom{N}{4}\lambda^4 + (12N - 16)\binom{N}{2}\lambda^2 \\
    &= \lambda^4N^4 + 6\lambda^2\left(1 - \lambda^2\right)N^3 + \left(3 - 11\lambda^2\right)\left(1 - \lambda^2\right)N^2 + (-2)\left(1 - 3\lambda^2\right)\left(1 - \lambda^2\right)N
\end{align}
\begin{align}
    \text{Tr}\left[\rho^{\otimes N}\left(16J_x^2J_y^2\right)\right] = \text{Tr}\left[\rho^{\otimes N}\left(16J_x^2J_z^2\right)\right] &= N^2\cdot 1 + 2N\binom{N}{2}\lambda^2 = \lambda^2N^3 + \left(1 - \lambda^2\right)N^2 \\
    \text{Tr}\left[\rho^{\otimes N}\left(16J_y^2J_x^2\right)\right] = \text{Tr}\left[\rho^{\otimes N}\left(16J_z^2J_x^2\right)\right] &= N^2\cdot 1 + 2N\binom{N}{2}\lambda^2 = \lambda^2N^3 + \left(1 - \lambda^2\right)N^2
\end{align}
\begin{align}
    \text{Tr}\left[\rho^{\otimes N}\left(16J_y^4\right)\right] = \text{Tr}\left[\rho^{\otimes N}\left(16J_z^4\right)\right] &= 3N^2 - 2N \\
    \text{Tr}\left[\rho^{\otimes N}\left(16J_y^2J_z^2\right)\right] = \text{Tr}\left[\rho^{\otimes N}\left(16J_z^2J_y^2\right)\right] &= (1-2\lambda^2)N^2 + 2\lambda^2N.
\end{align}

Adding these all together yields
\begin{equation}
    \text{Tr}\left[\rho^{\otimes N}\left(4J^2\right)^2\right] = \lambda^4N^4 + 2\lambda^2\left(5 - 3\lambda^2\right)N^3 + \left(15 - 22\lambda^2 + 11\lambda^4\right)N^2 - 6(1-\lambda)^2(1+\lambda)^2N,
\end{equation}
from which we conclude that
\begin{equation}
    \boxed{\mathbb{E}\left[N_C^4 + 4N_C^3 + 4N_C^2\right] = \lambda^4N^4 + 2\lambda^2\left(5 - 3\lambda^2\right)N^3 + \left(15 - 18\lambda^2 + 11\lambda^4\right)N^2 + 2\left(-3 + 4\lambda^2 - 3\lambda^4\right)N}.
\end{equation}
We do not need any higher moments of $N_C(N_C+2)$ for the computations we do in this paper, so we stop here.

\appsubsec{Going from Moments of $J^2$ to Moments of $N_C$}
{subsec:J2-moments-to-NC-moments}

So far, we have only shown how to compute the moments of $N_C(N_C+2)$. But how do we compute the moments of $N_C$ itself? As we observed in the previous subsection, considering the first $p$ moments of $N_C(N_C+2)$ yields $p$ equations for the first $2p$ moments of $N_C$, but we would need $2p$ equations to tease apart all these moments.

\vspace{0.5\baselineskip}

Here is an intriguing method that skirts the need for more equations. Since $N_C$ is a function of $N_C(N_C+2)$ via the identity
\begin{equation}
    N_C = \sqrt{N_C(N_C+2) + 1} - 1,
\end{equation}
we can consider the following question: given the moments of a random variable $X$, how do we compute the moments of $\sqrt{X}$? It turns out that there is a simple trick we can use. First, we use the following Taylor series expansion:
\begin{align}
    \sqrt{X} &= \sqrt{\mu}\sqrt{1 + \frac{X - \mu}{\mu}} \\
    &= \sqrt{\mu}\left[1 + \frac{1}{2}\frac{X - \mu}{\mu} + \left(-\frac{1}{8}\right)\frac{(X - \mu)^2}{\mu^2} + \frac{1}{16}\frac{(X - \mu)^3}{\mu^3} + \cdots\right] \\
    \therefore \mathbb{E}\left[\sqrt{X}\right] &= \sqrt{\mu}\left[1 + \frac{1}{2}\frac{\mathbb{E}[X - \mu]}{\mu} + \left(-\frac{1}{8}\right)\frac{\mathbb{E}\left[(X - \mu)^2\right]}{\mu^2} + \frac{1}{16}\frac{\mathbb{E}\left[(X - \mu)^3\right]}{\mu^3} + \cdots\right]
\end{align}
In the above formula, $\mu$ can in fact be anything. But a particularly natural choice is
\begin{equation}
    \mu \coloneq \mathbb{E}[X] \implies \mathbb{E}\left[\sqrt{X}\right] = \sqrt{\mu}\left[1 + \left(-\frac{1}{8}\right)\frac{\text{Var}(X)}{\mu^2} + \frac{1}{16}\frac{\mathbb{E}\left[(X - \mu)^3\right]}{\mu^3} + \cdots\right].
\end{equation}
In this way, we have successfully written a moment of $\sqrt{X}$ in terms of the mean of $X$ and the positive centered moments of $X$. Of course, this series may converge very slowly or not at all. However, if you have a distribution that is fairly concentrated around a positive value far from zero and which is also bounded away from zero, you can have a reasonable convergence rate on this series. In particular, for large $N$, the distribution of $N_C$ converges to a Gaussian with mean and variance both $\Theta(N)$, and such a distribution will broadly exhibit reasonable convergence.

\vspace{0.5\baselineskip}

We begin by defining $X$ and $\mu$ as follows:
\begin{align}
    X \coloneq \left(N_C + 1\right)^2 = N_C\left(N_C + 2\right) + 1 \\
    \mu \coloneq \mathbb{E}[X] = \text{Tr}\left[\rho^{\otimes N}\left(4J^2\right)\right] + 1.
\end{align}
We now compute the mean and variance of $X$:
\begin{align}
    \mathbb{E}\left[X\right] &= \text{Tr}\left[\rho^{\otimes N}\left(4J^2\right)\right] + 1 \\
    &= \lambda^2N^2 + \left(3 - \lambda^2\right)N + 1 \\
    \text{Var}\left[X\right] &= \text{Var}\left[N_C\left(N_C + 2\right)\right] \\
    &= \text{Tr}\left[\rho^{\otimes N}\left(4J^2\right)^2\right] - \text{Tr}\left[\rho^{\otimes N}\left(4J^2\right)\right]^2 \\
    &= \left[\lambda^4N^4 + 2\lambda^2\left(5 - 3\lambda^2\right)N^3 + \left(15 - 18\lambda^2 + 11\lambda^4\right)N^2 + 2\left(-3 + 4\lambda^2 - 3\lambda^4\right)N\right] \\
    & \quad - \left[\lambda^2N^2 + \left(3 - \lambda^2\right)N\right]^2 \\
    &= 4\lambda^2\left(1 - \lambda^2\right)N^3 + 2\left(3 - 6\lambda^2 + 5\lambda^4\right)N^2 + 2\left(-3 + 4\lambda^2 - 3\lambda^4\right)N.
\end{align}
We can then compute the quantity $\frac{\text{Var}(X)}{\mu^2}$ that appears in the second term in the infinite series:
\begin{equation}
    \frac{\text{Var}[X]}{\mu^2} = \frac{4\left(1 - \lambda^2\right)}{\lambda^2}N^{-1} + O\left(N^{-2}\right).
\end{equation}
We also compute $\sqrt{\mu}$ using the approximation $\sqrt{1+\epsilon} = 1 + \frac{\epsilon}{2} + O\left(\epsilon^2\right)$:
\begin{align}
    \mu &= \lambda^2N^2 + \left(3 - \lambda^2\right)N + 1 \\
    &= \lambda^2N^2\left[1 + \frac{3 - \lambda^2}{\lambda^2}\frac{1}{N} + \frac{1}{\lambda^2}\frac{1}{N^2}\right] \\
    \implies \sqrt{\mu} &= \lambda N\left[1 + \frac{3 - \lambda^2}{2\lambda^2}N^{-1} + O\left(N^{-2}\right)\right].
\end{align}
We can now put everything together as follows:
\begin{align}
    \mathbb{E}\left[\sqrt{X}\right] &= \sqrt{\mu}\left[1 - \frac{\text{Var}[X]}{8\mu^2} + O\left(N^{-2}\right)\right] \\
    &= \lambda N\left[1 + \frac{3 - \lambda^2}{2\lambda^2}N^{-1} + O\left(N^{-2}\right)\right]\left[1 - \frac{1 - \lambda^2}{2\lambda^2}N^{-1} + O\left(N^{-2}\right)\right] \\
    &= \lambda N\left[1 + \frac{1}{\lambda^2}N^{-1} + O\left(N^{-2}\right)\right] \\
    &= \lambda N + \frac{1}{\lambda} + O\left(N^{-1}\right).
\end{align}
Finally, since $\sqrt{X} = N_C-1$, all that remains is to subtract $1$ from $\mathbb{E}[X]$:
\begin{align}
    \mathbb{E}\left[N_C + 1\right] &= \lambda N + \frac{1}{\lambda} + O\left(N^{-1}\right) \\
    \implies \mathbb{E}\left[N_C\right] &= \lambda N + \frac{1-\lambda}{\lambda} + O\left(N^{-1}\right).
\end{align}
We have thus derived the $\mathbb{E}[N_C]$ term to $\Theta(1)$ precision and demonstrated the yield formula that is shown without proof in \cite{Cirac1999}. As a reminder, they define the yield as the \textit{fraction} of qubits remaining after the purification procedure, so they show the above expression divided by $N$.

\vspace{0.5\baselineskip}

This method has two key advantages. The first is that it computes the moments of $N_C$ using the moments of $N_C(N_C+2)$, without the need for the ``extra equations'' that we previously thought we might need. Another advantage is that it shows how one can systematically compute the moments of $N_C$ without guessing the answers in advance (via numerics or some other method). We also find this method to be pleasing for two other reasons. First, this method invokes the $J^2$ operator, which is a familiar object of interest to those well-versed in representation theory and quantum information theory. Second, this method bounds the error terms using the asymptotic normality of the distribution of $N_C$, which is another broadly important concept. In fact, the asymptotic normality of $N_C$ is a special case of a more general asymptotic normality result for the distribution of invariant subspaces of $\rho^{\otimes N}$ for a qudit state $\rho$, which are indexed by Young diagrams with at most $d$ rows \cite{Keyl2001}.

\vspace{0.5\baselineskip}

Unfortunately, this method also has two disadvantages. First, computing the moments of $J^2$ beyond the first two that we showed above becomes increasingly tedious. Second, it leaves open the possibility for inverse polynomial error terms in terms of $N$. However, it turns out that each positive integer moment of $N_C$ equals a polynomial in $N$ plus an \textit{exponentially small error}. We would never be able to see that using the above method. In fact, even revealing that the $\Theta\left(N^{-1}\right)$ term in $\mathbb{E}[N_C]$ is zero would require us to compute some non-obvious information about the third moment of $J^2$. To demonstrate that the positive integer moments of $N_C$ exhibit exponentially small errors beyond polynomials in $N$, we must use a different technique, which we show later in Appendix \ref{appendix:angular-momentum-moments}\ref{subsec:NC-moments-ese}.

\appsubsec{Going from Positive Moments of $N_C$ to Negative Moments of $N_C$}
{subsec:positive-moments-to-negative-moments}

So far, we have only shown positive moments of $N_C$. But how do we compute the negative moments of $N_C$? It turns out that we can recycle the infinite series trick we used in Appendix \ref{appendix:angular-momentum-moments}\ref{subsec:J2-moments-to-NC-moments}. First, we use the following geometric series expansion:
\begin{align}
    \frac{1}{X} &= \frac{1}{\mu}\frac{1}{1 + \frac{X - \mu}{\mu}} \\
    &= \frac{1}{\mu}\left[1 - \frac{X - \mu}{\mu} + \frac{(X - \mu)^2}{\mu^2} - \frac{(X - \mu)^3}{\mu^3} + \cdots\right] \\
    \therefore \mathbb{E}\left[\frac{1}{X}\right] &= \frac{1}{\mu}\left[1 - \frac{\mathbb{E}[X - \mu]}{\mu} + \frac{\mathbb{E}\left[(X - \mu)^2\right]}{\mu^2} - \frac{\mathbb{E}\left[(X - \mu)^3\right]}{\mu^3} + \cdots\right]
\end{align}
Once again, in the equation above, $\mu$ can be anything. But we again make a natural choice as follows:
\begin{equation}
    \mu \coloneq \mathbb{E}[X] \implies \mathbb{E}\left[\frac{1}{X}\right] = \frac{1}{\mu}\left[1 + \frac{\text{Var}(X)}{\mu^2} - \frac{\mathbb{E}\left[(X - \mu)^3\right]}{\mu^3} + \cdots\right].
\end{equation}
In this way, we have successfully written the mean of $X^{-1}$ in terms of the mean of $X$ and the positive centered moments of $X$.

\vspace{0.5\baselineskip}

To compute higher negative moments of $X$, we use the fact that, for all $p\in\Nbb$,
\begin{equation}
    (1-x)^{-p} = \sum_{n=0}^{\infty}\binom{n+p-1}{p-1}x^n,
\end{equation}
which can be obtained by differentiating the geometric series $(p-1)$ times (among many other ways). Therefore, using the same sequence of manipulations as we did to compute $\Ebb[X^{-1}]$, we conclude that
\begin{equation}
    \mathbb{E}\left[\frac{1}{X^p}\right] = \frac{1}{\mu^p}\left[1 - p\frac{\mathbb{E}[X - \mu]}{\mu} + \frac{p(p+1)}{2}\frac{\mathbb{E}\left[(X - \mu)^2\right]}{\mu^2} - \frac{p(p+1)(p+2)}{6}\frac{\mathbb{E}\left[(X - \mu)^3\right]}{\mu^3} + \cdots\right],
\end{equation}
which can again be simplified using a smart choice of $\mu$:
\begin{equation}
    \mu \coloneq \mathbb{E}[X] \implies\mathbb{E}\left[\frac{1}{X^p}\right] = \frac{1}{\mu^p}\left[1 + \frac{p(p+1)}{2}\frac{\text{Var}(X)}{\mu^2} - \frac{p(p+1)(p+2)}{6}\frac{\mathbb{E}\left[(X - \mu)^3\right]}{\mu^3} + \cdots\right].
\end{equation}

\vspace{0.5\baselineskip}

Similarly to the infinite series we showed in Subsection \ref{appendix:angular-momentum-moments}\ref{subsec:J2-moments-to-NC-moments}, these series may converge very slowly or not at all for a general distribution. However, the asymptotic normality of $N_C$ with mean and variance both $\Theta(N)$ will come to the rescue again and provide the same type of term-by-term decay that we saw previously. In particular, the $p^{\text{th}}$ centered moment of $N_C$ is $\Theta\left(N_C^{\lfloor p/2\rfloor}\right)$. As a result, in the above infinite series, the second term in the square brackets will be $O(1/N)$, the third and fourth terms will be $O(1/N^2)$, the fifth and sixth terms will be $O(1/N^3)$, and so on.

\vspace{0.5\baselineskip}

Using this method, we can finally compute the formulas for $\Ebb[N_C^{-p}]$ that we invoke in many other parts of this paper. In particular, using Appendices \ref{appendix:angular-momentum-moments}\ref{subsec:J2-integer-moments} and \ref{appendix:angular-momentum-moments}\ref{subsec:J2-moments-to-NC-moments}, we can systematically compute the positive moments of $N_C$, and thus also the positive centered moments of $N_C$. Once we have all these moments, we can plug them into the formulas above to compute $\Ebb\left[N_C^{-p}\right]$ to whatever precision we like. Up to $3^{\text{rd}}$-order precision, these negative moments look as follows:

\begin{lemma}[Negative moments of $N_C$ up to $\Theta(N^{-3})$ precision]
\label{lem:NC-negative-moments}
The first three negative moments of the distribution of $N_C$ look as follows:
\begin{align}
    \mathbb{E}\left[\frac{1}{N_C}\right] &= \left(\frac{1}{\lambda}\right)\frac{1}{N} + \left(\frac{1 - \lambda}{\lambda^2}\right)\frac{1}{N^2} + \left[\frac{(1 - \lambda)(1 + 2\lambda - \lambda^2)}{\lambda^4}\right]\frac{1}{N^3} + O\left(N^{-4}\right) \\
    \mathbb{E}\left[\frac{1}{N_C^2}\right] &= \left(\frac{1}{\lambda^2}\right)\frac{1}{N^2} + \left[\frac{(1 - \lambda)(1 + 3\lambda)}{\lambda^4}\right]\frac{1}{N^3} + O\left(N^{-4}\right) \\
    \mathbb{E}\left[\frac{1}{N_C^3}\right] &= \left(\frac{1}{\lambda^3}\right)\frac{1}{N^3} + O\left(N^{-4}\right).
\end{align}
\end{lemma}

However, we will hold off on proving Lemma \ref{lem:NC-negative-moments} for now, because although we \textit{can} prove it using only the information we have provided so far, it would be extremely tedious to do so. For example, computing $\mu^{-1}$ to three terms (to obtain $\Theta(N^{-3})$ precision) requires computing $\mu$ to $3$ terms as well (which corresponds to $\Theta(N^{-1})$ precision). The more precision we want on the negative moments, the more precision we must demand on the positive moments that get used in the above infinite series, but at the moment, those positive moments are themselves infinite series based on the method shown in Appendix \ref{appendix:angular-momentum-moments}\ref{subsec:J2-moments-to-NC-moments}.

\vspace{0.5\baselineskip}

To prove Lemma \ref{lem:NC-negative-moments} more efficiently, we ideally want an improved way to compute the positive moments of $N_C$, and this is precisely what we will show in Appendix \ref{appendix:angular-momentum-moments}\ref{subsec:NC-moments-ese}. Subsequently, in Appendix \ref{appendix:angular-momentum-moments}\ref{subsec:NC-negative-moments-full-computation}, we will plug these improved results (shown in Lemma \ref{lem:NC-moments-exp-small-error}) into the above infinite series to finally prove Lemma \ref{lem:NC-negative-moments}.

\appsubsec{Computing Positive Moments of $N_C$ with Exponentially Small Error}
{subsec:NC-moments-ese}

The method we used in Appendix \ref{appendix:angular-momentum-moments}\ref{subsec:J2-moments-to-NC-moments} to compute $\Ebb[N_C^p]$ may lead you to believe that $\Ebb[N_C^p]$ has an infinite power series expansion, in the sense that it has $N_C^q$ terms for all $-\infty < q\le p$. However, for positive integers $p$, this is actually not the case! It turns out that, for positive integers $p$, $\Ebb[N_C^p]$ is a degree-$p$ polynomial $P^{(p)}_\lambda(N)$ in $N$, up to an \textit{exponentially small error} (which we will commonly abbreviate to ``e.s.e.'' for convenience). In particular, we prove the following result for the first few values of $p$:

\begin{lemma}[Positive moments of $N_C$ with exponentially small error]
\label{lem:NC-moments-exp-small-error}
The first four positive moments of the distribution of $N_C$ can each be written as a polynomial in $N$ plus an exponentially small error. In particular:
\begin{align}
    \mathbb{E}\left[N_C\right] &= \lambda N + \frac{1-\lambda}{\lambda} + O\left(\left(1-\lambda^2\right)^{N/2}\right) \\
    \mathbb{E}\left[N_C^2\right] &= \lambda^2N^2 + (1-\lambda)(3+\lambda)N - \frac{2(1-\lambda)}{\lambda} + O\left(\left(1-\lambda^2\right)^{N/2}\right) \\
    \mathbb{E}\left[N_C^3\right] &= \lambda^3N^3 + 3\lambda(1-\lambda)(2+\lambda)N^2 + \frac{(1-\lambda)\left(3 - 6\lambda - 5\lambda^2 - 2\lambda^3\right)}{\lambda}N \\
    & \quad\quad + \frac{4(1-\lambda)}{\lambda} + O\left(\left(1-\lambda^2\right)^{N/2}\right) \\
    \mathbb{E}\left[N_C^4\right] &= \lambda^4N^4 + 2\lambda^2(1-\lambda)(5+3\lambda)N^3 + (1-\lambda)\left(15 - 9\lambda - 23\lambda^2 - 11\lambda^3\right)N^2 \\
    & \quad\quad + \frac{2(1-\lambda)(2+\lambda)(-3 + 3\lambda + \lambda^2 + 3\lambda^3)}{\lambda}N - \frac{8(1-\lambda)}{\lambda} + O\left(\left(1-\lambda^2\right)^{N/2}\right).
\end{align}
\end{lemma}

Before we prove these formulas, we will first observe some general patterns they exhibit, which in fact apply to $\mathbb{E}\left[N_C^p\right]$ for all nonnegative integers $p$. We will additionally point out when these patterns extend even to negative integers:
\begin{itemize}
    \item $\mathbb{E}\left[N_C^p\right]$ is a degree-$p$ polynomial in $N_C$ plus $O\left(\left(1-\lambda^2\right)^{N/2}\right)$. In fact, as we will see later, the error magnitude can be bounded above by $2^{p-1}\left(\frac{1-\lambda}{\lambda}\right)\left(1-\lambda^2\right)^{N/2}$, and the error term is positive for even values of $p$ and negative for odd values of $p$. \textcolor{red}{(nonnegative $p$ only)}
    \item The leading term in $\mathbb{E}\left[N_C^p\right]$ is $\lambda^pN^p$. This can be understood as a consequence of the fact that the typical values of $N_C$ are $\sim\lambda N$. \textcolor{blue}{(negative $p$ also)}
    \item All non-leading terms have a $(1-\lambda)$ factor. This is because, when $\lambda=1$, $N_C = N$ with probability $1$, so $\mathbb{E}\left[N_C^p\right]$ collapses to $N^p$. \textcolor{blue}{(negative $p$ also)}
    \item For all $q\ge\frac{p}{2}$, the coefficient of $N^q$ in $\mathbb{E}\left[N_C^p\right]$ has $\lambda^{p-2q}$ in the denominator. This can be understood as a consequence of the fact that the quantity $\lambda^2N$ serves as a good measure of whether one is in the asymptotic regime for Schur sampling. \textcolor{red}{(nonnegative $p$ only)}
    \item For all $q < \frac{p}{2}$, the coefficient of $N^q$ in $\mathbb{E}\left[N_C^p\right]$ has $\lambda$ in the denominator. This can be at least partially explained by the fact that, when $\lambda$ is \textit{exactly} zero, $\mathbb{E}\left[N_C^p\right] = \Theta\left(N^{p/2}\right)$, regardless of whether $p$ is odd or even. So there is no harm in having these terms diverge as $\lambda\rightarrow 0$. In fact, when $p$ is odd, the $N^{(p-1)/2}$ term becomes leading as $\lambda\rightarrow 0$, since by the previous item, all the higher powers of $N$ have coefficients going to $0$. Hence, the coefficient of $N^{(p-1)/2}$ must diverge, or else $\mathbb{E}\left[N_C^p\right]$ would be $O\left(N^{(p-1)/2}\right)$, contradicting the above fact. \textcolor{red}{(nonnegative $p$ only)}
\end{itemize}

\vspace{0.5\baselineskip}

The central idea is to write everything in terms of degree-$N$ homogeneous ``polynomials'' in $c_0 = \frac{1-\lambda}{2}$ and $c_1 = \frac{1+\lambda}{2}$. We put ``polynomials'' in quotation marks because we will also need to permit negative powers of $c_1$. In other words, we will write everything in terms of quantities of the form $c_1^{N-k}c_0^k = c_1^N(c_0/c_1)^k$ for nonnegative integers $k$. Since $c_0 + c_1 = 1$, if there is ever an expression of degree less than $N$, we will multiply it by the appropriate power of $(c_0+c_1)$ to make it have degree $N$. For the sake of brevity, we will only present the details for $\Ebb[N_C]$ and $\Ebb[N_C^2]$, but the same general method works for $\Ebb[N_C^3]$ and $\Ebb[N_C^4]$ as well.

\vspace{0.5\baselineskip}

For convenience, we will use the following two notations:
\begin{align}
    (a_0,a_1,a_2,\cdots,a_N) &= \sum_{k=0}^{N}a_kc_1^{N-k}c_0^k \\
    (a_0,a_1,a_2,\cdots,a_N \,|\, a_{N+1},a_{N+2},\cdots) &= \sum_{k=0}^{\infty}a_kc_1^{N-k}c_0^k.
\end{align}
The first notation does not permit negative powers of $c_1$, meaning that is a proper degree-$N$ homogeneous polynomial of $c_0$ and $c_1$. In contrast, the second notation does permit negative powers of $c_1$, meaning that it is an infinite power series. For convenience, we will often call these ``$(c_1,c_0)$ polynomials'' and ``$(c_1,c_0)$ extended polynomials'', respectively.

\vspace{0.5\baselineskip}

This notation is motivated by the fact that the probabilities of different $N_C$ values are given by such polynomials. In particular, the probability of getting $N_C=2j$ qubits when performing the Schur transform on $N=2J$ identical uncorrelated qubits with purity parameter $\lambda$ is as follows \cite{Cirac1999}:
\begin{align}
    p(N,N_C,\lambda) &= d_j(c_1c_0)^{J-j}\frac{c_1^{2j+1} - c_0^{2j+1}}{c_1-c_0} = d_j\sum_{k=0}^{2j}c_1^{J+j-k}c_0^{J-j+k} \\
    d_j &= \binom{2J}{J-j} - \binom{2J}{J-j-1}.
\end{align}

\vspace{0.5\baselineskip}

These values form a very nice pattern for computing moments of $N_C$. For example, consider $N=5$. The possible values of $N_C$ are $N_C=1,3,5$, and they have the following probabilities:
\begin{align}
    p(5,1,\lambda) = 5(c_1c_0)^2(c_1+c_0) &= (0,0,5,5,0,0) \\
    p(5,3,\lambda) = 4(c_1c_0)^1(c_1^3 + c_1^2c_0 + c_1c_0^2 + c_0^3) &= (0,4,4,4,4,0) \\
    p(5,5,\lambda) = 1(c_1c_0)^0(c_1^5 + c_1^4c_0 + c_1^3c_0^2 + c_1^2c_0^3 + c_1c_0^4 + c_0^5) &= (1,1,1,1,1,1).
\end{align}
We can now compute $\Ebb[N_C]$ as a $(c_1,c_0)$ polynomial by multiplying each $p(N,N_C,\lambda)$ by $N_C$ and adding them up:
\begin{align}
    \Ebb[N_C] &= 1\cdot p(5,1,\lambda) + 3\cdot p(5,3,\lambda) + 5\cdot p(5,5,\lambda) \\
    &= 1\cdot(0,0,5,5,0,0) + 3\cdot(0,4,4,4,4,0) + 5\cdot(1,1,1,1,1,1) \\
    &= (5,17,22,22,17,5).
\end{align}
In general, it is easy to see that $\Ebb[N_C^p]$ will always be symmetric $(c_1,c_0)$ polynomial $(a_0,\cdots,a_N)$, in the sense that $a_k = a_{N-k}$. In fact, we can easily write down the $(c_1,c_0)$ polynomial as follows:
\begin{align}
    \Ebb[N_C^p] &= \left(a^{(p)}_0, \cdots, a^{(p)}_N\right) \\
    a^{(p)}_k &= \sum_{j=J-k}^{J}d_j(2j)^p = \sum_{j=J-k}^{J}(2j)^p\left[\binom{2J}{J-j} - \binom{2J}{J-j-1}\right] \quad \left(k\le \lceil J\rceil\right) \\
    a^{(p)}_k &= a^{(p)}_{N-k} \quad \left(\lceil J\rceil+1\le k\le N\right).
\end{align}
Notice that we write $a^{(p)}_k$ as a convenient formula only for $0\le k\le \lceil J\rceil$, which accounts for just over half of the $N+1$ terms in the $(c_0,c_1)$. (For the rest of the values, we use the symmetry that we previously observed.) This fact will become important later.

\vspace{0.5\baselineskip}

On the other side, we also want to write our formulas for these expectation values (ignoring the exponentially small error) as $(c_1,c_0)$ polynomials. However, there are two wrinkles. The first is that we will generally deal with expressions of degree less than $N$, meaning that we will need to multiply them by suitable powers of $(c_1+c_0)$ to raise their degree to $N$. The second is that we will have expressions with $\lambda$ in the denominator, which can only be expressed with an extended $(c_1,c_0)$ polynomial. For example, the quantity $c_1^{N+1}/\lambda$ can be written as an infinite geometric series
\begin{equation}
    \frac{c_1^{N+1}}{\lambda} = \frac{c_1^N}{1-(c_0/c_1)} = \sum_{k=0}^{\infty}c_1^N\left(\frac{c_0}{c_1}\right)^k = \left(1,1,1,\cdots,1 \,|\, 1,1,1,\cdots\right).
\end{equation}
Once again, let us take $N=5$ as an example. We have the polynomial ansatz
\begin{equation}
    P^{(1)}_{\lambda}(N) = \lambda N + \frac{1-\lambda}{\lambda},
\end{equation}
which is what is claimed in Lemma \ref{lem:NC-moments-exp-small-error} to match $\Ebb[N_C]$ up to an exponentially small error. Then we can write
\begin{align}
    \lambda N = 5(c_1-c_0) &= 5(c_1-c_0)(c_1+c_0)^4 \\
    &= 5(c_1-c_0)(c_1^4 + 4c_1^3c_0 + 6c_1^2c_0^2 + 4c_1c_0^3 + c_0^4) \\
    &= 5(c_1^5 + 3c_1^4c_0 + 2c_1^3c_0^2 - 2c_1^2c_0^3 - 3c_1c_0^4 - c_0^5) \\
    &= (5,15,10,-10,-15,-5) \\
    \frac{1-\lambda}{\lambda} = \frac{2c_0}{\lambda} &= \frac{2c_0}{c_1-c_0}(c_1+c_0)^5 \\
    &= 2(c_1^5 + 5c_1^4c_0 + 10c_1^3c_0^2 + 10c_1^2c_0^3 + 5c_1c_0^4 + c_0^5)\sum_{k=1}^{\infty}\left(\frac{c_0}{c_1}\right)^k \\
    &= 2(0,1,6,16,26,31 \,|\, 32,32,32,\cdots) \\
    &= (0,2,12,32,52,62 \,|\, 64,64,64,\cdots) \\
    \therefore \lambda N + \frac{1-\lambda}{\lambda} &= (5,15,10,-10,-15,-5) + (0,2,12,32,52,62 \,|\, 64,64,64,\cdots) \\
    &= (5,17,22,22,37,57 \,|\, 64,64,64,\cdots).
\end{align}
In this way, we have expressed $\lambda N + \frac{1-\lambda}{\lambda}$ for $N=5$ as a $(c_1,c_0)$ extended polynomial.

\vspace{0.5\baselineskip}

Now let us compare our two $(c_1,c_0)$ (extended) polynomials:
\begin{align}
    \Ebb[N_C] &= (5,17,22,22,17,5) \\
    \lambda N + \frac{1-\lambda}{\lambda} &= (5,17,22,22,37,57 \,|\, 64,64,64,\cdots).
\end{align}
Notice that the first $4$ terms (that is, from $k=0$ to $k=3$) match. Furthermore, these are precisely the most dominant terms, since each successive term has an extra factor of $c_0/c_1$. Therefore, these two quantities are very nearly equal!

\vspace{0.5\baselineskip}

This is precisely the more general strategy we will adopt. We have already written $\Ebb[N_C^p]$ as a $(c_1,c_0)$ polynomial (see the formula for $a^{(p)}_k$ above). We will now write each polynomial ansatz $P^{(p)}_{\lambda}(N)$ as a $(c_1,c_0)$ extended polynomial. We will then show that the terms from $k=0$ to $k=\lceil J\rceil$ all match. (Recall that our above summation formula for $a^{(p)}_k$ only applies to $0\le k\le \lceil J\rceil$, with the rest being defined by the symmetry between $k$ and $N-k$. This is where that fact will come into play.) Finally, we will use the discrepancy from $k=\lceil J\rceil+1$ onward to show that the error must be exponentially small.

\vspace{0.5\baselineskip}

Let us first tackle the first moment, namely $\Ebb[N_C]$. Let $b^{(1)}_k$ denote the coefficient of $c_1^{N-k}c_0^k$ in the polynomial ansatz
\begin{equation}
    P^{(1)}_\lambda(N) = \lambda N + \frac{1-\lambda}{\lambda}.
\end{equation}
Generalizing the $N=5$ example above, we have the following:
\begin{align}
    \lambda N &= N(c_1-c_0)(c_1+c_0)^{N-1} \\
    &= N(c_1-c_0)\sum_{k=0}^{N-1}\binom{N-1}{k}c_1^{N-1-k}c_0^k \\
    &= \sum_{k=0}^{N}N\left[\binom{N-1}{k}-\binom{N-1}{k-1}\right]c_1^{N-k}c_0^k \\
    \frac{1-\lambda}{\lambda} &= 2(c_1+c_0)^N\sum_{k=1}^{\infty}\left(\frac{c_0}{c_1}\right)^k \\
    &= 2\left[\sum_{k=0}^{N}\binom{N}{k}c_1^{N-k}c_0^k\right]\left[\sum_{k=1}^{\infty}\left(\frac{c_0}{c_1}\right)^k\right] \\
    &= \sum_{k=0}^{\infty}2\left[\sum_{j=0}^{\min(k-1,N)}\binom{N}{j}\right]c_1^{N-k}c_0^k.
\end{align}
Adding these two expressions together and replacing $N$ with $2J$ for convenience gives us our $(c_1,c_0)$ polynomial:
\begin{align}
    P^{(1)}_\lambda(N) &= \lambda N + \frac{1-\lambda}{\lambda} = \left(b^{(1)}_0, \cdots, b^{(1)}_N \,|\, b^{(1)}_{N+1}, b^{(1)}_{N+2}, \cdots\right) \\
    b^{(1)}_k &= 2J\left[\binom{2J-1}{k} - \binom{2J-1}{k-1}\right] + 2\sum_{j=0}^{\min(k-1,2J)}\binom{2J}{j}.
\end{align}

\vspace{0.5\baselineskip}

We want to show that $a^{(1)}_k=b^{(1)}_k$ for all $0\le k\le\lceil N/2\rceil$. In this domain, our formulas simplify slightly as follows:
\begin{align}
    a^{(1)}_k &= \sum_{j=J-k}^{J}2j\left[\binom{2J}{J-j} - \binom{2J}{J-j-1}\right] \\
    b^{(1)}_k &= 2J\left[\binom{2J-1}{k} - \binom{2J-1}{k-1}\right] + 2\sum_{j=0}^{k-1}\binom{2J}{j}.
\end{align}
We now construct the following chain of equalities:
\begin{align}
    a^{(1)}_k &= \sum_{j=J-k}^{J}2j\left[\binom{2J}{J-j} - \binom{2J}{J-j-1}\right] \\
    &= 2(J-k)\binom{2J}{J-(J-k)} + 2\sum_{j=J-k+1}^{J}\binom{2J}{J-j} - 2J\binom{2J}{J-J-1} \\
    &= 2(J-k)\binom{2J}{k} + 2\sum_{j=0}^{k-1}\binom{2J}{j} \\
    &= 2J\left[\binom{2J-1}{k-1} + \binom{2J-1}{k}\right] - 2k\binom{2J}{k} + 2\sum_{j=0}^{k-1}\binom{2J}{j} \\
    &= 2J\left[\binom{2J-1}{k-1} + \binom{2J-1}{k}\right] - 4J\binom{2J-1}{k-1} + 2\sum_{j=0}^{k-1}\binom{2J}{j} \\
    &= 2J\left[\binom{2J-1}{k} - \binom{2J-1}{k-1}\right] + 2\sum_{j=0}^{k-1}\binom{2J}{j} \\
    &= b^{(1)}_k.
\end{align}
Let us clarify each step. The first equality matches $a^{(1)}_k$ with its formula. The second equality takes advantage of the fact that the previous expression is a telescoping sum. The third equality simplifies the previous expression and makes the substitution $j\mapsto J-j$ for the dummy index. The fourth equality uses Pascal's identity. The fifth equality uses the fact that $\binom{2J}{k} = \frac{2J}{k}\binom{2J-1}{k-1}$. The sixth equality merely simplifies the previous expression. The seventh equality matches $b^{(1)}_k$ with its formula.

\vspace{0.5\baselineskip}

The last task is to bound the error, which we call $e^{(1)}(N,\lambda)$ for convenience:
\begin{equation}
    e^{(1)}(N,\lambda) \coloneqq \Ebb[N_C] - P^{(1)}_{\lambda}(N).
\end{equation}
From the formulas for $a^{(1)}_k$ and $b^{(1)}_k$, we can infer that
\begin{align}
    \lceil J\rceil + 1 \le k \le N &\implies -2^{N+1} < a^{(2)}_k - b^{(2)}_k < 0 \\
    k \ge N+1 &\implies a^{(2)}_k - b^{(2)}_k = -2^{N+1}.
\end{align}
Therefore, the error $e^{(1)}(N,\lambda)$ is clearly negative, and we can upper bound its magnitude as follows:
\begin{align}
    \abs{e^{(1)}(N,\lambda)} &= \sum_{k=\lceil J\rceil+1}^{\infty}\left[b^{(2)}_k - a^{(2)}_k\right]c_1^{N-k}c_0^k \\
    &\le 2^{N+1}\sum_{k=\lceil J\rceil+1}^{\infty}c_1^{N-k}c_0^k \\
    &= 2^{N+1}c_1^{N-\lceil J\rceil-1}c_0^{\lceil J\rceil+1}\sum_{k=0}^{\infty}\left(\frac{c_0}{c_1}\right)^k \\
    &= 2^{N+1}c_1^{N-\lceil J\rceil-1}c_0^{\lceil J\rceil+1}\frac{c_1}{c_1-c_0} \\
    &= 2^{N+1}c_1^{N-\lceil J\rceil}c_0^{\lceil J\rceil}\frac{c_0}{c_1-c_0} \\
    &= \frac{1-\lambda}{\lambda}2^Nc_1^{N-\lceil J\rceil}c_0^{\lceil J\rceil} \\
    &\le \frac{1-\lambda}{\lambda}(4c_1c_0)^{N/2} \\
    &= \frac{1-\lambda}{\lambda}\left(1-\lambda^2\right)^{N/2}.
\end{align}
We thus conclude that the polynomial ansatz $P^{(1)}_\lambda(N)$ matches $\Ebb[N_C]$ up to $O\left(\left(1-\lambda^2\right)^{N/2}\right)$ error. In fact, we can say even more precisely that
\begin{equation}
    \Ebb[N_C] = \lambda N + \frac{1-\lambda}{\lambda} + e^{(1)}(N,\lambda),
\end{equation}
where the error is bounded above and below as follows:
\begin{equation}
    -\frac{1-\lambda}{\lambda}\left(1-\lambda^2\right)^{N/2} \le e^{(1)}(N,\lambda) \le 0.
\end{equation}
This concludes our proof of the first equation in Lemma \ref{lem:NC-moments-exp-small-error}.

\vspace{0.5\baselineskip}

Let us now tackle the second moment, namely $\Ebb[N_C^2]$. The efficient thing to do at this point is actually to invoke the formula for $\Ebb[4J^2] = \Ebb[N_C(N_C+2)]$ as shown in Appendix \ref{appendix:angular-momentum-moments}\ref{subsec:J2-integer-moments}:
\begin{align}
    \Ebb[N_C(N_C+2)] &= \lambda^2N^2 + \left(3 - \lambda^2\right)N \\
    \Ebb[N_C] &= \lambda N + \frac{1-\lambda}{\lambda} + e^{(1)}(N,\lambda) \\
    \therefore \Ebb[N_C^2] &= \Ebb[N_C(N_C+2)] - 2\cdot\Ebb[N_C] \\
    &= \lambda^2N^2 + (1-\lambda)(3+\lambda)N - \frac{2(1-\lambda)}{\lambda} - 2e^{(1)}(N,\lambda).
\end{align}
Therefore, we conclude that
\begin{equation}
    \Ebb[N_C^2] = \lambda^2N^2 + (1-\lambda)(3+\lambda)N - \frac{2(1-\lambda)}{\lambda} + e^{(2)}(N,\lambda),
\end{equation}
where the error $e^{(2)}(N,\lambda)$ satisfies
\begin{equation}
    e^{(2)}(\lambda) = -2e^{(1)}(\lambda) \implies 0 \le e^{(2)}(N,\lambda) \le \frac{2(1-\lambda)}{\lambda}\left(1-\lambda^2\right)^{N/2}.
\end{equation}
In this way, we quickly obtain $\Ebb[N_C^2]$ up to exponentially small error and bounds on the error. More generally, if we want to compute $\Ebb[N_C^p]$ for $1\le p\le P$ as efficiently as possible, we should start from $p=1$ and work upward, and we should alternate between using the $(c_1,c_0)$ extended polynomial method for odd $p$ values and using $\Ebb\left[(N_C(N_C+2))^{p/2}\right]$ for even $p$ values.

\vspace{0.5\baselineskip}

However, for the sake of thoroughness, let us prove the formula for $\Ebb[N_C^2]$ using the $(c_1,c_0)$ extended polynomial method anyway. Let $b^{(2)}_k$ denote the coefficient of $c_1^{N-k}c_0^k$ in the polynomial ansatz
\begin{equation}
    P^{(2)}_\lambda(N) = \lambda^2N^2 + (1-\lambda)(3+\lambda)N - \frac{2(1-\lambda)}{\lambda}.
\end{equation}
The first two terms can be written as $(c_1,c_0)$ polynomials. Furthermore, the last term can be written as a $(c_1,c_0)$ extended polynomial (other than the factor of $-2$, we already wrote it in this manner previously when computing $b^{(1)}_k$). We proceed as follows:
\begin{align}
    \lambda^2N^2 &= N^2(c_1-c_0)^2(c_1+c_0)^{N-2} \\
    &= N^2(c_1^2 - 2c_1c_0 + c_0^2)\left[\sum_{k=0}^{N-2}\binom{N-2}{k}c_1^{N-2-k}c_0^k\right] \\
    &= \sum_{k=0}^{N}N^2\left[\binom{N-2}{k} - 2\binom{N-2}{k-1} + \binom{N-2}{k-2}\right]c_1^{N-k}c_0^k \\
    (1-\lambda)(3+\lambda)N &= 4Nc_0(2c_1+c_0)(c_1+c_0)^{N-2} \\
    &= 4N(2c_1c_0 + c_0^2)\left[\sum_{k=0}^{N-2}\binom{N-2}{k}c_1^{N-2-k}c_0^k\right] \\
    &= \sum_{k=0}^{N}4N\left[2\binom{N-2}{k-1} + \binom{N-2}{k-2}\right]c_1^{N-k}c_0^k \\
    -\frac{2(1-\lambda)}{\lambda} &= \sum_{k=0}^{\infty}-4\left[\sum_{j=0}^{\min(k-1,N)}\binom{N}{j}\right]c_1^{N-k}c_0^k.
\end{align}
Adding these three expressions together and replacing $N$ with $2J$ for convenience gives us our $(c_1,c_0)$ polynomial:
\begin{align}
    P^{(2)}_\lambda(N) &= \lambda^2N^2 + (1-\lambda)(3+\lambda)N - \frac{2(1-\lambda)}{\lambda} = \left(b^{(2)}_0, \cdots, b^{(2)}_N \,|\, b^{(2)}_{N+1}, b^{(2)}_{N+2}, \cdots\right) \\
    b^{(2)}_k &= 4J^2\left[\binom{2J-2}{k} - 2\binom{2J-2}{k-1} + \binom{2J-2}{k-2}\right] + 8J\left[2\binom{2J-2}{k-1} + \binom{2J-2}{k-2}\right] - 4\sum_{j=0}^{\min(k-1,2J)}\binom{2J}{j}.
\end{align}

\vspace{0.5\baselineskip}

We want to show that $a^{(2)}_k=b^{(2)}_k$ for all $0\le k\le\lceil N/2\rceil$. In this domain, our formulas simplify slightly as follows:
\begin{align}
    a^{(2)}_k &= \sum_{j=J-k}^{J}(2j)^2\left[\binom{2J}{J-j} - \binom{2J}{J-j-1}\right] \\
    b^{(2)}_k &= N^2\left[\binom{2J-2}{k} - 2\binom{2J-2}{k-1} + \binom{2J-2}{k-2}\right] + N\left[8\binom{2J-2}{k-1} + 4\binom{2J-2}{k-2}\right] - 4\sum_{j=0}^{k-1}\binom{2J}{j}.
\end{align}
We now construct the following chain of equalities:
\begin{align}
    a^{(2)}_k &= \sum_{j=J-k}^{J}(2j)^2\left[\binom{2J}{J-j} - \binom{2J}{J-j-1}\right] \\
    &= 4(J-k)^2\binom{2J}{J-(J-k)} + 4\sum_{j=J-k+1}^{J}(2j-1)\binom{2J}{J-j} - 4J^2\binom{2J}{J-J-1} \\
    &= 4(J-k)^2\binom{2J}{k} + 4\sum_{j=0}^{k-1}(2(J-j)-1)\binom{2J}{j} \\
    &= 4(J-k)^2\binom{2J}{k} + 8J\sum_{j=0}^{k-1}\binom{2J}{j} - 16J\sum_{j=0}^{k-1}\binom{2J-1}{j-1} - 4\sum_{j=0}^{k-1}\binom{2J}{j} \\
    &= 4(J-k)^2\binom{2J}{k} + 8J\sum_{j=0}^{k-1}\left[\binom{2J-1}{j} - \binom{2J-1}{j-1}\right] - 4\sum_{j=0}^{k-1}\binom{2J}{j} \\
    &= 4(J-k)^2\binom{2J}{k} + 8J\binom{2J-1}{k-1} - 4\sum_{j=0}^{k-1}\binom{2J}{j} \\
    &= 4J^2\binom{2J}{k} - 4k(2J-k)\binom{2J}{k} + 8J\binom{2J-1}{k-1} - 4\sum_{j=0}^{k-1}\binom{2J}{j} \\
    &= 4J^2\binom{2J}{k} - 8J(2J-k)\binom{2J-1}{k-1} + 8J\binom{2J-1}{k-1} - 4\sum_{j=0}^{k-1}\binom{2J}{j} \\
    &= 4J^2\binom{2J}{k} - 16J^2\binom{2J-1}{k-1} + 8J(k-1)\binom{2J-1}{k-1} + 16J\binom{2J-1}{k-1} - 4\sum_{j=0}^{k-1}\binom{2J}{j} \\
    &= 4J^2\binom{2J}{k} - 16J(J-1)\binom{2J-1}{k-1} + 8J(2J-1)\binom{2J-2}{k-2} - 4\sum_{j=0}^{k-1}\binom{2J}{j} \\
    &= 4J^2\binom{2J-2}{k} + (-8J^2 + 16J)\binom{2J-2}{k-1} + (4J^2 + 8J)\binom{2J-2}{k-2} - 4\sum_{j=0}^{k-1}\binom{2J}{j} \\
    &= 4J^2\left[\binom{2J-2}{k} - 2\binom{2J-2}{k-1} + \binom{2J-2}{k-2}\right] + 8J\left[2\binom{2J-2}{k-1} + \binom{2J-2}{k-2}\right] - 4\sum_{j=0}^{k-1}\binom{2J}{j} \\
    &= b^{(2)}_k.
\end{align}
We do not explain each step in detail, but the manipulations used are quite similar to those used to prove that $a^{(1)}_k = b^{(1)}_k$.

\vspace{0.5\baselineskip}

The last task is to bound the error, which we call $e^{(2)}(N,\lambda)$ for convenience:
\begin{equation}
    e^{(1)}(N,\lambda) \coloneqq \Ebb[N_C^2] - P^{(2)}_{\lambda}(N).
\end{equation}From the formulas for $a^{(2)}_k$ and $b^{(2)}_k$, we can infer that
\begin{align}
    \lceil J\rceil + 1 \le k \le N &\implies 0 < a^{(2)}_k - b^{(2)}_k < 2^{N+2} \\
    k \ge N+1 &\implies a^{(2)}_k - b^{(2)}_k = 2^{N+2}.
\end{align}
Notice that these bounds are just $-2$ times the bounds used for $a^{(1)}_k - b^{(1)}_k$. Therefore, the error $e^{(2)}(N,\lambda)$ is clearly positive, and we can upper bound its magnitude as just double of what we obtained for $e^{(1)}(N,\lambda)$:
\begin{align}
    \abs{e^{(2)}(N,\lambda)} &= \sum_{k=\lceil J\rceil+1}^{\infty}\left[a^{(2)}_k - b^{(2)}_k\right]c_1^{N-k}c_0^k \\
    &\le 2^{N+2}\sum_{k=\lceil J\rceil+1}^{\infty}c_1^{N-k}c_0^k \\
    &\le \frac{2(1-\lambda)}{\lambda}\left(1-\lambda^2\right)^{N/2}.
\end{align}
We thus conclude that the polynomial ansatz $P^{(2)}_\lambda(N)$ matches $\Ebb[N_C^2]$ up to $O\left(\left(1-\lambda^2\right)^{N/2}\right)$ error. In fact, we can say even more precisely that
\begin{equation}
    \Ebb[N_C^2] = \lambda^2N^2 + (1-\lambda)(3+\lambda)N - \frac{2(1-\lambda)}{\lambda} + e^{(2)}(N,\lambda),
\end{equation}
where the error $e^{(2)}(N,\lambda)$ is bounded above and below as follows:
\begin{equation}
    0 \le e^{(2)}(N,\lambda) \le \frac{2(1-\lambda)}{\lambda}\left(1-\lambda^2\right)^{N/2}.
\end{equation}
This concludes our proof of the second equation in Lemma \ref{lem:NC-moments-exp-small-error}.

\vspace{0.5\baselineskip}

We do not present the details here, because they are very tedious, but the same general method works for $\Ebb[N_C^3]$ and $\Ebb[N_C^4]$ as well. Of course, it is best to only apply the $(c_1,c_0)$ polynomial method to $\Ebb[N_C^3]$, and then use $\Ebb\left[((N_C(N_C+2))^2\right]$ from Appendix \ref{appendix:angular-momentum-moments}\ref{subsec:J2-integer-moments} to compute $\Ebb[N_C^4]$.

\appsubsec{Full Computation of the Negative Moments of $N_C$}
{subsec:NC-negative-moments-full-computation}

To conclude this appendix, we combine the results from Appendix \ref{appendix:angular-momentum-moments}\ref{subsec:positive-moments-to-negative-moments} and Appendix \ref{appendix:angular-momentum-moments}\ref{subsec:NC-moments-ese} to prove Lemma \ref{lem:NC-negative-moments}. In particular, based on the infinite series we show in Appendix \ref{appendix:angular-momentum-moments}\ref{subsec:positive-moments-to-negative-moments}, we need to compute $\Ebb\left[N_C^{-p}\right]$, we need to compute the leading factor $\Ebb[N_C]^{-p}$, as well as the terms in the infinite series, which take the form
\begin{equation}
    \frac{\Ebb\left[\left(N_C - \Ebb[N_C]\right)^q\right]}{\Ebb[N_C]^q}.
\end{equation}

\vspace{0.5\baselineskip}

We begin by restating the first moment (average) of $N_C$ as shown in Lemma \ref{lem:NC-moments-exp-small-error}:
\begin{equation}
    \Ebb[N_C] = \lambda N + \frac{1-\lambda}{\lambda} + \text{e.s.e.}
\end{equation}
We can take the reciprocal of this quantity as follows:
\begin{align}
    \Ebb[N_C]^{-1} &= \left(\lambda N + \frac{1-\lambda}{\lambda} + \text{e.s.e.}\right)^{-1} \\
    &= (\lambda N)^{-1}\left(1 + \frac{1-\lambda}{\lambda^2}\frac{1}{N} + \text{e.s.e.}\right)^{-1} \\
    &= \frac{1}{\lambda N}\left[1 - \frac{1-\lambda}{\lambda^2}\frac{1}{N} + \frac{(1-\lambda)^2}{\lambda^4}\frac{1}{N^2} + O\left(N^{-3}\right)\right] \\
    &= \frac{1}{\lambda}\frac{1}{N} - \frac{1-\lambda}{\lambda^3}\frac{1}{N^2} + \frac{(1-\lambda)^2}{\lambda^5}\frac{1}{N^3} + O\left(N^{-4}\right)
\end{align}
We can now also compute the centered second moment (variance) of $N_C$ as follows, using both the first and second equations in Lemma \ref{lem:NC-moments-exp-small-error}:
\begin{align}
    \Ebb\left[\left(N_C - \Ebb[N_C]\right)^2\right] &= \Ebb[N_C^2] - \Ebb[N_C]^2 \\
    &= \left[\lambda^2N^2 + (1-\lambda)(3+\lambda)N - \frac{2(1-\lambda)}{\lambda} + \text{e.s.e.}\right] \\
    & \quad - \left(\lambda N + \frac{1-\lambda}{\lambda} + \text{e.s.e.}\right)^2 \\
    &= \left[\lambda^2N^2 + (1-\lambda)(3+\lambda)N - \frac{2(1-\lambda)}{\lambda} + \text{e.s.e.}\right] \\
    & \quad - \left[\lambda^2N^2 + 2(1-\lambda)N + \frac{(1-\lambda)^2}{\lambda^2} + \text{e.s.e.}\right] \\
    &= (1-\lambda^2)N - \frac{1-\lambda^2}{\lambda^2} + \text{e.s.e.}
\end{align}
We additionally compute the inverse square of the average of $N_C$:
\begin{align}
    \Ebb[N_C]^{-2} &= (\lambda N)^{-2}\left(1 + \frac{1-\lambda}{\lambda^2}\frac{1}{N} + \text{e.s.e.}\right)^{-2} \\
    &= \frac{1}{\lambda^2N^2}\left[1 - \frac{2(1-\lambda)}{\lambda^2}\frac{1}{N} + O\left(N^{-2}\right)\right] \\
    &= \frac{1}{\lambda^2}\frac{1}{N^2} - \frac{2(1-\lambda)}{\lambda^4}\frac{1}{N^3} + O\left(N^{-4}\right).
\end{align}
Combining the previous two results yields
\begin{align}
    \frac{\Ebb\left[\left(N_C - \Ebb[N_C]\right)^2\right]}{\Ebb[N_C]^2} &= \left[\frac{1}{\lambda^2}\frac{1}{N^2} - \frac{2(1-\lambda)}{\lambda^4}\frac{1}{N^3} + O\left(N^{-4}\right)\right]\left[(1-\lambda^2)N - \frac{1 - \lambda^2}{\lambda^2} + \text{e.s.e.}\right] \\
    &= \frac{1-\lambda^2}{\lambda^2}\frac{1}{N} - \frac{(1-\lambda^2)(3-2\lambda)}{\lambda^4}\frac{1}{N^2} + O\left(N^{-3}\right).
\end{align}
We now proceed to the third centered moment. This requires us to use the first three equations from Lemma \ref{lem:NC-moments-exp-small-error}:
\begin{align}
    \Ebb\left[\left(N_C - \Ebb[N_C]\right)^3\right] &= \Ebb[N_C^3] - 3\Ebb\left[N_C^2\right]\Ebb[N_C] + 3\Ebb[N_C]\Ebb[N_C]^2 - \Ebb[N_C]^3 \\
    &= \Ebb[N_C^3] - 3\Ebb\left[N_C^2\right]\Ebb[N_C] - 2\Ebb[N_C]^3 \\
    &= -2\lambda(1-\lambda^2)N + \frac{2(1-\lambda^2)}{\lambda^3} + \text{e.s.e.}
\end{align}
We additionally compute the inverse cube of the average of $N_C$:
\begin{align}
    \Ebb[N_C]^{-3} &= (\lambda N)^{-3}\left(1 + \frac{1-\lambda}{\lambda^2}\frac{1}{N} + \text{e.s.e.}\right)^{-3} \\
    &= \frac{1}{\lambda^3N^3}\left[1 + O\left(N^{-1}\right)\right] \\
    &= \frac{1}{\lambda^3}\frac{1}{N^3} + O\left(N^{-4}\right).
\end{align}
Combining the previous two results yields
\begin{align}
    \frac{\Ebb\left[(N_C - \Ebb[N_C])^3\right]}{\Ebb[N_C]^3} &= \left[\frac{1}{\lambda^3}\frac{1}{N^3} + O\left(N^{-4}\right)\right]\left[-2\lambda(1-\lambda^2)N + O(1)\right] \\
    &= -\frac{2(1-\lambda^2)}{\lambda^2}\frac{1}{N^2} + O\left(N^{-3}\right).
\end{align}
Finally, we proceed to the fourth centered moment. This requires us to use all four equations from Lemma \ref{lem:NC-moments-exp-small-error}:
\begin{align}
    \Ebb\left[\left(N_C - \Ebb[N_C]\right)^4\right] &= \Ebb[N_C^4] - 4\Ebb\left[N_C^3\right]\Ebb[N_C] + 6\Ebb\left[N_C^2\right]\Ebb[N_C]^2 \\
    & \quad - 4\Ebb[N_C]\Ebb[N_C]^3 + \Ebb[N_C]^4 \\
    &= \Ebb[N_C^4] - 4\Ebb\left[N_C^3\right]\Ebb[N_C] + 6\Ebb\left[N_C^2\right]\Ebb[N_C]^2 - 3\Ebb[N_C]^4 \\
    &= 3(1-\lambda^2)^2N^2 + \frac{2(1-\lambda^2)(-3 + 2\lambda^2 + 3\lambda^4)}{\lambda^2}N - \frac{(1-\lambda^2)(3+\lambda^2)}{\lambda^4} + \text{e.s.e.}
\end{align}
We additionally compute the inverse fourth power of the average of $N_C$:
\begin{align}
    \Ebb[N_C]^{-4} &= (\lambda N)^{-4}\left(1 + \frac{1-\lambda}{\lambda^2}\frac{1}{N} + \text{e.s.e.}\right)^{-4} \\
    &= \frac{1}{\lambda^4N^4}\left[1 + O\left(N^{-1}\right)\right] \\
    &= \frac{1}{\lambda^4}\frac{1}{N^4} + O\left(N^{-5}\right).
\end{align}
Combining the previous two results yields
\begin{align}
    \frac{\Ebb\left[(N_C - \Ebb[N_C])^4\right]}{\Ebb[N_C]^4} &= \left[\frac{1}{\lambda^4}\frac{1}{N^4} + O\left(N^{-5}\right)\right]\left[3(1-\lambda^2)^2N^2 + O(N)\right] \\
    &= \frac{3(1-\lambda^2)^2}{\lambda^4}\frac{1}{N^2} + O\left(N^{-3}\right).
\end{align}
A useful sanity check for all the above results is the Gaussian approximation for the distribution of $N_C$, which has mean $\sim\lambda N$ and variance $\sim(1-\lambda^2)N$ \cite{Keyl2001}. As a result, the centered moments of $N_C$ should always have leading-order behavior
\begin{equation}
    \Ebb\left[\left(N_C - \Ebb[N_C]\right)^p\right] = \begin{cases}
        (1-\lambda^2)^{p/2}N^{p/2} + O\left(N^{\frac{p}{2}-1}\right) & p\text{ even} \\
        O\left(N^{\frac{p-1}{2}}\right) & p\text{ odd}
    \end{cases}
\end{equation}
for any nonnegative integer $p$. This in turn implies that the terms in the infinite series always have leading-order behavior
\begin{equation}
    \frac{\Ebb\left[\left(N_C - \Ebb[N_C]\right)^p\right]}{\Ebb[N_C]^p} = \begin{cases}
        \frac{(1-\lambda^2)^{p/2}}{\lambda^p}N^{-p/2} + O\left(N^{-\frac{p}{2}-1}\right) & p\text{ even} \\
        O\left(N^{-\frac{p+1}{2}}\right) & p\text{ odd}
    \end{cases}
\end{equation}
for any positive integer $p\ge 2$.

\vspace{0.5\baselineskip}

We can now combine all of the above results to compute $\Ebb\left[N_C^{-1}\right]$ to $\Theta\left(N_C^{-3}\right)$ precision. Since we need to compute the first three terms of this expansion, this specific calculation is the setting where all the fancy tricks (and especially the positive moments up to exponentially small error, as shown in Lemma \ref{lem:NC-moments-exp-small-error}) are the most useful:
\begin{align}
    & \quad \Ebb\left[N_C^{-1}\right] \\
    &= \Ebb[N_C]^{-1}\left[1 + \frac{\Ebb\left[(N_C - \Ebb[N_C])^2\right]}{\Ebb[N_C]^2} - \frac{\Ebb\left[(N_C - \Ebb[N_C])^3\right]}{\Ebb[N_C]^3} + \frac{\Ebb\left[(N_C - \Ebb[N_C])^4\right]}{\Ebb[N_C]^4} + O\left(N^{-3}\right)\right] \\
    &= \left[\frac{1}{\lambda}\frac{1}{N} - \frac{1-\lambda}{\lambda^3}\frac{1}{N^2} + \frac{(1-\lambda)^2}{\lambda^5}\frac{1}{N^3} + O\left(N^{-4}\right)\right]\Bigg\{1 + \left[\frac{1-\lambda^2}{\lambda^2}\frac{1}{N} - \frac{(1-\lambda^2)(3-2\lambda)}{\lambda^4}\frac{1}{N^2} + O\left(N^{-3}\right)\right] \\
    & \quad - \left[-\frac{2(1-\lambda^2)}{\lambda^2}\frac{1}{N^2} + O\left(N^{-3}\right)\right] + \left[\frac{3(1-\lambda^2)^2}{\lambda^4}\frac{1}{N^2} + O\left(N^{-3}\right)\right] + O\left(N^{-3}\right)\Bigg\} \\
    &= \left[\frac{1}{\lambda}\frac{1}{N} - \frac{1-\lambda}{\lambda^3}\frac{1}{N^2} + \frac{(1-\lambda)^2}{\lambda^5}\frac{1}{N^3} + O\left(N^{-4}\right)\right]\Bigg\{1 + \frac{1-\lambda^2}{\lambda^2}\frac{1}{N} + \frac{(1-\lambda^2)(2-\lambda)}{\lambda^3}\frac{1}{N^2} + O\left(N^{-3}\right)\Bigg\} \\
    &= \frac{1}{\lambda}\frac{1}{N} + \frac{1-\lambda}{\lambda^2}\frac{1}{N^2} + \frac{(1-\lambda)(1+2\lambda-\lambda^2)}{\lambda^4}\frac{1}{N^3} + O\left(N^{-4}\right).
\end{align}
We conclude that
\begin{equation}
    \boxed{\Ebb\left[N_C^{-1}\right] = \frac{1}{\lambda}\frac{1}{N} + \frac{1-\lambda}{\lambda^2}\frac{1}{N^2} + \frac{(1-\lambda)(1+2\lambda-\lambda^2)}{\lambda^4}\frac{1}{N^3} + O\left(N^{-4}\right)},
\end{equation}
which precisely matches the first result in Lemma \ref{lem:NC-negative-moments}. Next, we compute $\Ebb\left[N_C^{-2}\right]$ to $\Theta\left(N_C^{-3}\right)$ precision, which is actually easier than the previous calculation, since here we only need to compute the first two terms of this expansion:
\begin{align}
    \Ebb\left[N_C^{-2}\right] &= \Ebb[N_C]^{-2}\left[1 + 3\frac{\Ebb\left[(N_C - \Ebb[N_C])^2\right]}{\Ebb[N_C]^2} + O\left(N^{-2}\right)\right] \\
    &= \left[\frac{1}{\lambda^2}\frac{1}{N^2} - \frac{2(1-\lambda)}{\lambda^4}\frac{1}{N^3} + O\left(N^{-4}\right)\right]\left[1 + \frac{3(1-\lambda^2)}{\lambda^2}\frac{1}{N} + O\left(N^{-3}\right)\right] \\
    &= \frac{1}{\lambda^2}\frac{1}{N^2} + \frac{(1-\lambda)(1+3\lambda)}{\lambda^4}\frac{1}{N^3} + O\left(N^{-4}\right).
\end{align}
We conclude that
\begin{equation}
    \boxed{\Ebb\left[N_C^{-2}\right] = \frac{1}{\lambda^2}\frac{1}{N^2} + \frac{(1-\lambda)(1+3\lambda)}{\lambda^4}\frac{1}{N^3} + O\left(N^{-4}\right)},
\end{equation}
which precisely matches the second result in Lemma \ref{lem:NC-negative-moments}. Finally, we compute $\Ebb\left[N_C^{-3}\right]$ to $\Theta\left(N_C^{-3}\right)$ precision, which of course is the easiest of all:
\begin{align}
    \Ebb\left[N_C^{-3}\right] &= \Ebb[N_C]^{-3}\left[1 + O\left(N^{-1}\right)\right] \\
    &= \left[\frac{1}{\lambda^3}\frac{1}{N^3} + O\left(N^{-4}\right)\right]\left[1 + O\left(N^{-1}\right)\right] \\
    &= \frac{1}{\lambda^3}\frac{1}{N^3} + O\left(N^{-4}\right).
\end{align}
We conclude that
\begin{equation}
    \boxed{\Ebb\left[N_C^{-3}\right] = \frac{1}{\lambda^3}\frac{1}{N^3} + O\left(N^{-4}\right)},
\end{equation}
which precisely matches the third and final result in Lemma \ref{lem:NC-negative-moments}.

%% file: app_P_vals.tex
\appsec{Understanding the $P_{w,w'}$ Values}
{appendix:understanding-P-vals}

Throughout this paper, we have to invoke the matrix representation of the post-Schur-sampling state $\rho_{N_C/2}$ when written in the basis of symmetrized Hamming weight states. In particular, we commonly refer to ``the $P_{w,w'}$ values'', which are defined as follows:
\begin{equation}
    P_{w,w'} \coloneqq \bra{w^{(s)}}\rho_{N_C/2}\ket{(w')^{(s)}}.
\end{equation}
This appendix is dedicated to proving the various useful facts about the $P_{w,w'}$ values that appear throughout this paper. This appendix is organized as follows:
\begin{itemize}
    \item In Appendix \ref{appendix:understanding-P-vals}\ref{subsec:P-vals-miscellaneous-facts}, we will prove some miscellaneous useful formulas for these values.
    \begin{itemize}
        \item This includes one formula that will be useful in Appendix \ref{appendix:perturbative-protocols} when studying the behavior of the optimal protocol in the equatorial case for $\lambda\approx 1$.
        \item This also includes one formula that is helpful for the numerical analysis of these coherence distillation protocols.
    \end{itemize}
    \item In Appendix \ref{appendix:understanding-P-vals}\ref{subsec:P-vals-moments-defining}, we explain why we organize the $P_{w,w'}$ values into those with a fixed ``offset'' $\alpha$ (meaning values of the form $P_{w-\alpha,w}$), and we explain why we only focus on $\alpha=0$ and $\alpha=1$ as the ones that matter for the asymptotic analysis of optimal single-shot coherence distillation, as explained in Appendices \ref{appendix:asymptotic-expansion}, \ref{appendix:1st-order-optimality}, \ref{appendix:2nd-order-optimality}, and \ref{appendix:equatorial-3rd-order-optimality}.
    \item In Appendix \ref{appendix:understanding-P-vals}\ref{subsec:P-vals-moments-offset0}, we show how to compute the relevant ``moments'' of the $P_{w,w}$ values (corresponding to offset $\alpha=0$).
    \item In Appendix \ref{appendix:understanding-P-vals}\ref{subsec:P-vals-moments-offset1-equatorial}, we show how to compute the relevant ``moments'' of the $P_{w-1,w}$ values (corresponding to offset $\alpha=1$) in the equatorial special case.
    \item In Appendix \ref{appendix:understanding-P-vals}\ref{subsec:P-vals-moments-offset1-general}, we show how to compute the relevant ``moments'' of the $P_{w-1,w}$ values (corresponding to offset $\alpha=1$) in the general case.
\end{itemize}

\appsubsec{Miscellaneous Useful Facts}
{subsec:P-vals-miscellaneous-facts}

Since we can write the Schur-sampled state as
\begin{equation}
    \rho_{N_C/2} = \frac{c_1-c_0}{c_1^{N_C+1}-c_0^{N_C+1}}\sum_{k=0}^{N_C}c_1^kc_0^{N-k}\ket{k_{\hat{n}}^{(s)}}\bra{k_{\hat{n}}^{(s)}},
\end{equation}
we can write $P_{w,w'}$ as follows:
\begin{equation}
    P_{w,w'} = \frac{c_1-c_0}{c_1^{N_C+1}-c_0^{N_C+1}}\sum_{k=0}^{N_C}c_1^kc_0^{N-k}Q_{wk}Q_{w'k},
\end{equation}
where for convenience, we have defined
\begin{align}
    Q_{wk} &\coloneqq \bra{w^{(s)}}\ket{k_{\hat{n}}^{(s)}}.
\end{align}
As a reminder, the computational basis states for a specific direction $\hat{n}$ are defined so that $\ket{0_{\hat{n}}}$ is a $-1$ eigenstate of $\hat{n}\cdot\vec{\sigma}$, while $\ket{1_{\hat{n}}}$ is a $+1$ eigenstate of $\hat{n}\cdot\vec{\sigma}$. (This convention may seem backward, but it is done to keep in line with the convention used by Cirac et al. \cite{Cirac1999}.) In particular, let $\hat{n} = \langle\sin\Theta,0,\cos\Theta\rangle$ for some $0\le\Theta\le\pi$ be some unit vector with azimuthal angle $\phi=0$. Then the computational basis states for direction $\hat{n}$ are
\begin{align}
    \ket{0_{\hat{n}}} &= \sin\frac{\Theta}{2}\ket{0} - \cos\frac{\Theta}{2}\ket{1} \\
    \ket{1_{\hat{n}}} &= \cos\frac{\Theta}{2}\ket{0} + \sin\frac{\Theta}{2}\ket{1}.
\end{align}
As a result, the inner product of a single-bit computational basis state with a single-bit state in some other direction given by polar angle $\Theta$ (assuming that it has azimuthal angle $\phi=0$) is
\begin{equation}
    \bra{b}\ket{b'_{\hat{n}}} = \begin{cases}
        \sin\frac{\Theta}{2} & b=0, b'=0 \\
        \cos\frac{\Theta}{2} & b=0, b'=1 \\
        -\cos\frac{\Theta}{2} & b=1, b'=0 \\
        \sin\frac{\Theta}{2} & b=1, b'=1.
    \end{cases}
\end{equation}
We can thus compute the $Q_{wk}$ values as follows:
\begin{align}
    Q_{wk} &= \bra{w^{(s)}}\ket{k_{\hat{n}}^{(s)}} \\
    &= \binom{N_C}{w}^{-1/2}\binom{N_C}{k}^{-1/2}\sum_{w(b)=w}\sum_{w(b')=k}\bra{b}\ket{b'_{\hat{n}}} \\
    &= \binom{N_C}{w}^{-1/2}\binom{N_C}{k}^{-1/2}\sum_{i=0}^{w}(-1)^i\left(\sin\frac{\Theta}{2}\right)^{N_C-k+w-2i}\left(\cos\frac{\Theta}{2}\right)^{k-w+2i}\binom{N_C}{k}\binom{N_C-k}{i}\binom{k}{w-i} \\
    &= \sqrt{\frac{\binom{N_C}{k}}{\binom{N_C}{w}}}\left(\sin\frac{\Theta}{2}\right)^{N_C-k+w}\left(\cos\frac{\Theta}{2}\right)^{k-w}\sum_{i=0}^{w}(-1)^i\left(\frac{1+\cos\Theta}{1-\cos\Theta}\right)^{i}\binom{N_C-k}{i}\binom{k}{w-i}.
\end{align}
For clarity, what we have done here is let the dummy index $i$ refer to the number of times a $\bra{1}$ from $\bra{b}$ collides with a $\ket{0_{\hat{n}}}$ from $\ket{b'_{\hat{n}}}$, which is what governs the number of $-1$ factors. There are $\binom{N_C}{k}$ ways to choose the string $b'$, then $\binom{N_C-k}{i}$ ways to choose the locations at which a $\bra{1}$ from $\bra{b}$ collides with a $\ket{0_{\hat{n}}}$ from $\ket{b'_{\hat{n}}}$, and then $\binom{k}{w-i}$ ways to choose the locations at which a $\bra{1}$ from $\bra{b}$ collides with a $\ket{1_{\hat{n}}}$ from $\ket{b'_{\hat{n}}}$.

\vspace{0.5\baselineskip}

In the equatorial case $\hat{n} = \hat{x}$, this formula simplifies somewhat as follows:
\begin{align}
    Q_{wk} &= \bra{w^{(s)}}\ket{k_X^{(s)}} \\
    &= \left(\frac{1}{\sqrt{\binom{N_C}{w}}}\sum_{w(b)=w}\bra{b}\right)\left(\frac{1}{\sqrt{\binom{N_C}{k}}}\sum_{w(b')=k}\ket{b'_X}\right) \\
    &= \sqrt{\frac{\binom{N_C}{k}}{\binom{N_C}{w}}}\sum_{w(b)=w}\bra{b}\left(\ket{+}^{\otimes k}\ket{-}^{\otimes(N_C-k)}\right) \\
    &= \sqrt{\frac{\binom{N_C}{k}}{\binom{N_C}{w}2^{N_C}}}\sum_{i=0}^{w}(-1)^i\binom{N_C-k}{i}\binom{k}{w-i}.
\end{align}
Notice that plugging $\Theta = \frac{\pi}{2}$ into the previous formula yields this result as well. This specific formula will actually find its greatest use in Appendix \ref{appendix:perturbative-protocols}, where we compute the behavior of the optimal coherence distillation protocol in the equatorial case for $\lambda$ very close to $1$. In that setting, we treat the protocol as a power series in $c_0 = \frac{1-\lambda}{2}$ up to second order, meaning that we only need $k=N_C,N_C-1,N_C-2$. In that setting, the somewhat complicated summation for $Q_{wk}$ simplifies considerably, since it only has $1,2,3$ nonzero terms (respectively).

\vspace{0.5\baselineskip}

At the moment, our formula for the $P_{w,w'}$ values may appear to be a triple summation, since it is a summation of the product $Q_{wk}Q_{w'k}$, and since each of $Q_{wk}$ and $Q_{w'k}$ is its own summation. However, it turns out that there is a much simpler formula for $P_{w,w'}$, which only requires a single summation:

\begin{lemma}
\label{lem:P-val-single-summation}
Consider an input state of the form
\begin{equation}
    \rho = \frac{\mathbb{I} + \lambda_xX + \lambda_zZ}{2} \quad\quad\quad \lambda_x = \lambda\sin\Theta_{\text{in}} \quad\quad \lambda_z = \lambda\cos\Theta_{\text{in}}.
\end{equation}
Then the resulting $P_{w-\alpha,w}$ values take the following form for all $0\le\alpha\le w\le N$:
\begin{equation}
    P_{w-\alpha,w} = \frac{c_1-c_0}{c_1^{N_C+1} - c_0^{N_C+1}}\left(\frac{\lambda_x}{1-\lambda_z}\right)^\alpha\left(\frac{1+\lambda_z}{2}\right)^N\left(\frac{1-\lambda_z}{1+\lambda_z}\right)^w\sqrt{\frac{\binom{N}{w}}{\binom{N}{w-\alpha}}}\sum_{j=0}^{N}\binom{w}{j+\alpha}\binom{N-w}{j}\left(\frac{\lambda_x^2}{1-\lambda_z^2}\right)^j.
\end{equation}
\end{lemma}

\begin{proof}
In this case, it actually helps to return to the perspective of the original tensor product state. In particular, we know that
\begin{equation}
    \left(\Pi_{j=N_C/2}\right)\rho^{\otimes N_C}\left(\Pi_{j=N_C/2}\right) = p\left(j=\frac{N_C}{2}\right)\rho_{N_C/2},
\end{equation}
where $p\left(j=N_C/2\right)$ equals the probability of getting the highest possible total angular momentum when performing Schur sampling \cite{Cirac1999}:
\begin{equation}
    p\left(j=\frac{N_C}{2}\right) = \frac{c_1^{N_C+1} - c_0^{N_C+1}}{c_1-c_0}.
\end{equation}
Therefore, we can write
\begin{align}
    p\left(j=\frac{N_C}{2}\right)P_{w-\alpha,w} &= p\left(j=\frac{N_C}{2}\right)\bra{(w-\alpha)^{(s)}}\rho_{N_C/2}\ket{w^{(s)}} \\
    &= \bra{(w-\alpha)^{(s)}}\left(\Pi_{j=N_C/2}\right)\rho^{\otimes N_C}\left(\Pi_{j=N_C/2}\right)\ket{w^{(s)}} \\
    &= \bra{(w-\alpha)^{(s)}}\rho^{\otimes N_C}\ket{w^{(s)}}.
\end{align}
From now on, we can consider the tensor product state $\rho^{\otimes N_C}$. We start by writing
\begin{equation}
    \rho^{\otimes N_C} = \frac{1}{2^{N_C}}\sum_{P\in\{I,X,Z\}^{\otimes N_C}}\lambda_x^{w_X(P)}\lambda_z^{w_Z(P)}P,
\end{equation}
where the summation is taken over all $N$-qubit Pauli strings $P$ that have only $I,X,Z$ symbols, and $w_X(P)$ denotes the $X$-weight of $P$, namely the number of $X$ symbols in $P$ (and $w_Z(P)$ is defined analogously). It follows that
\begin{align}
    p\left(j=\frac{N_C}{2}\right)P_{w-\alpha,w} &= \frac{1}{2^{N_C}}\sum_{P\in\{I,X,Z\}^{\otimes N_C}}\lambda_x^{w_X(P)}\lambda_z^{w_Z(P)}\bra{(w-\alpha)^{(s)}}P\ket{w^{(s)}} \\
    &= \frac{1}{2^{N_C}}\sum_{P\in\{I,X,Z\}^{\otimes N_C}}\lambda_x^{w_X(P)}\lambda_z^{w_Z(P)}\frac{1}{\sqrt{\binom{N_C}{w-\alpha}\binom{N_C}{w}}}\sum_{w(b)=w-\alpha}\sum_{w(b')=w}\bra{b}P\ket{b'}.
\end{align}

\vspace{0.5\baselineskip}

As a result, for each bit string $b'$ with Hamming weight $w$, we need to count the number of Pauli strings $P$ such that $P\ket{b'}$ is (up to sign) a computational basis state $b$ with Hamming weight $w-\alpha$. Since the $X$ symbols are the only ones (out of $I$, $X$, $Z$) that change the Hamming weight, the number of $X$ symbols acting on a $1$ bit has to be exactly $\alpha$ more than the number of $X$ symbols acting on a $0$ bit. As a result, there are
\begin{equation}
    \binom{w}{\frac{w_X(P)+\alpha}{2}}\binom{N_C-w}{\frac{w_X(P)-\alpha}{2}}
\end{equation}
ways to choose the locations of the $X$ symbols in the Pauli strings.

\vspace{0.5\baselineskip}

Let us consider what happens when we consider a Pauli string $P$ that has just $I$ and $X$ symbols, and then choose some subset of the $I$ symbols to replace with $Z$ symbols. Due to the $\lambda_z^{w_Z(P)}$ factor in the summation, and due to the fact that $Z\ket{0} = \ket{0}$ but $Z\ket{1} = -\ket{1}$, each $I$ replaced with a $Z$ acting on a $0$ bit in $b'$ changes the term associated with that Pauli string by a factor of $+\lambda_z$, while each $I$ replaced with a $Z$ acting on a $1$ bit in $b'$ changes the term associated with that Pauli string by a factor of $-\lambda_z$. As a result, when we sum over all Pauli strings $P$ with a fixed set of locations for the $X$ symbols but with all $2^{N-w_X(P)}$ possible choices of $I$ and $Z$ symbols, we accumulate a factor of
\begin{equation}
    (1+\lambda_z)^{N_C-w-j}(1-\lambda_z)^{w-j-\alpha},
\end{equation}
where we define $j\equiv \frac{w_X(P)-\alpha}{2}$ as a convenient summation index. By summing over all possible values of $j$, we can simplify our summation to
\begin{equation}
    p\left(j=\frac{N_C}{2}\right)P_{w-\alpha,w} = \frac{1}{2^{N_C}}\frac{1}{\sqrt{\binom{N_C}{w-\alpha}\binom{N_C}{w}}}\sum_{w(b')=w}\sum_{j=0}^{N_C}(1+\lambda_z)^{N_C-w-j}(1-\lambda_z)^{w-j-\alpha}\lambda_x^{2j+\alpha}.
\end{equation}

\vspace{0.5\baselineskip}

Finally, we note that there are $\binom{N_C}{w}$ input bit strings with Hamming weight $w$, and each of them contributes equally to the sum. We thus obtain
\begin{align}
    p\left(j=\frac{N_C}{2}\right)P_{w-\alpha,w} &= \frac{1}{2^{N_C}}\frac{\binom{N_C}{w}}{\sqrt{\binom{N_C}{w}\binom{N_C}{w-\alpha}}}\sum_{j=0}^{N_C}\binom{w}{j+\alpha}\binom{N_C-w}{j}(1+\lambda_z)^{N_C-w-j}(1-\lambda_z)^{w-j-\alpha}\lambda_x^{2j+\alpha} \\
    &= \frac{\lambda_x^\alpha}{2^{N_C}}\sqrt{\frac{\binom{N_C}{w}}{\binom{N_C}{w-\alpha}}}\sum_{j=0}^{N_C}\binom{w}{j+\alpha}\binom{N_C-w}{j}(1+\lambda_z)^{N_C-w-j}(1-\lambda_z)^{w-j-\alpha}\lambda_x^{2j} \\
    &= \left(\frac{\lambda_x}{1-\lambda_z}\right)^\alpha\left(\frac{1+\lambda_z}{2}\right)^{N_C}\left(\frac{1-\lambda_z}{1+\lambda_z}\right)^w\sqrt{\frac{\binom{N}{w}}{\binom{N_C}{w-\alpha}}}\sum_{j=0}^{N_C}\binom{w}{j+\alpha}\binom{N_C-w}{j}\left(\frac{\lambda_x^2}{1-\lambda_z^2}\right)^j,
\end{align}
exactly as desired.
\end{proof}

As a corollary, we can simplify the formula even a bit further in the equatorial case:

\begin{corollary}
\label{cor:P-val-single-summation-equatorial}
The equatorial case $\Theta_{\text{in}} = \frac{\pi}{2}$, which implies $\lambda_x = \lambda$ and $\lambda_z = 0$, yields the simplified formula
\begin{equation}
    P_{w-\alpha,w} = \frac{c_1-c_0}{c_1^{N_C+1} - c_0^{N_C+1}}\frac{\lambda^\alpha}{2^{N_C}}\sqrt{\frac{\binom{N_C}{w}}{\binom{N}{w-\alpha}}}\sum_{j=0}^{w-\alpha}\binom{w}{j+\alpha}\binom{N_C-w}{j}\lambda^{2j}.
\end{equation}
\end{corollary}

The primary advantage of Lemma \ref{lem:P-val-single-summation} is actually \textit{numerical}, rather than \textit{analytical}. In particular, even though the single summation is quite complicated, it is much more efficient to evaluate numerically. Furthermore, the binomial coefficients in the summation have an especially nice structure. In particular, one can write the summation in the $P_{w,w'}$ value in the form of a nested evaluation structure. For example, in the equatorial case (see Corollary \ref{cor:P-val-single-summation-equatorial}, the summation can be written as follows:
\begin{align}
    & A_0\Big\{1 + A_1\lambda^2\left[1 + A_2\lambda^2\left(1 + \cdots\{1 + A_{w-1}\lambda^2\}\cdots\right)\right]\Big\} \\
    & A_0 = w, \quad A_j = \frac{(N_C-w-j+1)(w-j)}{j(j+1)} \quad (1\le j\le w-1).
\end{align}
Using this nested evaluation structure, one can write a computer program to evaluate these $P_{w,w'}$ values very efficiently and also without directly invoking any large binomial coefficients. This is very nice, because it substantially increases the values of $N$ for which one can carry out effective numerics; otherwise, the act of computing large binomial coefficients explicitly would introduce numerical errors much sooner.

\appsubsec{Defining the ``Moments'' of the $P_{w-\alpha,w}$ Values}
{subsec:P-vals-moments-defining}

Now we discuss the properties of the $P_{w,w'}$ values that are the most important for our asymptotic analysis of single-shot qubit coherence distillation. In this appendix, we will always fix an integer $\alpha\in\{0,1\}$, which we call the \textbf{offset}, and then we will study the values $P_{w-\alpha,w}$ for $\alpha\le w\le N_C$. These values form a ``distribution'', with the caveat that they are not normalized for $\alpha\neq 0$. (For $\alpha=0$, these values are the diagonal entries, and hence they must add up to $1$. For small values of $\alpha$, they add up to only slightly less than $1$.) Hence we can study the moments of these ``distributions'', which turn out to play a crucial role in the asymptotic analysis of single-shot qubit coherence distillation, as shown in Appendices \ref{appendix:asymptotic-expansion}, \ref{appendix:1st-order-optimality}, \ref{appendix:2nd-order-optimality}, and \ref{appendix:equatorial-3rd-order-optimality}.

\vspace{0.5\baselineskip}

So why do we separate the $P_{w,w'}$ values by their offset, and why do we only focus on $\alpha=0$ and $\alpha=1$? As shown in Appendix \ref{appendix:boundary-value-problem}, because our output state is a single qubit, the only $P_{w,w'}$ values we care about are actually $P_{w-1,w}$ and $P_{w,w}$. In general, when a state with $U(1)$ asymmetry is written in the energy eigenbasis, each matrix entry will have an integer frequency at which it time evolves, and this is what we refer to as the ``modes'' of a coherent state. The state $\rho^{\otimes N}$ has modes ranging from $-N$ to $+N$ (and similarly, the Schur-sampled state $\rho_{N_C/2}$ has modes ranging from $-N_C$ to $+N_C$), but our single-qubit output state only has modes $-1,0,+1$. (These ideas are shown more explicitly in the first proof of the Kraus representation of an optimal coherence distillation protocol, which is shown in Appendix \ref{appendix:deriving-kraus-rep}\ref{subsec:distillation-protocol-form-proof-choi-matrices}.) It turns out that $P_{w,w'}$ corresponds precisely to a matrix entry that time evolves with frequency $w'-w$ (in other words, the offset is the same as the frequency of time evolution). However, since the $P_{w,w'}$ values are known to be symmetric in $w$ and $w'$, the distribution with offset $-\alpha$ is the same as the distribution with offset $+\alpha$.  Hence the only values that we need to study are $P_{w-1,w}$ (offset $\alpha=1$) and $P_{w,w}$ (offset $\alpha=0$).

\vspace{0.5\baselineskip}

In general, we would like to study centered moments of these distributions. Furthermore, it is a little bit nicer to consider moments of $w/N_C$, rather than just $w$ itself, since that helps make the decay of higher moments more apparent. As a result, it is natural to define the following quantity, which is the mean (first moment) of $w/N_C$ under the distribution with offset $0$:
\begin{equation}
    \mu \coloneqq \sum_{w=0}^{N_C}\frac{w}{N_C}P_{w,w}.
\end{equation}
Since the $P_{w,w}$ values concentrate around $w\approx\mu N$, the performance of a coherence distillation protocol is controlled primarily by how it acts in the vicinity of $w\approx\mu N$. Therefore, it makes sense to define this quantity as an anchor point and optimize the behavior of the protocol in this range. We explain how to do this more systematically in Appendix \ref{appendix:asymptotic-expansion}, and we carry out that optimization of the protocol in Appendices \ref{appendix:1st-order-optimality}, \ref{appendix:2nd-order-optimality}, and \ref{appendix:equatorial-3rd-order-optimality}.

\vspace{0.5\baselineskip}

As we will see in Appendix \ref{appendix:understanding-P-vals}\ref{subsec:P-vals-moments-offset0}, the quantity $\mu$ we defined above has the following form (where we will use ``e.s.e.'' from now on as an abbreviation for ``exponentially small error''):
\begin{equation}
    \mu = \frac{1 - C_{\text{in}}}{2} + \frac{C_{\text{in}}(1-\lambda)}{2\lambda}\frac{1}{N_C} + \text{e.s.e.}
\end{equation}
Based on this, we can define the $p^{\text{th}}$ centered moment of the distribution with offset $\alpha$ as follows:
\begin{equation}
    \mM_p^{(\alpha)} \coloneqq \sum_{w=0}^{N_C}\left(\frac{w-\frac{\alpha}{2}}{N_C}-\mu\right)^pP_{w-\alpha,w}.
\end{equation}
Notice that we subtract $\frac{\alpha}{2}$ from $w$ because we want the quantity whose moments we are computing to be symmetric in the indices $w-\alpha$ and $w$.

\appsubsec{Moments of the $P_{w,w}$ Values (Offset $\alpha=0$)}
{subsec:P-vals-moments-offset0}

Now that we have defined the moments of the $P_{w-\alpha,w}$ values, we are ready to state our first main result, which concerns the centered moments for offset $\alpha=0$:

\begin{lemma}[Centered moments of the $P_{w,w}$ values]
\label{lem:centered-moments-offset-0}
The centered moments for offset $\alpha=0$ take the following values up to $\Theta(N_C^{-2})$ precision:
\begin{align}
    \mM_0^{(0)} &= 1 \\
    \mM_1^{(0)} &= 0 \\
    \mM_2^{(0)} &= \frac{S_{\text{in}}^2}{4\lambda^2}\frac{1}{N_C} + \frac{(1-\lambda)\left(\left(C_{\text{in}}^2 - S_{\text{in}}^2\right) + C_{\text{in}}^2\lambda\right)}{4\lambda^2}\frac{1}{N_C^2} + \text{e.s.e.} \\
    \mM_3^{(0)} &= \frac{C_{\text{in}}S_{\text{in}}^2(3 - \lambda^2)}{8\lambda^2}\frac{1}{N_C^2} + O\left(N_C^{-3}\right) \\
    \mM_4^{(0)} &= \frac{3S_{\text{in}}^4}{16\lambda^2}\frac{1}{N_C^2} + O\left(N_C^{-3}\right).
\end{align}
\end{lemma}

To prove this, we will instead compute moments that are centered at $\frac{N_C}{2}$, instead of at the value $\mu$ stated previously. In particular, define the following quantity:
\begin{equation}
    \mU_p^{(0)} \coloneqq \sum_{w=0}^{N_C}\left(\frac{w}{N_C}-\frac{1}{2}\right)^pP_{w,w}
\end{equation}
The reason this quantity is nice is that, if we consider the total $z$-spin operator
\begin{equation}
    S_z \coloneqq 2J_z = \sum_{j=1}^{N_C}Z_j,
\end{equation}
then this operator satisfies
\begin{equation}
    S_z\ket{w^{(s)}} = (N_C-2w)\ket{w^{(s)}} \quad (0\le w\le N_C).
\end{equation}
Therefore, we can immediately write
\begin{equation}
    \mU_p^{(0)} = \sum_{w=0}^{N_C}\left(\frac{w}{N_C}-\frac{1}{2}\right)^pP_{w,w} = (-2N_C)^{-p}\sum_{w=0}^{N_C}(N_C-2w)^pP_{w,w} = (-2N_C)^{-p}\text{Tr}\left[\rho_{N_C/2}S_z^p\right].
\end{equation}

\vspace{0.5\baselineskip}

Now we need to find another way to compute the above trace. We can do it by having the powers of $S_z$ act on $\rho_{N_C/2}$ instead. Note that $\rho_{N_C/2}$ is diagonal in the basis of $\ket{k_{\hat{n}}^{(s)}}$ states for $0\le k\le N_C$. As a result, we just need to know how $S_z$ acts on such states. Fortunately, the action of $S_z$ on such states is actually relatively simple. This is because, in the spin-$j$ representation of $SU(2)$, $S_z$ is a diagonal operator, while $S_x$ and $S_y$ have entries only on the superdiagonal and subdiagonal. In this way, they only produce interactions between adjacent Hamming weights.

\vspace{0.5\baselineskip}

In particular, since we are considering symmetrized Hamming weight states based on the Bloch vector $\hat{n} = S_{\text{in}}\hat{x} + C_{\text{in}}\hat{z}$, the operator $S_z$ actually acts as $S_{\text{in}}S_x - C_{\text{in}}S_z$ on the basis of $\ket{k_{\hat{n}}^{(s)}}$ states. Therefore, we have the following equation:
\begin{align}
    S_z\ket{k_{\hat{n}}^{(s)}} &= S_{\text{in}}\sqrt{k(N_C-k+1)}\ket{(k-1)_{\hat{n}}^{(s)}} \\
    & \quad - C_{\text{in}}(N_C-2k)\ket{k_{\hat{n}}^{(s)}} \\
    & \quad + S_{\text{in}}\sqrt{(k+1)N_C-k)}\ket{(k+1)_{\hat{n}}^{(s)}}.
\end{align}
Due to the above fact, $S_z^p\ket{k_{\hat{n}}^{(s)}}$ will contain terms from $\ket{(k-p)_{\hat{n}}^{(s)}}$ to $\ket{(k+p)_{\hat{n}}^{(s)}}$. Each term can be produced by some path of length $p$, where at each step, you choose to increase the Hamming weight by $1$, decrease it by $1$, or keep the Hamming weight the same. We will call these options ``up'' (U), ``down'' (D), and ``middle'' (M), respectively. However, only the $\ket{k_{\hat{n}}^{(s)}}$ term will contribute to the trace, which corresponds to all possible paths that have net zero Hamming weight change (in other words, an equal number of U's and D's).

\vspace{0.5\baselineskip}

This yields a general strategy for computing $\mU_p^{(0)}$:
\begin{itemize}
    \item First, tabulate all paths of length $p$ of the form above (each step is drawn from $\{U,D,M\}$, and there are an equal number of U's and D's).
    \item Second, for each such path, write out the multiplicative factor that comes from each step. (For example, each ``up'' step produces a $\sqrt{(k+1)(N_C-k)}$ factor when moving from Hamming weight $k$ to Hamming weight $k+1$.)
    \item Third, multiply all these factors to obtain the overall contribution for that path, and add up these contributions over all the paths. This yields the quantity
    \begin{equation}
        \mT_p^{(0)} \coloneqq \bra{k_{\hat{n}}^{(s)}}S_z^p\ket{k_{\hat{n}}^{(s)}}.
    \end{equation}
    \item Fourth and finally, multiply this quantity by the prefactor
    \begin{equation}
        \bra{k_{\hat{n}}^{(s)}}\rho_{N_C/2}\ket{k_{\hat{n}}^{(s)}} = \frac{c_1-c_0}{c_1^{N_C+1} - c_0^{N_C+1}}c_1^kc_0^{N_C-k},
    \end{equation}
    and sum over all $0\le k\le N_C$.
\end{itemize}
Interestingly, if the input states are equatorial ($\Theta_{\text{in}} = \frac{\pi}{2}$), then the ``middle'' path closes, and only ``up'' and ``down'' steps are permitted. In this setting, the options reduce considerably. For example, $\mU_p^{(0)}$ is automatically equal to $0$ for odd $p$. If the output states are equatorial, then this simplification actually does not help us anyway, since bit-flip symmetry actually renders the $P_{w,w}$ values irrelevant, such that only the $P_{w-1,w}$ values contribute to the fidelity in a nontrivial way. However, if the input states are equatorial and the output states are non-equatorial, then this could provide some simplification.

\vspace{0.5\baselineskip}

Let us begin with the zeroth uncentered moment. Of course, this is actually trivial:
\begin{equation}
    \mU_0^{(0)} = \text{Tr}\left[\rho_{N_C/2}\right] = 1.
\end{equation}
And of course, converting this uncentered moment into a centered moment is also trivial, since $\mM^{(0)} = \mU_0^{(0)}$:
\begin{equation}
    \boxed{\mM_0^{(0)} = 1}.
\end{equation}

\vspace{0.5\baselineskip}

Now we proceed to the calculation of the first moment. By the definition of the centered moments, it is trivial to say that
\begin{equation}
    \boxed{\mM_1^{(0)} = 0}.
\end{equation}
So here we actually only care about computing the uncentered moment $\mU_1^{(0)}$, since that quantity will inform the mean $\mu$, which is then used for the centering for all other moments.

\vspace{0.5\baselineskip}

In this case, the only path of length $1$ that stays at the same Hamming weight is a single M step. As a result, we get an overall factor of $-(N_C-2k)C_{\text{in}}$, leading to the following summation:
\begin{align}
    \mU_1^{(0)} &= (-2N_C)^{-1}\frac{c_1-c_0}{c_1^{N_C+1}-c_0^{N_C+1}}\sum_{k=0}^{N_C}c_1^kc_0^{N_C-k}(-C_{\text{in}})(N_C-2k) \\
    &= \frac{\lambda}{c_1^{N_C+1}-c_0^{N_C+1}}\frac{C_{\text{in}}}{2N_C}\left[-N_C\sum_{k=0}^{N_C}c_0^kc_1^{N_C-k} + 2\sum_{k=0}^{N_C}kc_0^kc_1^{N_C-k}\right].
\end{align}
Note that we also invert the summation order with the mapping $k\mapsto N_C-k$, so that the exponent of $c_0$ increases with $k$, while the exponent of $c_1$ decreases with $k$.

\vspace{0.5\baselineskip}

Now we need to make a number of approximations. Fortunately, all of these approximations will have only exponentially small error. The first approximation is to approximate the denominator in the leading prefactor as follows (we will make this approximation repeatedly):
\begin{equation}
    c_1^{N_C+1} - c_0^{N_C+1} = c_1^{N_C+1}(1 + \text{e.s.e.}).
\end{equation}
The second approximation is to extend all of the summations to $k=\infty$, instead of terminating them at $k=N_C$. All of these summations have a decaying geometric sequence with common ratio $(c_0/c_1)$, times some polynomial of $k$, and these are all summations that can be computed using standard manipulations of the infinite geometric series. Furthermore, truncating any such summation after $N_C$ terms yields only exponentially small error. In particular, for the evaluation of $\mU_1^{(0)}$, we need the following summations:
\begin{align}
    \sum_{k=0}^{\infty}x^k &= \frac{1}{1-x} \\
    \sum_{k=0}^{\infty}kx^k &= \frac{x}{(1-x)^2}.
\end{align}
It follows that
\begin{align}
    \sum_{k=0}^{N_C}c_0^kc_1^{N_C-k} &= \frac{c_1^{N_C+1}}{\lambda}(1 + \text{e.s.e.}) \\
    \sum_{k=0}^{N_C}kc_0^kc_1^{N_C-k} &= \frac{c_1^{N_C+1}c_0}{\lambda^2}(1 + \text{e.s.e.}).
\end{align}
Plugging these results into the formula for $\mU_1^{(0)}$ and simplifying yields
\begin{align}
    \mU_1^{(0)} &= \frac{\lambda}{c_1^{N_C+1}-c_0^{N_C+1}}\frac{C_{\text{in}}}{2N_C}\left[-N_C\sum_{k=0}^{N_C}c_0^kc_1^{N_C-k} + 2\sum_{k=0}^{N_C}kc_0^kc_1^{N_C-k}\right] \\
    &= \frac{\lambda}{c_1^{N_C+1}}(1 + \text{e.s.e.})\frac{C_{\text{in}}}{2N_C}\left[-N_C\left(\frac{c_1^{N_C+1}}{\lambda}\right)(1 + \text{e.s.e.}) + 2\left(\frac{c_1^{N_C+1}c_0}{\lambda^2}\right)(1 + \text{e.s.e.})\right] \\
    &= \frac{C_{\text{in}}}{2N_C}\left[-N_C(1) + 2\left(\frac{c_0}{\lambda}\right)\right] + \text{e.s.e.} \\
    &= -\frac{C_{\text{in}}}{2} + \frac{(1-\lambda)C_{\text{in}}}{2\lambda}\frac{1}{N_C} + \text{e.s.e.}
\end{align}
Finally, note that $\mu$ is related to the uncentered moment $\mU_1^{(0)}$ as follows:
\begin{equation}
    \mu \coloneqq \sum_{w=0}^{N_C}\frac{w}{N_C}P_{w,w} = \mU_1^{(0)} + \frac{1}{2}.
\end{equation}
We conclude that $\mu$, the mean value of $w/N_C$ under the distribution given by the $P_{w,w}$ values, takes the following form:
\begin{equation}
    \boxed{\mu = \frac{1-C_{\text{in}}}{2} + \frac{(1-\lambda)C_{\text{in}}}{2\lambda}\frac{1}{N_C} + \text{e.s.e.}}
\end{equation}

\vspace{0.5\baselineskip}

We now proceed to computing the second uncentered and centered moments. Now, there are $3$ possible up-down-middle paths that have net zero Hamming weight change, and each of them contributes a term as follows:
\begin{align}
    (U,D): &\quad\quad (k+1)(N_C-k)S_{\text{in}}^2 \\
    (M,M): &\quad\quad (N_C-2k)C_{\text{in}}^2 \\
    (D,U): &\quad\quad k(N_C-k+1)S_{\text{in}}^2 \\
\end{align}
Adding these three terms yields
\begin{align}
    \mT_2^{(0)} &\coloneqq \bra{k_{\hat{n}}^{(s)}}S_z^2\ket{k_{\hat{n}}^{(s)}} \\
    &= \frac{1}{2}\left[N_C(N_C+2) - (N_C-2k)^2\right]S_{\text{in}}^2 + (N_C-2k)^2C_{\text{in}}^2 \\
    &= \left[N_C\left(-1 +(N_C+1)C_{\text{in}}^2\right)\right] + \left[2N_C\left(1 - 3C_{\text{in}}^2\right)\right]k + \left[2\left(-1 + 3C_{\text{in}}^2\right)\right]k^2.
\end{align}
This implies that
\begin{align}
    \mU_2^{(0)} &= (-2N_C)^{-2}\frac{c_1-c_0}{c_1^{N_C+1}-c_0^{N_C+1}}\quad \Bigg\{N_C\left(-1 +(N_C+1)C_{\text{in}}^2\right)\sum_{k=0}^{N_C}c_1^{N_C-k}c_0^k \\
    & \quad + 2N_C\left(1 - 3C_{\text{in}}^2\right)\sum_{k=0}^{N_C}kc_1^{N_C-k}c_0^k + 2\left(-1 + 3C_{\text{in}}^2\right)\sum_{k=0}^{N_C}k^2c_1^{N_C-k}c_0^k\Bigg\}.
\end{align}
Once again, we now need to make some approximations. We actually use exactly the same approximations  that we used for $\mU_1^{(0)}$, except now we add one extra. In particular, the infinite sum
\begin{equation}
    \sum_{k=0}^{\infty}k^2x^k = \frac{x(1+x)}{(1-x)^3}
\end{equation}
implies that
\begin{equation}
    \sum_{k=0}^{N_C}k^2c_0^kc_1^{N_C-k} = \frac{c_1^{N_C+1}c_0}{\lambda^3}(1 + \text{e.s.e.}).
\end{equation}
Plugging in all these approximations and simplifying yields 
\begin{equation}
    \mU_2^{(0)} = \frac{C_{\text{in}}^2}{4} + \left[\frac{1 + (-3+2\lambda)C_{\text{in}}^2}{2\lambda}\right]\frac{1}{N_C} + \left[\frac{(1-\lambda)(-2 + 3C_{\text{in}}^2)}{8\lambda^2}\right]\frac{1}{N_C^2} + \text{e.s.e.}
\end{equation}
Finally, we wish to relate the uncentered moment $\mU_2^{(0)}$ to the centered moment $\mM_2^{(0)}$. We can do so using the formula $\text{Var}[X] = \Ebb[X^2] - \Ebb[X]^2$, which implies that
\begin{equation}
    \mM_2^{(0)} = \mU_2^{(0)} - \left(\mU_1^{(0)}\right)^2.
\end{equation}
Simplifying this quantity reveals that the second centered moment of the $P_{w,w}$ values looks as follows:
\begin{equation}
    \boxed{\mM_2^{(0)} = \frac{S_{\text{in}}^2}{4\lambda}\frac{1}{N_C} + \frac{(1-\lambda)\left(\left(C_{\text{in}}^2 - S_{\text{in}}^2\right) + C_{\text{in}}^2\lambda\right)}{4\lambda^2}\frac{1}{N_C^2} + \text{e.s.e.}}
\end{equation}
This completes the proof of the $\mM_2^{(0)}$ formula that we stated in Lemma \ref{lem:centered-moments-offset-0}. The first term in this formula can actually be understood as a straightforward consequence of the fact that the $P_{w,w}$ values, as a distribution for the Hamming weight $w$, approximate a Gaussian with mean $\sim\frac{1-C_{\text{in}}}{2}N_C$ and variance $\sim\frac{S_{\text{in}}^2}{4\lambda}N_C$. However, the second term in this formula requires an understanding of the deviations away from this Gaussian behavior.

\vspace{0.5\baselineskip}

Now we proceed to the third moment. Here are all the up-down-middle paths of length $3$ and net zero Hamming weight change, along with their respective contributions:
\begin{align}
    (M,M,M): &\quad\quad -(N_C-2k)^3C_{\text{in}}^3 \\
    (U,M,D): &\quad\quad -(k+1)(N_C-k)(N_C-2(k+1))C_{\text{in}}S_{\text{in}}^2 \\
    (D,M,U): &\quad\quad -k(N_C-k+1)(N_C-2(k-1))C_{\text{in}}S_{\text{in}}^2 \\
    (U,D,M): &\quad\quad -(k+1)(N_C-k)(N_C-2k)C_{\text{in}}S_{\text{in}}^2 \\
    (M,D,U): &\quad\quad -(k+1)(N_C-k)(N_C-2k)C_{\text{in}}S_{\text{in}}^2 \\
    (D,U,M): &\quad\quad -k(N_C-k+1)(N_C-2k)C_{\text{in}}S_{\text{in}}^2 \\
    (M,D,U): &\quad\quad -k(N_C-k+1)(N_C-2k)C_{\text{in}}S_{\text{in}}^2
\end{align}
When we add these quantities together, we obtain
\begin{align}
    \mT_3^{(0)} &\coloneqq \bra{k_{\hat{n}}^{(s)}}S_z^3\ket{k_{\hat{n}}^{(s)}} \\
    &= -C_{\text{in}}\left(4k + 12k^3 - 2N_C - 6kN_C - 18k^2N_C + 3N_C^2 + 6kN_C^2\right) \\
    & \quad - C_{\text{in}}^3\left(-4k - 20k^3 + 2N_C + 6kN_C + 30k^2N_C - 3N_C^2 - 12kN_C^2 + N_C^3\right).
\end{align}
We now plug this quantity into our infinite summation:
\begin{equation}
    \mU_3^{(0)} = (-2N_C)^3\frac{c_1-c_0}{c_1^{N_C+1} - c_0^{N_C+1}}\sum_{k=0}^{N_C}c_1^kc_0^{N_C-k}\mT_3^{(0)}.
\end{equation}
After making the usual approximations and invoking the formulas for $\sum_{k=0}^{\infty}k^px^k$ for $p=0,1,2,3$, we eventually obtain
\begin{align}
    \mU_3^{(0)} &= \frac{C_{\text{in}}^3}{8} + \frac{3C_{\text{in}}\left[1 + C_{\text{in}}^2(-2+\lambda)\right]}{8\lambda}\frac{1}{N_C} \\
    & \quad + \frac{C_{\text{in}}\left[(-9 + 6\lambda + \lambda^2) + C_{\text{in}}^2(15 - 12\lambda - \lambda^2)\right]}{8\lambda^2}\frac{1}{N_C^2} \\
    & \quad + \frac{C_{\text{in}}(1-\lambda)\left[(9 - \lambda^2) + 3C_{\text{in}}^2(5 - \lambda^2)\right]}{8\lambda^3}\frac{1}{N_C^3} + \text{e.s.e.}
\end{align}
Finally, we wish to relate the uncentered moment $\mU_3^{(0)}$ to the centered moment $\mM_3^{(0)}$. We can do so using the following formula:
\begin{align}
    \Ebb\left[(X - \Ebb[X])^3\right] &= \Ebb[X^3] - 3\Ebb[X^2]\Ebb[X] + 3\Ebb[X]\Ebb[X]^2 - \Ebb[X]^3 \\
    &= \Ebb[X^3] - 3\Ebb[X^2]\Ebb[X] + 2\Ebb[X]^3,
\end{align}
which in turn implies that
\begin{equation}
    \mM_3^{(0)} = \mU_3^{(0)} - 3\mU_2^{(0)}\mU_1^{(0)} + 2\left(\mU_1^{(0)}\right)^3.
\end{equation}
Plugging in the values that we have already computed yields
\begin{equation}
    \boxed{\mM_3^{(0)} = \frac{C_{\text{in}}S_{\text{in}}^2(3 - \lambda^2)}{8\lambda^2}\frac{1}{N_C^2} + \frac{C_{\text{in}}(1-\lambda)\left[(-6 - 3\lambda + \lambda^2) + C_{\text{in}}^2(8 + 5\lambda - \lambda^2)\right]}{8\lambda^3}\frac{1}{N_C^3} + \text{e.s.e.}}
\end{equation}
We do not actually use the $N_C^{-3}$ term for anything in this paper, so in Appendix \ref{appendix:2nd-order-optimality}, we use the simplified formula
\begin{equation}
    \boxed{\mM_3^{(0)} = \frac{C_{\text{in}}S_{\text{in}}^2(3 - \lambda^2)}{8\lambda^2}\frac{1}{N_C^2} + O\left(N_C^{-3}\right)}.
\end{equation}
This completes the proof of the $\mM_3^{(0)}$ formula that we stated in Lemma \ref{lem:centered-moments-offset-0}. The curious reader may be interested to know that the computation of this third moment was the breaking point that finally compelled the first author of this paper to swallow his pride and use Mathematica.

\vspace{0.5\baselineskip}

Finally, we proceed to the fourth moment. We begin by tabulating all the up-down-middle paths of length $4$ and net zero Hamming weight change, along with their respective contributions. Now the number of paths is considerably larger. They can be organized as follows: $1$ path with four M's; 12 paths with one U, one D, and two M's; and $6$ paths with two U's and two D's. Here is the full list:
\begin{align}
    (M,M,M,M): &\quad\quad (N_C-2k)^4C_{\text{in}}^4 \\
    (M,M,U,D): &\quad\quad (k+1)(N_C-k)(N_C-2k)^2C_{\text{in}}^2S_{\text{in}}^2 \\
    (M,M,D,U): &\quad\quad k(N_C-k+1)(N_C-2k)^2C_{\text{in}}^2S_{\text{in}}^2 \\
    (U,D,M,M): &\quad\quad (k+1)(N_C-k)(N_C-2k)^2C_{\text{in}}^2S_{\text{in}}^2 \\
    (D,U,M,M): &\quad\quad k(N_C-k+1)(N_C-2k)^2C_{\text{in}}^2S_{\text{in}}^2 \\
    (M,U,D,M): &\quad\quad (k+1)(N_C-k)(N_C-2k)^2C_{\text{in}}^2S_{\text{in}}^2 \\
    (M,D,U,M): &\quad\quad k(N_C-k+1)(N_C-2k)^2C_{\text{in}}^2S_{\text{in}}^2 \\
    (M,U,M,D): &\quad\quad (k+1)(N_C-k)(N_C-2k)(N_C-2(k+1))C_{\text{in}}^2S_{\text{in}}^2 \\
    (M,D,M,U): &\quad\quad k(N_C-k+1)(N_C-2(k-1))(N_C-2k)C_{\text{in}}^2S_{\text{in}}^2 \\
    (U,M,M,D): &\quad\quad (k+1)(N_C-k)(N_C-2(k+1))^2C_{\text{in}}^2S_{\text{in}}^2 \\
    (D,M,M,U): &\quad\quad k(N_C-k+1)(N_C-2(k-1))^2C_{\text{in}}^2S_{\text{in}}^2 \\
    (U,M,D,M): &\quad\quad (k+1)(N_C-k)(N_C-2k)(N_C-2(k+1))C_{\text{in}}^2S_{\text{in}}^2 \\
    (D,M,U,M): &\quad\quad k(N_C-k+1)(N_C-2(k-1))(N_C-2k)C_{\text{in}}^2S_{\text{in}}^2 \\
    (U,U,D,D): &\quad\quad (k+1)(N_C-k)(k+2)(N_C-k-1)S_{\text{in}}^4 \\
    (U,D,U,D): &\quad\quad [(k+1)(N_C-k)]^2S_{\text{in}}^4 \\
    (U,D,D,U): &\quad\quad k(N_C-k+1)(k+1)(N_C-k)S_{\text{in}}^4 \\
    (D,U,U,D): &\quad\quad k(N_C-k+1)(k+1)(N_C-k)S_{\text{in}}^4 \\
    (D,U,D,U): &\quad\quad [k(N_C-k+1)]^2S_{\text{in}}^4 \\
    (D,D,U,U): &\quad\quad k(N_C-k+1)(k-1)(N_C-k+2)S_{\text{in}}^4
\end{align}
When we add these quantities together, we obtain
\begin{align}
    \mT_4^{(0)} &\coloneqq \bra{k_{\hat{n}}^{(s)}}S_z^4\ket{k_{\hat{n}}^{(s)}} \\
    &= \Big\{10k^2 + 6k^4 - 2N_C - 10kN_C + 6k^2N_C - 12k^3N_C \\
    & \quad + 3N_C^2 + 6kN_C^2 + 6k^2N_C^2\Big\} \\
    & \quad + C_{\text{in}}^2\Big\{-60k^2 - 60k^4 + 8N_C + 60kN_C + 36k^2N_C + 120k^3N_C \\
    & \quad - 14N_C^2 - 36kN_C^2 - 72k^2N_C^2 + 6N_C^3 + 12kN_C^3\Big\} \\
    & \quad + C_{\text{in}}^4\Big\{50k^2 + 70k^4 - 6N_C - 50kN_C - 30k^2N_C - 140k^3N_C \\
    & \quad + 11N_C^2 + 30kN_C^2 + 90k^2N_C^2 - 6N_C^3 - 20kN_C^3 + N_C^4\Big\}.
\end{align}
We now plug this quantity into our infinite summation:
\begin{equation}
    \mU_4^{(0)} = (-2N_C)^4\frac{c_1-c_0}{c_1^{N_C+1} - c_0^{N_C+1}}\sum_{k=0}^{N_C}c_1^kc_0^{N_C-k}\mT_4^{(0)}.
\end{equation}
After making the usual approximations and invoking the formulas for $\sum_{k=0}^{\infty}k^px^k$ for $p=0,1,2,3,4$, we eventually obtain
\begin{align}
    \mU_4^{(0)} &= \frac{C_{\text{in}}^4}{16} + \frac{C_{\text{in}}^2\left[3 + C_{\text{in}}^2(-5 + 2\lambda)\right]}{8\lambda}\frac{1}{N_C} \\
    & \quad + \frac{3 + 2C_{\text{in}}^2(-18 + 9\lambda + 2\lambda^2) + C_{\text{in}}^4(45 - 30\lambda - 4\lambda^2)}{16\lambda^2}\frac{1}{N_C^2} \\
    & \quad + \frac{(-9 + 6\lambda + \lambda^2) + 2C_{\text{in}}^2(45 - 36\lambda - 9\lambda^2 + 4\lambda^3) + C_{\text{in}}^4(-105 + 90\lambda + 25\lambda^2 - 16\lambda^3)}{16\lambda^3}\frac{1}{N_C^3} \\
    & \quad + \frac{(1-\lambda)\left[(9 - \lambda^2) + 30C_{\text{in}}^2(-3 + \lambda^2) + 15C_{\text{in}}^4(7 - 3\lambda^2)\right]}{16\lambda^4}\frac{1}{N_C^4} + \text{e.s.e.}
\end{align}
Finally, we wish to relate the uncentered moment $\mU_4^{(0)}$ to the centered moment $\mM_4^{(0)}$. We can do so using the following formula:
\begin{align}
    & \quad \Ebb\left[(X - \Ebb[X])^4\right] \\
    &= \Ebb[X^4] - 4\Ebb[X^3]\Ebb[X] + 6\Ebb[X^2]\Ebb[X]^2 - 4\Ebb[X]\Ebb[X]^3 + \Ebb[X]^4 \\
    &= \Ebb[X^4] - 4\Ebb[X^3]\Ebb[X] + 6\Ebb[X^2]\Ebb[X]^2 - 3\Ebb[X]^4,
\end{align}
which in turn implies that
\begin{equation}
    \mM_4^{(0)} = \mU_4^{(0)} - 4\mU_3^{(0)}\mU_1^{(0)} + 6\mU_2^{(0)}\left(\mU_1^{(0)}\right)^2 - 3\left(\mU_1^{(0)}\right)^4.
\end{equation}
Plugging in the values that we have already computed yields
\begin{subequations}
\begin{empheq}[box=\widefbox]{align}
    \mM_4^{(0)} &= \frac{3S_{\text{in}}^4}{16\lambda^2}\frac{1}{N_C^2} + \frac{S_{\text{in}}^2\left[(-9 + 6\lambda + \lambda^2) + 3C_{\text{in}}^2(9 - 2\lambda - 5\lambda^2)\right]}{16\lambda^3}\frac{1}{N_C^3} \\
    & \quad + \frac{(1-\lambda)\left[(9 - \lambda^2) + 4C_{\text{in}}^2(-15 - 6\lambda + 5\lambda^2 + \lambda^3) + 3C_{\text{in}}^4(20 + 11\lambda - 8 \lambda^2 - 3\lambda^3)\right]}{16\lambda^4}\frac{1}{N_C^4} + \text{e.s.e.}
\end{empheq}
\end{subequations}
We do not actually use the $N_C^{-3}$ and $N_C^{-4}$ terms for anything in this paper, so in Appendix \ref{appendix:2nd-order-optimality}, we use the simplified formula
\begin{equation}
    \boxed{\mM_4^{(0)} = \frac{3S_{\text{in}}^4}{16\lambda^2}\frac{1}{N_C^2} + O\left(N_C^{-3}\right)}.
\end{equation}
This completes the proof of the $\mM_4^{(0)}$ formula that we stated in Lemma \ref{lem:centered-moments-offset-0}. At this level of precision, the formula can actually be understood directly from the Gaussian approximation for the $P_{w,w}$ values.

\appsubsec{Moments of the $P_{w-1,w}$ Values (Offset $\alpha=1$) in the Equatorial Case}
{subsec:P-vals-moments-offset1-equatorial}

Our second main result shows the centered moments for offset $\alpha=1$ in the equatorial case, up to $\Theta(N_C^{-3})$ precision:

\begin{lemma}[Centered moments of the $P_{w-1,w}$ values in the equatorial case]
\label{lem:equatorial-offset1-moments0246}
In the equatorial special case, the even centered moments of the $P_{w-1,w}$ values take the following values up to $\Theta\left(N^{-3}\right)$ precision:
\begin{align}
    \mM_0^{(1)} &= 1 + \left(-\frac{1}{2\lambda}\right)\frac{1}{N_C} + \left(\frac{-3 + 4\lambda}{8\lambda^2}\right)\frac{1}{N_C^2} + \left(\frac{-9 + 12\lambda - 10\lambda^2}{16\lambda^3}\right)\frac{1}{N_C^3} + O\left(N_C^{-4}\right) \\
    \mM_2^{(1)} &= \left(\frac{1}{4\lambda}\right)\frac{1}{N_C} + \left(\frac{-3 + 2\lambda}{8\lambda^2}\right)\frac{1}{N_C^2} + \left(-\frac{3 + 4\lambda^2}{32\lambda^3}\right)\frac{1}{N_C^3} + O\left(N_C^{-4}\right) \\
    \mM_4^{(1)} &= \left(\frac{3}{16\lambda^2}\right)\frac{1}{N_C^2} + \left(\frac{-21 + 12\lambda + 2\lambda^2}{32\lambda^3}\right)\frac{1}{N_C^3} + O\left(N_C^{-4}\right) \\
    \mM_6^{(1)} &= \left(\frac{15}{64\lambda^3}\right)\frac{1}{N_C^3} + O\left(N_C^{-4}\right).
\end{align}
\end{lemma}

It is worth highlighting some consistent patterns in these formulas:
\begin{itemize}
    \item The leading-order asymptotic of $\mM_{2p}^{(1)}$ is given by $\mM_{2p}^{(1)} \sim \frac{(2p-1)!!}{(4\lambda)^p}\frac{1}{N_C^p}$. This is a direct consequence of the fact that the ``distribution'' of the $P_{w-1,w}$ values tends to a Gaussian with normalization $\sim 1$ and variance $\sim\frac{1}{4\lambda N_C}$.
    \item In general, each coefficient of $N_C^{-p}$ is a rational function of $\lambda$, where the denominator is a power of $2$ times $\lambda^p$, and the numerator is an integer-coefficient polynomial in $\lambda$ of degree at most $p-1$.
    \item The coefficient of $N_C^{-p}$ always has $\lambda^p$ in the denominator. This is a result of the fact that the quantity $\lambda N_C \sim \lambda^2N$ is an indicator of whether you should consider yourself to be in the ``large $N$'' asymptotic regime for this problem.
\end{itemize}

Let us now explain how we derive these formulas. The strategy we show here for the equatorial case will also apply to the general case, but it will become noticeably more complicated, which is why we first show it in this simpler setting. In Appendix \ref{appendix:understanding-P-vals}\ref{subsec:P-vals-moments-offset1-general}, we will explain how this strategy has to be adjusted for the general case.

\vspace{0.5\baselineskip}

First, define the total spin operator in the $\hat{n}$ direction:
\begin{equation}
    S_{\hat{n}} \coloneqq 2J_{\hat{n}} = \sum_{j=1}^{N_C}\left(\hat{n}\cdot\vec{\sigma}\right)_j.
\end{equation}
The central idea is that we will use two different methods to evaluate the quantity
\begin{equation}
    \text{Tr}\left[\rho_{N_C/2}S_z^pS_{\hat{n}}\right].
\end{equation}

\vspace{0.5\baselineskip}

For the first evaluation method, let us see how $S_{\hat{n}}$ acts on a symmetrized Hamming weight state in the computational basis:
\begin{align}
    S_{\hat{n}}\ket{w^{(s)}} &= S_{\text{in}}\sqrt{(w+1)(N_C-w)}\ket{(w+1)^{(s)}} \\
    & \quad + C_{\text{in}}(N_C-2w)\ket{w^{(s)}} \\
    & \quad + S_{\text{in}}\sqrt{w(N_C-w+1)}\ket{(w-1)^{(s)}}.
\end{align}
In the equatorial case $\hat{n} = \hat{x}$, this simplifies to
\begin{equation}
    S_x\ket{w^{(s)}} = \sqrt{(w+1)(N_C-w)}\ket{(w+1)^{(s)}} + \sqrt{w(N_C-w+1)}\ket{(w-1)^{(s)}}.
\end{equation}
Furthermore, the operator $S_z^p$ acts in a predictable way on such states:
\begin{equation}
    S_z^p\ket{w^{(s)}} = (N_C-2w)^p\ket{w^{(s)}}.
\end{equation}
Therefore, if we evaluate the trace of interest by expanding it in the symmetrized computational basis, we obtain the following:
\begin{align}
    & \quad \text{Tr}\left[\rho_{N_C/2}S_z^pS_x\right] \\
    &= \sum_{w=0}^{N_C}\bra{w^{(s)}}\rho_{N_C/2}S_z^pS_x\ket{w^{(s)}} \\
    &= \sum_{w=1}^{N_C}\sqrt{w(N_C-w+1)}\left[(N_C-2(w-1))^p + (N_C-2w)^p\right]\bra{(w-1)^{(s)}}\rho_{N_C/2}\ket{w^{(s)}} \\
    &= \sum_{w=1}^{N_C}\sqrt{w(N_C-w+1)}\left[(N_C-2(w-1))^p + (N_C-2w)^p\right]P_{w-1,w}.
\end{align}
We now define for convenience
\begin{equation}
    z \coloneqq \frac{1}{N_C+1}\left[\left(w - \frac{1}{2}\right) - \frac{N_C}{2}\right] = \frac{w}{N_C+1} - \frac{1}{2}.
\end{equation}
The result is that the quantity
\begin{equation}
    (N_C-2(w-1))^p + (N_C-2w)^p = [-2(N_C+1)z+1]^p + [-2(N_C+1)z-1]^p
\end{equation}
is some degree-$p$ polynomial in $z$ that only has terms of degree matching the parity of $p$. Furthermore, using a Taylor series, the quantity $\sqrt{w(N_C-w+1)}$ can be written as a power series in $z$. Therefore, the second term in the above quantity can be evaluated as some infinite series with higher and higher centered moments $\mM_p^{(1)}$. We can then put everything together to get a result for this trace.

\vspace{0.5\baselineskip}

Now we must evaluate this same trace in a different way. Note that $S_{\hat{n}}$ acts as follows on the symmetrized Hamming weight states with direction $\hat{n}$:
\begin{equation}
    S_{\hat{n}}\ket{k_{\hat{n}}^{(s)}} = (2k-N_C)\ket{k_{\hat{n}}^{(s)}}.
\end{equation}
Therefore, we can compute the trace of interest as follows:
\begin{equation}
    \text{Tr}\left[\rho_{N_C/2}S_z^pS_{\hat{n}}\right] = \frac{c_1-c_0}{c_1^{N_C+1} - c_0^{N_C+1}}\sum_{k=0}^{N_C}c_1^kc_0^{N_C-k}(2k-N_C)\bra{k_{\hat{n}}^{(s)}}S_z^p\ket{k_{\hat{n}}^{(s)}}.
\end{equation}
We already know how to compute $\bra{k_{\hat{n}}^{(s)}}S_z^p\ket{k_{\hat{n}}^{(s)}}$ using ``up-down-middle'' paths of length $p$, as shown in Appendix \ref{appendix:understanding-P-vals}\ref{subsec:P-vals-moments-offset0}. (In fact, as we explained there, in the equatorial case where $\hat{n} = \hat{x}$, the ``middle'' option closes, so we actually just need to consider ``up-down'' paths of length $p$. This also implies that the result is trivially zero for all odd $p$.) The result will be some polynomial in $N_C$. We can then plug that polynomial into the summation, and use the usual approximations where we invert the summation index as $k\mapsto N_C-k$ and then let the new summation index $k$ go to infinity instead of stopping at $N_C$. Using standard manipulations of the infinite geometric series, these quantities can all be evaluated up to exponentially small error.

\vspace{0.5\baselineskip}

We can now compare our two different ways of evaluating $\text{Tr}\left[\rho_{N_C/2}S_z^pS_x\right]$. On one hand, we wrote it as an infinite series of decaying moments. On the other hand, we evaluated it up to exponentially small error using ``up-down'' paths. In general, to get the moments up to $\Theta(N_C^{-p})$ precision, you need to evaluate this trace for all $0\le q\le p$ and put together the results, because in any one of these results, the $\Theta(N_C^{-q})$ contributions from different moments $\mM_r^{(1)}$ will collide, so you need to solve for them by collecting multiple equations.

\vspace{0.5\baselineskip}

This is all fairly involved, but it will become clearer once we actually go through the calculations. To this end, let us begin with $p=0$. Using the first method of evaluation, the trace we are interested in is
\begin{align}
    \text{Tr}\left[\rho_{N_C/2}S_x\right] &= \sum_{w=0}^{N_C}\bra{w^{(s)}}\rho_{N_C/2}S_x\ket{w^{(s)}} \\
    &= 2\sum_{w=1}^{N_C}\sqrt{w(N_C-w+1)}\bra{(w-1)^{(s)}}\rho_{N_C/2}\ket{w^{(s)}} \\
    &= 2\sum_{w=1}^{N_C}\sqrt{w(N_C-w+1)}P_{w-1,w}.
\end{align}
We can expand the square root as
\begin{align}
    \sqrt{w(N_C-w+1)} &= \frac{N_C+1}{2}\sqrt{1-4z^2} \\
    &= \frac{N_C+1}{2}\left[1 - 2z^2 - 2z^4 - 4z^6 + \cdots\right],
\end{align}
where we used the Taylor series
\begin{equation}
    \sqrt{1+\varepsilon} = 1 + \frac{\varepsilon}{2} - \frac{\varepsilon^2}{8} + \frac{\varepsilon^3}{16} + \cdots
\end{equation}
We now define the adjusted moments as follows. They are actually the same as the centered moments in the equatorial case, but we use powers of $(N_C+1)$ in the denominator instead of powers of $N_C$:
\begin{equation}
    \mU_p^{(1)} \coloneqq (N_C+1)^{-p}\sum_{w=1}^{N_C}\left(w-\frac{N_C+1}{2}\right)^pP_{w-1,w} = \left(\frac{N_C}{N_C+1}\right)^p\mM_p^{(1)}.
\end{equation}
In fact, it is nicer to derive most quantities in this subsection as power series in $(N_C+1)$ or $(N_C+1)^{-1}$, which is why we will do so from this point forward. Only at the end will we convert our final results into power series in $N_C^{-1}$, for the sake of consistency with the presentation of all other results throughout this paper.

\vspace{0.5\baselineskip}

Using this definition, and using the Taylor series above, we can express the trace as an infinite series of moments:
\begin{align}
    \text{Tr}\left[\rho_{N_C/2}S_x\right] &= \sum_{w=0}^{N_C}\bra{w^{(s)}}\rho_{N_C/2}S_x\ket{w^{(s)}} \\
    &= (N_C+1)\left[\mU_0^{(1)} - 2\mU_2^{(1)} - 2\mU_4^{(1)} - 4\mU_6^{(1)} + \cdots\right].
\end{align}
Due to this result, we define a new quantity that will encode this infinite series of adjusted moments:
\begin{align}
    \mV_{2p}^{(1)} &\coloneqq \sum_{k=0}^{\infty}(-4)^k\binom{1/2}{k}\mU_{2(p+k)}^{(1)} \\
    &= \mU_{2p}^{(1)} - 2\mU_{2(p+1)}^{(1)} - 2\mU_{2(p+2)}^{(1)} - 4\mU_{2(p+3)}^{(1)} + \cdots.
\end{align}
In particular, the coefficient of $\mU_{2(p+k)}^{(1)}$ is always the same as the coefficient of $z^k$ in the Taylor series for $\sqrt{1-4z^2}$ centered at $z=0$.

\vspace{0.5\baselineskip}

The evaluation of $\text{Tr}\left[\rho_{N_C/2}S_x\right]$ only involves $\mV_0^{(1)}$, but as we will see later, the evaluation of $\text{Tr}\left[\rho_{N_C/2}S_z^{2p}S_x\right]$ will involve $\mV_{2q}^{(1)}$ for all $0\le q\le p$. This is why we define $\mV_{2p}^{(1)}$ for all integers $p\ge 0$.

\vspace{0.5\baselineskip}

Now we turn to the second method of evaluation, which tells us that the trace we are interested in is
\begin{equation}
    \text{Tr}\left[\rho_{N_C/2}S_x\right] = \frac{c_1-c_0}{c_1^{N_C+1} - c_0^{N_C+1}}\sum_{k=0}^{N_C}(2k-N_C)c_1^kc_0^{N_C-k}.
\end{equation}
In particular, we need to use the summation
\begin{equation}
    \sum_{k=0}^{N_C}(2k-N_C)c_1^kc_0^{N_C-k} = \frac{c_1^{N_C+1}}{\lambda^2}\left[\lambda(N_C+1) - 1 + \text{e.s.e.}\right]
\end{equation}
to eventually obtain the result
\begin{equation}
    \text{Tr}\left[\rho_{N_C/2}S_x\right] = (N_C+1) - \frac{1}{\lambda} + \text{e.s.e.}
\end{equation}

\vspace{0.5\baselineskip}

Now that we have evaluated $\text{Tr}\left[\rho_{N_C/2}S_x\right]$ in two different ways, we can equate these two results. This gives us a formula for the quantity $\mV_0^{(1)}$, which encodes information about the moments of the $P_{w-1,w}$ values:
\begin{subequations}
\begin{empheq}[box=\widefbox]{align}
    \mV_0^{(1)} &\coloneqq \mU_0^{(1)} - 2\mU_2^{(1)} - 2\mU_4^{(1)} - 4\mU_6^{(1)} + \cdots \\
    &= 1 - \frac{1}{\lambda}\frac{1}{N_C+1} + \text{e.s.e.}
\end{empheq}
\end{subequations}
It may not be immediately clear what we can get from this. However, the concentration phenomenon means that $\mU_{2p}^{(\alpha)} = \Theta((N_C+1)^{-p})$. Therefore, when we match the $\Theta(1)$ terms on both sides of the above equation, we obtain
\begin{equation}
    \boxed{\mU_0^{(1)} = 1 + O\left((N_C+1)^{-1}\right)}.
\end{equation}
In this way, we have found the $\Theta(1)$ term in the adjusted moment $\mU_0^{(1)}$. In fact, if we also match the $\Theta\left((N_C+1)^{-1}\right)$ terms on both sides, we get
\begin{equation}
    \boxed{\mU_0^{(1)} - 2\mU_2^{(1)} = 1 - \frac{1}{\lambda}\frac{1}{N_C+1} + O\left((N_C+1)^{-2}\right)}.
\end{equation}
In this sense, we also have some information about the $\Theta\left((N_C+1)^{-1}\right)$ terms in both $\mU_0^{(1)}$ and $\mU_2^{(1)}$. However, we will need a second equation to fully solve for them.

\vspace{0.5\baselineskip}

This is the part where continuing this process further gives us more information. When we proceed to $p=2$, meaning that we evaluate $\text{Tr}\left[\rho_{N_C/2}S_z^2S_x\right]$, we will obtain another equation. That equation will also have the $\Theta\left((N_C+1)^{-1}\right)$ terms of $\mU_0^{(1)}$ and $\mU_2^{(1)}$ mixed together. But now that we have two different equations relating the $\Theta\left((N_C+1)^{-1}\right)$ terms of $\mU_0^{(1)}$ and $\mU_2^{(1)}$, we can solve for both of them. Similarly, when we proceed to $p=4$, we can solve for the $\Theta\left((N_C+1)^{-2}\right)$ terms of $\mU_0^{(1)}$, $\mU_2^{(1)}$, and $\mU_4^{(2)}$. Using this process, we can eventually compute all the moments to whatever precision we desire, and Lemma \ref{lem:equatorial-offset1-moments0246} is the result.

\vspace{0.5\baselineskip}

Let us proceed to $p=2$ now. Using the first method of evaluation, the trace of interest is
\begin{align}
    \text{Tr}\left[\rho_{N_C/2}S_z^2S_x\right] &= \sum_{w=1}^{N_C}\sqrt{w(N_C-w+1)}\left[(N_C-2(w-1))^2 + (N_C-2w)^2\right]P_{w-1,w}.
\end{align}
The polynomial in the summation is
\begin{equation}
    (N_C-2(w-1))^2 + (N_C-2w)^2 = 2\left[4(N_C+1)^2z^2 + 1\right].
\end{equation}
Using this and the previously derived expression for $\sqrt{w(N_C-w+1)}$ as a Taylor series in $z$, we obtain
\begin{align}
    & \quad \sqrt{w(N_C-w+1)}\left[(N_C-2(w-1))^2 + (N_C-2w)^2\right] \\
    &= \left[4(N_C+1)^3z^2 + (N_C+1)\right]\left[1 - 2z^2 - 2z^4 - 4z^6 + \cdots\right] \\
    &= 4(N_C+1)^3\left[z^2 - 2z^4 - 2z^6 - 4z^8 + \cdots\right] \\
    & \quad + (N_C+1)\left[1 - 2z^2 - 2z^4 - 4z^6 + \cdots\right].
\end{align}
Therefore, the trace of interest is
\begin{align}
    \text{Tr}\left[\rho_{N_C/2}S_z^2S_x\right] = 4(N_C+1)^3 & \left[\mU_2^{(1)} - 2\mU_4^{(1)} - 2\mU_6^{(1)} - 4\mM_8^{(1)} + \cdots\right],
\end{align}
which can be rewritten succinctly as
\begin{equation}
    \text{Tr}\left[\rho_{N_C/2}S_z^2S_x\right] = 4(N_C+1)^3\mV_2^{(1)} + (N_C+1)\mV_0^{(1)}.
\end{equation}
As you can see, whereas $\text{Tr}\left[\rho_{N_C/2}S_z^2S_x\right]$ only involved $\mV_0^{(1)}$, $\text{Tr}\left[\rho_{N_C/2}S_z^2S_x\right]$ involves both $\mV_0^{(1)}$ and $\mV_2^{(1)}$. When we continue the pattern, we will find that $\text{Tr}\left[\rho_{N_C/2}S_z^{2p}S_x\right]$ involves $\mV_{2q}^{(1)}$ for all $0\le q\le p$. This is why it made sense earlier to define $\mV_{2p}^{(1)}$ for all $p\ge 0$.

\vspace{0.5\baselineskip}

Now we evaluate the trace of interest using the second method. We begin as follows:
\begin{align}
    \text{Tr}\left[\rho_{N_C/2}S_z^2S_x\right] &= \text{Tr}\left[S_x\rho_{N_C/2}S_z^2\right] \\
    &= \frac{c_1-c_0}{c_1^{N_C+1} - c_0^{N_C+1}}\sum_{k=0}^{N_C}c_1^kc_0^{N-k}(2k-N_C)\text{Tr}\left[\ket{k_X^{(s)}}\bra{k_X^{(s)}}S_z^2\right]
\end{align}
We can evaluate $\text{Tr}\left[\ket{k_X^{(s)}}\bra{k_X^{(s)}}S_z^2\right]$ using ``up-down'' paths of length $2$. In this case, there are just $2$ possible paths, which have the following contributions:
\begin{align}
    (U,D): &\quad\quad (k+1)(N_C-k) \\
    (D,U): &\quad\quad k(N_C-k+1)
\end{align}
Adding these quantities together yields
\begin{equation}
    \text{Tr}\left[\ket{k_X^{(s)}}\bra{k_X^{(s)}}S_z^2\right] = -2k^2 + N_C + 2kN_C.
\end{equation}
Also, as always, we will make frequent use of the approximation
\begin{equation}
    \frac{c_1-c_0}{c_1^{N_C+1} - c_0^{N_C+1}} = \frac{\lambda}{c_1^{N_C+1}}\left(1 + \text{e.s.e.}\right).
\end{equation}
Putting this all together yields
\begin{equation}
    \text{Tr}\left[\rho_{N_C/2}S_z^2S_x\right] = \frac{\lambda}{c_1}\sum_{k=0}^{\infty}\left(\frac{c_0}{c_1}\right)^k(N_C-2k)(-2k^2 + N_C + 2kN_C) + \text{e.s.e.}
\end{equation}
We now employ the same set of tricks that served us back in Appendix \ref{appendix:understanding-P-vals}\ref{subsec:P-vals-moments-offset0}. In particular, we extend the summation to $k=\infty$ (which only imposes an exponentially small error) and use the formulas for $\sum_{k=0}^{\infty}k^qx^k$ for $q=0,1,2$, which can be obtained using standard manipulations of the infinite geometric series. The end result is
\begin{equation}
    \text{Tr}\left[\rho_{N_C/2}S_z^2S_x\right] = \frac{1}{\lambda}(N_C+1)^2 + \frac{-3 + \lambda^2}{\lambda^2}(N_C+1) + \frac{3 - 2\lambda^2}{\lambda^3} + \text{e.s.e.}
\end{equation}
Therefore, equating our two different expressions for $\text{Tr}\left[\rho_{N_C/2}S_z^2S_x\right]$, and using the fact that we have already solved for $\mV_0^{(1)}$, we can solve for $\mV_2^{(1)}$:
\begin{subequations}
\begin{empheq}[box=\widefbox]{align}
    \mV_2^{(1)} &\coloneqq \mU_2^{(1)} - 2\mU_4^{(1)} - 2\mU_6^{(1)} - 4\mU_8^{(1)} + \cdots \\
    &= \frac{1}{4\lambda}\frac{1}{N_C+1} - \frac{3}{4\lambda^2}\frac{1}{(N_C+1)^2} + \frac{3 - \lambda^2}{4\lambda^3}\frac{1}{(N_C+1)^3} + \text{e.s.e.}
\end{empheq}
\end{subequations}
At this point, we now have enough information to solve for the $\Theta\left(N_C+1)^{-1}\right)$ terms in both $\mU_0^{(1)}$ and $\mU_2^{(1)}$. However, we will actually save this task for later, because as we will show at the very end, once we compute $\mV_{2q}^{(1)}$ for all $0\le q\le p$ up to exponentially small error, there is a more clever way to solve for $\mU_{2q}^{(1)}$ for all $0\le q\le p$ up to $\Theta\left((N_C+1)^{-p}\right)$ precision.

\vspace{0.5\baselineskip}

Let us proceed to $p=4$ now. Using the first method of evaluation, the trace of interest is
\begin{align}
    \text{Tr}\left[\rho_{N_C/2}S_z^4S_x\right] &= \sum_{w=1}^{N_C}\sqrt{w(N_C-w+1)}\left[(N_C-2(w-1))^4 + (N_C-2w)^4\right]P_{w-1,w}.
\end{align}
The polynomial in the summation is
\begin{equation}
    (N_C-2(w-1))^4 + (N_C-2w)^4 = 2\left[16(N_C+1)^4z^4 + 24(N_C+1)^2z^2 + 1\right].
\end{equation}
Using this and the previously derived expression for $\sqrt{w(N_C-w+1)}$ as a Taylor series in $z$, we obtain
\begin{align}
    & \quad \sqrt{w(N_C-w+1)}\left[(N_C-2(w-1))^4 + (N_C-2w)^4\right] \\
    &= \left[16(N_C+1)^5z^4 + 24(N_C+1)^3z^2 + (N_C+1)\right]\left[1 - 2z^2 - 2z^4 - 4z^6 + \cdots\right] \\
    &= 16(N_C+1)^5\left[z^4 - 2z^6 - 2z^8 - 4z^{10} + \cdots\right] \\
    & \quad + 24(N_C+1)^3\left[z^2 - 2z^4 - 2z^6 - 4z^8 + \cdots\right] \\
    & \quad + (N_C+1)\left[1 - 2z^2 - 2z^4 - 4z^6 + \cdots\right].
\end{align}
Therefore, the trace of interest is
\begin{align}
    \text{Tr}\left[\rho_{N_C/2}S_z^4S_x\right] = 16(N_C+1)^5 & \left[\mU_4^{(1)} - 2\mU_6^{(1)} - 2\mU_8^{(1)} - 4\mM_{10}^{(1)} + \cdots\right] \\
    + 24(N_C+1)^3 & \left[\mU_2^{(1)} - 2\mU_4^{(1)} - 2\mU_6^{(1)} - 4\mM_8^{(1)} + \cdots\right] \\
    + (N_C+1) & \left[\mU_0^{(1)} - 2\mU_2^{(1)} - 2\mU_4^{(1)} - 4\mM_6^{(1)} + \cdots\right],
\end{align}
which can be rewritten succinctly as
\begin{equation}
    \text{Tr}\left[\rho_{N_C/2}S_z^4S_x\right] = 16(N_C+1)^5\mV_4^{(0)} + 24(N_C+1)^3\mV_2^{(1)} + (N_C+1)\mV_0^{(1)}.
\end{equation}

\vspace{0.5\baselineskip}

Now we evaluate the trace of interest using the second method. Just as before, we begin as follows:
\begin{align}
    \text{Tr}\left[\rho_{N_C/2}S_z^4S_x\right] &= \text{Tr}\left[S_x\rho_{N_C/2}S_z^4\right] \\
    &= \frac{c_1-c_0}{c_1^{N_C+1} - c_0^{N_C+1}}\sum_{k=0}^{N_C}c_1^kc_0^{N-k}(2k-N_C)\text{Tr}\left[\ket{k_X^{(s)}}\bra{k_X^{(s)}}S_z^4\right].
\end{align}
We can evaluate $\text{Tr}\left[\ket{k_X^{(s)}}\bra{k_X^{(s)}}S_z^4\right]$ using ``up-down'' paths of length $4$. In this case, there are $\binom{4}{2} = 6$ possible paths, which have the following contributions:
\begin{align}
    (U,U,D,D): &\quad\quad [(k+1)(N_C-k)][(k+2)(N_C-k-1)] \\
    (U,D,U,D): &\quad\quad [(k+1)(N_C-k)]^2 \\
    (U,D,D,U): &\quad\quad [k(N_C-k+1)][(k+1)(N_C-k)] \\
    (D,U,U,D): &\quad\quad [k(N_C-k+1)][(k+1)(N_C-k)] \\
    (D,U,D,U): &\quad\quad [k(N_C-k+1)]^2 \\
    (D,D,U,U): &\quad\quad [(k-1)(N_C-k+2)][k(N_C-k+1)]
\end{align}
Adding these quantities together yields
\begin{equation}
    \text{Tr}\left[\ket{k_X^{(s)}}\bra{k_X^{(s)}}S_z^4\right] = 10k^2 + 6k^4 - 2N_C - 10kN_C - 6k^2N_C - 12k^3N_C + 3N_C^2 + 6kN_C^2 + 6k^2N_C^2.
\end{equation}
Combining this with the usual approximation for the prefactor yields
\begin{align}
    \text{Tr}\left[\rho_{N_C/2}S_z^4S_x\right] &= \frac{\lambda}{c_1}\sum_{k=0}^{\infty}(N_C-2k)(10k^2 + 6k^4 - 2N_C - 10kN_C  \\
    & \quad - 6k^2N_C - 12k^3N_C + 3N_C^2 + 6kN_C^2 + 6k^2N_C^2) + \text{e.s.e.}
\end{align}
By extending the summation to $k=\infty$ and invoking the formulas for $\sum_{k=0}^{\infty}k^qx^k$ for integers $0\le q\le 4$, we obtain
\begin{align}
    \text{Tr}\left[\rho_{N_C/2}S_z^4S_x\right] &= \frac{3}{\lambda^2}(N_C+1)^3 + \frac{-18 + 7\lambda^2}{\lambda^3}(N_C+1)^2 \\
    & \quad + \frac{45 - 33\lambda^2 + \lambda^4}{\lambda^4}(N_C+1) + \frac{-45 + 48\lambda^2 - 8\lambda^4}{\lambda^5} + \text{e.s.e.}
\end{align}
Therefore, equating our two different expressions for $\text{Tr}\left[\rho_{N_C/2}S_z^4S_x\right]$, and using the fact that we have already solved for $\mV_0^{(1)}$ and $\mV_2^{(1)}$, we can solve for $\mV_4^{(1)}$:
\begin{subequations}
\begin{empheq}[box=\widefbox]{align}
    \mV_4^{(1)} &\coloneqq \mU_4^{(1)} - 2\mU_6^{(1)} - 2\mU_8^{(1)} - 4\mU_{10}^{(1)} + \cdots \\
    &= \frac{3}{16\lambda^2}\frac{1}{(N_C+1)^2} + \frac{-18 + \lambda^2}{16\lambda^3}\frac{1}{(N_C+1)^3} \\
    & \quad + \frac{15(3 - \lambda^2)}{16\lambda^4}\frac{1}{(N_C+1)^4} + \frac{-45 + 30\lambda^2 - \lambda^4}{16\lambda^5}\frac{1}{(N_C+1)^5} + \text{e.s.e.}
\end{empheq}
\end{subequations}

\vspace{0.5\baselineskip}

Finally, we proceed to $p=6$. Using the first method of evaluation, the trace of interest is
\begin{align}
    \text{Tr}\left[\rho_{N_C/2}S_z^6S_x\right] &= \sum_{w=1}^{N_C}\sqrt{w(N_C-w+1)}\left[(N_C-2(w-1))^6 + (N_C-2w)^6\right]P_{w-1,w}.
\end{align}
The polynomial in the summation is
\begin{equation}
    (N_C-2(w-1))^6 + (N_C-2w)^6 = 2\left[64(N_C+1)^6z^6 + 240(N_C+1)^4z^4 + 60(N_C+1)^2z^2 + 1\right].
\end{equation}
Using this and the previously derived expression for $\sqrt{w(N_C-w+1)}$ as a Taylor series in $z$, we obtain
\begin{align}
    & \quad \sqrt{w(N_C-w+1)}\left[(N_C-2(w-1))^6 + (N_C-2w)^6\right] \\
    &= \left[64(N_C+1)^6z^6 + 240(N_C+1)^4z^4 + 60(N_C+1)^2z^2 + 1\right]\left[1 - 2z^2 - 2z^4 - 4z^6 + \cdots\right] \\
    &= 64(N_C+1)^7\left[z^6 - 2z^8 - 2z^{10} - 4z^{12} + \cdots\right] \\
    & \quad + 240(N_C+1)^5\left[z^4 - 2z^6 - 2z^8 - 4z^{10} + \cdots\right] \\
    & \quad + 60(N_C+1)^3\left[z^2 - 2z^4 - 2z^6 - 4z^8 + \cdots\right] \\
    & \quad + (N_C+1)\left[1 - 2z^2 - 2z^4 - 4z^6 + \cdots\right].
\end{align}
Therefore, the trace of interest is
\begin{align}
    \text{Tr}\left[\rho_{N_C/2}S_z^6S_x\right] = 64(N_C+1)^7 & \left[\mU_6^{(1)} - 2\mU_8^{(1)} - 2\mU_{10}^{(1)} - 4\mM_{12}^{(1)} + \cdots\right] \\
    + 240(N_C+1)^5 & \left[\mU_4^{(1)} - 2\mU_6^{(1)} - 2\mU_8^{(1)} - 4\mM_{10}^{(1)} + \cdots\right] \\
    + 60(N_C+1)^3 & \left[\mU_2^{(1)} - 2\mU_4^{(1)} - 2\mU_6^{(1)} - 4\mM_8^{(1)} + \cdots\right] \\
    + (N_C+1) & \left[\mU_0^{(1)} - 2\mU_2^{(1)} - 2\mU_4^{(1)} - 4\mM_6^{(1)} + \cdots\right],
\end{align}
which can be rewritten succinctly as
\begin{equation}
    \text{Tr}\left[\rho_{N_C/2}S_z^6S_x\right] = 64(N_C+1)^7\mV_6^{(1)} + 240(N_C+1)^5\mV_4^{(0)} + 60(N_C+1)^3\mV_2^{(1)} + (N_C+1)\mV_0^{(1)}.
\end{equation}

\vspace{0.5\baselineskip}

Now we evaluate the trace of interest using the second method. As always, we begin as follows:
\begin{align}
    \text{Tr}\left[\rho_{N_C/2}S_z^6S_x\right] &= \text{Tr}\left[S_x\rho_{N_C/2}S_z^6\right] \\
    &= \frac{c_1-c_0}{c_1^{N_C+1} - c_0^{N_C+1}}\sum_{k=0}^{N_C}c_1^kc_0^{N-k}(2k-N_C)\text{Tr}\left[\ket{k_X^{(s)}}\bra{k_X^{(s)}}S_z^6\right].
\end{align}
We can evaluate $\text{Tr}\left[\ket{k_X^{(s)}}\bra{k_X^{(s)}}S_z^6\right]$ using ``up-down'' paths of length $4$. In this case, there are $\binom{6}{3} = 20$ possible paths, which have the following contributions:
\begin{align}
    (U,U,U,D,D,D): &\quad\quad [(k+1)(N_C-k)][(k+2)(N_C-k-1)][(k+3)(N_C-k-2)] \\
    (U,U,D,U,D,D): &\quad\quad [(k+1)(N_C-k)][(k+2)(N_C-k-2)]^2 \\
    (U,D,U,U,D,D): &\quad\quad [(k+1)(N_C-k)]^2[(k+2)(N_C-k-2)] \\
    (D,U,U,U,D,D): &\quad\quad [k(N_C-k+1)][(k+1)(N_C-k)][(k+2)(N_C-k-2)] \\
    (U,U,D,D,U,D): &\quad\quad [(k+1)(N_C-k)]^2[(k+2)(N_C-k-2)] \\
    (U,D,U,D,U,D): &\quad\quad [(k+1)(N_C-k)]^3 \\
    (D,U,U,D,U,D): &\quad\quad [k(N_C-k+1)][(k+1)(N_C-k)]^2 \\
    (U,D,D,U,U,D): &\quad\quad [k(N_C-k+1)][(k+1)(N_C-k)]^2 \\
    (D,U,D,U,U,D): &\quad\quad [k(N_C-k+1)]^2[(k+1)(N_C-k)] \\
    (D,D,U,U,U,D): &\quad\quad [(k-1)(N_C-k+2)][k(N_C-k+1)][(k+1)(N_C-k)] \\
    (U,U,D,D,D,U): &\quad\quad [k(N_C-k+1)][(k+1)(N_C-k)][(k+2)(N_C-k-1)] \\
    (U,D,U,D,D,U): &\quad\quad [k(N_C-k+1)][(k+1)(N_C-k)]^2 \\
    (D,U,U,D,D,U): &\quad\quad [k(N_C-k+1)]^2[(k+1)(N_C-k)] \\
    (U,D,D,U,D,U): &\quad\quad [k(N_C-k+1)]^2[(k+1)(N_C-k)] \\
    (D,U,D,U,D,U): &\quad\quad [k(N_C-k+1)]^3 \\
    (D,D,U,U,D,U): &\quad\quad [(k-1)(N_C-k+2)][k(N_C-k+1)]^2 \\
    (U,D,D,D,U,U): &\quad\quad [(k-1)(N_C-k+2)][k(N_C-k+1)][(k+1)(N_C-k)] \\
    (D,U,D,D,U,U): &\quad\quad [(k-1)(N_C-k+2)][k(N_C-k+1)]^2 \\
    (D,D,U,D,U,U): &\quad\quad [(k-1)(N_C-k+2)]^2[k(N_C-k+1)] \\
    (D,D,D,U,U,U): &\quad\quad [(k-2)(N_C-k+3)][(k-1)(N_C-k+2)][k(N_C-k+1)]
\end{align}
Adding these quantities together yields
\begin{align}
    \text{Tr}\left[\ket{k_X^{(s)}}\bra{k_X^{(s)}}S_z^6\right] &= -112k^2 - 140k^4 - 20k^6 \\
    & \quad + 16N_C + 112kN_C + 90k^2N_C + 280k^3N_C + 30k^4N_C + 60k^5N_C \\
    & \quad - 30N_C^2 - 90kN_C^2 - 180k^2N_C^2 - 60k^3N_C^2 - 60k^4N_C^2 \\
    & \quad + 15N_C^3 + 40kN_C^3 + 30k^2N_C^3 + 20k^3N_C^3.
\end{align}
Combining this with the usual approximation for the prefactor yields
\begin{align}
    \text{Tr}\left[\rho_{N_C/2}S_z^6S_x\right] &= \frac{\lambda}{c_1}\sum_{k=0}^{\infty}(N_C-2k)(-112k^2 - 140k^4 - 20k^6 \\
    & \quad + 16N_C + 112kN_C + 90k^2N_C + 280k^3N_C + 30k^4N_C + 60k^5N_C \\
    & \quad - 30N_C^2 - 90kN_C^2 - 180k^2N_C^2 - 60k^3N_C^2 - 60k^4N_C^2 \\
    & \quad + 15N_C^3 + 40kN_C^3 + 30k^2N_C^3 + 20k^3N_C^3) + \text{e.s.e.}
\end{align}
By extending the summation to $k=\infty$ and invoking the formulas for $\sum_{k=0}^{\infty}k^qx^k$ for integers $0\le q\le 6$, we obtain
\begin{align}
    \text{Tr}\left[\rho_{N_C/2}S_z^6S_x\right] &= \frac{15}{\lambda^3}(N_C+1)^4 + \frac{30(-5 + 2\lambda^2)}{\lambda^4}(N_C+1)^3 + \frac{675 -510\lambda^2 + 31\lambda^4}{\lambda^5}(N_C+1)^2 \\
    & \quad + \frac{-1575 + 1725\lambda^2 - 333\lambda^4 + \lambda^6}{\lambda^6}(N_C+1) + \frac{1575 - 2250\lambda^2 + 768\lambda^4 - 32\lambda^6}{\lambda^7} + \text{e.s.e.}
\end{align}
Therefore, equating our two different expressions for $\text{Tr}\left[\rho_{N_C/2}S_z^6S_x\right]$ and invoking the previously derived results for $\mV_0^{(1)}$, $\mV_2^{(1)}$, and $\mV_4^{(1)}$, we can solve for $\mV_6^{(1)}$:
\begin{subequations}
\begin{empheq}[box=\widefbox]{align}
    \mV_6^{(1)} &\coloneqq \mU_6^{(1)} - 2\mU_8^{(1)} - 2\mU_{10}^{(1)} - 4\mU_{12}^{(1)} + \cdots \\
    &= \frac{15}{64\lambda^3}\frac{1}{(N_C+1)^3} + \frac{15(-10 + \lambda^2)}{64\lambda^4}\frac{1}{(N_C+1)^4} \\
    & \quad + \frac{675 - 240\lambda^2 + \lambda^4}{64\lambda^5}\frac{1}{(N_C+1)^5} + \frac{21(-75 + 50\lambda^2 - 3\lambda^4)}{64\lambda^6}\frac{1}{(N_C+1)^6} \\
    & \quad + \frac{1575 - 1575\lambda^2 + 273\lambda^4 - \lambda^6}{64\lambda^7}\frac{1}{(N_C+1)^7} + \text{e.s.e.}
\end{empheq}
\end{subequations}

\vspace{0.5\baselineskip}

Now that we have computed all the $\mV_{2p}^{(1)}$ quantities that are relevant at $\Theta\left((N_C+1)^{-3}\right)$ precision, we will explain how to convert them into the adjusted moments $\mU_{2p}^{(1)}$ up to $\Theta\left((N_C+1)^{-3}\right)$ precision. Recall that
\begin{align}
    \mV_0^{(1)} &\coloneqq \mU_0^{(1)} - 2\mU_2^{(1)} - 2\mU_4^{(1)} - 4\mU_6^{(1)} + \cdots \\
    \mV_2^{(1)} &\coloneqq \mU_2^{(1)} - 2\mU_4^{(1)} - 2\mU_6^{(1)} - 4\mU_8^{(1)} + \cdots \\
    \mV_4^{(1)} &\coloneqq \mU_4^{(1)} - 2\mU_6^{(1)} - 2\mU_8^{(1)} - 4\mU_{10}^{(1)} + \cdots \\
    \mV_6^{(1)} &\coloneqq \mU_6^{(1)} - 2\mU_8^{(1)} - 2\mU_{10}^{(1)} - 4\mU_{12}^{(1)} + \cdots
\end{align}
These equations can be inverted to solve for the adjusted moments $\mU_{2p}^{(1)}$ as follows:
\begin{align}
    \mU_0^{(1)} &= \mV_0^{(1)} + 2\mV_2^{(1)} + 6\mV_4^{(1)} + 20\mV_6^{(1)} + \cdots \\
    \mU_2^{(1)} &= \mV_2^{(1)} + 2\mV_4^{(1)} + 6\mV_6^{(1)} + 20\mV_8^{(1)} + \cdots \\
    \mU_4^{(1)} &= \mV_4^{(1)} + 2\mV_6^{(1)} + 6\mV_8^{(1)} + 20\mV_{10}^{(1)} + \cdots \\
    \mU_6^{(1)} &= \mV_6^{(1)} + 2\mV_8^{(1)} + 6\mV_{10}^{(1)} + 20\mV_{12}^{(1)} + \cdots
\end{align}
For example, to solve for $\mU_0^{(1)}$, start with $\mV_0^{(1)}$, which has the desired $\mU_0^{(1)}$ as its first term. Then add a suitably chosen multiple of $\mV_2^{(1)}$ to eliminate the presence of $\mU_2^{(1)}$. Since $\mV_0^{(1)}$ contains $-2\mU_2^{(1)}$, we need to add $2\mV_2^{(1)}$. Then add a suitably chosen multiple of $\mV_4^{(1)}$ to eliminate the presence of $\mU_4^{(1)}$. Since $\mV_0^{(1)}$ contains $-2\mU_4^{(1)}$ and $\mV_2^{(1)}$ contains $-2\mU_4^{(1)}$, we need to add $6\mV_4^{(1)}$. We can repeat this process as many times as needed to write $\mU_0^{(1)}$ as an infinite series in terms of the $\mV_{2p}^{(1)}$ quantities.

\vspace{0.5\baselineskip}

In our case, because $\mV_{2p}^{(1)} = O\left((N_C+1)^{-p}\right)$, the relevant expressions up to $\Theta\left((N_C+1)^{-3}\right)$ precision are as follows:
\begin{align}
    \mU_0^{(1)} &= \mV_0^{(1)} + 2\mV_2^{(1)} + 6\mV_4^{(1)} + 20\mV_6^{(1)} + O\left((N_C+1)^{-4}\right) \\
    \mU_2^{(1)} &= \mV_2^{(1)} + 2\mV_4^{(1)} + 6\mV_6^{(1)} O\left((N_C+1)^{-4}\right) \\
    \mU_4^{(1)} &= \mV_4^{(1)} + 2\mV_6^{(1)} + O\left((N_C+1)^{-4}\right) \\
    \mU_6^{(1)} &= \mV_6^{(1)} + O\left((N_C+1)^{-4}\right).
\end{align}
After plugging our solutions for $\mV_{2p}^{(1)}$ for $q=0,1,2,3$ into the above equations, we obtain the adjusted moments $\mU_{2p}^{(1)}$ for $p=0,1,2,3$ up to $\Theta\left((N_C+1)^{-3}\right)$ precision, exactly as desired:
\begin{equation}
    \boxed{\mU_6^{(1)} = \frac{15}{64\lambda^3}\frac{1}{(N_C+1)^3} + O\left((N_C+1)^{-4}\right)}
\end{equation}

\begin{equation}
    \boxed{\mU_4^{(1)} = \frac{3}{16\lambda^2}\frac{1}{(N_C+1)^2} + \frac{-21 + 2\lambda^2}{32\lambda^3}\frac{1}{(N_C+1)^3} + O\left((N_C+1)^{-4}\right)}
\end{equation}

\begin{equation}
    \boxed{\mU_2^{(1)} = \frac{1}{4\lambda}\frac{1}{N_C+1} - \frac{3}{8\lambda^2}\frac{1}{(N_C+1)^2} - \frac{3 + 4\lambda^2}{32\lambda^3}\frac{1}{(N_C+1)^3} + O\left((N_C+1)^{-4}\right)}
\end{equation}

\begin{equation}
    \boxed{\mU_0^{(1)} = 1 - \frac{1}{2\lambda}\frac{1}{N_C+1} - \frac{3}{8\lambda^2}\frac{1}{(N_C+1)^2} - \frac{9 + 2\lambda^2}{16\lambda^3}\frac{1}{(N_C+1)^3} + O\left((N_C+1)^{-4}\right)}.
\end{equation}

\vspace{0.5\baselineskip}

Our final task is to convert the adjusted moments $\mU_p^{(1)}$ into the desired centered moments $\mM_p^{(1)}$. (As a reminder, in the equatorial case, we have $\mU_p^{(1)} = \mM_p^{(1)} = 0$ for odd $p$, so we do not need to worry about those.) The formula relating the two is very simple:
\begin{equation}
    \mM_p^{(1)} = \left(\frac{N_C+1}{N_C}\right)^p\mU_p^{(1)}
\end{equation}
So all we need to do is to apply this formula and also rewrite the powers of $(N_C+1)$ in terms of powers of $N_C$ as needed. We begin with $p=6$:
\begin{align}
    \mM_6^{(1)} &= \left(\frac{N_C+1}{N_C}\right)^6\mU_6^{(1)} \\
    &= \frac{15}{64\lambda^3}\frac{(N_C+1)^3}{N_C^6} + O\left((N_C+1)^2N_C^{-6}\right) \\
    &= \frac{15}{64\lambda^3}\frac{1}{N_C^3} + O\left(N_C^{-4}\right).
\end{align}
We now handle $p=4$:
\begin{align}
    \mM_4^{(1)} &= \left(\frac{N_C+1}{N_C}\right)^4\mU_4^{(1)} \\
    &= \frac{3}{16\lambda^2}\frac{(N_C+1)^2}{N_C^4} + \frac{-21 + 2\lambda^2}{32\lambda^3}\frac{(N_C+1)}{N_C^4} + O\left(N_C^{-4}\right) \\
    &= \left[\frac{3}{16\lambda^2}\frac{1}{N_C^2} + \frac{3}{8\lambda^2}\frac{1}{N_C^3} + O\left(N_C^{-4}\right)\right] + \left[\frac{-21 + 2\lambda^2}{32\lambda^3}\frac{1}{N_C^3} + O\left(N_C^{-4}\right)\right] + O\left(N_C^{-4}\right) \\
    &= \frac{3}{16\lambda^2}\frac{1}{N_C^2} + \frac{-21 + 12\lambda + 2\lambda^2}{32\lambda^3}\frac{1}{N_C^3} + O\left(N_C^{-4}\right).
\end{align}
Next, we move on to $p=2$:
\begin{align}
    \mM_2^{(1)} &= \left(\frac{N_C+1}{N_C}\right)^2\mU_2^{(1)} \\
    &= \frac{1}{4\lambda}\frac{(N_C+1)}{N_C^2} - \frac{3}{8\lambda^2}\frac{1}{N_C^2} - \frac{3 + 4\lambda^2}{32\lambda^3}\frac{1}{N_C^2(N_C+1)} + O\left((N_C+1)^{-2}N_C^{-2}\right) \\
    &= \left[\frac{1}{4\lambda}\frac{1}{N_C} + \frac{1}{4\lambda}\frac{1}{N_C^2}\right] + \left[-\frac{3}{8\lambda^2}\frac{1}{N_C^2}\right] + \left[-\frac{3 + 4\lambda^2}{32\lambda^3}\frac{1}{N_C^3} + O\left(N_C^{-4}\right)\right] + O\left(N_C^{-4}\right) \\
    &= \frac{1}{4\lambda}\frac{1}{N_C} + \frac{-3 + 2\lambda}{8\lambda^2}\frac{1}{N_C^2} - \frac{3 + 4\lambda^2}{32\lambda^3}\frac{1}{N_C^3} + O\left(N_C^{-4}\right).
\end{align}
Finally, we tackle $p=0$:
\begin{align}
    \mM_0^{(1)} &= \mU_0^{(1)} \\
    &= 1 - \frac{1}{2\lambda}\frac{1}{N_C+1} + \frac{3}{8\lambda^2}\frac{1}{(N_C+1)^2} - \frac{9 + 2\lambda^2}{16\lambda^3}\frac{1}{(N_C+1)^3} + O\left((N_C+1)^{-4}\right) \\
    &= 1 - \frac{1}{2\lambda}\left[\frac{1}{N_C} - \frac{1}{N_C^2} + \frac{1}{N_C^3} + O\left(N_C^{-4}\right)\right] - \frac{3}{8\lambda^2}\left[\frac{1}{N_C^2} - \frac{2}{N_C^3} + O\left(N_C^{-4}\right)\right] \\
    & \quad - \frac{9 + 2\lambda^2}{16\lambda^3}\left[\frac{1}{N_C^3} + O\left(N_C^{-4}\right)\right] + O\left(N_C^{-4}\right) \\
    &= 1 - \frac{1}{2\lambda}\frac{1}{N_C} + \frac{-3 + 4\lambda}{8\lambda^2}\frac{1}{N_C^2} + \frac{-9 + 12\lambda - 10\lambda^2}{16\lambda^3}\frac{1}{N_C^3} + O\left(N_C^{-4}\right).
\end{align}
Let us write out our final results for these centered moments, just for cleanliness:
\begin{equation}
    \boxed{\mM_6^{(1)} = \left(\frac{15}{64\lambda^3}\right)\frac{1}{N_C^3} + O\left(N_C^{-4}\right)}
\end{equation}

\begin{equation}
    \boxed{\mM_4^{(1)} = \left(\frac{3}{16\lambda^2}\right)\frac{1}{N_C^2} + \left(\frac{-21 + 12\lambda + 2\lambda^2}{32\lambda^3}\right)\frac{1}{N_C^3} + O\left(N_C^{-4}\right)}
\end{equation}

\begin{equation}
    \boxed{\mM_2^{(1)} = \left(\frac{1}{4\lambda}\right)\frac{1}{N_C} + \left(\frac{-3 + 2\lambda}{8\lambda^2}\right)\frac{1}{N_C^2} + \left(-\frac{3 + 4\lambda^2}{32\lambda^3}\right)\frac{1}{N_C^3} + O\left(N_C^{-4}\right)}
\end{equation}

\begin{equation}
    \boxed{\mM_0^{(1)} = 1 + \left(-\frac{1}{2\lambda}\right)\frac{1}{N_C} + \left(\frac{-3 + 4\lambda}{8\lambda^2}\right)\frac{1}{N_C^2} + \left(\frac{-9 + 12\lambda - 10\lambda^2}{16\lambda^3}\right)\frac{1}{N_C^3} + O\left(N_C^{-4}\right)}.
\end{equation}
These four equations agree precisely with the formulas stated in Lemma \ref{lem:equatorial-offset1-moments0246}, so the proof is complete.

\appsubsec{Moments of the $P_{w-1,w}$ Values (Offset $\alpha=1$) in the General Case}
{subsec:P-vals-moments-offset1-general}

Our third and final main result shows the centered moments for offset $\alpha=1$ in the general case, up to $\Theta(N_C^{-2})$ precision:

\begin{lemma}[Centered moments of the $P_{w-1,w}$ values in the general case]
\label{lem:centered-moments-offset-1-general}
The centered moments for offset $\alpha=1$ take the following values up to $\Theta(N_C^{-2})$ precision:
\begin{align}
    \mM_0^{(1)} &= 1 - \frac{1}{2S_{\text{in}}^2\lambda}\frac{1}{N_C} + \left(-\frac{C_{\text{in}}^2}{2S_{\text{in}}^4} + \frac{1}{2S_{\text{in}}^2\lambda} - \frac{3}{8S_{\text{in}}^4\lambda^2}\right)\frac{1}{N_C^2} + O\left(N_C^{-3}\right) \\
    \mM_1^{(1)} &= \frac{C_{\text{in}}(1 + \lambda^2)}{4S_{\text{in}}^2\lambda^2}\frac{1}{N_C^2} + O\left(N_C^{-3}\right) \\
    \mM_2^{(1)} &= \frac{S_{\text{in}}^2}{4\lambda}\frac{1}{N_C} + \left[-\frac{1}{8} + \frac{(1-\lambda)\left(-1 + \left(C_{\text{in}}^2 - S_{\text{in}}^2\right)(2+\lambda)\right)}{8\lambda^2}\right]\frac{1}{N_C^2} + O\left(N_C^{-3}\right) \\
    \mM_3^{(1)} &= \frac{C_{\text{in}}S_{\text{in}}^2(3 - \lambda^2)}{8\lambda^2}\frac{1}{N_C^2} + O\left(N_C^{-3}\right) \\
    \mM_4^{(1)} &= \frac{3S_{\text{in}}^4}{16\lambda^2}\frac{1}{N_C^2} + O\left(N_C^{-3}\right).
\end{align}
\end{lemma}

Once again, it is worth highlighting some consistent patterns in these formulas:
\begin{itemize}
    \item The leading-order asymptotic of each even moment $\mM_{2p}^{(1)}$ is given by $\mM_{2p}^{(1)} \sim \frac{(2p-1)!!S_{\text{in}}^{2p}}{(4\lambda)^p}\frac{1}{N_C^p}$. This is a direct consequence of the fact that the ``distribution'' of the $P_{w-1,w}$ values tends to a Gaussian with normalization $\sim 1$ and variance $\sim\frac{S_{\text{in}}^2}{4\lambda N_C}$.
    \item Plugging in $C_{\text{in}} = 0$ and $S_{\text{in}} = 1$ immediately recovers the formulas in the equatorial special case as stated in Lemma \ref{lem:equatorial-offset1-moments0246}. This includes the fact that all the odd moments $\mM_{2p+1}^{(1)}$ vanish in this case.
    \item In general, each coefficient of $N_C^{-p}$ is a rational function of $\lambda$, where the denominator is a power of $2$ times $\lambda^p$, and the numerator is a polynomial in $\lambda$ of degree at most $p$ (where the coefficients can depend on $\Theta_{\text{in}}$). Furthermore, in the equatorial case $\Theta_{\text{in}} = \frac{\pi}{2}$, the degree of the numerator becomes at most $p-1$. (For example, the coefficient of $N_C^{-2}$ in $\mM_2^{(1)}$ can be written as a rational function of $\lambda$ with denominator $8\lambda^2$ and numerator $(4C_{\text{in}}^2 - 3) + 2S_{\text{in}}^2\lambda - 2C_{\text{in}}^2\lambda^2$, which is quadratic in $\lambda$. Plugging in $\Theta_{\text{in}} = \frac{\pi}{2}$ makes the numerator reduce to $-3 + 2\lambda$, which is linear in $\lambda$.)
    \item The coefficient of $N_C^{-p}$ always has $\lambda^p$ in the denominator. This is a result of the fact that the quantity $\lambda N_C \sim \lambda^2N$ is an indicator of whether you should consider yourself to be in the ``large $N$'' asymptotic regime for this problem.
\end{itemize}

\vspace{0.5\baselineskip}

Let us now explain how we derive these formulas. Deriving these formulas was truly the most painful part of this work by far, which is why we forgo the full details. However, we at least explain the general strategy, which is very similar to what we showed in Appendix \ref{appendix:understanding-P-vals}\ref{subsec:P-vals-moments-offset1-equatorial}, but with some slight complications.

\vspace{0.5\baselineskip}

In the equatorial case, we computed the trace
\begin{equation}
    \text{Tr}\left[\rho_{N_C/2}S_z^pS_x\right]
\end{equation}
in two different ways. Here, we adjust this quantity to
\begin{equation}
    \text{Tr}\left[\rho_{N_C/2}(S_z-a\Ibb)^pS_{\hat{n}}\right],
\end{equation}
where $a$ is some quantity that we will define later.

\vspace{0.5\baselineskip}

For the first evaluation method, we begin by observing how $S_{\hat{n}}$ acts on a symmetrized Hamming weight state in the computational basis:
\begin{align}
    S_{\hat{n}}\ket{w^{(s)}} &= S_{\text{in}}\sqrt{(w+1)(N_C-w)}\ket{(w+1)^{(s)}} \\
    & \quad + C_{\text{in}}(N_C-2w)\ket{w^{(s)}} \\
    & \quad + S_{\text{in}}\sqrt{w(N_C-w+1)}\ket{(w-1)^{(s)}}.
\end{align}
In the equatorial case, we did not have to worry about the middle term. Furthermore, the operator $(S_z-a\Ibb)$ acts in a predictable way on such states:
\begin{equation}
    (S_z-a\Ibb)^p\ket{w^{(s)}} = (N_C-2w-a)^p\ket{w^{(s)}}.
\end{equation}
Therefore, if we evaluate this trace by expanding it in the symmetrized computational basis, we obtain the following:
\begin{align}
    & \quad \text{Tr}\left[\rho_{N_C/2}(S_z-a\Ibb)^pS_{\hat{n}}\right] \\
    &= \sum_{w=0}^{N_C}\bra{w^{(s)}}\rho_{N_C/2}(S_z-a\Ibb)^pS_{\hat{n}}\ket{w^{(s)}} \\
    &= C_{\text{in}}\sum_{w=0}^{N_C}(N_C-2w)(N_C-2w-a)^p\bra{w^{(s)}}\rho_{N_C/2}\ket{w^{(s)}} \\
    & \quad + S_{\text{in}}\sum_{w=1}^{N_C}\sqrt{w(N_C-w+1)}\left[(N_C-2(w-1)-a)^p + (N_C-2w-a)^p\right]\bra{(w-1)^{(s)}}\rho_{N_C/2}\ket{w^{(s)}} \\
    &= C_{\text{in}}\sum_{w=0}^{N_C}\left(N_C-2w-a)^{p+1} + a(N_C-2w-a)^p\right]P_{w,w} \\
    & \quad + S_{\text{in}}\sum_{w=1}^{N_C}\sqrt{w(N_C-w+1)}\left[(N_C-2(w-1)-a)^p + (N_C-2w-a)^p\right]P_{w-1,w} \\
    &= C_{\text{in}}(-2N_C)^p\left[(-2N_C)\mM_{p+1}^{(0)} + a\mM_p^{(0)}\right] \\
    & \quad + S_{\text{in}}\sum_{w=1}^{N_C}\sqrt{w(N_C-w+1)}\left[(N_C-2(w-1)-a)^p + (N_C-2w-a)^p\right]P_{w-1,w}.
\end{align}
The first line in this final expression was not present in the equatorial case. Fortunately, it is written in terms of $\mM_p^{(0)}$ values, which we already know how to evaluate from Appendix \ref{appendix:understanding-P-vals}\ref{subsec:P-vals-moments-offset0}. Evaluating the second line in this final expression will be similar to the equatorial case. The factor $(N_C-2(w-1)-a)^p + (N_C-2w-a)^p$ is some degree-$p$ polynomial in $z$, and the factor $\sqrt{w(N_C-w+1)}$ can be written as a Taylor series in $z$. Therefore, the second term in the above quantity can be evaluated as some infinite series with higher and higher centered moments $\mM_p^{(1)}$. We can then put everything together to get a result for this trace as an infinite series of higher and higher centered moments $\mM_p^{(1)}$. In the equatorial case, both of these quantities were even functions of $z$, but now this will not be the case, so the infinite series will be considerably more complicated.

\vspace{0.5\baselineskip}

For the second evaluation method, we proceed as follows:
\begin{equation}
    \text{Tr}\left[\rho_{N_C/2}(S_z-a\Ibb)^pS_{\hat{n}}\right] = \frac{c_1-c_0}{c_1^{N_C+1} - c_0^{N_C+1}}\sum_{k=0}^{N_C}c_1^kc_0^{N_C-k}(2k-N_C)\bra{k_{\hat{n}}^{(s)}}(S_z-a\Ibb)^p\ket{k_{\hat{n}}^{(s)}}.
\end{equation}
Notice that
\begin{equation}
    \bra{k_{\hat{n}}^{(s)}}(S_z-a\Ibb)^p\ket{k_{\hat{n}}^{(s)}} = \sum_{l=0}^{p}\binom{p}{l}(-a)^{p-l}\bra{k_{\hat{n}}^{(s)}}S_z^l\ket{k_{\hat{n}}^{(s)}},
\end{equation}
where we know how to compute $\bra{k_{\hat{n}}^{(s)}}S_z^l\ket{k_{\hat{n}}^{(s)}}$ using ``up-down-middle'' paths of length $l$, as shown in Appendix \ref{appendix:understanding-P-vals}\ref{subsec:P-vals-moments-offset0}. In this case, we need to consider ``up-down-middle'' paths of every length $0\le l\le p$. The result will be some polynomial in $N_C$. We can then plug that polynomial into the summation, and use the usual approximations where we invert the summation index as $k\mapsto N_C-k$ and then let the new summation index $k$ go to infinity instead of stopping at $N_C$. These quantities can all be evaluated using standard manipulations of the infinite geometric series. The conclusion is that we know how to evaluate $\text{Tr}\left[\rho_{N_C/2}(S_z-a\Ibb)^pS_{\hat{n}}\right]$ up to exponentially small error.

\vspace{0.5\baselineskip}

Since we need to solve for a greater number of quantities, we actually need to perform both evaluation methods above for two different values of $a$ for each nonnegative integer $p$. It is now finally time to make smart choices for the value of $a$. The first convenient choice is $a=0$, because it makes the calculation using ``up-down-middle'' paths the easiest in the second evaluation method. In particular, we only need to consider ``up-down-middle'' paths of length $p$, as opposed to having to consider them for all lengths $0\le l\le p$. The second convenient choice of is $a = N_C(1-2\mu)$. The reason is that the degree-$p$ polynomial in $z$ that appears in the first evaluation method now simplifies as follows:
\begin{equation}
    (N_C-2(w-1)-a)^p + (N_C-2w-a)^p = [-2(N_C+1)z+1]^p + [-2(N_C+1)z-1]^p.
\end{equation}
In particular, it will now only have terms of degree matching the parity of $p$, thus simplifying the calculation somewhat.

\vspace{0.5\baselineskip}

We can now compare our two different ways of evaluating $\text{Tr}\left[\rho_{N_C/2}(S_z-a\Ibb)^pS_{\hat{n}}\right]$. On one hand, we wrote it as an infinite series of decaying moments. On the other hand, we evaluated it up to exponentially small error using ``up-down-middle'' paths. In general, to get the moments up to $\Theta(N_C^{-p})$ precision, you need to evaluate this trace for all $0\le q\le p$ and put together the results, because in any one of these results, the $\Theta(N_C^{-q})$ contributions from different moments $\mM_r^{(1)}$ will collide, so you need to solve for them by collecting multiple equations. However, once all these calculations are done, we can (with great difficulty) extract the moments $\mM_p^{(1)}$ from them, and Lemma \ref{lem:centered-moments-offset-1-general} is the result.

%% file: app_PH_dissipation.tex
\appsec{Dissipation of Purity of Coherence}
{appendix:PH-dissipation}

We know that purity of coherence is a resource that cannot be created under TI channels, only destroyed. As such, when distilling coherence, we should closely watch how much of this valuable resource we waste. This raises a natural question: how much purity of coherence is ``dissipated'' by our optimal distillation protocol?

\vspace{0.5\baselineskip}

Given that the bound on fidelity imposed by monotonicity of purity of coherence is tight at the leading order, one might intuitively expect that purity of coherence is essentially conserved. This intuition turns out to be correct, but we will compute the loss of purity of coherence to slightly higher precision to allow us to observe at least some wastage. For simplicity, we will restrict this discussion to the equatorial special case $\Theta_{\text{in}} = \Theta_{\text{out}} = \frac{\pi}{2}$.

\vspace{0.5\baselineskip}

There is a very important subtlety that we must observe. Because purity of coherence and fidelity are both monotonically increasing functions of the purity parameter, maximizing the output qubit fidelity indeed maximizes the purity of coherence of the output state. However, while fidelity is a linear function of the purity parameter, purity of coherence is a nonlinear function of the purity parameter. As a result, the expected value of the purity of coherence is not just the purity of coherence of the average output state. Therefore, the average purity of coherence of the output state is a function of not only the quantum channel but also the specific realization of that quantum channel. In particular, post-selection plays a key role. If our output state is a weighted average of output states resulting from different measurement outcomes, then we retain more purity of coherence if we remember the measurement outcome, since purity of coherence is convex.

\vspace{0.5\baselineskip}

For the realization of our equatorial distillation protocol, we will consider two possibilities: one in which we keep track of the value $j$ we obtained in the total angular momentum measurement as part of the initial Schur sampling step \cite{Cirac1999}, and one in which we forget this value. We begin with the naive implementation, in which we forget this value:

\begin{theorem}[Dissipation of purity of coherence, assuming naive implementation]
\label{thm:PH-dissipation-naive}
If one forgets the outcome of the angular momentum measurement, then the optimal $N\rightarrow 1$ equatorial distillation protocol produces an average dissipation of purity of coherence of $6 + O\left(N^{-1}\right)$.
\end{theorem}

\begin{proof}

For convenience, we will define the quantity
\begin{equation}
    \eta \equiv P_H(\rho) = \frac{4\lambda^2}{1-\lambda^2}.
\end{equation}
The purity of coherence of the input state is
\begin{equation}
    P_{H_{\text{in}}}\left(\rho^{\otimes N}\right) = \frac{4\lambda^2}{1-\lambda^2}N = \eta N.
\end{equation}
We now compute the purity of coherence of the output state. We can naively evaluate it by taking the optimal purity parameter $\tilde{\lambda}$ that we expressed as a power series in $1/N$ and then passing that value through the purity-of-coherence function $f(x) = \frac{4x^2}{1-x^2}$:
\begin{align}
    \tilde{\lambda}_{\text{opt}} &= 1 - \frac{2}{\eta}\frac{1}{N} - \frac{6}{\eta^2}\frac{1}{N^2} + O\left(N^{-3}\right) \\
    \implies \tilde{\lambda}_{\text{opt}}^2 &= 1 - \frac{4}{\eta}\frac{1}{N} - \frac{8}{\eta^2}\frac{1}{N^2} + O\left(N^{-3}\right) \\
    \therefore P_H\left(\rho\left(\tilde{\lambda}_{\text{opt}}\right)\right) &= \frac{4\tilde{\lambda}_{\text{opt}}^2}{1 - \tilde{\lambda}_{\text{opt}}^2} \\
    &= \frac{4\left[1 - \frac{4}{\eta}\frac{1}{N} - \frac{8}{\eta^2}\frac{1}{N^2} + O\left(N^{-3}\right)\right]}{\frac{4}{\eta}\frac{1}{N} + \frac{8}{\eta^2}\frac{1}{N^2} + O\left(N^{-3}\right)} \\
    &= \eta N\left[1 - \frac{4}{\eta}\frac{1}{N} - \frac{8}{\eta^2}\frac{1}{N^2} + O\left(N^{-3}\right)\right]\left[1 - \frac{2}{\eta}\frac{1}{N} + O\left(N^{-2}\right)\right] \\
    &= \eta N\left[1 - \frac{6}{\eta}\frac{1}{N} + O\left(N^{-2}\right)\right] \\
    &= \frac{4\lambda^2}{1-\lambda^2}N - 6 + O\left(N^{-1}\right).
\end{align}
The first term is in fact just the purity of coherence of the input state, which tells us that the optimal distillation protocol conserves purity of coherence at the leading order, as expected. Subtracting the above amount from the original purity of coherence, we see that the dissipated purity of coherence is $6 + O\left(N^{-1}\right)$, as desired.

\end{proof}

The fact that there is nonzero purity of coherence wastage at the $\Theta(1)$ order (which is the order immediately following the leading order), is directly tied to a fact that we previously observed, namely that the minimum infidelity fails to saturate the bound set by purity of coherence at the $\Theta(N^{-2})$ order (which is again the order immediately following the leading order), as we showed in Appendix \ref{appendix:2nd-order-optimality}\ref{subsec:2nd-order-optimality-special-case-equatorial}.

\vspace{0.5\baselineskip}

However, one can object to this naive calculation. The optimal purity parameter $\tilde{\lambda}_{\text{opt}}$, when written in terms of $N$, is in fact an average over the different angular momentum outcomes. In other words, we average $\tilde{\lambda}_{\text{opt}}(N_C)$ over different outcomes $N_C = 2j$ to obtain $\tilde{\lambda}_{\text{opt}}$, and then pass this value through the purity-of-coherence function $f(x) = \frac{4x^2}{1-x^2}$. However, we always know what the angular momentum measurement outcome was, so when computing the average purity of coherence, we should not blindly average together the purity parameters resulting from the protocol following the angular momentum measurement. In particular, $f$ is a convex function in $x$ (this is a special case of the broader fact that purity of coherence is a convex resource). So by Jensen's inequality, $\mathbb{E}[f(X)] \ge f\left(\mathbb{E}[X]\right)$. This means that we should actually take our values $\tilde{\lambda}_{\text{opt}}(N_C)$ that are written as a power series in $1/N_C$, pass these values through the purity-of-coherence function $f(x) = \frac{4x^2}{1-x^2}$, and then average those values together to get the average purity of coherence of the distilled output qubit. Assuming that we remember our angular momentum outcome (which of course we would), this is a more accurate reflection of the average purity of coherence that our distilled qubit contains.

\begin{theorem}[Dissipation of purity of coherence, assuming angular momentum post-selection]
\label{thm:PH-dissipation-post-cirac}
If one keeps track of the outcome of the angular momentum measurement but implements the post-Schur-sampling distillation protocol naively, then the optimal $N\rightarrow 1$ equatorial distillation protocol produces an average dissipation of purity of coherence of $2 + O\left(N^{-1}\right)$.
\end{theorem}

\begin{proof}

We start with the power series for $\tilde{\lambda}_{\text{opt}}(N_C)$, the optimal purity parameter for a specific angular momentum outcome, which we derived in a previous section:
\begin{align}
    \tilde{\lambda}_{\text{opt}}(N_C) &= 1 - \frac{1-\lambda^2}{2\lambda}\frac{1}{N_C} - \frac{(1-\lambda)^2(1+\lambda)(3-\lambda)}{8\lambda^2}\frac{1}{N_C^2} + O\left(N_C^{-3}\right) \\
    \implies \tilde{\lambda}_{\text{opt}}(N_C)^2 &= 1 - \frac{1-\lambda^2}{\lambda}\frac{1}{N_C} - \frac{(1-\lambda)^3(1+\lambda)}{2\lambda^2}\frac{1}{N_C^2} + O\left(N_C^{-3}\right).
\end{align}
We now compute $P(N_C)$, the purity of coherence for a specific total angular momentum outcome by passing this value through the purity-of-coherence function:
\begin{align}
    P(N_C) &= f\left(\tilde{\lambda}_{\text{opt}}(N_C)\right) \\
    &= \frac{4\tilde{\lambda}_{\text{opt}}(N_C)^2}{1 - \tilde{\lambda}_{\text{opt}}(N_C)^2} \\
    &= \frac{4\left[1 - \frac{1-\lambda^2}{\lambda}\frac{1}{N_C} - \frac{(1-\lambda)^3(1+\lambda)}{2\lambda^2}\frac{1}{N_C^2} + O\left(N_C^{-3}\right)\right]}{\frac{1-\lambda^2}{\lambda}\frac{1}{N_C} + \frac{(1-\lambda)^3(1+\lambda)}{2\lambda^2}\frac{1}{N_C^2} + O\left(N_C^{-3}\right)} \\
    &= \frac{4\lambda}{1-\lambda^2}N_C\left[1 - \frac{1-\lambda^2}{\lambda}\frac{1}{N_C} - \frac{(1-\lambda)^3(1+\lambda)}{2\lambda^2}\frac{1}{N_C^2} + O\left(N_C^{-3}\right)\right]\left[1 - \frac{(1-\lambda)^2}{2\lambda}\frac{1}{N_C} + O\left(N_C^{-2}\right)\right] \\
    &= \frac{4\lambda}{1-\lambda^2}N_C - \frac{2(3+\lambda)}{1+\lambda} + O\left(N_C^{-1}\right).
\end{align}
And only now do we take the average over different angular momentum outcomes. This requires us to use the formula for $\Ebb[N_C]$, which we derive up to exponentially small error (abbreviated ``e.s.e.'' for convenience) in Appendix \ref{appendix:angular-momentum-moments}. The resulting average purity of coherence is
\begin{align}
    \mathbb{E}\left[P(N_C)\right] &= \frac{4\lambda}{1-\lambda^2}\mathbb{E}\left[N_C\right] - \frac{2(3+\lambda)}{1+\lambda} + \mathbb{E}\left[O\left(N_C^{-1}\right)\right] \\
    &= \frac{4\lambda}{1-\lambda^2}\left[\lambda N + \frac{1-\lambda}{\lambda} + \text{e.s.e.}\right] - \frac{2(3+\lambda)}{1+\lambda} + O\left(N^{-1}\right) \\
    &= \frac{4\lambda^2}{1-\lambda^2}N - 2 + O\left(N^{-1}\right).
\end{align}
Just as before, the first term is the purity of coherence of the input state. Subtracting the above quantity from the initial purity of coherence yields a dissipation of $2 + O\left(N^{-1}\right)$, as desired.

\end{proof}

As an aside, we find the fact that the leading-order dissipation of purity of coherence has no dependence on $\lambda$ or $N$ and is just a plain number in both cases ($6$ for naive implementation, and $2$ for angular momentum post-selection) oddly amusing.

\vspace{0.5\baselineskip}

The dissipated purity of coherence dropping from $6 + O\left(N^{-1}\right)$ to $2 + O\left(N^{-1}\right)$ shows that there is some advantage, however slight, of keeping track of the angular momentum measurement outcome along with the single purified coherent qubit, rather than naively forgetting that result. In general, for any convex resource such as purity of coherence, the power of post-selection is that it preserves more of the resource. Naively averaging together states from different measurement outcomes wastes the resource. Another useful intuition is that it is better to optimize distillation procedures for each individual measurement outcome, rather than optimizing just one distillation procedure for the average.

\vspace{0.5\baselineskip}

But this naturally raises the question: why stop at only one ``layer'' of post-selection? Right now, we have the protocol (following Schur sampling) written abstractly in terms of Kraus operators. However, if we could find a more physical interpretation of the post-Schur-sampling protocol, where the last step would involve averaging over different outcomes of some additional measurement, then we could take advantage of post-selection on those measurement outcomes too and thereby dissipate even less purity of coherence. It would be especially interesting if there were a way to use this additional post-selection to completely eliminate the dominant term (that is, the $\Theta(1)$ term) in the purity of coherence wastage, and thereby produce only $O\left(N^{-1}\right)$ wastage.

%% file: app_entanglement_breaking.tex
\appsec{Entanglement-Breaking Distillation Protocols}
{appendix:entanglement-breaking}

One interesting question we might ask is the following: how well can we do if we restrict ourselves to only entanglement-breaking (aka measure-and-prepare) protocols? At the leading order (that is, minimizing the $\Theta\left(N^{-1}\right)$ term of the infidelity), this question was already answered by \cite{Marvian2020}, which describes a TI measure-and-prepare channel $\mE_{\text{TI-MP}}$ whose infidelity factor is $\delta_1\left(\mE_{\text{TI-MP}}\right) = \frac{1}{4\lambda^2}$. This is known to be the smallest possible infidelity factor for a single-shot entanglement-breaking distillation protocol, due to a bound that states that the RLD Fisher information of the output for an entanglement-breaking protocol will be upper-bounded by the SLD Fisher information of the input \cite{Marvian2020}. However, we can analyze this question at higher orders as well. For simplicity, we restrict this discussion to the equatorial case $\Theta_{\text{in}} = \Theta_{\text{out}} = \frac{\pi}{2}$.

\appsubsec{Characterizing Entanglement-Breaking Permutation-Invariant Distillation Protocols}
{subsec:characterizing-eb-protocols}

Let us first characterize which single-shot protocols are entanglement-breaking. We will still restrict our attention to permutation-invariant protocols, so that we can focus on the $N_C$-qubit symmetric subspace and think about the protocol as something that we do after performing Schur sampling and obtaining an $N_C$-qubit fully symmetric state. We start by writing the Choi matrix on the $N_C$-qubit symmetric subspace for a TI distillation protocol, which we derived in Appendix \ref{appendix:deriving-kraus-rep}\ref{subsec:distillation-protocol-form-proof-choi-matrices}:

\begin{equation}
    J(\mE) = \begin{array}{|cc|cc|cc|cc|cc|cc|}
        \hline
        \cos^2\theta_0 & & & A_1 & & & & & & & & \\
        & \sin^2\theta_0 & & & & & & & & & & \\
        \hline
        & & \cos^2\theta_1 & & & A_2 & & & & & & \\
        \overline{A}_1 & & & \sin^2\theta_1 & & & & & & & & \\
        \hline
        & & & & \cos^2\theta_2 & & & \ddots & & & & \\
        & & \overline{A}_2 & & & \sin^2\theta_2 & & & & & & \\
        \hline
        & & & & & & \ddots & & & A_{N_C-1} & & \\
        & & & & \ddots & & & \ddots & & & & \\
        \hline
        & & & & & & & & \cos^2\theta_{N_C-1} & & & A_{N_C} \\
        & & & & & & \overline{A}_{N_C-1} & & & \sin^2\theta_{N_C-1} & & \\
        \hline
        & & & & & & & & & & \cos^2\theta_{N_C} & \\
        & & & & & & & & \overline{A}_{N_C} & & & \sin^2\theta_{N_C} \\
        \hline
    \end{array}.
\end{equation}

As a brief aside, it is worth mentioning that this channel does NOT necessarily have the Kraus representation stated in Theorem \ref{thm:optimal-protocol-kraus-rep}. The reason is that, because we now want to ensure that the channel is entanglement-breaking, we must sacrifice some of the ``obvious'' optimizations we made from the above form. In general, as we showed in Appendix \ref{appendix:deriving-kraus-rep}\ref{subsec:distillation-protocol-form-proof-kraus-operators}, a TI channel $\mE$ mapping the $N_C$ qubits in the symmetric subspace to $1$ qubit will necessarily have Kraus operators
\begin{equation}
    K_{(w,\alpha)} = c(w-1,\alpha)\ket{0}\bra{(w-1)^{(s)}} + d(w,\alpha)\ket{1}\bra{w^{(s)}} \quad (0\le w\le N_C+1)
\end{equation}
for some complex numbers $c(w,\alpha)$ and $d(w,\alpha)$, and where we set $\ket{(-1)^{(s)}} = \ket{(N_C+1)^{(s)}} = 0$ by convention. However, the multiplicity index $\alpha$ means that this channel does not match the form stated in Theorem \ref{thm:optimal-protocol-kraus-rep}.

\vspace{0.5\baselineskip}

Let us return to studying the Choi matrix above. This Choi matrix is block-diagonal with $2$ blocks of size $1\times 1$ and $N_C$ blocks of size $2\times 2$. From the requirement that the Choi matrix be positive semidefinite (PSD) to represent a completely positive (CP) map, we derive the conditions
\begin{equation}
    \abs{A_w} \le \cos\theta_{w-1}\sin\theta_w \quad 1\le w\le N_C.
\end{equation}

\vspace{0.5\baselineskip}

It is known that a Choi matrix represents an entanglement-breaking channel, also known as a super-positive (SP) map, if and only if it corresponds to a separable state (up to normalization). Furthermore, a necessary condition for this to hold is the positive partial transpose (PPT) condition, namely that the partial transpose of the Choi matrix is also PSD. The partial transpose of the above Choi matrix can be computed easily by taking the transpose of each block (this corresponds to the partial transpose on the output qubit):

\begin{equation}
    J(\mE)^{T_{\text{out}}} = \begin{array}{|cc|cc|cc|cc|cc|cc|}
        \hline
        \cos^2\theta_0 & & & & & & & & & & & \\
        & \sin^2\theta_0 & A_1 & & & & & & & & & \\
        \hline
        & \overline{A}_1 & \cos^2\theta_1 & & & & & & & & & \\
        & & & \sin^2\theta_1 & A_2 & & & & & & & \\
        \hline
        & & & \overline{A}_2 & \cos^2\theta_2 & & & \ddots & & & & \\
        & & & & & \sin^2\theta_2 & & & & & & \\
        \hline
        & & & & & & \ddots & & & & & \\
        & & & & \ddots & & & \ddots & A_{N_C-1} & & & \\
        \hline
        & & & & & & & \overline{A}_{N_C-1} & \cos^2\theta_{N_C-1} & & & \\
        & & & & & & & & & \sin^2\theta_{N_C-1} & A_{N_C} & \\
        \hline
        & & & & & & & & & \overline{A}_{N_C} & \cos^2\theta_{N_C} & \\
        & & & & & & & & & & & \sin^2\theta_{N_C} \\
        \hline
    \end{array}.
\end{equation}
Amazingly, this matrix is also block-diagonal with $2$ blocks of size $1\times 1$ and $N_C$ blocks of size $2\times 2$, but the blocks are a bit different. This matrix is PSD if and only if
\begin{equation}
    \abs{A_w} \le \sin\theta_{w-1}\cos\theta_w \quad (1\le w\le N_C).
\end{equation}
Therefore, for us to have a valid protocol that is also entanglement-breaking, we must have
\begin{equation}
    \abs{A_w} \le \text{min}\{\cos\theta_{w-1}\sin\theta_w, \sin\theta_{w-1}\cos\theta_w\} \quad (1\le w\le N_C).
\end{equation}
Of course, if we want the best possible performance, we should have $A_w$ be nonnegative real and saturate this bound:
\begin{equation}
    A_w = \text{min}\{\cos\theta_{w-1}\sin\theta_w, \sin\theta_{w-1}\cos\theta_w\} \quad (1\le w\le N_C).
\end{equation}

\vspace{0.5\baselineskip}

So now we can ask: what is the best possible entanglement-breaking channel? Recall that the output purity parameter $\tilde{\lambda} = 1 - 2\mI(\mE)$ is given by
\begin{equation}
    \tilde{\lambda} = 2\sum_{w=1}^{N_C}P_{w-1,w}A_w.
\end{equation}
However, we can easily construct the upper bound
\begin{align}
    A_w &= \text{min}\{\cos\theta_{w-1}\sin\theta_w, \sin\theta_{w-1}\cos\theta_w\} \\
    &\le \sqrt{\left(\cos\theta_{w-1}\sin\theta_w\right)\left(\sin\theta_{w-1}\cos\theta_w\right)} \\
    &= \frac{1}{2}\sqrt{\sin\left(2\theta_{w-1}\right)\sin\left(2\theta_w\right)} \\
    &\le \frac{1}{2}.
\end{align}
Clearly this bound is saturated if and only if $\theta_w = \frac{\pi}{4}$, such that $\cos^2\theta_w = \sin^2\theta_w = \frac{1}{2}$, for all $0\le w\le N_C$. Hence we immediately obtain the optimal time-translation invariant, permutation-invariant, measure-and-prepare channel:
\begin{equation}
    \sin^2\theta_w = \frac{1}{2} \quad (0\le w\le N_C).
\end{equation}
In general, the PPT condition is merely a necessary condition for a channel to be entanglement-breaking, not a sufficient one. However, in Appendix \ref{appendix:entanglement-breaking}\ref{subsec:explicit-separability}, we close this gap in this specific case by showing a manifestly separable expression for the above Choi matrix. As a result, the protocol given by $\sin^2\theta_w = \frac{1}{2}$ for all $0\le w\le N_C$ is truly the optimal entanglement-breaking protocol in the equatorial case.

\vspace{0.5\baselineskip}

Interestingly, if we consider the optimal first-order protocol in the asymptotic regime, which we derived in Appendix \ref{appendix:1st-order-optimality}, and we take the $\lambda\rightarrow 0$ limit of that protocol, we actually get $\sin^2\theta_w = \frac{1-\cos\Theta_{\text{out}}}{2}$ for all $0\le w\le N_C$. In the equatorial case, this corresponds precisely to the protocol we derived above. Intuitively, as $\lambda\rightarrow 0$, the protocol becomes closer to entanglement-breaking, which agrees with the observation made by Marvian that the minimum infidelity factor of an entanglement-breaking protocol $\delta_1(\mE_{\text{opt},EB}) = \frac{1}{4\lambda^2}$ gets closer and closer to saturating the bound for a more general protocol $\delta_1(\mE_{\text{opt}}) = \frac{1-\lambda^2}{4\lambda^2}$ in the high-noise limit $\lambda\rightarrow 0$ \cite{Marvian2020}.

\vspace{0.5\baselineskip}

From the above protocol, we obtain the following output purity parameter:
\begin{equation}
    \tilde{\lambda} = \sum_{w=1}^{N_C}P_{w-1,w}.
\end{equation}
This is precisely the ``zeroth moment'' of the $P_{w-1,w}$ values, as computed in Appendix \ref{appendix:understanding-P-vals}\ref{subsec:P-vals-moments-offset1-equatorial}. Hence this zeroth moment is more than just a mathematically useful quantity that we use to compute the optimal distillation protocol. It has an operational meaning as the optimal purity parameter of a measure-and-prepare distillation protocol. In particular, in Appendix \ref{appendix:understanding-P-vals}, we state that
\begin{equation}
    \mM_0^{(1)} = 1 + \left(-\frac{1}{2\lambda}\right)\frac{1}{N_C} + \left(\frac{-3 + 4\lambda}{8\lambda^2}\right)\frac{1}{N_C^2} + \left(\frac{-9 + 12\lambda - 10\lambda^2}{16\lambda^3}\right)\frac{1}{N_C^3} + O\left(N_C^{-4}\right).
\end{equation}
Therefore, the minimal infidelity of an entanglement-breaking protocol in the equatorial case is
\begin{equation}
    \mI(\mE_{\text{opt,EB}}) = \frac{1-\mM_0^{(1)}}{2} = \left(\frac{1}{4\lambda}\right)\frac{1}{N_C} + \left(\frac{3 - 4\lambda}{16\lambda^2}\right)\frac{1}{N_C^2} + \left(\frac{9 - 12\lambda + 10\lambda^2}{32\lambda^3}\right)\frac{1}{N_C^3} + O\left(N_C^{-4}\right).
\end{equation}
We now use the expressions for $\Ebb\left[N_C^{-p}\right]$ for $p=1,2,3$, which we take from Lemma \ref{lem:NC-negative-moments}:
\begin{align}
    \mathbb{E}\left[\frac{1}{N_C}\right] &= \left(\frac{1}{\lambda}\right)\frac{1}{N} + \left(\frac{1 - \lambda}{\lambda^2}\right)\frac{1}{N^2} + \left[\frac{(1 - \lambda)(1 + 2\lambda - \lambda^2)}{\lambda^4}\right]\frac{1}{N^3} + O\left(N^{-4}\right) \\
    \mathbb{E}\left[\frac{1}{N_C^2}\right] &= \left(\frac{1}{\lambda^2}\right)\frac{1}{N^2} + \left[\frac{(1 - \lambda)(1 + 3\lambda)}{\lambda^4}\right]\frac{1}{N^3} + O\left(N^{-4}\right) \\
    \mathbb{E}\left[\frac{1}{N_C^3}\right] &= \left(\frac{1}{\lambda^3}\right)\frac{1}{N^3} + O\left(N^{-4}\right).
\end{align}
Using the above formulas, we can convert the aforementioned power series in $N_C^{-1}$ into a power series in $N^{-1}$:
\begin{align}
    \mI(\mE_{\text{opt,EB}}) &= \left(\frac{1}{4\lambda}\right)\Bigg\{\left(\frac{1}{\lambda}\right)\frac{1}{N} + \left(\frac{1 - \lambda}{\lambda^2}\right)\frac{1}{N^2} + \left[\frac{(1 - \lambda)(1 + 2\lambda - \lambda^2)}{\lambda^4}\right]\frac{1}{N^3}\Bigg\} \\
    & \quad\quad + \left(\frac{3 - 4\lambda}{16\lambda^2}\right)\Bigg\{\left(\frac{1}{\lambda^2}\right)\frac{1}{N^2} + \left[\frac{(1 - \lambda)(1 + 3\lambda)}{\lambda^4}\right]\frac{1}{N^3}\Bigg\} \\
    & \quad\quad + \left(\frac{9 - 12\lambda + 10\lambda^2}{32\lambda^3}\right)\Bigg\{\left(\frac{1}{\lambda^3}\right)\frac{1}{N^3}\Bigg\} + O\left(N^{-4}\right) \\
    &= \left(\frac{1}{4\lambda^2}\right)\frac{1}{N} + \left(\frac{3 - 4\lambda^2}{16\lambda^4}\right)\frac{1}{N^2} + \left(\frac{15 - 16\lambda^2 + 8\lambda^4}{32\lambda^6}\right)\frac{1}{N^3} + O(N^{-4}).
\end{align}

\vspace{0.5\baselineskip}

Another interesting question we can ask is the following: suppose we choose the true optimal $A_w$, disregarding the entanglement-breaking condition:
\begin{equation}
    A_w = \cos\theta_{w-1}\sin\theta_w \quad (1\le w\le N).
\end{equation}
In that case, what conditions must the $\theta_w$ values satisfy for the protocol to be entanglement-breaking?

\vspace{0.5\baselineskip}

It turns out that the PPT condition provides an amazingly simple necessary condition. In particular, if we set $A_w = \cos\theta_{w-1}\sin\theta_w$, then we must have
\begin{align}
    \cos\theta_{w-1}\sin\theta_w &\le \sin\theta_{w-1}\cos\theta_w \quad 1\le w\le N \\
    \iff \tan\theta_w &\le \tan\theta_{w-1} \quad 1\le w\le N \\
    \iff \theta_{w-1} &\le \theta_w \quad 1\le w\le N.
\end{align}
This is equivalent to saying that the $\theta_w$ values are decreasing. Hence we reach the following necessary condition: \textbf{if a permutation-invariant TI channel that satisfies all ``obvious'' optimizations (namely that the $A_w$ values saturate the bounds needed for the channel to be completely positive) is entanglement-breaking, then the $\theta_w$ values must be decreasing.} (We suspect that this condition is also sufficient, although we do not yet have a proof.)

\vspace{0.5\baselineskip}

For now, we have assumed that the $\theta_w$ values can be anything in $\left[0,\frac{\pi}{2}\right]$, and we have just maximized the off-diagonal elements $A_w$ conditioned on the $\theta_w$ values. However, the last ``obvious'' optimization (before we are forced to do the complicated boundary value problem stuff) is $\theta_0 = 0$ and $\theta_N = \frac{\pi}{2}$. But this immediately implies that the $\theta_w$ values cannot be decreasing for the optimal distillation protocol! We thus achieve the following corollary: \textbf{an absolutely optimal single-shot qubit coherence distillation protocol is NEVER entanglement-breaking.}

\appsubsec{Explicit Demonstration of Separability of Choi Matrix}
{subsec:explicit-separability}

In general, it is very hard to discern whether a quantum channel is entanglement-breaking just from the Choi matrix. However, for the protocol we derived in the previous subsection, which is optimal among all those satisfying the PPT condition, we can write out a manifestly separable formula for the Choi matrix. In particular, we will write $J(\mE)$ as a convex combination of pure states that are tensor products of pure states on the $(N_C+1)$-dimensional input space and pure states on the $2$-dimensional output space.

\vspace{0.5\baselineskip}

For convenience, define the following families of states on the input and output Hilbert spaces, respectively:
\begin{align}
    \ket{\Psi(\phi)} &\equiv \frac{1}{\sqrt{N_C+1}}\sum_{w=0}^{N_C}e^{i\phi w}\ket{w^{(s)}} \\
    \ket{\eta(\phi)} &\equiv \frac{\ket{0} + e^{i\phi}\ket{1}}{\sqrt{2}}.
\end{align}
Notice that $\ket{\Psi(\phi)}$ is the \textbf{forward} time evolution of $\ket{\Psi(0)}$ under the Hamiltonian $H = Z_1 + \cdots + Z_{N_C}$, and $\ket{\eta(\phi)}$ is similarly the \textbf{forward} time evolution of $\ket{\eta(0)}$ under the Hamiltonian $H = Z$. However, let us now consider $\ket{\eta(-\phi)}$, which is the \textbf{backward} time evolution of $\ket{\eta(0)}$ under the Hamiltonian $H = Z$. The tensor product of $\ket{\Psi(\phi)}$ and $\ket{\eta(-\phi)}$ looks as follows:
\begin{align}
    & \quad \ket{\Psi(\phi)}\bra{\Psi(\phi)}\otimes\ket{\eta(-\phi)}\bra{\eta(-\phi)} \\
    &= \frac{1}{2(N_C+1)}\sum_{w=0}^{N_C}\sum_{w'=0}^{N_C}\sum_{b=0}^{1}\sum_{b'=0}^{1}e^{i\phi\left[(w-w') - (b-b')\right]}\ket{w^{(s)}}\bra{(w')^{(s)}}\otimes\ket{b}\bra{b'}.
\end{align}
Now we integrate the above state over the unit circle. Using the fact that $\int_{0}^{2\pi}e^{i\phi n}d\phi = 2\pi\delta_{n,0}$, we conclude that
\begin{equation}
    \int_{0}^{2\pi}\ket{\Psi(\phi)}\bra{\Psi(\phi)}\otimes\ket{\eta(-\phi)}\bra{\eta(-\phi)}\,d\phi = \frac{\pi}{N_C+1}\sum_{w-w'=b-b'}\ket{w^{(s)}}\bra{(w')^{(s)}}\otimes\ket{b}\bra{b'}.
\end{equation}
But notice that this is just a rescaling of our desired Choi matrix! This matrix has $\frac{\pi}{N_C+1}$ for all entries such that $w - w' = b - b'$, while $J(\mE)$ has $\frac{1}{2}$ for all such entries. Therefore,
\begin{equation}
    J(\mE) = \frac{N_C+1}{2\pi}\int_{0}^{2\pi}\ket{\Psi(\phi)}\bra{\Psi(\phi)}\otimes\ket{\eta(-\phi)}\bra{\eta(-\phi)}\,d\phi,
\end{equation}
which is a manifestly separable expression for the Choi matrix, since it is $(N_C+1)$ times a convex combination of the pure product states $\ket{\Psi(\phi)}\bra{\Psi(\phi)}\otimes\ket{\eta(-\phi)}\bra{\eta(-\phi)}$. We conclude that the distillation channel we derived in the previous subsection is indeed entanglement-breaking.

\vspace{0.5\baselineskip}

This manifestly separable expression for the Choi matrix also immediately implies a measure-and-prepare implementation scheme for the optimal entanglement-breaking channel. More generally, a separable Choi matrix (representing an entanglement-breaking channel) can be written in the form
\begin{equation}
    \mJ(\mE) = \sum_{j}c_j\ket{\Psi_j}\bra{\Psi_j}\otimes\ket{\eta_j}\bra{\eta_j},
\end{equation}
where $\ket{\Psi_j}$ are pure states on the input Hilbert space, $\ket{\eta_j}$ are pure states on the output Hilbert space, and $c_j$ are nonnegative real numbers satisfying $\sum_{j}c_j = d_{\text{in}}$, where $d_{\text{in}}$ is the dimension of the Hilbert space. Furthermore, suppose that
\begin{equation}
    \sum_{j}\frac{c_j}{d_{\text{in}}}\ket{\Psi_j}\bra{\Psi_j} = \Ibb_{d_{\text{in}}}.
\end{equation}
Then such a quantum channel has a Kraus representation of the form
\begin{equation}
    K_j = \sqrt{\frac{c_j}{d_{\text{in}}}}\ket{\eta_j^*}\bra{\Psi_j},
\end{equation}
where $\ket{\eta_j^*}$ denotes the complex conjugate of $\ket{\eta_j}$. Furthermore, this channel can be implemented by applying the POVM
\begin{equation}
    M_j = \frac{c_j}{d_{\text{in}}}\ket{\Psi_j}\bra{\Psi_j}
\end{equation}
and then preparing the state $\ket{\eta_j^*}\bra{\eta_j^*}$ based on the POVM outcome $j$. The optimal entanglement-breaking single-shot qubit coherence distillation protocol shown above has exactly this form (except we have an integral instead of a sum). Therefore, we can implement this protocol by applying the continuous POVM
\begin{equation}
    dM = \frac{N_C+1}{2\pi}\ket{\Psi(\phi)}\bra{\Psi(\phi)}\,d\phi
\end{equation}
and then preparing the state $\ket{\eta(\phi)}\bra{\eta(\phi)}$ based on the POVM outcome $\phi$.

%% file: app_perturbative_protocols.tex
\appsec{Optimal Equatorial Protocols in the Low-Noise Limit}
{appendix:perturbative-protocols}

One interesting aside we can consider is what the optimal protocol looks like in the regime of $\lambda$ values extremely close to $1$ (low-noise limit) or $0$ (high-noise limit). In particular, instead of treating $N$ perturbatively (consider the limit $N\rightarrow\infty$, and expand in powers of $N^{-1}$), we will treat $\lambda$ perturbatively. As far as we know, there is no obvious information-theoretic consequence for the behavior of the optimal protocol in these regimes. However, we believe it offers some nice insight into the single-shot qubit coherence distillation problem in its own right. For convenience, we will restrict to the equatorial special case $\Theta_{\text{in}} = \Theta_{\text{out}} = \frac{\pi}{2}$.

In this appendix, we will focus on the low-noise limit $\lambda\rightarrow 1$, since that limit is easier to analyze. Numerically solving the boundary value problem in Appendix \ref{appendix:boundary-value-problem} reveals that, for $\lambda\approx 1$, the $\sin^2\theta_w$ values follow a relatively smooth curve. This naturally motivates the idea of expanding the protocol perturbatively in powers of $c_0 = \frac{1-\lambda}{2}$. In particular, we will prove the following result:

\begin{theorem}[Optimal equatorial distillation protocol in the low-noise limit]
\label{thm:equatorial-lam1-perturbative}
In the low-noise limit, the optimal $\lambda\rightarrow 1$ equatorial distillation protocol can be expanded as a power series in $c_0 = \frac{1-\lambda}{2}$, where the zeroth, first, and second orders look as follows:
\begin{align}
    S_w = \sin^2\theta_w &= f_0(w) + f_1(w)c_0 + f_2(w)c_0^2 + O\left(c_0^3\right) \\
    f_0(w) &= \frac{w}{N_C} \\
    f_1(w) &= \frac{4}{N_C^2(N_C-2)}w(N_C-w)(N_C-2w) \\
    f_2(w) &= Aw(N_C-w)(N_C-2w)(w^2 - N_Cw + B) \\
    A &= \frac{16(3N_C-4)}{N_C^3(N_C-1)(N_C-2)(N_C-4)} \\
    B &= \frac{3N_C^3 - 7N_C^2 + 16}{4(3N_C-4)}.
\end{align}
\end{theorem}

\appsubsec{First-Order Perturbative Protocol for $\lambda\rightarrow 1$}
{subsec:lam1-order1}

To study the first-order perturbative protocol, we should write everything in terms of powers of $c_0$ and ignore all quantities that are quadratic or higher in $c_0$. In particular, we want to write the protocol as
\begin{equation}
    \sin^2\theta_w = f_0(w) + f_1(w)c_0 + O\left(c_0^2\right).
\end{equation}
From Appendix \ref{appendix:boundary-value-problem}, we already know that $f_0(w) = \frac{w}{N_C}$, so we now need to find $f_1(w)$.

\vspace{0.5\baselineskip}

To accomplish this, we must first write the $P_{w,w'}$ values to the first order in $c_0$. We begin by approximating the prefactor corresponding to the probability of the $j=N_C/2$ outcome in the angular momentum measurement when performing Schur sampling:
\begin{equation}
    \frac{c_1 - c_0}{c_1^{N_C+1} - c_0^{N_C+1}} = \frac{1 - 2c_0}{(1 - c_0)^{N_C+1} - c_0^{N_C+1}} = 1 + (N_C-1)c_0 + O\left(c_0^2\right).
\end{equation}
Next, we need to write the $Q_{wk}$ values (which are defined in Appendix \ref{appendix:understanding-P-vals}\ref{subsec:P-vals-miscellaneous-facts}) perturbatively as well. Because the formula for $P_{w,w'}$ has a factor of $c_0^{N_C-k}$ in front of the term involving $Q_{wk}$ and $Q_{w'k}$, we only need to consider $k=N_C$ and $k=N_C-1$, because all other $Q_{wk}$ values will be quadratic or higher in $c_0$ and can thus be ignored. We thus obtain
\begin{align}
    Q_{w,N_C} &= \sqrt{\frac{1}{\binom{N_C}{w}2^{N_C}}}\left[\binom{N_C}{w}\right] = \sqrt{\frac{\binom{N_C}{w}}{2^{N_C}}} \\
    Q_{w,N_C-1} &= \sqrt{\frac{N_C}{\binom{N_C}{w}2^{N_C}}}\left[\binom{N_C}{w} - \binom{N_C-1}{w-1}\right] = \sqrt{\frac{\binom{N_C}{w}}{2^{N_C}}}\frac{N_C-2w}{\sqrt{N_C}}.
\end{align}
Putting this all together, we get the desired first-order perturbative formula for the $P_{w,w'}$ values:
\begin{align}
    P_{w,w'} = \,\, & \left[1 + (N_C-1)c_0 + O\left(c_0^2\right)\right] \\
    & \left[\left(1 - N_Cc_0 + O\left(c_0^2\right)\right)\frac{\sqrt{\binom{N_C}{w}\binom{N_C}{w'}}}{2^{N_C}} + \left(c_0 + O\left(c_0^2\right)\right)\frac{\sqrt{\binom{N_C}{w}\binom{N_C}{w'}}}{2^{N_C}}(N_C-2w)(N_C-2w')\right] \\
    = \,\, & \frac{\sqrt{\binom{N_C}{w}\binom{N_C}{w'}}}{2^{N_C}}\left[1 + c_0\left(\frac{(N_C-2w)(N_C-2w')}{N_C} - 1\right) + O\left(c_0^2\right)\right].
\end{align}
Finally, we take the ratios of suitably chosen $P_{w,w'}$ values to find the $R_w$ values to first order in $c_0$:
\begin{align}
    R_w = \frac{P_{w-1,w}}{P_{w,w+1}} &= \sqrt{\frac{\binom{N}{w-1}}{\binom{N}{w+1}}}\times\frac{1 + c_0\left(\frac{(N-2(w-1))(N-2w)}{N} - 1\right) + O\left(c_0^2\right)}{1 + c_0\left(\frac{(N-2w)(N-2(w+1))}{N} - 1\right) + O\left(c_0^2\right)} \\
    &= \sqrt{\frac{w(w+1)}{(N-w)(N-w+1)}}\left[1 + \frac{4(N-2w)}{N}c_0 + O\left(c_0^2\right)\right].
\end{align}
As we show in Appendix \ref{appendix:boundary-value-problem}\ref{subsec:brute-force-optimization}, the optimal protocol in the equatorial case is fully determined by the $R_w$ values, so these are really the only values we need.

\vspace{0.5\baselineskip}

We must now expand various trigonometric functions of the $\theta_w$ values to first order in $c_0$. We need to do this because we need to see how the first-order changes in the $R_w$ values, when plugged into the three-angle relation, shifts the $\sin^2\theta_w$ values. Since the three-angle relation involves $\tan\theta_w$, $\cos\theta_{w-1}$, and $\sin\theta_{w+1}$, we want to write these quantities in terms of $\sin^2\theta_w$. We do so as follows:
\begin{align}
    \sin^2\theta_w &= \frac{w}{N_C} + f_1(w)c_0 + O\left(c_0^2\right) \\
    \therefore \tan\theta_w &= \sqrt{\frac{w + N_Cf_1(w)c_0 + O\left(c_0^2\right)}{N_C - w - N_Cf_1(w)c_0 + O\left(c_0^2\right)}} \\
    &= \sqrt{\frac{w}{N_C-w}}\left[1 + \frac{N_Cf_1(w)}{2w}c_0 + O\left(c_0^2\right)\right]\left[1 + \frac{N_Cf_1(w)}{2(N_C-w)}c_0 + O\left(c_0^2\right)\right] \\
    &= \sqrt{\frac{w}{N_C-w}}\left[1 + \frac{N_C^2f_1(w)}{2w(N_C-w)}c_0 + O\left(c_0^2\right)\right] \\
    \cos\theta_{w-1} &= \sqrt{\frac{N_C-w+1}{N_C} - f_1(w-1)c_0 + O\left(c_0^2\right)} \\
    &= \sqrt{\frac{N_C-w+1}{N_C}}\left[1 - \frac{N_Cf_1(w-1)}{2(N_C-w+1)}c_0 + O\left(c_0^2\right)\right] \\
    \sin\theta_{w+1} &= \sqrt{\frac{w+1}{N_C} + f(w+1)c_0 + O\left(c_0^2\right)} \\
    &= \sqrt{\frac{w+1}{N_C}}\left[1 + \frac{N_Cf_1(w+1)}{2(w+1)}c_0 + O\left(c_0^2\right)\right].
\end{align}
Plugging these formulas into the three-angle relation yields
\begin{align}
    \tan\theta_w &= \frac{\cos\theta_{w-1}}{\sin\theta_{w+1}}R_w \\
    \iff 1 + \frac{N_C^2f_1(w)}{2w(N_C-w)}c_0 + O\left(c_0^2\right) &= \frac{1 - \frac{N_Cf_1(w-1)}{2(N_C-w+1)}c_0 + O\left(c_0^2\right)}{1 + \frac{N_Cf_1(w+1)}{2(w+1)}c_0 + O\left(c_0^2\right)}\left[1 + \frac{4(N_C-2w)}{N_C}c_0 + O\left(c_0^2\right)\right] \\
    \iff \frac{N_C^2f_1(w)}{2w(N_C-w)} &= -\frac{N_Cf_1(w-1)}{2(N_C-w+1)} - \frac{N_Cf_1(w+1)}{2(w+1)} + \frac{4(N_C-2w)}{N_C}.
\end{align}

\vspace{0.5\baselineskip}

In other words, we have turned our three-angle relation for $\sin^2\theta_w$ into a three-value relation for $f_1(w)$. Furthermore, the boundary conditions and the bit-flip symmetry condition for the equatorial special case also immediately translate as follows:
\begin{itemize}
    \item boundary conditions: $\theta_0 = 0, \, \theta_{N_C} = \frac{\pi}{2} \implies f_1(0) = f_1(N) = 0$
    \item bit-flip symmetry: $\theta_w + \theta_{N_C-w} = \frac{\pi}{2} \implies f_1(N_C-w) = -f_1(w)$.
\end{itemize}
Turning one boundary value problem into another may not feel like progress. But the boundary value problem for $f_1(w)$ is a lot easier than the one for $\sin^2\theta_w$, because its solution is much more guessable. One way to guess the solution, inspired by the numerics, is to guess that $f_1(w)$ is scale-invariant, in the sense that it only depends on $w$ as a fraction of $N_C$. Doing this dramatically simplifies the three-value relation above:
\begin{align}
    f_1(w) &\approx h\left(\frac{w}{N}\right) \\
    \implies \frac{h(x)}{2x(1-x)} &= -\frac{h(x)}{2(1-x)} - \frac{h(x)}{2x} + 4(1-2x) \\
    \implies h(x) &= 4x(1-x)(1-2x).
\end{align}
If you try to plug in this result for $f_1(w)$, it will not work. However, it is also natural to try this functional form, but with an adjustment to the leading factor that can depend on $N_C$. Such an ansatz looks as follows:
\begin{align}
    & f(w) = c(N_C)w(N_C-w)(N_C-2w) \\
    \iff & \frac{c(N_C)N_C^2(N_C-2w)}{2} = -\frac{c(N_C)N_C(w-1)(N_C-2w+2)}{2} \\
    & \quad\quad\quad\quad\quad\quad\quad\quad\quad\quad\quad - \frac{c(N_C)N_C(N_C-w-1)(N_C-2w-2)}{2} + \frac{4(N_C-2w)}{N_C} \\
    \iff & c(N_C)N_C\left[N_C(N_C-2w) + (w-1)(N_C-2w+2) + (N_C-w-1)(N_C-2w-2)\right] = \frac{8(N_C-2w)}{N_C} \\
    \iff & 2N_C(N_C-2)(N_C-2w)c(N_C) = \frac{8(N_C-2w)}{N_C} \\
    \iff & c(N_C) = \frac{4}{N_C^2(N_C-2)}.
\end{align}
Sure enough, adjusting the leading factor allowed us to find a solution to the boundary value problem! We conclude that the first-order perturbation to the optimal protocol in the low-noise limit looks as follows:
\begin{align}
    f_1(w) &= \frac{4w(N_C-w)(N_C-2w)}{N_C^2(N_C-2)} \\
    \sin^2\theta_w &= \frac{w}{N_C} + \frac{4w(N_C-w)(N_C-2w)}{N_C^2(N_C-2)}c_0 + O\left(c_0^2\right).
\end{align}

\vspace{0.5\baselineskip}

It is worthwhile to make a few nice observations regarding $f_1(w)$:
\begin{itemize}
    \item First, it is indeed a cubic function in $w$, which one may have suspected from the numerics.
    \item Second, this cubic has roots $w=0,\frac{N_C}{2},N_C$, which are natural consequences of the boundary conditions and bit-flip symmetry.
    \item Third, although it is not exactly scale-invariant (in the sense of depending only on $\frac{w}{N_C}$), it approaches scale invariance in the $N_C\rightarrow\infty$ limit, which one may also have suspected from the numerics.
    \item Fourth, it is well-defined for all $N_C$ except for $N_C=0$ and $N_C=2$. This is fine because the cases with $N_C\le 2$ do not allow for perturbations anyway, in the sense that the optimal protocol is the same regardless of $\lambda$. For all interesting cases ($N_C\ge 3$), this perturbation is perfectly well-behaved.
\end{itemize}

\appsubsec{Second-Order Perturbative Protocol for $\lambda\rightarrow 1$}
{subsec:lam1-order2}

It turns out that computing the second-order perturbative protocol in the $\lambda\rightarrow 1$ limit, while more tedious than computing the first-order perturbative protocol, is indeed doable. Furthermore, it has a fairly satisfying form. The zeroth-order protocol (expressed using $\sin^2\theta_w$) is linear in $w$, while the first-order perturbative protocol is cubic in $w$. This loosely suggests that the second-order protocol is quintic (i.e., a degree-$5$ polynomial) in $w$. This idea turns out to be precisely correct, as we will now show.

\vspace{0.5\baselineskip}

We will follow much the same workflow that we used to compute the first-order perturbative protocol in the $\lambda\rightarrow 1$ limit. We begin by approximating the prefactor corresponding to the probability of the $j=N_C/2$ outcome in the angular momentum measurement when performing Schur sampling: 
\begin{align}
    \frac{c_1 - c_0}{c_1^{N_C+1} - c_0^{N_C+1}} &= \frac{1 - 2c_0}{(1 - c_0)^{N_C+1} - c_0^{N_C+1}} \\
    &= \frac{1 - 2c_0}{1 - (N_C+1)c_0 + \binom{N_C+1}{2}c_0^2 + O\left(c_0^3\right)} \\
    &= \left(1 - 2c_0\right)\left(1 + (N_C+1)c_0 + \binom{N_C+2}{2}c_0^2 + O\left(c_0^3\right)\right) \\
    &= 1 + (N_C-1)c_0 + \frac{(N_C+1)(N_C-2)}{2}c_0^2 + O\left(c_0^3\right).
\end{align}
Next, we need to compute how $c_1^kc_0^{N_C-k}$ is written up to second order in $c_0$. Since we discard all terms that are cubic or higher in $c_0$, we only need to consider $k=N_C$, $k=N_C-1$, and $k=N_C-2$:
\begin{align}
    k = N_C \rightarrow c_1^kc_0^{N_C-k} &= \left(1 - c_0\right)^N_C = 1 - N_Cc_0 + \frac{N_C(N_C-1)}{2}c_0^2 + O\left(c_0^3\right) \\
    k = N_C-1 \rightarrow c_1^kc_0^{N_C-k} &= \left(1 - c_0\right)^{N_C-1}c_0 = c_0 - (N_C-1)c_0^2 + O\left(c_0^3\right) \\
    k = N_C-2 \rightarrow c_1^kc_0^{N_C-k} &= \left(1 - c_0\right)^{N_C-2}c_0^2 = c_0^2 + O\left(c_0^3\right).
\end{align}
In Appendix \ref{appendix:perturbative-protocols}\ref{subsec:lam1-order1}, we already computed $Q_{wk}$ for $k=N_C$ and $k=N_C-1$:
\begin{align}
    Q_{w,N_C} &= \sqrt{\frac{\binom{N_C}{w}}{2^{N_C}}} \\
    Q_{w,N_C-1} &= \sqrt{\frac{\binom{N_C}{w}}{N_C\cdot 2^{N_C}}}(N_C-2w).
\end{align}
Now that we are studying the second-order perturbation, we need to introduce $k=N_C-2$ as well:
\begin{align}
    Q_{w,N_C-2} &= \sqrt{\frac{\binom{N_C}{2}}{\binom{N_C}{w}2^{N_C}}}\left[\binom{2}{0}\binom{N_C-2}{w} - \binom{2}{1}\binom{N_C-2}{w-1} + \binom{2}{2}\binom{N_C-2}{w-2}\right] \\
    &= \sqrt{\frac{\binom{N_C}{2}}{\binom{N_C}{w}2^{N_C}}}\binom{N_C-2}{w}\left[1 - 2\cdot\frac{w}{N_C-w-1} + \frac{w(w-1)}{(N_C-w)(N_C-w-1)}\right] \\
    &= \sqrt{\frac{\binom{N_C}{2}}{\binom{N_C}{w}2^{N_C}}}\binom{N_C}{w}\frac{(N_C-w)(N_C-w-1)}{N_C(N_C-1)}\left[1 - \frac{2w}{N_C-w-1} + \frac{w(w-1)}{(N_C-w)(N_C-w-1)}\right] \\
    &= \sqrt{\frac{\binom{N_C}{w}}{2N_C(N_C-1)2^{N_C}}}\left[(N_C-w)(N_C-w-1) - 2w(N_C-w) + w(w-1)\right] \\
    &= \sqrt{\frac{\binom{N_C}{w}}{2N_C(N_C-1)2^{N_C}}}\left[(N_C-2w)^2 - N_C\right].
\end{align}

We can now put everything together and compute $P_{w,w'}$ to second order in $c_0$:
\begin{align}
    P_{w,w'} = \,\, & \frac{\sqrt{\binom{N_C}{w}\binom{N_C}{w'}}}{2^{N_C}}\left[1 + (N_C-1)c_0 + \frac{(N_C+1)(N_C-2)}{2}c_0^2 + O\left(c_0^3\right)\right] \\
    & \Bigg\{\left[1 - N_Cc_0 + \frac{N_C(N_C-1)}{2}c_0^2 + O\left(c_0^3\right)\right] \\
    & + \left[c_0 - (N_C-1)c_0^2 + O\left(c_0^3\right)\right]\frac{1}{N_C}(N_C-2w)(N_C-2w') \\
    & + \left[c_0^2 + O\left(c_0^3\right)\right]\frac{1}{2N_C(N_C-1)}[(N_C-2w)^2 - N_C][(N_C-2w')^2 - N_C]\Bigg\}.
\end{align}
The above formula simplifies as follows:
\begin{align}
    P_{w,w'} &= \frac{\sqrt{\binom{N_C}{w}\binom{N_C}{w'}}}{2^{N_C}}\Bigg\{1 + \left[\frac{(N_C-2w)(N_C-2w')}{N_C} - 1\right]c_0 + (\cdots)c_0^2 + O\left(c_0^3\right)\Bigg\} \\
    \left[c_0^2\text{ coeff}\right] &= \frac{[(N_C-2w)^2 - N_C][(N_C-2w')^2 - N_C]}{2N_C(N_C-1)} - 1,
\end{align}
where we write the $c_0^2$ coefficient on a separate line to make it easier to read.

\vspace{0.5\baselineskip}

Finally, we need to take the ratios of suitable $P_{w,w'}$ values to compute the $R_w$ values. For two general quantities written as second-order perturbation series in $c_0$, their quotient looks as follows:
\begin{align}
    \frac{1 + \alpha_1c_0 + \alpha_2c_0^2 + O\left(c_0^3\right)}{1 + \beta_1c_0 + \beta_2c_0^2 + O\left(c_0^3\right)} &= \left[1 + \alpha_1c_0 + \alpha_2c_0^2 + O\left(c_0^3\right)\right]\left[1 - \beta_1c_0 + \left(\beta_1^2 - \beta_2\right)c_0^2 + O\left(c_0^3\right)\right] \\
    &= 1 + \left(\alpha_1 - \beta_1\right)c_0 + \left[\left(\alpha_2 - \beta_2\right) - \beta_1\left(\alpha_1 - \beta_1\right)\right]c_0^2 + O\left(c_0^3\right).
\end{align}
After some tedious but ultimately straightforward calculation, the quotient formula and the formula for the $P_{w,w'}$ values yield
\begin{align}
    R_w = \frac{P_{w-1,w}}{P_{w,w+1}} &= \sqrt{\frac{w(w+1)}{(N_C-w)(N_C-w+1)}}\left[1 + \frac{4(N_C-2w)}{N_C}c_0 + (\cdots)c_0^2 + O\left(c_0^3\right)\right] \\
    \left[c_0^2\text{ coeff}\right] &= \frac{4(N_C-2w)}{N_C}\left[\frac{(N_C-2w)^2}{N_C(N_C-1)} + \frac{2(N_C-2w)}{N_C} - \frac{1}{N_C-1}\right],
\end{align}
where we once again separate out the $c_0^2$ coefficient to make it easier to read. Just as in the first-order perturbation analysis, the $R_w$ values are the only ones we need to focus on, since they fully determine the optimal protocol.

\vspace{0.5\baselineskip}

We now need to examine the perturbation in the optimal protocol itself. First, we write $\sin^2\theta_w$ as a second-order perturbation series in $c_0$, where we derived the zeroth-order contribution in Appendix \ref{appendix:boundary-value-problem}, and we derived the first-order contribution in Appendix \ref{appendix:perturbative-protocols}\ref{subsec:lam1-order1}:
\begin{align}
    \sin^2\theta_w &= f_0(w) + f_1(w)c_0 + f_2(w)c_0^2 + O\left(c_0^3\right) \\
    f_0(w) &= \frac{w}{N_C} \\
    f_1(w) &= \frac{4}{N_C^2(N_C-2)}w(N_C-w)(N_C-2w).
\end{align}
Our objective is to solve for $f_2(w)$. To do so, we must expand various trigonometric functions of the $\theta_w$ values to second order in $c_0$. We use the following Taylor series:
\begin{align}
    \left(1 + \epsilon\right)^{1/2} &= 1 + \frac{\epsilon}{2} - \frac{\epsilon^2}{8} + O\left(\epsilon^3\right) \\
    \left(1 - \epsilon\right)^{-1/2} &= 1 - \frac{\epsilon}{2} + \frac{3}{8}\epsilon^2 + O\left(\epsilon^3\right).
\end{align}
Here, we have to expand these Taylor series out to the $\epsilon^2$ term, whereas in Appendix \ref{appendix:perturbative-protocols}\ref{subsec:lam1-order1}, we only expanded them out to the $\epsilon$ term. These Taylor series allow us to eventually compute all of the trigonometric functions we need:
\begin{align}
    \sin^2\theta_w &= f_0(w)\left[1 + \frac{f_1(w)}{f_0(w)}c_0 + \frac{f_2(w)}{f_0(w)}c_0^2 + O\left(c_0^3\right)\right] \\
    \sin\theta_w &= \sqrt{f_0(w)}\left[1 + \frac{f_1(w)}{2f_0(w)}c_0 + \left(\frac{f_2(w)}{2f_0(w)} - \frac{f_1(w)^2}{8f_0(w)^2}\right)c_0^2 + O\left(c_0^3\right)\right] \\
    \csc\theta_w &= \frac{1}{\sqrt{f_0(w)}}\left[1 - \frac{f_1(w)}{2f_0(w)}c_0 + \left(\frac{3f_1(w)^2}{8f_0(w)^2} - \frac{f_2(w)}{2f_0(w)}\right)c_0^2 + O\left(c_0^3\right)\right]
\end{align}

\begin{align}
    \cos^2\theta_w &= \left(1 - f_0(w)\right)\left[1 - \frac{f_1(w)}{1 - f_0(w)}c_0 - \frac{f_2(w)}{1 - f_0(w)}c_0^2 + O\left(c_0^3\right)\right] \\
    \cos\theta_w &= \sqrt{1 - f_0(w)}\left[1 - \frac{f_1(w)}{2\left(1 - f_0(w)\right)}c_0 - \left(\frac{f_2(w)}{2\left(1 - f_0(w)\right)} + \frac{f_1(w)^2}{8\left(1 - f_0(w)\right)^2}\right)c_0^2 + O\left(c_0^3\right)\right]
\end{align}

\begin{align}
    \tan\theta_w &= \frac{\sin\theta_w}{\cos\theta_w} \\
    &= \sqrt{\frac{f_0(w)}{1 - f_0(w)}}\left[1 + \frac{f_1(w)}{2f_0(w)\left(1 - f_0(w)\right)}c_0 + (\text{stuff})c_0^2 + O\left(c_0^3\right)\right] \\
    (\text{stuff}) &= -\beta_1\left(\alpha_1 - \beta_1\right) + \left(\alpha_2 - \beta_2\right) \\
    &= \frac{f_1(w)}{2\left(1 - f_0(w)\right)}\cdot\frac{f_1(w)}{2f_0(w)\left(1 - f_0(w)\right)} + \frac{f_2(w)}{2f_0(w)\left(1 - f_0(w)\right)} + \frac{f_1(w)^2\left(2f_0(w) - 1\right)}{8f_0(w)^2\left(1 - f_0(w)\right)^2} \\
    &= \frac{f_2(w)}{2f_0(w)\left(1 - f_0(w)\right)} + \frac{f_1(w)^2\left(4f_0(w) - 1\right)}{8f_0(w)^2\left(1 - f_0(w)\right)^2},
\end{align}
where in the above calculation of $\tan\theta_w$, we once again used the formula for the quotient of two second-order perturbation series.

\vspace{0.5\baselineskip}

Now we need to plug all of these results into our three-angle recurrence relation
\begin{equation}
    \tan\theta_w = \frac{\cos\theta_{w-1}}{\sin\theta_{w+1}}R_w = \cos\theta_{w-1}\csc\theta_{w+1}R_w.
\end{equation}
To compute the right side, we need to multiply three different second-order perturbation series: one for $\cos\theta_{w-1}$, one for $\csc\theta_{w+1}$, and one for $R_w$. In the most general case, this is done as follows:
\begin{align}
    & \left[1 + \alpha_1c_0 + \alpha_2c_0^2 + O\left(c_0^3\right)\right]\left[1 + \beta_1c_0 + \beta_2c_0^2 + O\left(c_0^3\right)\right]\left[1 + \gamma_1c_0 + \gamma_2c_0^2 + O\left(c_0^3\right)\right] \\
    = \quad & 1 + \left(\alpha_1 + \beta_1 + \gamma_1\right)c_0 + \left(\alpha_2 + \beta_2 + \gamma_2 + \alpha_1\beta_1 + \alpha_1\gamma_1 + \beta_1\gamma_1\right)c_0^2 + O\left(c_0^3\right).
\end{align}
This will also be the time to actually plug in the known values for $f_0(w)$ and $f_1(w)$. When we do so, the zeroth-order and first-order terms on both sides of the three-angle recurrence relation will match, since we solved for the $f_0(w)$ and $f_1(w)$ that would make them match. When we compare the second-order terms on both sides, we obtain
\begin{align}
    & \frac{N_C^2}{2w(N_C-w)}f_2(w) + \frac{N_C}{2(N_C-w+1)}f_2(w-1) + \frac{N_C}{2(w+1)}f_2(w+1) \\
    = \,\, & \frac{2(N_C-2w)^2(N_C-4w)}{N_C(N_C-2)^2} + \frac{4(N_C-2w)}{N_C}\left[\frac{(N_C-2w)^2 - N_C}{N_C(N_C-1)} + \frac{4(N_C-2w)}{N_C(N_C-2)}\right] \\
    & + \frac{2(N_C-4)(N_C-2w)[3(N_C-w-1)(N_C-2w-2) - (w-1)(N_C-2w+2)]}{N_C^2(N_C-2)^2}.
\end{align}
We have thus derived a three-value relation for the second-order perturbation $f_2(w)$. Along with the boundary conditions $f_2(0) = f_2(N_C) = 0$, this defines a boundary value problem for $f_2(w)$. It is also worth noting that the boundary conditions and three-value relation themselves satisfy bit-flip symmetry.

\vspace{0.5\baselineskip}

Just as we did for the first-order perturbation, we now need to somehow guess the solution to the above boundary value problem. Based on the the numerics, and similar to what we did for the first-order perturbation, it makes sense to guess that $f_2(w)$ is approximately scale-invariant. This assumption dramatically simplifies the three-value relation as follows:
\begin{align}
    f_2(w) &\approx g\left(\frac{w}{N_C}\right) \\
    \implies \frac{g(x)}{2x(1-x)} + \frac{g(x)}{2(1-x)} + \frac{g(x)}{2x} &= 2(1-2x)^2(1-4x) + 4(1-2x)[(1-2x)^2 + 0] \\
    & \quad\quad + 2(1-2x)[3(1-x)(1-2x) - x(1-2x)] \\
    \iff g(x)\left[\frac{1}{2x(1-x)} + \frac{1}{2(1-x)} + \frac{1}{2x}\right] &= 2(1-2x)^2(1-4x) + 4(1-2x)^3 + 2(1-2x)^2(3-4x) \\
    \iff \frac{g(x)}{x(1-x)} &= 2(1-2x)^2\left((1-4x) + (2-4x) + (3-4x)\right) \\
    \iff g(x) &= 12x(1-x)(1-2x)^3.
\end{align}

Based on this solution for $g(x)$, it is tempting to use a guess of the form $f_2(w) = Aw(N_C-w)(N_C-2w)^3$ for some constant $A$. However, some quick checks will reveal that no value of $A$ will allow the above recurrence relation connecting $f_2(w-1)$, $f_2(w)$, and $f_2(w+1)$ to be satisfied exactly.

\vspace{0.5\baselineskip}

As a result, we need to consider a slightly more general ansatz. We take the smallest amount of generalization we can help, which is the following:
\begin{equation}
    f_2(w) = Aw(N_C-w)(N_C-2w)(w^2 - N_Cw + B)
\end{equation}
This ansatz is very sensible, since the ansatz still needs to have roots at $w = 0, \frac{N_C}{2}, N_C$, and since the ansatz still needs to be odd with respect to $w = \frac{N_C}{2}$. Furthermore, the asymptotic approximation given by $g(x) = 12x(1-x)(1-2x)^3$ suggests the following limiting behaviors for the unknown parameters $A$ and $B$:
\begin{equation}
    A \sim \frac{48}{N_C^5}, \quad B \sim \frac{N_C^2}{4}.
\end{equation}

Expanding the three-term relation for $f_2(w)$ using the above quintic ansatz for $f_2(w)$ yields a polynomial in $w$ on both the left-hand side (LHS) and the right-hand side (RHS). We begin by matching the leading coefficients on both sides, which at first appear to be the $w^4$ coefficients:
\begin{align}
    (\text{LHS }w^3\text{ coeff}) &= 0 + \frac{AN_C}{2}(-2) + \frac{AN_C}{2}(+2) = 0 \\
    (\text{RHS }w^3\text{ coeff}) &= 0.
\end{align}
But this turns out to be useless. In particular, both sides are in fact just cubic polynomials in $w$. We must thus move on to match the $w^3$ coefficients. We derive
\begin{align}
    (\text{LHS }w^3\text{ coeff}) = \,\, & \frac{AN_C^2}{2}(-2) \\
    & + \frac{AN_C}{2}\left[(-1)(-2)(1) + (1)(N_C+2)(1) + (1)(-2)(-N_C-2)\right] \\
    & + \frac{AN_C}{2}\left[(N_C-1)(-2)(1) + (-1)(N_C-2)(1) + (-1)(-2)(-N_C+2)\right] \\
    = \,\, & -AN_C^2 + \frac{AN_C}{2}(3N_C+8) + \frac{AN_C}{2}(-5N_C+8) \\
    = \,\, & -2AN_C(N_C-4)
\end{align}

\begin{align}
    (\text{RHS }w^3\text{ coeff}) &= \frac{2(-2)^2(-4)}{N_C(N_C-2)^2} + \frac{4(-2)^3}{N_C^2(N_C-1)} + \frac{2(N_C-4)(-2)[3(-1)(-2) - (1)(-2)]}{N_C^2(N_C-2)^2} \\
    &= -\frac{32}{N_C(N_C-2)^2} - \frac{32}{N_C^2(N_C-1)} - \frac{32(N_C-4)}{N_C^2(N_C-2)^2} \\
    &= -\frac{32}{N_C^2(N_C-1)(N_C-2)^2}\left[N_C(N_C-1) + (N_C-2)^2 + (N_C-1)(N_C-4)\right] \\
    &= -\frac{32}{N_C^2(N_C-1)(N_C-2)^2}\left[(3N_C-4)(N_C-2)\right] \\
    &= -\frac{32(3N_C-4)}{N_C^2(N_C-1)(N_C-2)}.
\end{align}
Matching these two expressions lets us solve for $A$:
\begin{align}
    (\text{LHS }w^3\text{ coeff}) &= (\text{RHS }w^3\text{ coeff}) \\
    \implies -2AN_C(N_C-4) &= -\frac{32(3N_C-4)}{N_C^2(N_C-1)(N_C-2)} \\
    \implies A &= \frac{16(3N_C-4)}{N_C^3(N_C-1)(N_C-2)(N_C-4)}.
\end{align}

We now proceed to solve for $B$. Matching the constant terms (i.e., the $w^0$ coefficients) on both sides appears to be the easiest way to do this. We first compute them as functions of $B$, as follows:
\begin{align}
    (\text{LHS }w^0\text{ coeff}) &= \frac{AN_C^2}{2}(N_C)(B) + \frac{AN_C}{2}(-1)(N_C+2)(N_C+1+B) \\
    & \quad\quad + \frac{AN_C}{2}(N_C-1)(N_C-2)(-N_C+1+B) \\
    &= \frac{AN_C}{2}B\left[N_C^2 - (N_C+2) + (N_C-1)(N_C-2)\right] \\
    & \quad\quad - \frac{AN_C}{2}(N_C+1)(N_C+2) - \frac{AN_C}{2}(N_C-1)^2(N_C-2) \\
    &= \frac{AN_C^2}{2}\left[2(N_C-2)B - (N_C^2 - 3N_C + 8)\right]
\end{align}

\begin{align}
    (\text{RHS }w^0\text{ coeff}) &= \frac{2(N_C)^2(N_C)}{N_C(N_C-2)^2} + \frac{4N_C}{N_C}\left[\frac{N_C^2 - N_C}{N_C(N_C-1)} + \frac{4N_C}{N_C(N_C-2)}\right] \\
    & \quad\quad + \frac{2(N_C-4)(N_C)[3(N_C-1)(N_C-2) - (-1)(N_C+2)]}{N_C^2(N_C-2)^2} \\
    &= \frac{2N_C^2}{(N_C-2)^2} + 4\left(1 + \frac{4}{N_C-2}\right) + \frac{2(N_C-4)(3N_C^2 - 8N_C + 8)}{N_C(N_C-2)^2} \\
    &= \frac{2}{N_C(N_C-2)^2}\left[N_C^3 + 2N_C(N_C-2)^2 + 8N_C(N_C-2) + (N_C-4)(3N_C^2 - 8N_C + 8)\right] \\
    &= \frac{2}{N_C(N_C-2)^2}\left[2(N_C-2)(3N_C^2 - 4N_C + 8)\right] \\
    &= \frac{4(3N_C^2 - 4N_C + 8)}{N_C(N_C-2)}.
\end{align}

Matching these two quantities lets us solve for $B$:
\begin{align}
    (\text{LHS }w^0\text{ coeff}) &= (\text{RHS }w^0\text{ coeff}) \\
    \implies \frac{AN_C^2}{2}\left[2(N_C-2)B - (N_C^2 - 3N_C + 8)\right] &= \frac{4(3N_C^2 - 4N_C + 8)}{N_C(N_C-2)} \\
    \implies B &= \frac{3N_C^3 - 7N_C^2 + 16}{4(3N_C-4)}.
\end{align}
One good sanity check at this point is that both $A$ and $B$ indeed satisfy the asymptotic approximations that we expected them to, namely
\begin{equation}
    A \sim \frac{48}{N_C^5}, \quad B \sim \frac{N_C^2}{4}.
\end{equation}
The last step is to verify that the $w$ and $w^2$ coefficients of the LHS and the RHS actually match with these given values of $A$ and $B$. This is a tedious but ultimately straightforward process.

\vspace{0.5\baselineskip}

We have thus computed the the $c_0^2$ contribution to the optimal protocol. Putting everything together, here is the expansion of $\sin^2\theta_w$ as a power series in $c_0$, out to the $c_0^2$ term:
\begin{align}
    S_w = \sin^2\theta_w &= f_0(w) + f_1(w)c_0 + f_2(w)c_0^2 + O\left(c_0^3\right) \\
    f_0(w) &= \frac{w}{N_C} \\
    f_1(w) &= \frac{4}{N_C^2(N_C-2)}w(N_C-w)(N_C-2w) \\
    f_2(w) &= Aw(N_C-w)(N_C-2w)(w^2 - N_Cw + B) \\
    A &= \frac{16(3N_C-4)}{N_C^3(N_C-1)(N_C-2)(N_C-4)} \\
    B &= \frac{3N_C^3 - 7N_C^2 + 16}{4(3N_C-4)}.
\end{align}
This completes the proof of Theorem \ref{thm:equatorial-lam1-perturbative}.

\vspace{0.5\baselineskip}

Just as we did for $f_1(w)$, it is worthwhile to make a few nice observations regarding $f_2(w)$:
\begin{itemize}
    \item First, $f_2(w)$ is indeed a quintic function in $w$, which one may have suspected from the numerics.
    \item Second, $f_2(w)$ has $w=0,\frac{N_C}{2},N_C$ as three of its five roots. These are natural consequences of the boundary conditions and bit-flip symmetry.
    \item Third, although $f_2(w)$ is not exactly scale-invariant (in the sense of depending only on $\frac{w}{N_C}$), it approaches scale invariance in the $N_C\rightarrow\infty$ limit, which one may also have suspected from the numerics. Interestingly, the two remaining roots of the quintic tend toward $w=\frac{N_C}{2}$ in the limit, which makes the perturbation very small in the region $w\approx\frac{N_C}{2}$.
    \item Fourth, the prefactor $A$ is well-defined for all $N_C$ except for $N_C=0,1,2,4$. For the same reason as before, the $N_C=0,1,2$ cases are not an issue. But one might be concerned about the $N_C=4$ case. Fortunately, in this case, the only valid values of $w$, which are integers $0\le w\le 4$, all yield zero when plugged into the rest of the expression, namely $w(N_C-w)(N_C-2w)(w^2 - N_Cw + B)$. As a result, one gets an indeterminate $0/0$ form, rather than a form that would necessarily diverge or be otherwise problematic. In other words, our formula tells us nothing about the second-order perturbation in the $N_C=4$ case. Fortunately, $N_C=4$ is small enough to be solved exactly, and the second-order perturbation can always be computed manually from that solution.
\end{itemize}